\documentclass[sigconf]{acmart}
\AtBeginDocument{%
  \providecommand\BibTeX{{%
    Bib\TeX}}}

\usepackage{algorithmic}
\usepackage{booktabs}
\usepackage{paralist}
\usepackage{amsthm}
\usepackage{mathdots}
\usepackage[toc,page,header]{appendix}
\usepackage{amsmath}
\usepackage{amsfonts}
\usepackage{mathtools}
\usepackage{balance} % for balancing columns on the final page
\usepackage{graphicx}
\usepackage{amsmath,amsfonts}
\usepackage{algorithmic}
\usepackage[ruled, vlined, commentsnumbered, linesnumbered]{algorithm2e}
\usepackage{textcomp}
\usepackage{mathabx}
\usepackage{xcolor} % No need for [dvipsnames]
\definecolor{BrickRed}{rgb}{0.8, 0.25, 0.33}  % Define BrickRed manually
\definecolor{RoyalBlue}{rgb}{0.25, 0.41, 0.88}  % Manually define the color
\usepackage{enumitem}

\usepackage[normalsize]{subfigure}
\usepackage{color}

\usepackage{amsthm}
\usepackage{mdframed}
\theoremstyle{plain}
\newmdenv[
  topline=true,
  bottomline=true,
  rightline=true,
  leftline=true,
  linecolor=black,
  linewidth=0.8pt,
  backgroundcolor=white,
  skipabove=10pt,
  skipbelow=10pt
]{boxedtheorem}

\usepackage{amsthm}

\newtheorem{reptheorem}{Theorem}
\newtheorem{relemma}{Lemma}
\newtheorem{property}{\bf Property}
\newtheorem{reproperty}{\bf Property}

\newtheorem{theorem}{\bf{Theorem}}
\newtheorem{definition}{\bf{Definition}}

\newtheorem{lemma}{\bf{Lemma}}

\newtheorem{reproposition}{\bf{Proposition}}
\newtheorem{proposition}{\bf{Proposition}}
\usepackage{amsthm}   % For the theorem environment
\usepackage{tcolorbox} % For creating colored boxes
\usepackage{xr} % or \usepackage{xr-hyper}
\usepackage{microtype}      % microtypography
\usepackage{lipsum}

\newcommand{\DEL}[1]{\iffalse #1 \fi}

\newcommand{\rd}{\color{black}}
\newcommand{\rev}{\color{black}}
\newcommand{\bl}{\color{black}}

\newcommand{\squishlist}{
\begin{list}{$\bullet$}
  { \setlength{\itemsep}{0pt}
     \setlength{\parsep}{0pt}
     \setlength{\topsep}{0pt}
     \setlength{\partopsep}{0pt}
     \setlength{\leftmargin}{0em}
     \setlength{\labelwidth}{0em}
     \setlength{\labelsep}{0.2em} } }
\copyrightyear{2025}
\acmYear{2025}
\setcopyright{cc}
\setcctype{by}
\acmConference[CCS '25]{Proceedings of the 2025 ACM SIGSAC Conference on Computer and Communications Security}{October 13--17, 2025}{Taipei, Taiwan}
\acmBooktitle{Proceedings of the 2025 ACM SIGSAC Conference on Computer and Communications Security (CCS '25), October 13--17, 2025, Taipei, Taiwan}\acmDOI{10.1145/3719027.3765042}
\acmISBN{979-8-4007-1525-9/2025/10}

\title{\textsc{PAnDA}: Rethinking Metric Differential Privacy Optimization at Scale with Anchor-Based Approximation}

\author{Ruiyao Liu}
\affiliation{
  \institution{University of North Texas}
  \city{Denton, Texas}
  \country{USA}}
\email{RuiyaoLiu@my.unt.edu}

\author{Chenxi Qiu}\authornote{Chenxi Qiu is the corresponding author.}
\affiliation{
  \institution{University of North Texas}
  \city{Denton, Texas}
  \country{USA}}
\email{chenxi.qiu@unt.edu}

\begin{abstract}
\emph{Metric Differential Privacy (mDP)} extends the local differential privacy (LDP) framework to metric spaces, enabling more nuanced privacy protection for data such as geo-locations. However, existing mDP optimization methods, particularly those based on linear programming (LP), face scalability challenges due to the quadratic growth in decision variables. 

In this paper, we propose \textit{\underline{P}erturbation via \underline{An}chor-based \underline{D}istributed \underline{A}pproximation (\textsc{PAnDA})}, a scalable two-phase framework for optimizing metric differential privacy (mDP). To reduce computational overhead, \textsc{PAnDA} allows each user to select a small set of anchor records, enabling the server to solve a compact linear program over a reduced domain. We introduce three anchor selection strategies, \emph{exponential decay (\textsc{PAnDA}-e)}, \emph{power-law decay (\textsc{PAnDA}-p)}, and \emph{logistic decay (\textsc{PAnDA}-l)}, and establish theoretical guarantees under a relaxed privacy notion called \emph{probabilistic mDP (PmDP)}. Experiments on real-world geo-location datasets demonstrate that \textsc{PAnDA} scales to secret domains with up to 5,000 records, two times larger than prior LP-based methods, while providing theoretical guarantees for both privacy and utility. \looseness = -1

\end{abstract}


\newcommand{\BibTeX}{\rm B\kern-.05em{\sc i\kern-.025em b}\kern-.08em\TeX}

\begin{document}

%%% The following commands remove the headers in your paper. For final 
%%% papers, these will be inserted during the pagination process.

\pagestyle{fancy}
\fancyhead{}

%%% The next command prints the information defined in the preamble.

\maketitle 

%%%%%%%%%%%%%%%%%%%%%%%%%%%%%%%%%%%%%%%%%%%%%%%%%%%%%%%%%%%%%%%%%%%%%%%%

\section{Introduction}
% \vspace{-0.10in}
\label{sec:intro}

\emph{Local Differential Privacy (LDP)}~\cite{Duchi-FOCS2013} has emerged as a preferred paradigm for privacy-preserving data collection, especially in scenarios where users do not fully trust centralized aggregators. By ensuring that each user's reported data is statistically indistinguishable across all possible inputs, LDP offers strong and provable privacy guarantees. However, LDP requires a uniform level of indistinguishability across all input pairs, which limits its applicability in contexts that demand more nuanced privacy control. For instance, in location-based services (LBSs), the objective is often to obscure a user’s exact location within a certain geographic range \cite{ImolaUAI2022}. Standard LDP fails to capture such contextual nuances, for example, it does not differentiate between a user being within 1 kilometer or 100 kilometers of a reference point, even though these cases may warrant vastly different levels of perturbation. Moreover, the indiscriminate protection of all input pairs as equally sensitive often results in excessive noise, severely degrading the utility of the perturbed data for downstream tasks~\cite{Qardaji-CCS2013}.

% This limitation becomes especially evident in domains with structured or continuous input spaces, such as geographic locations or work embeddings  \cite{ImolaUAI2022}. For example, in a location-based service (LBS), reporting a random point satisfying LDP might significantly degrade the data utility for downstream applications like navigation or local recommendations \cite{Qiu-TMC2022}. % While this worst-case approach ensures robustness, it often leads to excessive noise in practice, significantly degrading data utility \cite{Qardaji-CCS2013}. 

To address these limitations, \textbf{metric Differential Privacy (mDP)} \cite{Chatzikokolakis-PETS2013} was introduced as a generalization of LDP that enables more nuanced levels of indistinguishability between inputs. Instead of applying a uniform privacy guarantee, mDP incorporates a distance metric to modulate the strength of protection: inputs that are close under the metric must remain highly indistinguishable, while those that are farther apart may be more readily distinguished. This distance-aware relaxation enables privacy mechanisms to inject less noise, thereby improving utility while still offering meaningful privacy guarantees. This enhancement broadens the flexibility and applicability of LDP across various data domains, including geo-location perturbation in LBSs \cite{Andres-CCS2013, Yu-NDSS2017, Shokri-PoPETs2015, Chatzikokolakis-PoPETs2015} and text perturbation in natural language processing (NLP) \cite{ImolaUAI2022,Carvalho2021TEMHU}.

\vspace{0.05in}
\noindent \textbf{Related work.} Compared to traditional LDP, optimizing for mDP introduces additional complexity due to its non-uniform privacy constraints and the diverse utility loss under perturbation. Specifically, mDP requires varying privacy guarantees between any two records, and utility loss can depend heavily not only on the magnitude but also the direction of the perturbation \cite{Qiu-TMC2022}. While predefined noise distribution mechanisms such as the \emph{Laplace mechanism} \cite{Andres-CCS2013} and the \emph{exponential mechanism (EM)} \cite{Carvalho2021TEMHU} can satisfy mDP, they typically generate noise based on the perturbation magnitude. This approach often overlooks directional variations in utility loss across the output space, leading to suboptimal utility performance.

To better account for utility sensitivity, recent research on mDP has increasingly focused on \emph{optimization-based mechanisms}, particularly those using \emph{linear programming (LP)}. These approaches discretize both the input (secret) domain $\mathcal{X}$ and the output (perturbed) domain $\mathcal{Y}$ into finite sets, allowing for explicit modeling of utility loss for each possible perturbation~\cite{ImolaUAI2022}. The perturbation mechanism is then optimized by solving an LP that minimizes the expected utility loss while satisfying mDP constraints for all neighboring record pairs~\cite{Bordenabe-CCS2014}. This formulation requires optimizing a probability distribution over perturbed outputs for each real input, resulting in $|\mathcal{X}||\mathcal{Y}|$ decision variables. However, this scalability poses a major computational bottleneck, for instance, optimizing the perturbation distributions for thousands of discrete locations in a small geographic region can lead to millions of LP variables, resulting in prohibitively high computational overhead~\cite{Qiu-TMC2022}. \looseness = -1

\begin{figure}[t]
\centering
\hspace{0.00in}
\begin{minipage}{0.45\textwidth}
  \subfigure{
\includegraphics[width=1.00\textwidth]{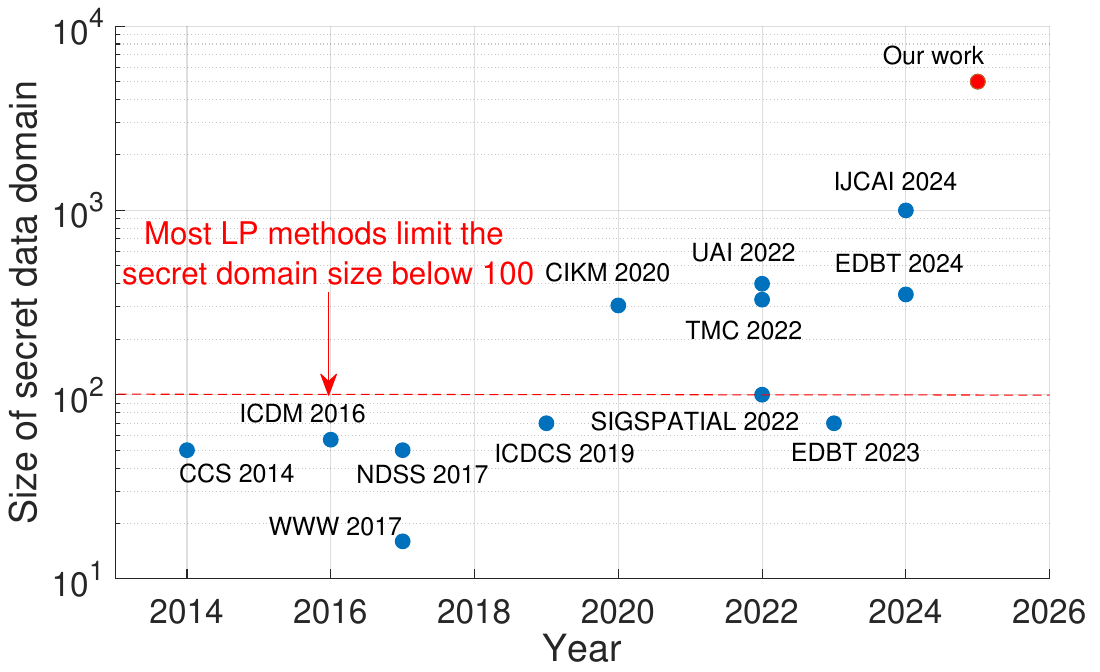}}
\vspace{-0.25in}
\end{minipage}
\caption{Scalability comparison of related works and ours. 
\vspace{0.00in}
\newline \small CCS 2014 \cite{Bordenabe-CCS2014}, ICDM 2016 \cite{Wang-CIDM2016}, WWW 2017 \cite{Wang-WWW2017}, NDSS 2017 \cite{Yu-NDSS2017}, ICDCS 2019 \cite{Qiu-ICDCS2019}, CIKM 2020 \cite{Qiu-CIKM2020}, TMC 2022 \cite{Qiu-TMC2022}, SIGSPATIAL 2022 \cite{Qiu-SIGSPATIAL2022}, UAI 2022 \cite{ImolaUAI2022}, EDBT 2023 \cite{Pappachan-EDBT2023}, EDBT 2024 \cite{Qiu-EDBT2024}, IJCAI 2024 \cite{Qiu-IJCAI2024}.}
\label{fig:relatedwork}
\vspace{-0.20in}
\end{figure}

% However, optimizing data utility under mDP introduces new technical challenges. In particular, minimizing expected utility loss while satisfying mDP constraints across all record pairs naturally leads to a linear programming (LP) formulation. Unfortunately, this LP problem becomes intractable at scale: the number of variables grows quadratically with the size of the record domain, and the number of constraints scales with the number of all neighboring record pairs. As a result, existing LP-based mDP methods are limited to datasets with a few hundred or, at most, a few thousand records, far from sufficient for real-world deployments.

Due to such high computational overhead, as shown in Fig.\ref{fig:relatedwork}, most existing LP-based approaches \cite{Shokri-CCS2012, Bordenabe-CCS2014, Wang-CIDM2016, Wang-WWW2017, Yu-NDSS2017, Qiu-ICDCS2019, Qiu-SIGSPATIAL2022} restrict the size of the secret data domain $\mathcal{X}$ to no more than 100 records. Although more recent methods have extended the scalability of LP-based optimization by using hierarchical index structures~\cite{AhujaEDBT2019}, optimization decomposition techniques~\cite{Qiu-CIKM2020, Qiu-TMC2022, Qiu-IJCAI2024}, or combining EM with LP~\cite{ImolaUAI2022}, the size of $\mathcal{X}$ remains limited to approximately 1,000 records. This constraint makes LP-based approaches impractical in moderate-dimensional settings and inadequate for real-world datasets that evolve or expand over time~\cite{Feyisetan-WDSM2020}.

% mDP originated in the domain of geo-location privacy protection \cite{Andres-CCS2013}, requiring ``\emph{geo-indistinguishability}'' for each pair of locations with the Euclidean distance lower than a predetermined threshold. In other words, mDP defines ``neighboring locations'' based on Euclidean distance, diverging from the original DP which relies on Hamming distance. Over time, mDP has been explored across a spectrum of metric choices, including % Euclidean distance \cite{Andres-CCS2013}, Manhattan distance \cite{Chatzikokolakis-PETS2013}, Hyperbolic distance \cite{Feyisetan-ICDM2019}, Haversine distance \cite{Pappachan-EDBT2023}, Word Mover’s distance \cite{Fernandes-PST2019},  and others \cite{Feyisetan-2021-private}.

\vspace{-0.05in}
\subsection*{Our Work}

\textbf{Contribution 1: Anchor-Based Approximation in Place of Full-Domain LP.}  
To address the identified research gap, this paper proposes a new framework named  
\textit{\underline{P}erturbation via \underline{An}chor-based \underline{D}istributed \underline{A}pproximation (\textsc{PAnDA})}.  
Instead of solving an LP over the full secret data domain $\mathcal{X}$, \textsc{PAnDA} reduces the problem size by allowing users to select a small set of representative \emph{anchor records}. This significantly reduces the LP complexity while preserving utility. The \textsc{PAnDA} framework proceeds in two phases:  
\begin{itemize}[left=1.2em, labelsep=0.5em] 
    \item [\textbf{Ph-I:}] Each user probabilistically selects and uploads a set of anchor records to the server. The server then optimizes a lightweight LP problem that only covers the anchor records from all the users. 
    \item [\textbf{Ph-II:}] Each user downloads the optimized anchor perturbation matrix, from which the user obtains a \emph{surrogate perturbation vector} {\rev to perform} data perturbation.
\end{itemize}
% Because the total number of anchor records is substantially smaller than the size of the original secret dataset $\mathcal{X}$, \textsc{PAnDA} significantly lowers the computation overhead of perturbation optimization. However, this design introduces two key challenges:  
Moreover, \textsc{PAnDA} adopts a relaxed privacy definition called $(\epsilon, \delta)$-\emph{Probabilistic mDP (PmDP)}, which allows a small failure probability $\delta$ in satisfying the mDP constraints. This relaxation reduces the need to account for rare worst-case record pairs, thereby mitigating the high computation overhead of enforcing strict mDP constraints for all pairs of records. 

However, this two-phase perturbation optimization design introduces two key challenges:  

First, since data disclosure occurs in both phases---during anchor selection and record perturbation---each step potentially leaks information about a user's true data. Hence, a critical challenge is to ensure that the \emph{accumulated privacy cost from both phases remains within a specified budget $\epsilon$ under the mDP framework (\textbf{Challenge 1})}. \looseness = -1

Second, \textsc{PAnDA} approximates each user's perturbation vector using a surrogate vector constructed from a small set of anchor records. While this approximation significantly reduces computational complexity, it also introduces error that may lead to violations of mDP guarantees. Thus, the second challenge is to \emph{rigorously quantify and bound the potential mDP violations resulting from the use of surrogate perturbation vectors (\textbf{Challenge 2})}. \looseness = -1

\vspace{0.03in}
\noindent \textbf{Contribution 2: Two-Phase Perturbation Design with Theoretical Privacy and Utility Guarantees.} To address \textbf{Challenge 1}, that is, controlling the cumulative privacy cost across the two phases of \textsc{PAnDA}, we present \textbf{Theorem~\ref{thm:composition}}, which establishes that the \emph{sequential composition property} of mDP holds under the PmDP framework. This theoretical result enables us to split the total privacy budget $\epsilon$ between the two phases of the algorithm. Then, in the phase \textbf{Ph-I}, we introduce three probabilistic anchor selection methods, \emph{exponential decay}, \emph{power-law decay}, and \emph{logistic}, of which the bounded posterior leakage is given in \textbf{Proposition \ref{prop:posteriorbound}}. In the phase \textbf{Ph-II}, the anchor perturbation matrix is optimized to ensure that the additional privacy cost of this phase does not exceed the remaining budget. \looseness = -1

To address \textbf{Challenge 2}, i.e., bounding the mDP violation ratio introduced by surrogate perturbation vectors, \textsc{PAnDA} incorporates a privacy \emph{safety margin} $\xi$ into the mDP constraints applied to anchor pairs. This margin tightens the privacy budget during optimization, serving as a buffer to account for discrepancies between users’ true records and their surrogates. By reserving this margin in the optimization process, \textsc{PAnDA} ensures that the resulting mechanism satisfies $(\epsilon, \delta)$-PmDP, or satisfies $\epsilon$-mDP with high probability $1-\delta$, as formally established in \textbf{Proposition~\ref{prop:samplebudget}}.

Additionally, to evaluate the effectiveness of \textsc{PAnDA}’s optimization under this approximation, we derive in \textbf{Proposition~\ref{prop:ULbound}} a lower bound on the expected utility loss, to assess how closely \textsc{PAnDA} approaches the optimal utility loss.

% \textsc{PAnDA} theoretically derives and calibrates $\xi$ to ensure that the overall privacy leakage for each real record pair still satisfies the relaxed privacy guarantee of $(\epsilon, \delta)$-PmDP, even when surrogate distances are used. This is achieved by computing the probability that the anchor pairs sufficiently cover each record pair under the safety margin, and selecting the smallest $\xi$ such that the success probability exceeds $1 - \delta$.

\vspace{0.03in}
\noindent \textbf{Contribution 3: Empirical Validation of Efficiency and Privacy-Utility Tradeoff.} We conducted an extensive empirical study to evaluate the performance of \textsc{PAnDA}, comparing it against several state-of-the-art perturbation algorithms, including predefined noise distribution mechanism EM \cite{Chatzikokolakis-PoPETs2015}, LP-based methods \cite{Bordenabe-CCS2014, Qiu-TMC2022}, and hybrid methods combining LP and EM \cite{ImolaUAI2022, Chatzikokolakis-PETS2017}, across real-world map datasets from Rome, Italy, London, UK, and new york city US \cite{openstreetmap}.  \looseness = -1

The experimental results show that \textsc{PAnDA} can optimize a secret dataset with at least {\bl 5,000} records, achieving a capacity \emph{three times of the LP-based benchmarks \cite{Bordenabe-CCS2014, Qiu-IJCAI2024}, while maintaining comparable utility loss relative to the optimal solution (i.e., the approximation ratio is {\bl1.3649} on average, where a ratio
close to 1 indicates a near-optimal solution). 
% \textbf{Contribution 3: Empirical Validation of Efficiency and Privacy-Utility Tradeoff.}  We conduct extensive experiments on multiple real-world datasets to evaluate the performance of \textsc{PAnDA}. The results demonstrate that \textsc{PAnDA} increases size of secret data domain by \textbf{xx.xx\%} compared to the current state of the art LP methods. 
Moreover, \textsc{PAnDA} achieves higher utility compared to methods applicable to large-scale data domain, e.g., \textsc{PAnDA} attains {\bl76.40\%} lower utility loss compared to \cite{Chatzikokolakis-PoPETs2015, Chatzikokolakis-PETS2017}}. Additionally, even with a small anchor set (e.g., {\bl14.81\%} of $\mathcal{X}$ on average), \textsc{PAnDA} maintains mDP guarantees with high probability, i.e., the mDP violation rate of \textsc{PAnDA} is below {\bl $10^{-7}$} across different tested scenarios.

The remainder of the paper is organized as follows: Section \ref{sec:preliminary} introduces the preliminaries of mDP optimization. Section \ref{sec:method} describes the algorithm design and Section \ref{sec:performance} evaluates the algorithm's performance. Section \ref{sec:related} presents the related work. %Section \ref{sec:discussions} and 
Section \ref{sec:conclude} makes a conclusion.

\vspace{-0.00in}
\section{Preliminaries}
\label{sec:preliminary}
\vspace{-0.00in}

In this section, we introduce the preliminary knowledge of \emph{mDP}, the \emph{LP formulation}, and the \emph{remote computation framework}. 
The main notations used throughout this paper are listed in \textbf{Table \ref{tab:notation} in Appendix}. \looseness = -1

\vspace{0.03in}
\noindent \textbf{mDP}. In its original definition, LDP requires data perturbation to maintain a uniform level of \emph{indistinguishability} for all possible inputs \cite{Duchi-FOCS2013}. %The classification of "neighboring databases" depends on their \emph{Hamming distance}; that is, two databases are considered "neighbors" if they differ by at most one record. 
\emph{mDP} extends the LDP concept by considering distances between any two records within a general metric space. 

Definition \ref{def:mDP} formally defines $\epsilon$-mDP. 
\begin{definition}[$\epsilon$-mDP]
\label{def:mDP}
Let $\mathcal{X}$ denote the secret data domain, equipped with a distance function $d_{\mathcal{X}}: \mathcal{X} \times \mathcal{X} \rightarrow \mathbb{R}_{\geq 0}$, abbreviated as $d$. For any two records $x, x' \in \mathcal{X}$, let $d_{x,x'} = d(x, x')$ denote their distance.

A randomized mechanism $\mathcal{M}: \mathcal{X} \rightarrow \mathcal{Y}$ satisfies $\epsilon$-mDP if  % for all measurable subsets $\mathcal{Y}' \subseteq \mathcal{Y}$ and 
for all $x, x' \in \mathcal{X}$,
\vspace{-0.12in}
\begin{equation}
\label{eq:mDP}
\frac{\Pr[\mathcal{M}(x) \in \mathcal{Y}']}{\Pr[\mathcal{M}(x') \in \mathcal{Y}']} \leq e^{\epsilon d_{x,x'}}, \forall \mathcal{Y}' \subseteq \mathcal{Y}
\end{equation}
Here, $\epsilon > 0$ is the privacy budget that controls the level of indistinguishability based on the distance between inputs.
\end{definition}

For simplicity, in what follows we use $\mathcal{M}(x) \stackrel{\epsilon}{\approx} \mathcal{M}(x')$ to represent that Eq. (\ref{eq:mDP}) is satisfied. 

Intuitively, mDP ensures that small changes in the input $x$ of the perturbation method $\mathcal{M}$ result in bounded changes in the distribution of the output $\mathcal{M}(x)$, thus providing privacy guarantees in the corresponding distance metric space. A lower $\epsilon$ signifies a tighter bound, meaning less information can be revealed about $\mathcal{M}(x)$. 

\DEL{
\vspace{0.05in}
\noindent \textbf{Threat model (posterior leakage)}. Like the existing works \cite{Andres-CCS2013}, we assume that both perturbed data $\mathcal{M}(X)$ and perturbation function $\mathcal{M}$ are known by attackers. An attacker can use Bayes' formula \cite{Yu-NDSS2017} to derive the \emph{posterior} of the secret data $X$, i.e., $P\left(X = x | \mathcal{M}(X) = y\right)$, $\forall x \in \mathcal{X}$. In this case, the information disclosure caused by the perturbed data $\mathcal{M}(X) = y$ can be quantified by the posterior leakage  \cite{Kifer-PODS2012} (Definition \ref{def:PL}): \looseness = -1
\begin{definition}
\label{def:PL}
(Posterior leakage (PL)) Given the perturbation function $\mathcal{M}$, the PL of any two records $x_n, x_m\in \mathcal{X}$ is defined by 
\vspace{-0.00in}
\normalsize
% \small 
\begin{eqnarray}
\label{eq:PL}
\mathrm{PL}\left((x_n, x_m); \mathcal{M}\right) =  \underbrace{\frac{\Pr\left(X = x_n | \mathcal{M}(X) = y\right)}{\Pr\left(X = x_m| \mathcal{M}(X) = y\right)}}_{\mbox{posterior ratio}}\left\slash\underbrace{\frac{\Pr\left(X=x_n\right)}{\Pr\left(X=x_m\right)}}_{\mbox{prior ratio}}\right. 
\end{eqnarray}
\normalsize
\end{definition}
We require 
\begin{equation}
\label{eq:PLbound}
e^{-\epsilon d_{n, m}} \leq \mathrm{PL}\left((x_n, x_m); \mathcal{M}\right) \leq e^{\epsilon d_{n, m}} 
\end{equation}

Intuitively, the prior ratio and the posterior ratio in Eq. (\ref{eq:PL}) reflects the record $X$'s probabilities of being $x_n$ and $x_m$ \textbf{before and after the observation of the perturbed data $\mathcal{M}(X) =y$}. If %the prior ratio and the posterior ratio are closer, then % 
$\mathrm{PL}\left((x_n, x_m); \mathcal{M}\right)$ has a lower value, it implies that the attacker can obtain less additional information of $X$ by observing $\mathcal{M}(X) =y$, therefore achieving a higher privacy level. 

As a countermeasure of the Bayesian inference attacks, the perturbation function $\mathcal{M}$ is designed to enforce the posterior leakage between any $x_n$ and $x_m$ to be bounded: \looseness = -1
\begin{equation}
\label{eq:PLcriterion}
% \small 
\mathrm{PL}\left((x_n, x_m); \mathcal{M}\right) \leq \epsilon d_{x_n,x_m}, \forall x_n, x_m\in \mathcal{X}.   
\end{equation}
As proved by \cite{Kifer-TDS2014}, meeting mDP as defined in Eq. (\ref{eq:mDP}) is \textbf{equivalent} to the PL bound in Eq. (\ref{eq:PLcriterion}).}

\vspace{0.05in}
\noindent \textbf{LP formulation}. Like \cite{ImolaUAI2022, Qiu-TMC2022}, we consider the case where both $\mathcal{X}$ and $\mathcal{Y}$ are finite sets. In this case, the perturbation function $\mathcal{M}$ can be represented as the \emph{perturbation matrix} $\mathbf{Z}_{\mathcal{X}} = \{z_{x,y}\}_{(x,y) \in \mathcal{X} \times \mathcal{Y}}$, where each entry $z_{x,y}$ denotes the probability of selecting $y \in \mathcal{Y}$ as the perturbed output given the original record $x \in \mathcal{X}$. We define the \emph{perturbation vector} of $x$ as $\mathbf{z}_x = [z_{x,y}]_{y\in \mathcal{Y}}$, which represents the probability distribution over perturbed outputs for a given input $x$. Consequently, the mDP constraints formulated in Eq. (\ref{eq:mDP}) can be written as the following linear constraints: For each pair of records $x, x' \in \mathcal{X}$ 
\vspace{-0.06in}
\begin{equation}
\label{eq:mDPdiscrete}
% \textstyle 
z_{x,y} - e^{\epsilon d_{x,x'}} z_{x',y} \leq 0, ~\forall y \in \mathcal{Y}. 
\vspace{-0.02in}
\end{equation}
Additionally, the sum probability of perturbed record $y \in \mathcal{Y}$ for each real record $x$ should be equal to 1, i.e., 
\vspace{-0.00in}
\begin{equation}
\label{eq:um}
\textstyle
\sum_{y \in \mathcal{Y}}z_{x,y} = 1,~ \forall x \in \mathcal{X}.
\end{equation}
We use $c_{x,y}$ to represent the \emph{data utility loss} of the downstream decision-making caused by the perturbed record $y$ when the real record is $x$. Therefore, the \emph{loss function} of $\mathbf{Z}_{\mathcal{X}}$, measuring the expected utility loss caused by the perturbation matrix $\mathbf{Z}_{\mathcal{X}}$, can be defined as 
\vspace{-0.03in}
\begin{equation}
\textstyle 
\mathcal{L}(\mathbf{Z}_{\mathcal{X}}) = \sum_{(x,y) \in \mathcal{X}\times \mathcal{Y}}  p_x c_{x,y} z_{x,y},
\end{equation}
where $p_x = \Pr\left(X = x\right)$ is the prior probability of the real record being $x$. The goal of \emph{perturbation matrix optimization} is to minimize $\mathcal{L}(\mathbf{Z}_{\mathcal{X}})$ while satisfying both the mDP (Eq. (\ref{eq:mDPdiscrete})) and the probability unit measure (Eq. (\ref{eq:um})), which can be formulated as the LP problem: 
\begin{equation}
\label{eq:PPO}
\min ~\mathcal{L}(\mathbf{Z}_{\mathcal{X}}) ~\mathrm{s.t.} ~\mbox{Eq. (\ref{eq:mDPdiscrete})(\ref{eq:um})) are satisfied}, 
 % \\ \label{eq:LPconstraint2}
%& \mbox{and other constraints in Table \ref{Tb:constraints} are satisfied}
\end{equation}
where each entry $z_{x,y}$ in $\mathbf{Z}_{\mathcal{X}}$ satisfies $0 \leq z_{x,y} \leq 1$.

Although the LP problem formulated in Eq. (\ref{eq:PPO}) can be solved using classical LP algorithms such as the simplex method \cite{Linear&Nonlinear}, it is hampered by high computational costs. The time complexity of an LP problem depends on the number of decision variables and the number of linear constraints \cite{Linear&Nonlinear}. In Eq. (\ref{eq:PPO}), the decision matrix $\mathbf{Z}_{\mathcal{X}}$ encompasses $O(|\mathcal{X}||\mathcal{Y}|)$ decision variables, where it must adhere to the mDP constraints for every pair of neighboring locations in $\mathcal{X}$, resulting in $O(|\mathcal{X}|^2|\mathcal{Y}|)$ linear constraints.

% the calculation of LR-Geo still needs to be migrated to the server since (i) the computational demands of LR-Geo remain relatively high for mobile devices, and (ii) LR-Geo's LP formulation involves assessing data utility for downstream decision-making, a task typically handled by the server rather than individual users \cite{Wang-WWW2017}.

\begin{figure}[t]
\centering
\hspace{0.00in}
\begin{minipage}{0.45\textwidth}
  \subfigure{
\includegraphics[width=1.00\textwidth]{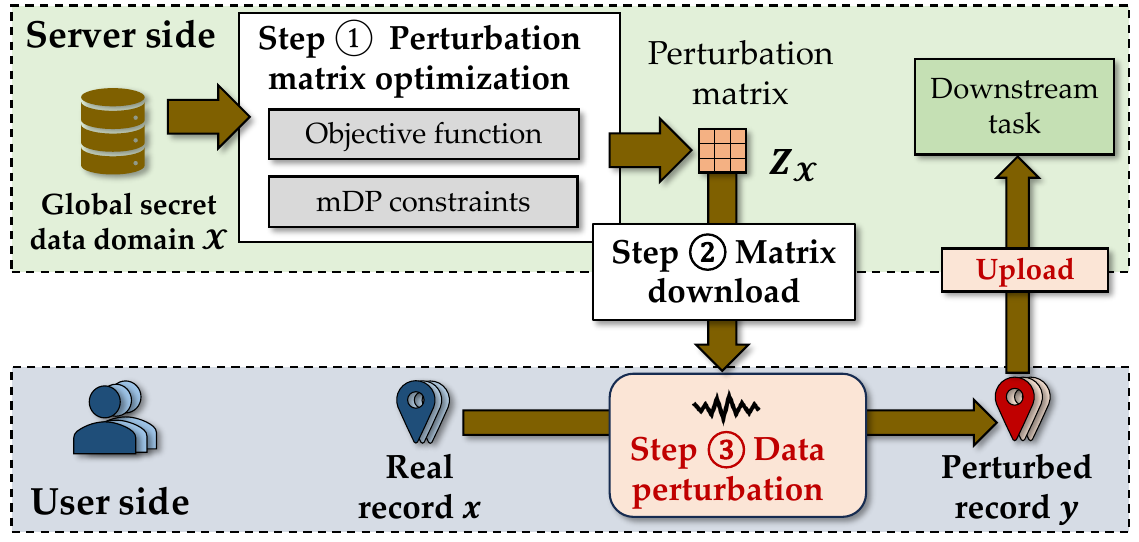}}
\vspace{-0.18in}
\end{minipage}
\caption{Remote perturbation optimization framework.}
\label{fig:remoteopt}
\vspace{-0.18in}
\end{figure}

\vspace{0.03in}
\noindent \textbf{Remote computing framework}. Due to the high computation overhead of LP, as illustrated in Fig. \ref{fig:remoteopt}, many prior works \cite{Wang-WWW2017, Qiu-TMC2022} adopt a remote computing framework, where the server first optimizes the perturbation matrix $\mathbf{Z}_{\mathcal{X}}$ (\textbf{Step \textcircled{1}}), after which users download $\mathbf{Z}_{\mathcal{X}}$ from the server (\textbf{Step \textcircled{2}}) and use it to perturb their records before reporting them to the server (\textbf{Step \textcircled{3}}). Importantly, although the server generates the matrix $\mathbf{Z}_{\mathcal{X}}$, users' exact records remain hidden from the server \cite{Wang-WWW2017}. Matrix $\mathbf{Z}_{\mathcal{X}}$ is designed to satisfy mDP, ensuring that even if an attacker can obtain users' perturbed records, it is challenging to distinguish their exact records from other records as the mDP constraints are satisfied. \looseness = -1

While servers have higher computation capability, the high computational requirements make LP-based geo-obfuscation infeasible for applications involving a large-scale of secret data domain. As shown in Fig. \ref{fig:relatedwork}, existing LP-based methods typically limit the size of secret data domain to up to 1,000 records, i.e., $|\mathcal{X}| \leq 1,000$.

% We let $\mathcal{X}_{i} = \left\{x_m\in\mathcal{X}\left|d_{x_n, x_m} \leq r \right.\right\}$ denote the set of $x_n$'s neighbors. 

% \noindent \textbf{More notations}. We call $x_m$ a neighbor of $x_n$ if $d_{x_n, x_m} \leq r$. We let $\mathcal{X}_{i}$ (or $\mathcal{N}^{(1)}_{x_n}$) denote the neighbor set of $x_n$
%\begin{equation}
%\mathcal{X}_{i} = \left\{x_m \in \mathcal{X} \left|d_{x_n, x_m} \leq r\right.\right\}. 
%\end{equation} 
%We call $x_m$ a 2nd-order neighbor if $x_n$, denoted by $x_m \in \mathcal{N}^{(2)}_{x_n}$, if 
%\begin{itemize}
%\item $x_m \notin \mathcal{X}_{i}$, i.e., $x_m$ is not a neighbor of $x_n$
%\item $\exists x_\ell \in \mathcal{X}_{i}$ such that $x_m \in \mathcal{N}_{x_\ell}$, i.e., $x_m$ is a neighbor of $x_n$'s neighbor. 
%\end{itemize}

\vspace{-0.05in}
\section{Methodology}
\vspace{-0.00in}
\label{sec:method}
In this section, we present \textsc{PAnDA}. Section~\ref{subsec:motivation} discusses the motivations for our design, followed by an overview of the framework in Section~\ref{subsec:framework}. Sections~\ref{subsec:PhaseI} and~\ref{subsec:PhaseII} then detail \textsc{PAnDA}’s two phases, respectively.

\vspace{-0.05in}
\subsection{Motivations}
\label{subsec:motivation}
\vspace{-0.00in}
% \noindent 
We consider a scenario in which a set of participating users $1, ..., N$ need to report their records to a server. We use $x_n$ to represent the secret record index held by each user $n$ ($n = 1,..., N$). 

As shown in Fig.~\ref{fig:remoteopt}, previous works such as \cite{Bordenabe-CCS2014,Wang-WWW2017,Qiu-PETS2025} follow a common approach: the server first computes a global perturbation matrix $\mathbf{Z}_{\mathcal{X}}$, which is then downloaded by each user. To perturb their private data $x_n$, user $n$ only needs the corresponding row $\mathbf{z}_{x_n}$ from this matrix. However, it is not feasible for a user to ask the server to optimize only their own row $\mathbf{z}_{x_n}$, due to two reasons:  
(1) the mDP constraints couple $\mathbf{z}_{x_n}$ with the perturbation vectors of other records, making isolated optimization invalid;  
(2) the server does not know the user’s true record $x_n$, and therefore cannot evaluate the utility loss for different perturbation choices.  
As a result, the server needs to solve a large LP that includes all possible real records in $\mathcal{X}$, leading to substantial computational overhead.

\subsubsection{Motivation 1: Anchor-Based Approximation in Place of Full Secret Data Domain $\mathcal{X}$.} 
{\rev Although} the perturbation vectors are coupled by the mDP constraints, this dependency weakens with distance. According to the definition of mDP in Eq.~(\ref{eq:mDP}), the further apart two secret records are, the less influence they have on each other due to diminishing privacy constraints. For instance, if two users' locations differ by 100 kilometers, the constraint requiring their perturbation distributions to be similar is nearly negligible. Consequently, the perturbation vector of one user exerts minimal influence on the optimization of the other’s perturbation vector.

Leveraging this observation, \textsc{PAnDA} avoids the need to optimize the full perturbation matrix $\mathbf{Z}_{\mathcal{X}}$. Instead, it introduces an \emph{anchor perturbation matrix} for each user. Specifically, for each user $n$ ($n = 1, \dots, N$), \textsc{PAnDA} optimizes a smaller matrix $\mathbf{Z}_{\mathcal{A}_n}$ that only includes the perturbation vectors of a subset of \emph{anchor records} $\mathcal{A}_n \subseteq \mathcal{X}$. These anchor records are selected to capture the most relevant constraints for user $n$. Across all users, the total number of anchor records, $|\bigcup_n \mathcal{A}_n|$, is chosen to be much smaller than the total number of real records, i.e., $|\bigcup_n \mathcal{A}_n| \ll |\mathcal{X}|$. This drastically reduces the size of the resulting LPs and leads to significant computational savings.

Intuitively, the anchor sets $\mathcal{A}_1, \dots, \mathcal{A}_N$ should satisfy the following two key properties:
\vspace{-0.02in}
\begin{itemize}[left=0.3em, labelsep=0.5em] 
    \item[\textbf{(a)}] \textbf{Accuracy of approximate matrices:} Each submatrix $\mathbf{Z}_{\mathcal{A}_n}$ should closely approximate the perturbation vectors that are highly relevant to the perturbation vector $\mathbf{z}_{x_n}$, enabling user $n$ to maintain the required level of indistinguishability from other records while effectively reducing the utility loss introduced by perturbation. \textit{Considering that the anchor set $\mathcal{A}_n$ is relevant to the true record $x_n$, $\mathcal{A}_n$ should be constructed locally by user $n$ based on $x_n$.}
    \item[\textbf{(b)}] \textbf{Bounded privacy leakage caused by selected anchor sets:} The anchor sets $\mathcal{A}_n$ and $\mathcal{A}_m$ across different users should collectively ensure strong indistinguishability between users’ true records $x_n$ and $x_m$, thereby preserving privacy from the server's perspective. \textit{This requires anchor records to be selected in a probabilistic manner} (as formally justified in \textbf{Proposition~\ref{prop:wsmall}}), where the selection probability reflects each anchor record’s influence on the optimization of $x_n$’s perturbation vector.
\end{itemize}
\vspace{-0.10in}
\subsubsection{Probablistic mDP instead of worst-case mDP.} The original definition of mDP (in \textbf{Definition \ref{def:mDP}}) considers the adversarial worst-case scenario, evaluating the maximum divergence between  
\newline $\Pr[\mathcal{M}(x_n) \in \mathcal{Y}']$ and $\Pr[\mathcal{M}(x_m) \in \mathcal{Y}']$ for any pair of candidate inputs $x_n$ and $x_m$. Although this worst-case guarantee ensures strong privacy protection, it introduces significant computational overhead~\cite{Xiao-CRYPTO2023}, as the LP formulation needs to include all input pairs in the domain $\mathcal{X}$, leading to $O(|\mathcal{X}|^2|\mathcal{Y}|)$ mDP constraints. 

To balance privacy and practicality, we introduce a relaxed variant of mDP, called \textbf{Probabilistic mDP (PmDP)} in \textbf{Definition \ref{def:approxmDP}}. 
This formulation tolerates a small probability $\delta$ that the privacy guarantee may not hold, thereby reducing the need to optimize for rare worst-case input pairs 
and enabling more efficient perturbation optimization. 
\vspace{-0.02in}
\begin{definition}[$(\epsilon,\delta)$-PmDP]
\label{def:approxmDP}
A randomized mechanism $\mathcal{M}: \mathcal{X} \rightarrow \mathcal{Y}$ satisfies $(\epsilon,\delta)$-PmDP if, for all $x, x' \in \mathcal{X}$,
\begin{equation}
\label{eq:PmDP}
\Pr\left[\mathcal{M}(x) \stackrel{\epsilon}{\approx} \mathcal{M}(x')\right] \geq 1 - \delta.
\end{equation}
Here, the parameter $\delta \in [0, 1)$ represents the allowable \emph{failure probability} for the mDP constraint $\mathcal{M}(x) \stackrel{\epsilon}{\approx} \mathcal{M}(x')$. 
That is, with probability at least $1 - \delta$, the mechanism guarantees $\epsilon$-indistinguishability between outputs for $x$ and $x'$; with probability at most $\delta$, this guarantee may not hold.
\end{definition}
To simplify notation, we write $\mathcal{M}(x) \stackrel{(\epsilon, \delta)}{\sim} \mathcal{M}(x')$ when Eq.~(\ref{eq:PmDP}) is satisfied. 
When $\delta = 0$, PmDP recovers the original (pure) mDP:
\begin{equation}
\mathcal{M}(x) \stackrel{\epsilon}{\approx} \mathcal{M}(x') \quad \equiv \quad \mathcal{M}(x) \stackrel{(\epsilon, 0)}{\sim} \mathcal{M}(x').
\end{equation}
In practice, $\delta$ can be set to reflect the level of tolerance for rare but potentially high-utility data disclosures that slightly violate the worst-case mDP condition. % This probabilistic relaxation facilitates the design of \textsc{PAnDA} without significantly compromising overall privacy.

\vspace{0.03in}
{\rev 
\noindent \textbf{Remark.} Notably, PmDP (\textbf{Definition \ref{def:approxmDP}}) is related to approximate DP \cite{Dwork-STOC2009}, they are distinct in how they relax strict privacy guarantees. Approximate DP, denoted by $(\epsilon, \delta)$-DP,  requires that for all neighboring datasets $x$ and  $x'$ and all measurable sets $\mathcal{Y}' \subseteq \mathcal{Y}$, $\Pr[\mathcal{M}(x) \in \mathcal{Y}'] \leq e^{\epsilon} \cdot \Pr[\mathcal{M}(x') \in \mathcal{Y}'] + \delta$, meaning that the output distributions of a mechanism are close up to a multiplicative factor $e^\epsilon$ with an additive slack $\delta$. Differently, PmDP requires that for any pair of inputs $ x, x' $, the mechanism $\mathcal{M}$ satisfies, i.e., $\Pr\left[\mathcal{M}(x) \stackrel{\epsilon}{\approx} \mathcal{M}(x')\right] \geq 1 - \delta$, where $\mathcal{M}(x) \stackrel{\epsilon}{\approx} \mathcal{M}(x')$ means that for all measurable $\mathcal{Y}' \subseteq \mathcal{Y}$, $\Pr[\mathcal{M}(x) \in \mathcal{Y}'] \leq e^{\epsilon d_{x, x'}} \Pr[\mathcal{M}(x') \in \mathcal{Y}']$. This means that the $\epsilon$-indistinguishability condition must hold with high probability over the randomness of the mechanism $\mathcal{M}$.}

\begin{figure}[t]
\centering
\hspace{0.00in}
\begin{minipage}{0.46\textwidth}
  \subfigure{
\includegraphics[width=1.00\textwidth]{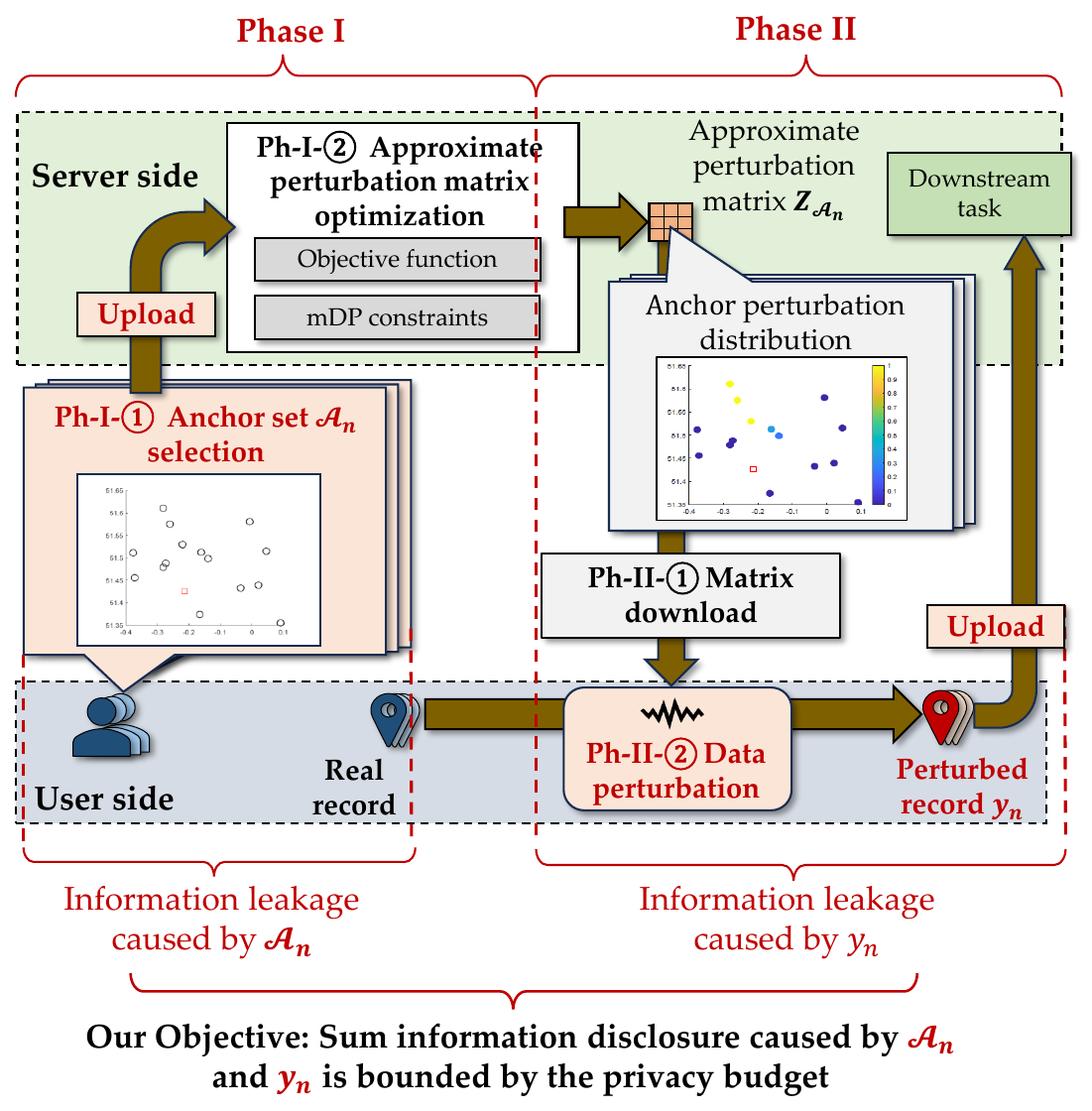}}
\vspace{-0.25in}
\end{minipage}
\caption{Entire \textsc{PAnDA} procedure in two phases.}
\label{fig:framework}
\vspace{-0.15in}
\end{figure}

\subsection{Overview of \textsc{PAnDA}}
\label{subsec:framework}
Guided by the \textbf{properties (a)} and \textbf{(b)}, we design the \textsc{PAnDA} procedure in two phases, as shown in Fig.~\ref{fig:framework}:
\begin{itemize}[left=1.2em, labelsep=0.5em] 
\item [\textbf{Ph-I}:] Each user $n$ applies a randomized mechanism $\mathcal{M}_{\mathrm{I}}: \mathcal{X} \rightarrow 2^{\mathcal{X}}$ to select a set of anchor records $\mathcal{A}_{n} \subseteq \mathcal{X}$ based on their true record $x_n$ and uploads $\mathcal{A}_{n}$ to the server (Step Ph-I-\textcircled{1}). The server then optimizes the perturbation probabilities for the union of all uploaded anchor sets $\mathcal{A}_{[1, N]} = \bigcup_{n=1}^N \mathcal{A}_{n}$ (Step Ph-I-\textcircled{2}), producing the \emph{joint approximate  perturbation matrix} $\mathbf{Z}_{\mathcal{A}_{[1, N]}} = \left\{z_{x,y}\right\}_{(x,y) \in \mathcal{A}_{[1, N]} \times \mathcal{Y}}$. {\rev Here, we assume an honest-but-curious model in which all users follow the protocol correctly when selecting their anchor records.} 
\item [\textbf{Ph-II}:] Each user $n$ downloads the anchor perturbation matrix $\mathbf{Z}_{\mathcal{A}_n} = \left\{z_{x,y}\right\}_{(x,y) \in \mathcal{A}_n \times \mathcal{Y}}$,
corresponding to their anchor set $\mathcal{A}_n$ (Step Ph-II-\textcircled{1}), and uses it to perturb their true record to obtain $y_n$, which is then reported to the server (Step Ph-II-\textcircled{2}).

Notably, a user’s true record $x_n$ may not be included in their anchor set ($x_n \notin \mathcal{A}_n$). In this case, the perturbation of $x_n$ is approximated using the perturbation vector $\mathbf{z}_{\hat{x}_n}$ of its nearest anchor $\hat{x}_n$ in $\mathcal{A}_n$, defined as $\hat{x}_n = \arg\min_{x \in \mathcal{A}_n} d_{x_n, x}$.
We refer to $\hat{x}_n$ as \emph{surrogate} of $x_n$ and refer to $\mathbf{z}_{\hat{x}_n}$ as \emph{surrogate perturbation vector} of $x_n$.

\end{itemize}

The \textsc{PAnDA} framework includes two rounds of information released by user $n$, each managed by a randomized mechanism:
\begin{itemize}[left=0.3em, labelsep=0.5em] 
\item In \textbf{Ph-I-\textcircled{1}}, the randomized mechanism $\mathcal{M}_{\mathrm{I}}: \mathcal{X} \rightarrow 2^{\mathcal{X}}$ is used to select the anchor record set $\mathcal{A}_n$. 
\item In \textbf{Ph-II-\textcircled{2}}, the randomized mechanism $\mathcal{M}_{\mathrm{II}}: \mathcal{X} \rightarrow \mathcal{Y}$, which is implemented via the perturbation matrix $\mathbf{Z}_{\mathcal{A}_n}$, is used to obtain the perturbed record $y_n$. 
\end{itemize}
Therefore, \textsc{\textsc{PAnDA}} can be viewed as the composition of two randomized mechanisms, defined as $\mathcal{M}_{\mathrm{\textsc{PAnDA}}}(x) = (\mathcal{M}_{\mathrm{I}}(x), \mathcal{M}_{\mathrm{II}}(x))$. This composition forms a probabilistic mapping $\mathcal{M}_{\mathrm{\textsc{PAnDA}}}: \mathcal{X} \rightarrow 2^{\mathcal{X}} \times \mathcal{Y}$, where $2^{\mathcal{X}}$ and $\mathcal{Y}$ denote the domains of the anchor record set and the perturbed record, respectively.

\vspace{0.05in}
{\rev \noindent \textbf{Remark.} Since each user independently selects and uploads their anchor set in Phase I, the framework is naturally robust to non-reporting participants, users who do not participate are simply excluded from the optimization without impacting others. When new users are added or existing users drop out, the server can re-optimize the anchor perturbation matrix over the updated union of anchor sets. If users report a sequence of true records over time (e.g., trajectory data), anchor selection can either be repeated for each record or reused with adjustments. When the records are spatially or temporally close (e.g., consecutive locations along a path), anchor sets selected for earlier records can often be reused or incrementally expanded to cover nearby regions, reducing overhead. }

\subsubsection{Unique Challenges of \textsc{PAnDA}.}

\textbf{Challenge I: privacy composition across phases.}  
In \textsc{PAnDA}, data disclosures occur in both phases: during anchor selection and during record perturbation. Each of these releases may leak information about the user's true record $x_n$. Thus, the \emph{first challenge} is to ensure that the \emph{total privacy cost resulting from both phases remains bounded by a given privacy budget $\epsilon$ under the mDP framework}.

\textbf{Challenge II: Approximation-induced mDP violation.}  
While \textsc{\textsc{PAnDA}} employs a surrogate perturbation vector $\mathbf{z}_{\hat{x}_n}$ to approximate the perturbation vector $\mathbf{z}_{x_n}$, such approximation inherently introduces error. This raises the \emph{second challenge}: \emph{how to rigorously control the approximation error to avoid violations of mDP, even under worst-case input scenarios}.

% although each anchor perturbation matrix $\mathbf{Z}_{\mathcal{A}_n}$ aims accurately, 

Next, we introduce \textbf{Theorem~\ref{thm:composition}} to prove the \emph{sequential composition} of PmDP,  enabling us to split the total privacy budget $\epsilon$ between the two phases. After that, we detail the two randomized mechanisms, $\mathcal{M}_{\mathrm{I}}$ and $\mathcal{M}_{\mathrm{II}}$ in the two phases, in Sections~\ref{subsec:PhaseI} and~\ref{subsec:PhaseII}, to address the above two challenges.
%\subsection{Compositon Property for \textsc{PAnDA}}
%\label{subsec:budgetalloc}
% \vspace{-0.00in}
% \subsubsection{Probabilistic mDP.} 
% \subsubsection{Composiability of PmDP for \textsc{PAnDA}.} \textbf{Theorem~\ref{thm:composition}} establishes that the \textbf{sequential composition} property of mDP is preserved under the PmDP framework.
\begin{theorem}[Sequential Composition of PmDP for \textsc{PAnDA}]
\label{thm:composition}
\begin{eqnarray}
&& \mathcal{M}_{\mathrm{I}}(x_n) \stackrel{(\epsilon_{1}, \delta_{1})}{\sim} \mathcal{M}_{\mathrm{I}}(x_m)~\text{and}~ \mathcal{M}_{\mathrm{II}}(x_n) \stackrel{(\epsilon_{2}, \delta_{2})}{\sim} \mathcal{M}_{\mathrm{II}}(x_m) \\
&\Rightarrow& \mathcal{M}_{\mathrm{\textsc{PAnDA}}}(x_n) \stackrel{(\epsilon_1+\epsilon_2,\, \delta_1+\delta_2)}{\sim} \mathcal{M}_{\mathrm{\textsc{PAnDA}}}(x_m)
\end{eqnarray}
\end{theorem}
\vspace{-0.2in}
% \noindent \textbf{Proof sketch.} We first prove that $\mathcal{M}_{\mathrm{I}}$ and $\mathcal{M}_{\mathrm{II}}$ satisfy the sequential composition property under mDP (see \textbf{Lemma~\ref{lem:comp}} in the Appendix). Building on this result, we then extend the composition to the PmDP setting as stated in \textbf{Theorem~\ref{thm:composition}}. For full proofs of all theorems and propositions, please refer to \textbf{Appendix~\ref{sec:proofs}}.
\noindent{\rev \begin{proof}[Proof Sketch] 
% The key idea of the proof is to leverage the \emph{sequential composition} property of mDP and extend it to PmDP. We begin by establishing in \textbf{Lemma~\ref{lem:comp}} (introduced in \textbf{Appendix \ref{subsec:proof:thm:composition}}) that if both anchor selection mechanism $\mathcal{M}_{\mathrm{I}}$ and anchor perturbation mechanism $\mathcal{M}_{\mathrm{II}}$ individually satisfy $\epsilon$-mDP, then their composition $\mathcal{M}_{\mathrm{\textsc{PAnDA}}}$ satisfies $(\epsilon_1 + \epsilon_2)$-mDP. Next, we extend this result to the PmDP setting. Under $(\epsilon, \delta)$-PmDP, each mechanism is allowed to violate mDP constraints with a small probability $\delta$. Therefore, the probability that \emph{both} $\mathcal{M}_{\mathrm{I}}$ and $\mathcal{M}_{\mathrm{II}}$ satisfy their respective $\epsilon$-mDP bounds is at least $1 - (\delta_1 + \delta_2)$ by the union bound. Thus, with probability at least $1 - (\delta_1 + \delta_2)$, the composition $\mathcal{M}_{\mathrm{\textsc{PAnDA}}}$ satisfies $(\epsilon_1 + \epsilon_2)$-mDP, leading to the final guarantee of $(\epsilon_1 + \epsilon_2, \delta_1 + \delta_2)$-PmDP for the entire mechanism. The detailed proof can be found in \textbf{Appendix \ref{subsec:proof:thm:composition}}.  
The proof builds on the \emph{sequential composition} property of mDP and extends it to PmDP. We first show in \textbf{Lemma~\ref{lem:comp}} (Appendix~\ref{subsec:proof:thm:composition}) that if $\mathcal{M}_{\mathrm{I}}$ and $\mathcal{M}_{\mathrm{II}}$ each satisfy $\epsilon$-mDP, then their composition $\mathcal{M}_{\textsc{PAnDA}}$ satisfies $(\epsilon_1 + \epsilon_2)$-mDP. Extending to PmDP, by the union bound, the composition satisfies $(\epsilon_1 + \epsilon_2, \delta_1 + \delta_2)$-PmDP. The detailed proof can be found in \textbf{Appendix \ref{subsec:proof:thm:composition}}.  
\end{proof}
}

\noindent \textbf{Remark.} Although the two mechanisms $\mathcal{M}_\mathrm{I}$ and $\mathcal{M}_\mathrm{II}$ are not independent, since the anchor set selected in Phase I directly affects the perturbation mechanism in Phase II, the sequential composition property of (P)mDP does not require independence \cite{dwork2014algorithmic}. % Thus, it remains valid to decompose the overall privacy budget $\epsilon$ into phase-wise components $\epsilon_1$ and $\epsilon_2$, as formalized in Theorem~\ref{thm:composition}.

% \noindent \textbf{Remark.} Although the two mechanisms $\mathcal{M}_{\mathrm{I}}$ and $\mathcal{M}_{\mathrm{II}}$ are not independent, since the anchor records selected by $\mathcal{M}_{\mathrm{I}}$ determine the perturbation matrix $\mathbf{Z}_{\mathcal{A}_n}$, which is subsequently used by $\mathcal{M}_{\mathrm{II}}$ to generate the perturbed record, the sequential composition property does not require independence between mechanisms {\rev[ref]}. % In fact, composition can still hold under dependency, although tighter or alternative privacy bounds may be needed depending on the level of correlation.

\DEL{
\begin{theorem}
[Sequential Composition of PmDP]
\label{thm:composition} 
Consider a class of randomized mechanisms $\mathscr{M}$ satisfying the sequential composition property for $\epsilon$-mDP. Given $L$ randomized mechanisms $\mathcal{M}_1, \dots, \mathcal{M}_L \in \mathscr{M}$ with each mechanism $\mathcal{M}_{\ell}$ satisfying $(\epsilon_{\ell}, \delta_{\ell})$-PmDP, then the composed mechanism 
$\mathcal{M} = (\mathcal{M}_1, \dots, \mathcal{M}_L)$ satisfies $(\sum_{\ell=1}^L \epsilon_{\ell}, \sum_{\ell=1}^L \delta_{\ell})$-PmDP. Formally, for all $x, x' \in \mathcal{X}$,
\begin{eqnarray}
&& \mathcal{M}_{\ell}(x) \stackrel{(\epsilon_{\ell}, \delta_{\ell})}{\sim} \mathcal{M}_{\ell}(x'), \quad \forall \ell = 1, \dots, L  
\\ 
&\Rightarrow& \mathcal{M}(x) \stackrel{(\sum_{\ell=1}^L \epsilon_{\ell}, \sum_{\ell=1}^L \delta_{\ell})}{\sim} \mathcal{M}(x').
\end{eqnarray}
\end{theorem}
\begin{theorem}[Sequential Composition of PmDP]
\label{thm:composition} 
Let $\mathscr{M}$ be a class of randomized mechanisms that satisfy all the features or conditions required for sequential composition in $\epsilon$-metric Differential Privacy (mDP). Suppose we are given $L$ mechanisms $\{\mathcal{M}_\ell\}_{\ell=1}^L \subset \mathscr{M}$, and each $\mathcal{M}_\ell$ satisfies $(\epsilon_\ell, \delta_\ell)$-PmDP. Define the composed mechanism 
\begin{equation}
\mathcal{M}(x) \;=\; \bigl(\mathcal{M}_1(x), \dots, \mathcal{M}_L(x)\bigr).
\end{equation}
Then, for all $x, x' \in \mathcal{X}$,
\begin{equation}
\forall \,\ell = 1, \dots, L: 
\quad 
\mathcal{M}_\ell(x) 
\;\stackrel{(\epsilon_\ell, \delta_\ell)}{\sim}\;
\mathcal{M}_\ell(x') 
\quad\Longrightarrow\quad 
\mathcal{M}(x) 
\;\stackrel{\bigl(\sum_{\ell=1}^L \epsilon_\ell,\;\sum_{\ell=1}^L \delta_\ell\bigr)}{\sim}\;
\mathcal{M}(x').
\end{equation}
\end{theorem}
}
\DEL{
\begin{equation}
\mathcal{M}_{\ell}(x) \stackrel{\epsilon_{\ell}}{\approx} \mathcal{M}_{\ell}(x'), \quad \forall \ell = 1, \dots, L  
\Rightarrow \mathcal{M}(x) \stackrel{\sum_{\ell=1}^L \epsilon_{\ell}}{\approx} \mathcal{M}(x').
\end{equation}}

\DEL{
\begin{proposition}
[Sequential composition of mDP for \textsc{PAnDA}] 
\label{prop:\textsc{PAnDA}comp}
If $\mathcal{M}_{\mathrm{I}}(x_n) \stackrel{\epsilon_{1}}{\sim} \mathcal{M}_{\mathrm{I}}(x_m)$ and $\mathcal{M}_{\mathrm{II}}(x_n) \stackrel{\epsilon_{2}}{\sim} \mathcal{M}_{\mathrm{II}}(x_m)$ $\forall x_n, x_m \in \mathcal{X}$, then $\mathcal{M}_{\mathrm{\textsc{PAnDA}}}(x_n) \stackrel{\epsilon_{1}+\epsilon_{2}}{\sim} \mathcal{M}_{\mathrm{\textsc{PAnDA}}}(x_m)$. 
\end{proposition}}
% \vspace{0.05in}
% \noindent \textbf{Privacy budgets allocation for the two phases of \textsc{PAnDA}}. 

Theorem~\ref{thm:composition} establishes that the sequential composition property of PmDP holds within the \textsc{PAnDA} framework, despite the fact that the two mechanisms, anchor selection and record perturbation, are not independent. This result provides a theoretical foundation for partitioning the total privacy budget $\epsilon$ into two disjoint components: $\epsilon_1$ for Phase I (anchor selection) and $\epsilon_2 = \epsilon - \epsilon_1$ for Phase II (record perturbation). Similarly, the overall failure probability $\delta$ can be split as $\delta = \delta_1 + \delta_2$ across the two phases. This composability enables the design of the randomized mechanisms $\mathcal{M}{\mathrm{I}}$ and $\mathcal{M}{\mathrm{II}}$ in each phase relatively independently, while still satisfying the overall privacy guarantee.

% With this decomposition, \textsc{PAnDA} can enforce privacy constraints within each phase independently, while ensuring that the full mechanism satisfies the global $(\epsilon, \delta)$-PmDP guarantee. 

\subsection{Phase I: Local Anchor Selection}
\label{subsec:PhaseI}
In this part, we present the details of the randomized mechanism $\mathcal{M}_{\mathrm{I}}$ in Phase I, to select anchor records $\mathcal{A}_n$ for each user $n$. {\rev In particular}, we let $\mathcal{M}_{\mathrm{I}}$ select each anchor record from the domain $\mathcal{X}$ independently. Let $w_{x_n,x}$ denote the probability of selecting record $x \in \mathcal{X}$ as an anchor given the user's true record $x_n$, i.e.,
\begin{equation}
w_{x_n,x} = \Pr\left[x \in \mathcal{M}_{\mathrm{I}}(x_n) \mid X_n = x_n\right].
\end{equation}
where $X_n$ is a random variable to represent user $n$'s real record. Then, the probability that $\mathcal{M}_{\mathrm{I}}(x_n)$ selects a specific anchor set $\mathcal{A} \subseteq \mathcal{X}$ is given by
\begin{equation}
\label{eq:M_I}
\textstyle \Pr\left[\mathcal{M}_{\mathrm{I}}(x_n) = \mathcal{A}\right] = \prod_{x \in \mathcal{A}} w_{x_n,x} \prod_{x \notin \mathcal{A}} (1 - w_{x_n,x}).
\end{equation}
{\rev Here, Eq.~(\ref{eq:M_I}) treats the inclusion of each anchor $x \in \mathcal{X}$ as an independent Bernoulli trial with success probability $w_{x_n,x}$. That is, for any particular subset $\mathcal{A} \subseteq \mathcal{X}$, the probability that $\mathcal{M}_{\mathrm{I}}(x_n)$ selects exactly $\mathcal{A}$ corresponds to the probability of including all records in $\mathcal{A}$ and excluding all records not in $\mathcal{A}$.  Due to the independence of these inclusion decisions, the total probability is the product of the individual probabilities for each record, leading to the multiplicative form in Eq.~(\ref{eq:M_I}).  }

 % In this case, the anchor record selection method can be described as a matrix $\mathbf{W} = \left\{w_{x_n,x_\ell}\right\}_{(x_n, x_\ell)\in \mathcal{X}^2}$. 

% $\mathcal{M}_{\mathrm{\textsc{PAnDA}}}$ comprises of two random functions in the two phases $\mathcal{M}_{\mathrm{I}}: \mathcal{X}\rightarrow 2^{\mathcal{X}}$ in \textbf{Ph-I}, which outputs a set of anchor records $\mathcal{A}$, and $\mathcal{M}_{\mathrm{II}}: \mathcal{X}\rightarrow {\mathcal{Y}}$ in \textbf{Ph-II}, which outputs a perturbed record $y$. Accordingly, the mDP constraints in Eq. (\ref{eq:mDP\textsc{PAnDA}}) can be also written as 
% \begin{equation}
%\frac{\Pr\left[\mathcal{M}_{\mathrm{I}}(x_n) = \mathcal{A} \wedge \mathcal{M}_{\mathrm{II}}(x_n) = y\right]}{\Pr\left[\mathcal{M}_{\mathrm{I}}(x_m) = \mathcal{A} \wedge \mathcal{M}_{\mathrm{II}}(x_m) = y\right]} \leq e^{\epsilon d_{x_n, x_m}}, \forall (\mathcal{A}, y) \in 2^{\mathcal{X}} \times \mathcal{Y}. 
% \end{equation} 

{\rev To formally characterize the privacy leakage introduced by this anchor selection process, we analyze how distinguishable two users' anchor sets can be when their true records differ. This analysis is captured in \textbf{Proposition \ref{prop:posteriorbound}}, which bounds the worst-case privacy leakage due to anchor set selection.}
\begin{proposition}
\label{prop:posteriorbound}
For any two records $x_n, x_m \in \mathcal{X}$, we define $\mathcal{A}_{n,m} = \left\{x \in \mathcal{X} \left|w_{x_n,x} > w_{x_m,x}\right.\right\}$ and $\overline{\mathcal{A}}_{n,m} = \left\{x \in \mathcal{X} \left|w_{x_n,x} < w_{x_m,x}\right.\right\}$, 
then either $\epsilon_{x_n,x_m,\mathcal{A}_{n,m}}$ or $\epsilon_{x_n,x_m,\overline{\mathcal{A}}_{n,m}}$ provides the least upper bound of the privacy cost $\epsilon_{x_n,x_m,\mathcal{A}}$, i.e., 
\begin{eqnarray}
\label{eq:PLupperbound2}
\max\left\{\epsilon_{x_n,x_m,\mathcal{A}_{n,m}}, \epsilon_{x_n,x_m,\overline{\mathcal{A}}_{n,m}}\right\} = \sup_{\mathcal{A}\in 2^{\mathcal{X}}} \epsilon_{x_n,x_m,\mathcal{A}}
%\frac{\Pr\left[\mathcal{M}_{\mathrm{I}}(x_n) = \mathcal{A}\right]}{\Pr\left[\mathcal{M}_{\mathrm{I}}(x_m) = \mathcal{A}\right]} \leq \frac{\Pr\left[\mathcal{M}_{\mathrm{I}}(x_n) = \mathcal{A}_{n,m}\right]}{\Pr\left[\mathcal{M}_{\mathrm{I}}(x_m) = \mathcal{A}_{n,m}\right]}, ~\forall \mathcal{A} \in 2^{\mathcal{X}}. 
\end{eqnarray}
\vspace{-0.00in}
{\rev where $\epsilon_{x_n,x_m,\mathcal{A}} = \log \left( \frac{\Pr[\mathcal{M}_{\mathrm{I}}(x_n) = \mathcal{A}]}{\Pr[\mathcal{M}_{\mathrm{I}}(x_m) = \mathcal{A}]} \right)$ measures the distinguishability between two users $n$ and $m$ based on the probability that each user selects the same anchor set $\mathcal{A}$. A larger value of $\epsilon_{x_n,x_m,\mathcal{A}}$ indicates that the selected anchor set $\mathcal{A}$ leaks more information about whether the true input was $x_n$ or $x_m$.}

\end{proposition}
% \noindent \textbf{Proof sketch.} 
\DEL{Here, we can obtain the privacy cost between $x_n$ and $x_m$ given each selected anchor record $x \in \mathcal{X}$, denoted by $\epsilon_{x_n,x_m,x}$, which satisyfies
\begin{eqnarray}
e^{-\epsilon_{x_n,x_m,x}d_{x_n,x_m}}\leq \frac{w_{x_n,x}}{w_{x_m,x}} \leq e^{\epsilon_{x_n,x_m,x}d_{x_n,x_m}}
\end{eqnarray}
\begin{eqnarray}
e^{-\epsilon_{x_n,x_m,x}d_{x_n,x_m}}\leq \frac{1-w_{x_n,x}}{1-w_{x_m,x}} \leq e^{\epsilon_{x_n,x_m,x}d_{x_n,x_m}}
\end{eqnarray}} 
% The detailed proof can be found in \textbf{Section \ref{subsec:prop:posteriorbound} in Appendix}. 
\vspace{-0.05in}
{\rev 
\noindent
\begin{proof}[Proof Sketch] 
% We aim to identify the anchor set $\mathcal{A}$ that yields the worst-case privacy cost (i.e., the largest likelihood ratio between $\mathcal{M}_{\mathrm{I}}(x_n)$ and $\mathcal{M}_{\mathrm{I}}(x_m)$). 
% Let $W(\mathcal{A}, \mathbf{w}_{x_n}, \mathbf{w}_{x_m})$ denote the likelihood ratio for a specific anchor set $\mathcal{A}$, where the ratio is computed based on the selection probabilities of each anchor under $x_n$ and $x_m$. 
The key idea is that the likelihood ratio between $\mathcal{M}_{\mathrm{I}}(x_n)$ and $\mathcal{M}_{\mathrm{I}}(x_m)$ is maximized when the anchor set favors one record over the other—either $\mathcal{A}_{n,m} = \left\{x \mid w_{x_n,x} > w_{x_m,x}\right\}$ or $\overline{\mathcal{A}}_{n,m} = \left\{x \mid w_{x_n,x} < w_{x_m,x}\right\}$. Assuming for contradiction that another set $\tilde{\mathcal{A}}$ yields a higher ratio, we construct a new set $\hat{\mathcal{A}}$ by modifying $\tilde{\mathcal{A}}$ based on weight differences, and show this leads to an even higher ratio, contradicting the optimality of $\tilde{\mathcal{A}}$. Therefore, the worst-case privacy cost is bounded by the maximum likelihood ratio over $\mathcal{A}_{n,m}$ and $\overline{\mathcal{A}}_{n,m}$. The full proof is in \textbf{Appendix \ref{subsec:prop:posteriorbound}}.
\end{proof}
}

For simplicity, in what follows, we let 
{\rev \begin{eqnarray}
\label{eq:epsilon_overline1}
\overline{\epsilon}_{x_n,x_m} &=& \max\left\{\epsilon_{x_n,x_m,\mathcal{A}_{n,m}}, \epsilon_{x_n,x_m,\overline{\mathcal{A}}_{n,m}}\right\} \\  
\label{eq:epsilon_overline2}
&=& \sup_{\mathcal{A}\in 2^{\mathcal{X}}} \epsilon_{x_n,x_m,\mathcal{A}}.
\end{eqnarray}}

\DEL{When deriving $\mathcal{M}_{\mathrm{II}}$ in Phase II, the server has obtained the anchor record set $\mathcal{A}$, then the server can estimate the privacy cost caused by $\mathcal{A}$ for any pair of real records $x_n$ and $x_m$, denoted by $\epsilon_{x_n,x_m,\mathcal{A}}$ % to represent the privacy cost caused by $\frac{\Pr\left[\mathcal{M}_{\mathrm{I}}(x_n) = \mathcal{A}\right]}{\Pr\left[\mathcal{M}_{\mathrm{I}}(x_m) = \mathcal{A}\right]}$
\begin{equation}
\epsilon_{x_n,x_m,\mathcal{A}} = \frac{\ln \left(\frac{\Pr\left[\mathcal{M}_{\mathrm{I}}(x_n) = \mathcal{A}\right]}{\Pr\left[\mathcal{M}_{\mathrm{I}}(x_m) = \mathcal{A}\right]}\right)}{d_{x_n,x_m}}.
\end{equation}}

\vspace{-0.10in}
\begin{lemma}
\label{lem:PhI_privacycost}
For each pair of neighbors $x_n, x_m \in \mathcal{X}$, $\mathcal{M}_{\mathrm{I}}(x_n) \stackrel{(\overline{\epsilon}_{x_n,x_m},0)}{\sim}  \mathcal{M}_{\mathrm{I}}(x_m)$. 
\end{lemma}
\vspace{-0.10in}
{\rev \begin{proof}[Proof Sketch]
% The detailed proof can be found in \textbf{Section \ref{subsec:proof:prop:M2budget} in Appendix}. 
We show that $\mathcal{M}_{\mathrm{I}}$ satisfies $(\overline{\epsilon}_{x_n,x_m}, 0)$-PmDP for any $x_n, x_m \in \mathcal{X}$ by bounding the log-likelihood ratio over all anchor sets $\mathcal{A} \subseteq \mathcal{X}$. Since $\overline{\epsilon}_{x_n,x_m}$ is defined as the worst-case privacy leakage, the output distributions differ by at most a factor of $e^{\overline{\epsilon}_{x_n,x_m} d_{x_n,x_m}}$, ensuring pure PmDP with $\delta = 0$.
The full proof is in \textbf{Appendix~\ref{subsec:proof:lem:PhI_privacycost}}.
\end{proof}}

\vspace{-0.05in}
\subsubsection{Heuristic anchor selection algorithms.} Before designing the heuristic anchor selection algorithms, we first introduce two \emph{design criteria} to follow.  
\begin{proposition}
\label{prop:wsmall}
To ensure that $\mathcal{M}_{\mathrm{I}}(x_n) \stackrel{\epsilon}{\approx} \mathcal{M}_{\mathrm{I}}(x_m)$ for any finite $\epsilon > 0$, it is necessary that if $w_{x_n,x} = 1$, then $w_{x_m,x} = 1$ must also hold.
\end{proposition}
{\rev \begin{proof}[Proof Sketch]
%This proposition establishes a necessary condition for ensuring that the anchor selection mechanism $\mathcal{M}_{\mathrm{I}}$ satisfies finite $\epsilon$-mDP. 
If an anchor $x$ is always selected by $x_n$ but not by $x_m$, then $\mathcal{M}_{\mathrm{I}}(x_n)$ and $\mathcal{M}_{\mathrm{I}}(x_m)$ assign zero and nonzero probabilities to sets excluding $x$, leading to an infinite likelihood ratio and violating mDP. To ensure bounded privacy loss, any always-selected anchor must be shared by both.
See \textbf{Appendix~\ref{subsec:proof:prop:wsmall}} for the full proof. \looseness = -1
\end{proof}}

Proposition~\ref{prop:wsmall} indicates that if a true location $x_n$ selects an anchor $x$ with probability 1 (i.e., $w_{x_n,x} = 1$), then to maintain a bounded privacy cost under mDP, every other record $x_m \in \mathcal{X}$ must also select $x$ with probability 1 (i.e., $w_{x_m,x} = 1$). This requirement contradicts \textbf{Property (a)} of anchor selection, which emphasizes that anchors in $\mathcal{A}_n$ should be locally relevant to the user's true record $x_n$. Forcing all users to deterministically select the same anchor disregards both locality and utility considerations. 

Therefore, \emph{\textbf{Criterion 1}, suggested by Proposition~\ref{prop:wsmall}, is to require each anchor to be selected in a probabilistic manner---i.e., $w_{x_n,x} < 1$ for all $x \in \mathcal{X}$, including the case $x = x_n$.}

\begin{table}[t]
\caption{Anchor record selection methods. }
\vspace{-0.10in}
\label{Tb:anchorrecordselection}
\centering
\normalsize 
\small 
\begin{tabular}{l|l}
\toprule
Selection method                  & Probability density function $h(d_{x_n,x})$ \\
\hline
\hline
Exponential decay           & $h(d_{x_n,x}) = \alpha e^{-\lambda d_{x_n, x}}$ 
\\ 
Modified power law decay          & $h(d_{x_n,x}) = 
   \frac{\alpha}{1+d_{x_n,x}^{\lambda}}$  \\ 
Logistic function          & $h(d_{x_n,x}) = 
   \frac{2\alpha}{1+e^{\lambda d_{x_n,x}}}$  \\ 
\hline
\end{tabular}
\normalsize
\vspace{-0.10in}
\end{table}

\begin{figure}[t]
\centering
\hspace{0.00in}
\begin{minipage}{0.49\textwidth}
  \subfigure[Exponential decay]{
\includegraphics[width=0.32\textwidth, height = 0.10\textheight]{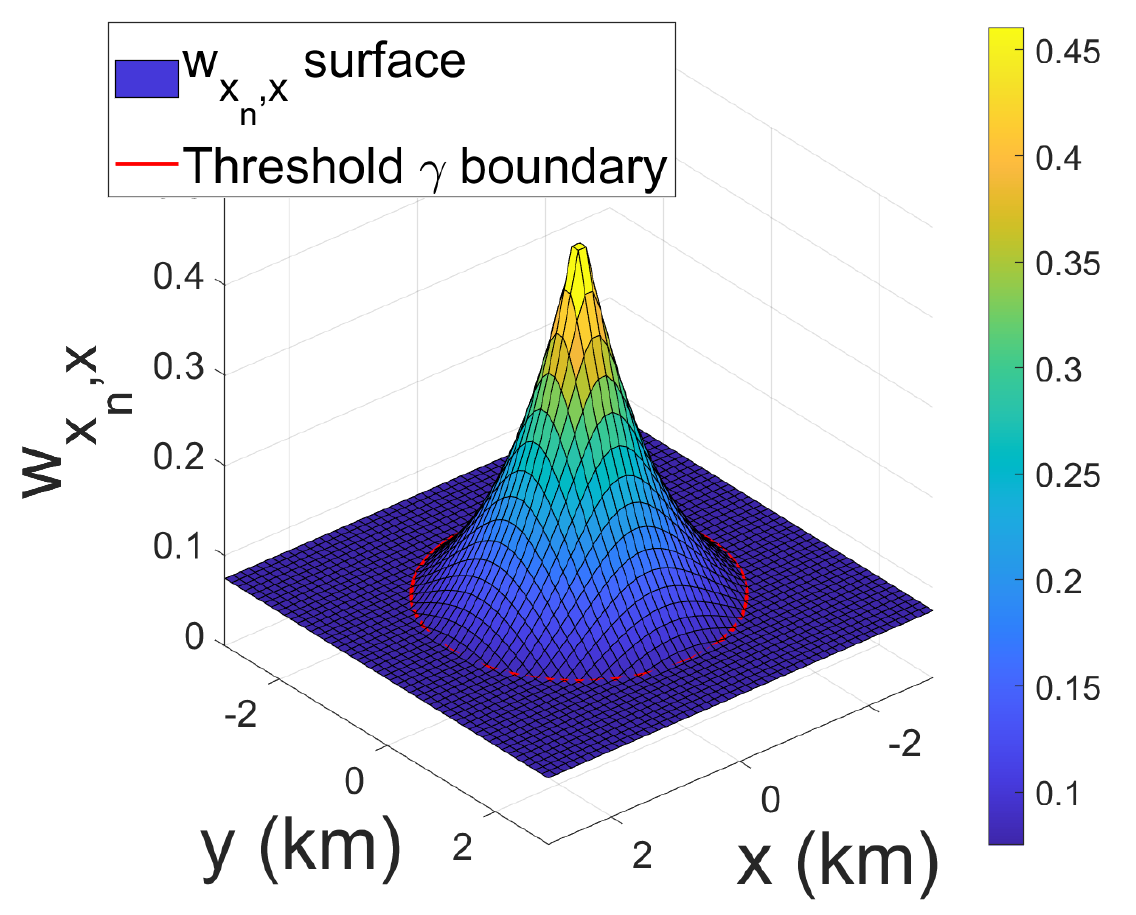}}
  \subfigure[Power law decay]{
\includegraphics[width=0.32\textwidth, height = 0.10\textheight]{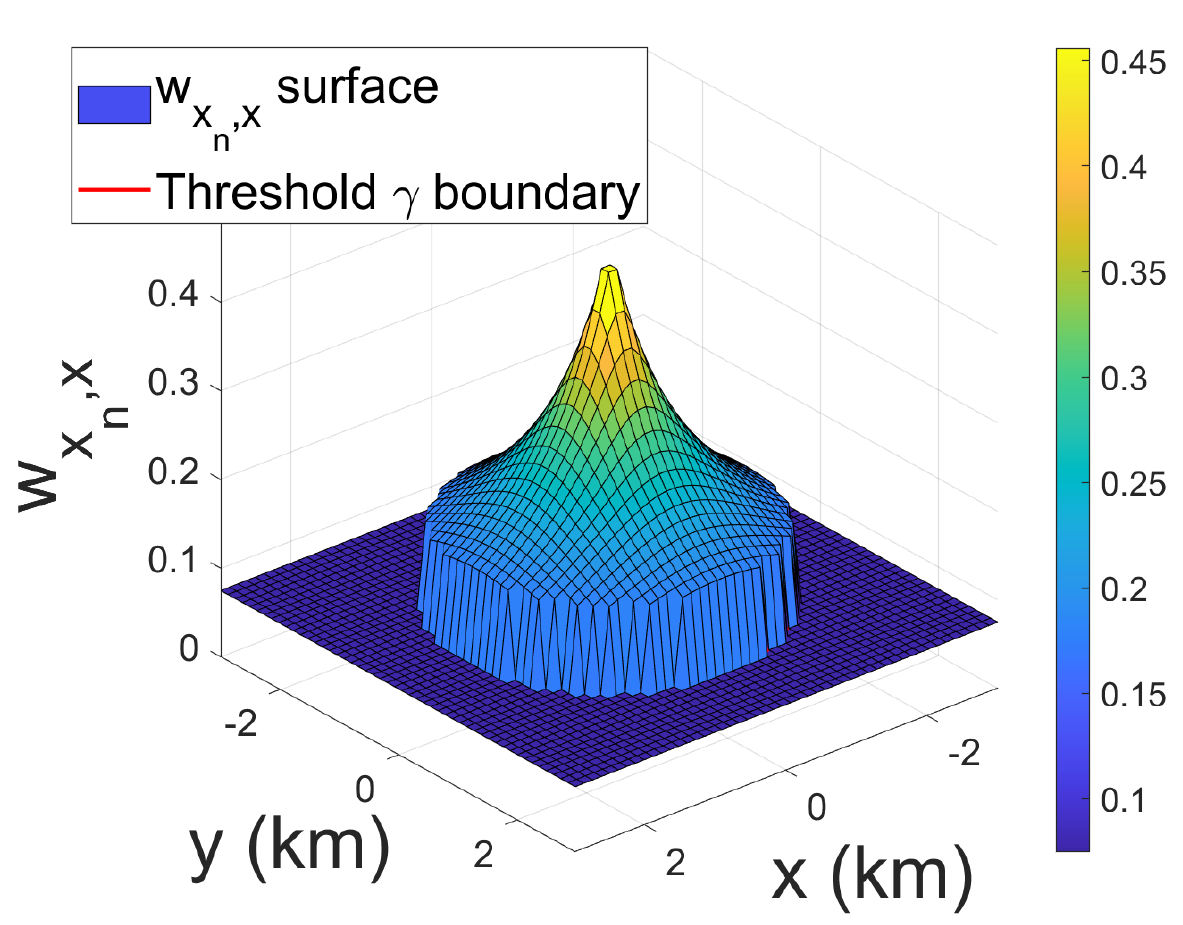}}
  \subfigure[Logistic]{
\includegraphics[width=0.32\textwidth, height = 0.10\textheight]{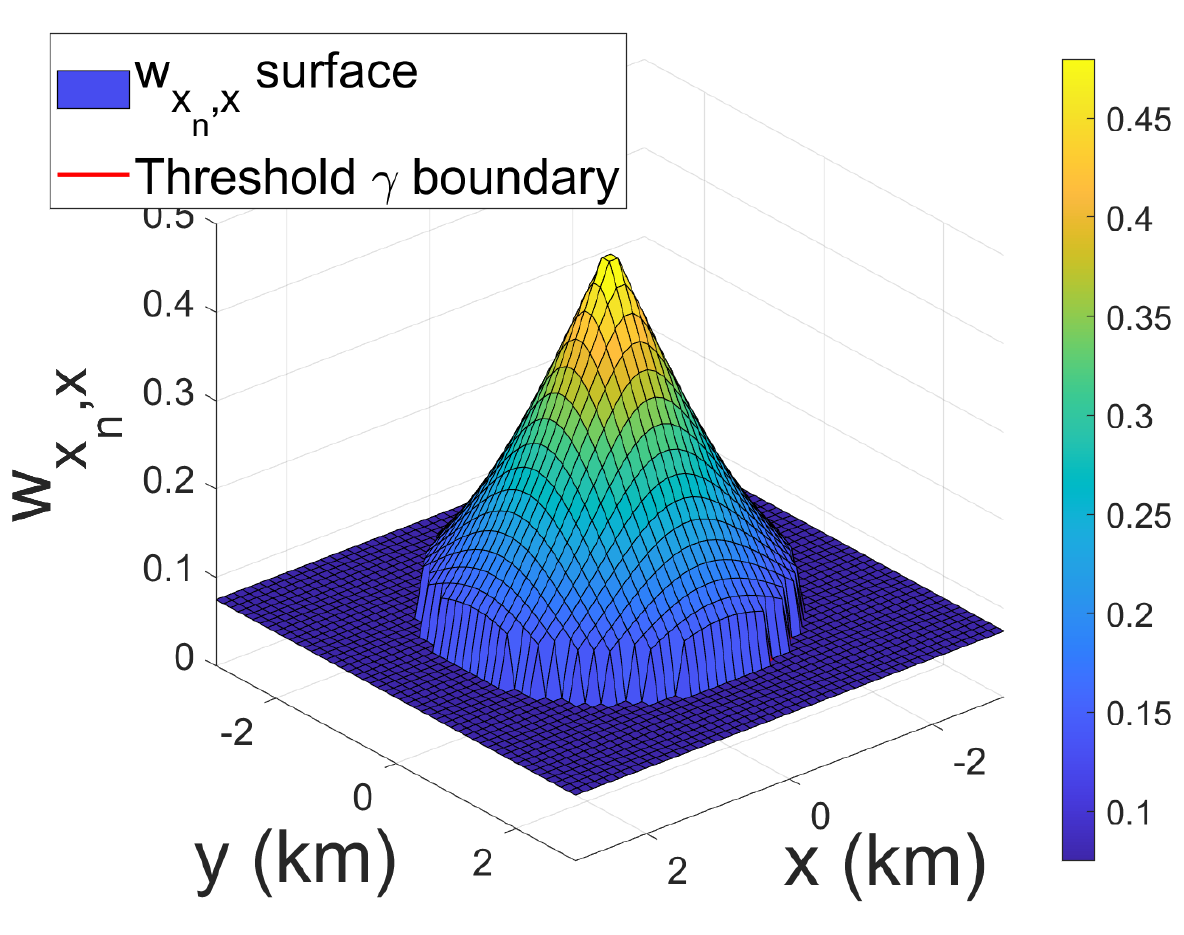}}
\vspace{-0.15in}
\end{minipage}
\caption{Examples of anchor record selection methods {\rev ($\alpha = 0.95$, $\lambda = 0.5$km$^{-1}$, and $\gamma = 2.0$km)}.}
\label{fig:surface}
\vspace{-0.15in}
\end{figure}

\emph{\textbf{Criterion 2} requires each anchor's selection probabiltiy is based on its distance to the true record $x_n$}, for two main reasons:

% \textbf{Accurate Approximation of mDP Constraints.} 
First, under mDP, the strength of the privacy constraint between two records $x_n$ and $x_m$ depends on their distance $d_{x_n,x_m}$. The closer the records, the tighter the constraint. To approximate the full mDP constraint space using a reduced anchor set, it is important to include nearby records, as they contribute the most significant constraints. % Ignoring them may lead to poor approximation or potential violations of mDP.
    
% \textbf{Faithful estimation of utility loss.} 
Second, the utility loss experienced by a user depends on how their true record $x_n$ is perturbed. When the perturbation matrix is constructed over a set of anchors, utility loss is typically estimated using proxy distances (e.g., between a true location and the anchor center). If anchors are selected far from $x_n$, these utility estimates may poorly reflect actual user loss. Distance-based selection ensures that utility evaluations remain relevant and accurate.

% Therefore, distance-aware anchor selection allows the surrogate perturbation matrix to simultaneously preserve critical privacy constraints and offer high utility, aligning with both \textbf{Property (a)} and the overall design goals of \textsc{PAnDA}.

To achieve the above two criteria, we define the selection probability of each candidate anchor $x \in \mathcal{X}$ as a function of its distance to the user's true record $x_n$: $w_{x_n,x} = h(d_{x_n,x})$, where $h(\cdot)$ is a non-increasing function. To prevent selection probabilities for distant records from becoming too small (and thus reducing diversity in anchor sets), we introduce a distance threshold $\gamma$, and define the final selection rule as:
\begin{equation}
\label{eq:w_exp}
w_{x_n,x} = 
\begin{cases}
h(d_{x_n,x}), & \text{if } 0 < d_{x_n,x} < \gamma \\
h(\gamma), & \text{if } d_{x_n,x} \geq \gamma
\end{cases}
\end{equation}
That is, when the distance $d_{x_n,x}$ exceeds the threshold $\gamma$, the selection probability is clipped to a minimum value $h(\gamma)$, ensuring non-negligible selection probability for even distant anchors.

We instantiate $h(d)$ using three widely used decay functions, \emph{exponential decay}, \emph{modified power law decay}, and \emph{logistic function}, as described in Table \ref{Tb:anchorrecordselection}, each of which allows us to control the selection sharpness through a decay parameter $\lambda > 0$, and ensures probabilistic selection with a scaling parameter $\alpha < 1$. 

Intuitively, \emph{exponential decay} strongly favors nearby records and rapidly suppresses the selection probability for distant ones, resulting in highly localized anchor sets. \emph{power-law decay} decays more slowly and assigns non-negligible probabilities to moderately distant anchors, thus promoting more diverse anchor coverage. \emph{logistic decay} adopts a sigmoid form, offering a smoother and more gradual decay than the exponential function while still preserving locality. These three methods provide flexible control over the trade-off between local precision and global coverage, with the decay parameter $\lambda$ governing the sharpness of the decay.

% These functions provide flexible control over how sharply anchor selection probability decays with distance, and allow for empirical tuning to balance utility and privacy in different application settings.

%\noindent \textbf{(2) Modified power law decay}. 
%\begin{equation}
%\label{eq:w_power}
%   h(d_{x_n,x}) = \frac{\alpha}{1+d_{x_n,x}^{\lambda}}
%\end{equation}

%\noindent \textbf{(3) Logistic function}. 
%\begin{equation}
%\label{eq:w_log}
%   h(d_{x_n,x}) = \frac{2\alpha}{1+e^{\lambda d_{x_n,x}}}
%\end{equation}

\DEL{
After deriving $w_{n,1}, ..., w_{n,K}$ using Eq. (\ref{eq:w_exp}), we scale each $w_{x_n,x_\ell}$ with $d_{x_n,x_\ell}>\eta$ by $\frac{\Gamma - |\mathcal{N}_n|}{\sum_{x_\ell \notin \mathcal{N}_n} w_{x_n,x_\ell}}$, i.e., 
\begin{equation}
\label{eq:w_scale}
\overline{w}_{n,l} = \frac{\Gamma - |\mathcal{N}_n|}{\sum_{k} w_{x_n,x_\ell}} w_{x_n,x_\ell},
\end{equation}
and select each record $x_\ell$ with probability $\overline{w}_{n,l}$, so that the expected number of anchor records of each real record $x_n$ is $\Gamma$.}

\vspace{-0.05in}
\subsubsection{Discussion: Connection to Coarse-Grained Methods.} Coarse-grained data perturbation methods \cite{Bordenabe-CCS2014, Wang-WWW2017} that approximate user locations to fixed grid centers can be viewed as a special case of anchor-based approximation. In these methods, each location is deterministically mapped to the center of a predefined grid cell, which effectively serves as a surrogate anchor. However, unlike \textsc{PAnDA}’s probabilistic anchor selection, these surrogates are chosen in a deterministic and server-defined manner, independent of the user’s local data.

As a result, such approaches do not require a two-phase framework. Since anchor (grid center) assignment is fixed and globally known, there is no need for a Phase~I anchor selection process. The entire mechanism reduces to a single-phase perturbation strategy performed directly over the predefined anchors, at the cost of reduced flexibility, higher mDP violation ratio, and higher utility loss, which is demonstrated in the experimental results in Tables \ref{Tb:exp:ULscalability}\&\ref{Tb:exp:ULscalability_realdata} and Fig. \ref{fig:mDPfailture}(a)(b)(c).

\DEL{
\begin{figure}[t]
\centering
\hspace{0.00in}
\begin{minipage}{0.42\textwidth}
  \subfigure{
\includegraphics[width=1.00\textwidth]{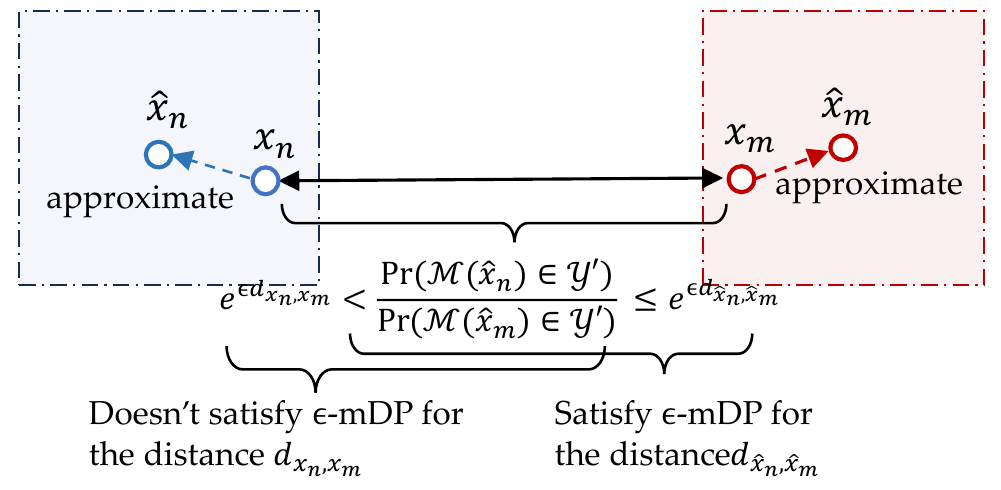}}
\vspace{-0.10in}
\end{minipage}
\caption{mDP violation of coarse-grained data perturbation.}
\label{fig:DPlimit}
\vspace{-0.10in}
\end{figure}}

\vspace{-0.05in}
\subsection{Phase II: Anchor Perturbation {\rev Optimization}}
\label{subsec:PhaseII}

\begin{figure}[t]
\centering
\hspace{0.00in}
\begin{minipage}{0.48\textwidth}
  \subfigure[$d_{\hat{x}_n, \hat{x}_m} \leq \frac{(\epsilon - \overline{\epsilon}_{x_n,x_m}) d_{x_n,x_m}}{\epsilon - \overline{\epsilon}_{\hat{x}_n,\hat{x}_m}}$: the surrogate perturbation vectors $\mathbf{z}_{\hat{x}_n}$ and $\mathbf{z}_{\hat{x}_m}$ satisfy tighter mDP constraints for $x_n$ and $x_m$]{
\includegraphics[width=0.48\textwidth]{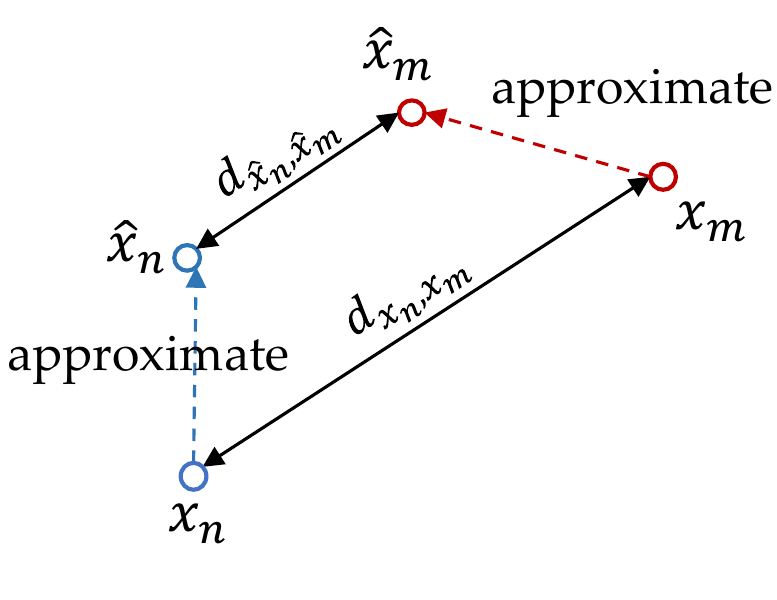}}
\hspace{0.05in}
  \subfigure[$d_{\hat{x}_n, \hat{x}_m}>\frac{(\epsilon - \overline{\epsilon}_{x_n,x_m}) d_{x_n,x_m}}{\epsilon - \overline{\epsilon}_{\hat{x}_n,\hat{x}_m}}$: the surrogate perturbation vectors $\mathbf{z}_{\hat{x}_n}$ and $\mathbf{z}_{\hat{x}_m}$ satisfy {\rev looser} mDP constraints for $x_n$ and $x_m$]{
\includegraphics[width=0.46\textwidth]{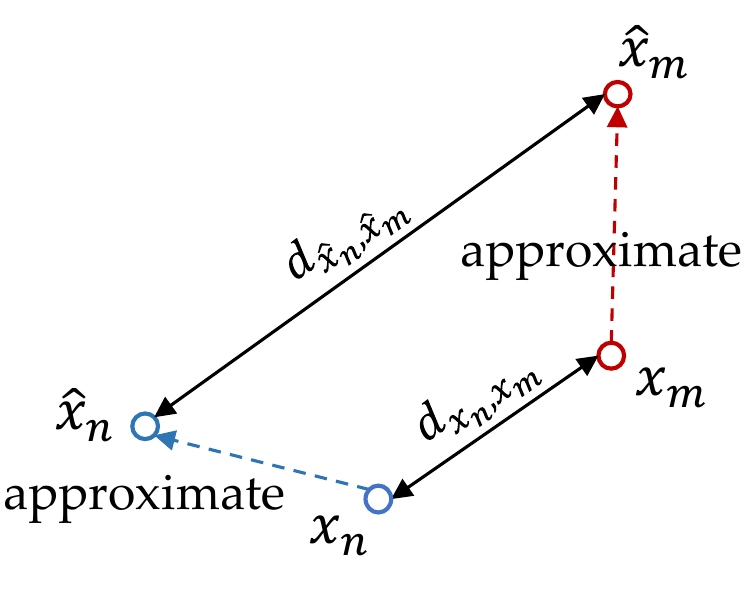}}
\vspace{-0.15in}
\end{minipage}
\caption{When the distance $d_{\hat{x}_n, \hat{x}_m}$ between surrogates is different from the real distance $d_{x_n, x_m}$.}
\label{fig:anchorviolate}
\vspace{-0.25in}
\end{figure}

In this part, we present the details of the randomized mechanism $\mathcal{M}_{\mathrm{II}}$ in Phase II, whose objective is to optimize the anchor perturbation for each user while ensuring that the privacy cost of the perturbed record does not exceed the remaining privacy budget.

Recall that after Phase I, each user $n$ (where $n = 1, \ldots, N$) uploads their corresponding anchor record $\mathcal{A}_{n}$ to the server, while keeping their private record $x_n$ hidden. Upon receiving $\mathcal{A}_{[1, N]} = \bigcup_{n=1}^N \mathcal{A}_{n}$ from the $N$ users, the server optimizes the joint anchor perturbation matrix $\mathbf{Z}_{\mathcal{A}_{[1, N]}} = \{ z_{x,y} \}_{(x,y)\in \mathcal{A}_{[1, N]}\times \mathcal{Y}}$, which is a submatrix of the full perturbation matrix $\mathbf{Z}_{\mathcal{X}}$ containing only the rows corresponding to the anchor records. Each user then downloads their corresponding submatrix $\mathbf{Z}_{\mathcal{A}_n}$, identifies the nearest anchor $\hat{x}_n \in \mathcal{A}_n$ to their true record $x_n$, that is, $\hat{x}_n = \arg\min_{x \in \mathcal{A}_n} d_{x_n, x}$, and approximates $\mathbf{z}_{x_n}$ by $\mathbf{z}_{\hat{x}_n}$. Here, we refer to the nearest anchor $\hat{x}_n$ as the \emph{surrogate} of the true record $x_n$.

\begin{proposition}
\label{prop:M2budget}
According to the sequential composition property of PmDP (\textbf{Theorem~\ref{thm:composition}}), to ensure
$\mathcal{M}_{\mathrm{\textsc{PAnDA}}}(x_n) \stackrel{(\epsilon,\delta)}{\sim} \mathcal{M}_{\mathrm{\textsc{PAnDA}}}(x_m)$, $\forall x_n, x_m \in \mathcal{X}$, 
it suffices to enforce
$\mathcal{M}_{\mathrm{II}}(x_n) \stackrel{(\epsilon - \overline{\epsilon}_{x_n,x_m},\delta)}{\sim} \mathcal{M}_{\mathrm{II}}(x_m)$, $\forall x_n, x_m \in \mathcal{X}$. 
\end{proposition}
\vspace{-0.10in}
{\rev 
\begin{proof}[Proof Sketch]
% This proposition applies the sequential composition property of probabilistic PmDP to the two-phase structure of \textsc{PAnDA}.  The idea is to split the total privacy budget $\epsilon$ across the two mechanisms: $\mathcal{M}_{\mathrm{I}}$ and $\mathcal{M}_{\mathrm{II}}$. From Lemma~\ref{lem:PhI_privacycost}, we know the privacy cost of $\mathcal{M}_{\mathrm{I}}$ is at most $\overline{\epsilon}_{x_n,x_m}$ with failure probability $0$.  To ensure the total cost across both mechanisms remains within the target $(\epsilon,\delta)$-PmDP guarantee, we allocate the remaining budget $\epsilon - \overline{\epsilon}_{x_n,x_m}$ (with failure $\delta$) to $\mathcal{M}_{\mathrm{II}}$. By invoking Theorem~\ref{thm:composition}, this composition preserves the $(\epsilon,\delta)$-PmDP guarantee for the overall mechanism $\mathcal{M}_{\mathrm{\textsc{PAnDA}}}$. The detailed proof can be found in \textbf{Appendix \ref{subsec:proof:prop:M2budget}}. 
This proposition applies sequential composition of probabilistic PmDP to \textsc{PAnDA}'s two-phase structure. 
The total budget $\epsilon$ is split between $\mathcal{M}_{\mathrm{I}}$ and $\mathcal{M}_{\mathrm{II}}$: 
from Lemma~\ref{lem:PhI_privacycost}, $\mathcal{M}_{\mathrm{I}}$ costs at most $\overline{\epsilon}_{x_n,x_m}$ (with failure $0$), 
so the remaining $\epsilon - \overline{\epsilon}_{x_n,x_m}$ (with failure $\delta$) is assigned to $\mathcal{M}_{\mathrm{II}}$. 
By Theorem~\ref{thm:composition}, the overall mechanism $\mathcal{M}_{\textsc{PAnDA}}$ satisfies $(\epsilon,\delta)$-PmDP. 
The full proof appears in \textbf{Appendix~\ref{subsec:proof:prop:M2budget}}.
\end{proof}}
\vspace{-0.10in}
Based on Proposition~\ref{prop:M2budget}, the remaining privacy budget available in Phase II for any pair $(x_n, x_m)$ is $\epsilon - \overline{\epsilon}_{x_n,x_m}$, and we allow a failure probability of $\delta$.

\subsubsection{Analysis of mDP Violation Due to Surrogate Perturbation.} Note that for each pairs of users $n$ and $m$, the server does not know which anchor serves as the surrogate for the users' true records $x_n$ and $x_m$. Therefore, it enforces the mDP constraints for each pair of anchor records $(x, x') \in \mathcal{A}_{n} \times \mathcal{A}_{m}$. 

Note that if the server allocates a privacy leakage bound of $(\epsilon - \overline{\epsilon}_{x,x'})d_{x,x'}$ to each anchor pair $x$ and $x'$ based on their distance, it might cause potential mDP violation as the approximated distance value $d_{x,x'}$ can be different from the real distance between $x_n$ and $x_m$. % As illustrated by Fig. \ref{fig:anchorviolate}, 
\DEL{as shown in Fig. \ref{fig:anchorviolate}(a), if $d_{\hat{x}_n, \hat{x}_m} \leq \frac{(\epsilon - \overline{\epsilon}_{x_n,x_m}) d_{x_n,x_m}}{\epsilon - \overline{\epsilon}_{\hat{x}_n,\hat{x}_m}}$
we can derive that 
\begin{equation}
\underbrace{\overline{\epsilon}_{x_n,x_m} d_{x_n,x_m}}_{\text{\footnotesize Max privacy cost in Phase I}} + \underbrace{(\epsilon - \overline{\epsilon}_{\hat{x}_n,\hat{x}_m}) d_{\hat{x}_n,\hat{x}_m}}_{\text{\footnotesize Privacy cost in Phase II}} \leq
\underbrace{\epsilon d_{x_n,x_m}}_{\text{\footnotesize Privacy budget}}.
\end{equation}
\begin{equation}
\label{eq:mDPPh-II-1}
z_{x,y} - e^{(\epsilon - \overline{\epsilon}_{x,x'}) d_{x,x'}} z_{x',y} \leq 0,
\end{equation}
then the total privacy cost between the surrogates $\hat{x}_n$ and $\hat{x}_m$ is possibly different from the privacy cost between the real records $x_n$ and $x_m$, as illustrated by Fig. \ref{fig:anchorviolate}. }
Specifically, as shown in Fig. \ref{fig:anchorviolate}(b), the surrogate perturbation vectors applied to
the true records $x_n$ and $x_m$ may exceed the overall privacy budget when $d_{\hat{x}_n,\hat{x}_m} > \frac{(\epsilon - \epsilon_{x_n,x_m,\mathcal{A}}) d_{x_n,x_m}}{\epsilon - \overline{\epsilon}_{\hat{x}_n,\hat{x}_m}}$, {\rev {because}
\begin{equation}
\overline{\epsilon}_{x_n,x_m} \geq \epsilon_{x_n,x_m,\mathcal{A}} \quad \text{(by the definition of $\overline{\epsilon}_{x_n,x_m}$ in Eq.~(\ref{eq:epsilon_overline1})(\ref{eq:epsilon_overline2}))},
\end{equation}
and we have the following derivation (note that $\overline{\epsilon}_{\hat{x}_n,\hat{x}_m} < \epsilon$):
\begin{align}
& d_{\hat{x}_n,\hat{x}_m} > \frac{(\epsilon - \epsilon_{x_n,x_m,\mathcal{A}}) d_{x_n,x_m}}{\epsilon - \overline{\epsilon}_{\hat{x}_n,\hat{x}_m}} \\
% \Rightarrow\quad & d_{\hat{x}_n,\hat{x}_m} \left(\epsilon - \overline{\epsilon}_{\hat{x}_n,\hat{x}_m} \right) > (\epsilon - \epsilon_{x_n,x_m,\mathcal{A}}) d_{x_n,x_m} \\
\Rightarrow\quad & \epsilon_{x_n,x_m,\mathcal{A}} d_{x_n,x_m} + d_{\hat{x}_n,\hat{x}_m} \left( \epsilon - \overline{\epsilon}_{\hat{x}_n,\hat{x}_m} \right) > \epsilon d_{x_n,x_m} \\
\Rightarrow\quad & \underbrace{\overline{\epsilon}_{x_n,x_m}  d_{x_n,x_m}}_{\text{\scriptsize Max privacy cost in Phase I}} + \underbrace{(\epsilon - \overline{\epsilon}_{\hat{x}_n,\hat{x}_m})  d_{\hat{x}_n,\hat{x}_m}}_{\text{\scriptsize Privacy cost in Phase II}} > \underbrace{\epsilon  d_{x_n,x_m}}_{\text{\scriptsize Total privacy budget}}.
\end{align}
% \underbrace{\epsilon_{x_n,x_m,\mathcal{A}} d_{x_n,x_m}}_{\text{\footnotesize Privacy cost in Phase I}} + \underbrace{(\epsilon - \overline{\epsilon}_{\hat{x}_n,\hat{x}_m}) d_{\hat{x}_n,\hat{x}_m}}_{\text{\footnotesize Privacy cost in Phase II}} > \underbrace{\epsilon d_{x_n,x_m}}_{\text{\footnotesize Total privacy budget}}.
}
\vspace{-0.15in}

\subsubsection{Privacy Safety Margin $\xi$ To Guarantee mDP Success Ratio.} 
\label{subsubsec:safetymargin}
To address this issue, we introduce a \emph{privacy safety margin} $\xi_{x_n, x_m}$ for each pair of anchors $(x, x') \in \mathcal{A}_n \times \mathcal{A}_m$, allowing the server to enforce the following constraint:
\begin{equation}
\label{eq:mDPPh-II-1-revised}
z_{x,y} -  e^{(\epsilon - \overline{\epsilon}_{x, x'}) d_{x, x'} - \xi_{x_n,x_m}}z_{x',y} \leq 0.
\end{equation}
This ensures that, as long as the distance between a pair of anchors $(x, x')$, where the true surrogates $(\hat{x}_n, \hat{x}_m)$ are guaranteed to be included in $\mathcal{A}_n$ and $\mathcal{A}_m$, satisfies:
\begin{equation}
d_{x,x'} \leq \frac{(\epsilon - \overline{\epsilon}_{x_n,x_m}) d_{x_n,x_m} + \xi_{x_n,x_m}}{\epsilon - \overline{\epsilon}_{x,x'}},
\end{equation}
the $\epsilon$-mDP constraint is satisfied:
\begin{equation}
\label{eq:mDP_xi}
\underbrace{\overline{\epsilon}_{x_n,x_m} d_{x_n,x_m}}_{\text{\footnotesize Max privacy cost in Phase I}} + 
\underbrace{(\epsilon - \overline{\epsilon}_{x,x'}) d_{x,x'} - \xi_{x_n,x_m}}_{\text{\footnotesize Privacy cost in Phase II}} 
\leq 
\underbrace{\epsilon d_{x_n,x_m}}_{\text{\footnotesize Privacy budget}}.
\end{equation}

Ideally, the safety margin $\xi_{x_n, x_m}$ only needs to apply to the true surrogates $(\hat{x}_n, \hat{x}_m)$. However, since the server does not observe the surrogates directly and only has access to the anchor sets $\mathcal{A}_n$ and $\mathcal{A}_m$, it must apply the margin to anchor pairs instead. While the surrogates are included in these anchor sets by design, strictly enforcing Eq.~(\ref{eq:mDP_xi}) for \emph{all} pairs $(x, x') \in \mathcal{A}_n \times \mathcal{A}_m$ would require a conservatively large safety margin $\xi_{x_n, x_m}$, which could significantly reduce the remaining privacy budget, i.e., $(\epsilon - \overline{\epsilon}_{x,x'}) d_{x,x'} - \xi_{x_n, x_m}$, available for optimizing the perturbation probabilities $z_{x,y}$ and $z_{x',y}$. Fortunately, since the mDP definition allows for a small failure probability $\delta$, it is sufficient to enforce the constraint for only a subset of anchor pairs with high probability, thereby allowing more budget to be allocated toward utility optimization.

In the following, we analyze how to choose $\xi_{x_n, x_m}$ to ensure compliance with $(\epsilon, \delta)$-PmDP. % demonstrating that sufficient privacy budget can still be maintained in Phase II while satisfying the relaxed guarantee.

\DEL{
\begin{eqnarray}
&& \underbrace{e^{\epsilon_{x_n,x_m,\mathcal{A}}d_{x_n,x_m}}}_{\mbox{\footnotesize Bound in Phase-I}}  \times \underbrace{e^{(\epsilon-\overline{\epsilon}_{\hat{x}_n,\hat{x}_m}) d_{\hat{x}_n,\hat{x}_m}-\xi}}_{\mbox{\footnotesize Bound in Phase-II}} \\
&\leq& e^{\epsilon_{x_n,x_m,\mathcal{A}}d_{x_n,x_m}}  \times e^{(\epsilon-\overline{\epsilon}_{x_n,x_m}) d_{x_n,x_m}}\\ 
% &=& e^{\epsilon_{x_n,x_m,\mathcal{A}}d_{x_n,x_m} + (\epsilon-\overline{\epsilon}_{x_n,x_m}) d_{x_n,x_m}-\xi} \\
&\leq& e^{\epsilon d_{x_n,x_m}}.
\end{eqnarray}}

\subsubsection{Safety Margin Searching Algorithm.} 
First, for each pair of records $x_n, x_m \in \mathcal{X}$, we compute the failure probability of the mDP constraints between $x_n$ and $x_m$ given a safety margin value $\xi$. Specifically, we define $\mathcal{H}_{x_n,x_m}(\xi)$ as the set of anchor record pairs $(x, x')$ that satisfy the condition:
 $d_{x,x'} \leq \frac{(\epsilon - \overline{\epsilon}_{x_n,x_m})d_{x_n,x_m}+\xi}{\epsilon-\overline{\epsilon}_{x,x'}}$
\DEL{
\begin{eqnarray}
&& \mathcal{H}_{x_n,x_m}\left(\xi\right) \\ 
&=& \left\{\left(\hat{x}_{n}, \hat{x}_{m}\right) \in \mathcal{X}^2\left|\begin{array}{ll}&(\epsilon-\overline{\epsilon}_{\hat{x}_n,\hat{x}_m}) d_{\hat{x}_n,\hat{x}_m}-\xi \\ 
    \leq & (\epsilon - \overline{\epsilon}_{x_n,x_m}) d_{x_n,x_m}\end{array}\right.\right\},
\end{eqnarray}}
\begin{eqnarray}
\label{eq:H}
&& \mathcal{H}_{x_n,x_m}\left(\xi\right) \\ 
&=& \left\{\left(x, x'\right) \in \mathcal{X}^2\left|\begin{array}{l}d_{x,x'} \leq \frac{(\epsilon - \overline{\epsilon}_{x_n,x_m})d_{x_n,x_m}+\xi}{\epsilon-\overline{\epsilon}_{x,x'}}
 \\  
 d_{x_n,x} < \gamma,  d_{x_m,x'} < \gamma
\end{array}\right.\right\}
\end{eqnarray}
in which each $\left(x, x'\right)$ is \emph{sufficient} to achieve $\epsilon$-mDP between $\left(x_{n}, x_{m}\right)$ (as demonstrated by Eq. (\ref{eq:mDP_xi})). 
\begin{proposition}
\label{prop:failuremDPprob}
For each pair of real records $x_n$ and $x_m$, the probability of success is lower bounded by 
\begin{eqnarray}
\textstyle h_{x_n,x_m}(\xi) = \sum_{\left(\hat{x}_{n},\hat{x}_{m}\right) \in \mathcal{H}_{x_n,x_m}(\xi)}v_{x_{n}, \hat{x}_{n}} v_{x_{m}, \hat{x}_{m}}
\end{eqnarray}
where each $v_{x_{q}, \hat{x}_{q}} = w_{x_q,\hat{x}_q} \prod_{x \in \mathcal{U}_{x_q,\hat{x}_q}}(1-w_{x_q,x}) (q = n, m)$ represents the probability that $\hat{x}_{q}$ is $x_{q}$'s surrogate and $$\mathcal{U}_{x_{q},\hat{x}_q} = \left\{x\in \mathcal{X}| d_{x_{q},x} < d_{x_{q},\hat{x}_q}\right\}$$ represents the set of records that is closer to $x_{q}$ compared to $\hat{x}_q$.
\end{proposition}
\vspace{-0.10in}
{\rev 
\begin{proof}[Proof Sketch]
Privacy holds if the chosen surrogates $(\hat{x}_n,\hat{x}_m)$ lie in $\mathcal{H}_{x_n,x_m}(\xi)$, which enforces tightened constraints. 
Their selection probability is the product of being included in the anchor set and not superseded by a closer record. 
Summing over valid pairs gives the lower bound $h_{x_n,x_m}(\xi)$, showing that larger $\xi$ increases the likelihood of satisfying mDP. 
The full proof is in \textbf{Appendix~\ref{subsec:proof:prop:failuremDPprob}}.
\end{proof}}

\DEL{
{\bl When mDP is satisfied:  
\begin{eqnarray}
&& \underbrace{e^{\epsilon_{x_n,x_m,\mathcal{A}}d_{x_n,x_m}}}_{\mbox{\footnotesize Bound in Phase-I}}  \times \underbrace{e^{(\epsilon-\overline{\epsilon}_{\hat{x}_n,\hat{x}_m}) d_{\hat{x}_n,\hat{x}_m}-\xi}}_{\mbox{\footnotesize Bound in Phase-II}} \leq e^{\epsilon d_{x_n,x_m}} \\
&\Leftrightarrow& \epsilon_{x_n,x_m,\mathcal{A}}d_{x_n,x_m} + (\epsilon-\epsilon_{\hat{n},\hat{m},\mathcal{A}}) d_{\hat{x}_n,\hat{x}_m}-\xi \leq \epsilon d_{x_n,x_m} \\
&\Leftrightarrow& \epsilon d_{\hat{x}_n,\hat{x}_m} -\epsilon_{\hat{n},\hat{m},\mathcal{A}}d_{\hat{x}_n,\hat{x}_m} -\xi \leq \epsilon d_{x_n,x_m} - \epsilon_{x_n,x_m,\mathcal{A}}d_{x_n,x_m} \\
&\Leftrightarrow& \epsilon (d_{\hat{x}_n,\hat{x}_m} - d_{x_n,x_m})-\xi \leq  \epsilon_{\hat{n},\hat{m},\mathcal{A}}d_{\hat{x}_n,\hat{x}_m} - \epsilon_{x_n,x_m,\mathcal{A}}d_{x_n,x_m} \\
&\Leftrightarrow& \epsilon (d_{\hat{x}_n,\hat{x}_m} - d_{x_n,x_m})-\xi \\
&\leq&  \sum_{x_\ell \in \mathcal{A}} \ln \left(\frac{w_{\hat{n},l}}{w_{\hat{m},l}}\right)+\sum_{x_\ell \notin \mathcal{A}} \ln \left(\frac{1-w_{\hat{n},l}}{1-w_{\hat{m},l}}\right) \\
&-&  \sum_{x_\ell \in \mathcal{A}} \ln \left(\frac{w_{x_n,x_\ell}}{w_{x_m,x_\ell}}\right) - \sum_{x_\ell \notin \mathcal{A}} \ln \left(\frac{1-w_{x_n,x_\ell}}{1-w_{x_m,x_\ell}}\right) \\
&=& \sum_{x_\ell \in \mathcal{A}} \ln \left(\frac{w_{\hat{n},l}w_{x_m,x_\ell}}{w_{\hat{m},l}w_{x_n,x_\ell}}\right)+\sum_{x_\ell \notin \mathcal{A}} \ln \left(\frac{(1-w_{\hat{n},l})(1-w_{x_m,x_\ell})}{(1-w_{\hat{m},l})(1-w_{x_n,x_\ell})}\right)
\end{eqnarray}}}

\vspace{-0.10in}
\begin{property}
\label{prop:xi}
(1) $h_{x_n,x_m}(\xi)$ is \emph{monotonically non-decreasing} in $\xi$, and (2) by setting  $\xi_{x_n, x_m} = h_{x_n,x_m}^{-1}(1-\delta)$
we can guarantee $(\epsilon, \delta)$-PmDP for $(x_n, x_m)$. 
\end{property}
\vspace{-0.10in}
{\rev \begin{proof}[Proof Sketch] 
% \emph{(1) Monotonicity.} The success probability function $h_{x_n,x_m}(\xi)$ is defined as the cumulative probability of selecting surrogate anchor pairs that satisfy a relaxed mDP constraint parameterized by $\xi$. As $\xi$ increases, the set $\mathcal{H}_{x_n,x_m}(\xi)$—which contains valid anchor pairs—only grows larger because the constraint becomes looser. Formally, if $\xi_1 < \xi_2$, then $\mathcal{H}_{x_n,x_m}(\xi_1) \subseteq \mathcal{H}_{x_n,x_m}(\xi_2)$. Since $h(\xi)$ sums only non-negative terms over this set, the function is monotonically non-decreasing in $\xi$. \emph{(2) Guaranteeing $(\epsilon, \delta)$-PmDP.} Because $h_{x_n,x_m}(\xi)$ increases with $\xi$, we can select the smallest value $\xi_{x_n,x_m}$ such that the cumulative probability reaches $1 - \delta$. That is, by solving $h_{x_n,x_m}(\xi_{x_n,x_m}) = 1 - \delta$, we ensure that with probability at least $1 - \delta$, the surrogate pair satisfies the privacy constraint. Hence, applying this $\xi_{x_n,x_m}$ in the anchor-based optimization guarantees $(\epsilon, \delta)$-PmDP for the pair $(x_n, x_m)$. The detailed proof can be found in \textbf{Appendix \ref{subsec:proof:prop:xi}}. 
(1) The success probability $h_{x_n,x_m}(\xi)$ is monotonically non-decreasing since the relaxed set $\mathcal{H}_{x_n,x_m}(\xi)$ grows with $\xi$, and $h_{x_n,x_m}(\xi)$ sums non-negative terms. 
(2) Choosing the smallest $\xi_{x_n,x_m}$ with $h_{x_n,x_m}(\xi_{x_n,x_m}) = 1 - \delta$ guarantees the surrogate pair satisfies the $\epsilon$-mDP constraint with probability at least $1 - \delta$, yielding $(\epsilon,\delta)$-PmDP. 
The full proof is in \textbf{Appendix~\ref{subsec:proof:prop:xi}}.
\end{proof}}
\vspace{-0.10in}
According to Proposition \ref{prop:xi}, when the server formulates the mDP constraints between each anchor pair $(x, x') \in \mathcal{A}_{n} \times \mathcal{A}_{m}$, we can reserve the safety margin $\xi_{x_n, x_m} = h_{x_n,x_m}^{-1}(1-\delta)$ to ensure $(\epsilon, \delta)$-PmDP for $(x_n, x_m)$. Specifically, we have:
\begin{equation}
\frac{z_{\hat{x}_{n},y}}{z_{\hat{x}_{m},y}} \leq e^{(\epsilon-\overline{\epsilon}_{\hat{x}_{n},\hat{x}_{m}}) d_{\hat{x}_{n},\hat{x}_{m}} - \overbrace{h_{x_n,x_m}^{-1}(1-\delta)}^{\mbox{\footnotesize safety margin } \xi_{x_n, x_m}}}.
\end{equation}
{\rev However, computing $h^{-1}_{x_n,x_m}(1 - \delta)$ requires knowledge of $x_n$ and $x_m$, which are unknown to the server. To address this, when estimating the safety margin for each pair $(x, x')$, the server identifies the nearest $\Gamma$ records\footnote{In our experiments, the parameter $\Gamma$ is set in the range of {\rev 5} to 20.} for both $x$ and $x'$, denoted by $\mathcal{S}_x$ and $\mathcal{S}_{x'}$, respectively. 

The value of $\Gamma$ is determined based on the distance threshold $\gamma$ used during anchor selection (Eq.~(\ref{eq:w_exp})), beyond which the probability of selecting an anchor becomes negligible. For each anchor $x$, we choose the $\Gamma$ closest records (adjusted for local density) that could have selected $x$ within this range. This ensures that the estimated safety margin conservatively covers the user's true location with high probability. These neighborhoods thus serve as effective proxies, allowing the server to compute a worst-case safety margin without requiring access to actual user records.} 
% Here, the value of $\Gamma$ is chosen based on the anchor selection probability defined in Eq.~(\ref{eq:w_exp}), i.e., the probability $\Pr[(x_n, x_m) \in \mathcal{S}_x \times \mathcal{S}_{x'}] \geq 1 - \delta$, which guarantees that likely surrogates are covered within $\mathcal{S}_x$ and $\mathcal{S}_{x'}$. These neighborhoods serve as proxies for the unknown true locations of users whose surrogates are $x$ and $x'$, allowing the server to conservatively estimate the worst-case safety margin.}
% However, computing $h_{x_n,x_m}^{-1}(1-\delta)$ requires knowledge of $x_n$ and $x_m$, which are unknown to the server. To address this, when estimating the safety margin for each pair $(x, x')$, the server first identifies the nearest $\Gamma$ records\footnote{In our experiments, the parameter $\Gamma$ is set in the range of 10 to 20.} for both $x$ and $x'$, denoted by $\mathcal{S}_{x}$ and $\mathcal{S}_{x'}$, respectively. 

Given $\mathcal{S}_x$ and $\mathcal{S}_{x'}$ for each $(x, x') \in \mathcal{A}_n\times \mathcal{A}_m$, the safety margin $\xi_{x_n, x_m}$ is then approximated as
\begin{equation}
\label{eq:estxi}
\textstyle \hat{\xi}_{x,x'} = \sup_{(\tilde{x}, \tilde{x}') \in \mathcal{S}_{x} \times \mathcal{S}_{x'}} h_{\tilde{x}, \tilde{x}'}^{-1}(1-\delta_{x, x'}),
\end{equation} 
where $\delta_{x, x'}$ is the adjusted violation ratio due to approximation:
\begin{equation}
\delta_{x, x'} = 1 - \frac{1 - \delta}{\Pr\left[(x_n,x_m) \in \mathcal{S}_{x} \times \mathcal{S}_{x'}|x\in \mathcal{A}_n, x'\in \mathcal{A}_m\right] }.
\end{equation}
Note that to ensure the violation ratio $\delta_{x, x'} \in [0, 1]$, $\Gamma$ should be sufficiently large, so that $\Pr\left[(x_n,x_m) \in \mathcal{S}_{x} \times \mathcal{S}_{x'}|x\in \mathcal{A}_n, x'\in \mathcal{A}_m\right]  > 1-\delta$. Here, $\Pr\left[(x_n,x_m) \in \mathcal{S}_{x} \times \mathcal{S}_{x'}|x\in \mathcal{A}_n, x'\in \mathcal{A}_m\right] $ can be calculated using Bayes' formula (details are given in Eq. (\ref{eq:PPP1})--Eq. (\ref{eq:PPP2})) in Appendix.

\vspace{-0.05in}
\begin{proposition}
\label{prop:samplebudget}
By applying the privacy budget $(\epsilon-\overline{\epsilon}_{x,x'}) d_{x,x'}-\hat{\xi}_{x,x'}$, it can guarantee $(\epsilon, \delta)$-PmDP for $(x_n, x_m)$.
\end{proposition}
\vspace{-0.05in}
{\rev 
\begin{proof}[Proof Sketch]
% This proposition shows that applying a conservative privacy budget—adjusted using the precomputed threshold $\hat{\xi}_{x,x'}$—still ensures that the overall mechanism satisfies $(\epsilon, \delta)$-PmDP for any record pair $(x_n, x_m)$. The main idea is to leverage clustering: records are grouped into clusters $\mathcal{S}_x$ and $\mathcal{S}_{x'}$, and the value $\hat{\xi}_{x,x'}$ is chosen to cover all record pairs within these clusters. That is, for any $x_n \in \mathcal{S}_x$ and $x_m \in \mathcal{S}_{x'}$, we have $\hat{\xi}_{x,x'} \geq \xi_{x_n,x_m}$. Thus, if $(x_n,x_m)$ belong to $\mathcal{S}_x \times \mathcal{S}_{x'}$, then the applied privacy budget $(\epsilon-\overline{\epsilon}_{x,x'}) d_{x,x'} - \hat{\xi}_{x,x'}$ guarantees the success probability of mDP for that pair exceeds $1 - \delta_{x,x'}$. By appropriately calibrating $\delta_{x,x'}$ such that this local guarantee aggregates to a global success probability at least $1 - \delta$, we ensure that \textsc{PAnDA} as a whole satisfies $(\epsilon, \delta)$-PmDP. This completes the sketch. The detailed proof can be found in \textbf{Appendix \ref{subsec:proof:prop:samplebudget}}. 
This proposition shows that using a conservative privacy budget, based on the threshold $\hat{\xi}_{x,x'}$, ensures the mechanism satisfies $(\epsilon, \delta)$-PmDP. By clustering records into $\mathcal{S}_x$ and $\mathcal{S}_{x'}$, and ensuring $\hat{\xi}_{x,x'} \geq \xi_{x_n,x_m}$ for all $(x_n, x_m) \in \mathcal{S}_x \times \mathcal{S}_{x'}$, the applied budget guarantees a success probability of at least $1 - \delta_{x,x'}$. Properly aggregating these local guarantees yields global $(\epsilon, \delta)$-PmDP.
The detailed proof can be found in \textbf{Appendix \ref{subsec:proof:prop:samplebudget}}.
\end{proof}}

\DEL{
Therefore, by enforcing $\sum_{\left(\hat{x}_{n}, \hat{x}_{m}\right)\in \mathcal{H}_{x_n,x_m}}w_{x_n,\hat{x}_n} w_{x_m,\hat{x}_m} \geq 1-\delta$, we can guarantee the success ratio is not lower than $1-\delta$, or  
\begin{eqnarray}
%&&\sum_{\left(\hat{x}_{n}, \hat{x}_{m}\right)\notin \mathcal{H}_{x_n,x_m}}w_{x_n,\hat{x}_n} w_{x_m,\hat{x}_m} \leq \delta \\
% &\Rightarrow&  
\sum_{\left(\hat{x}_{n}, \hat{x}_{m}\right)\notin \mathcal{H}_{x_n,x_m}}\alpha h\left(d_{n,\hat{n}}\right) \alpha  h\left(d_{m,\hat{m}}\right) \leq \delta 
% \\ &\Rightarrow&  \alpha \leq \sqrt{\frac{\delta}{\sum_{\left(\hat{x}_{n}, \hat{x}_{m}\right)\in \mathcal{H}_{x_n,x_m}}h(d_{n,\hat{n}})h(d_{m,\hat{m}})}}
\end{eqnarray}
{\bl Remark: we need to pay attention as $\alpha$ might be lower bounded ... }. }

\DEL{
\begin{proposition}
\label{prop:adjmDP}
When determining the perturbation matrix $\mathbf{Z}_{\mathcal{X}}$, to achieve mDP constraints (in Eq. (\ref{eq:mDPdiscrete})), it is sufficient to achieve the following constraints: 
\begin{eqnarray}
\frac{z_{x,y}}{z_{m,k}} \leq e^{\epsilon d_{\mathcal{N}_i, \mathcal{N}_j} - \ln \beta_{\mathcal{N}_i, \mathcal{N}_j}}.  
\end{eqnarray}
\end{proposition}}

Next, we introduce the algorithm to calculate $\hat{\xi}_{x,x'}$ for each pair of anchors $(x,x') \in \mathcal{A}_n\times \mathcal{A}_m$, where the detailed pseudo code can be found in \textbf{Algorithms \ref{al:recompute} and \ref{al:xisearch} in Section \ref{sec:algorithmdetail} in Appendix}. 

For each pair of true records $x_n$ and $x_m$, we first pre-compute a set of \emph{candidate safety margins}: $\Upsilon_{x_n,x_m} = \left\{ \Delta_{x,x'} \mid (x, x') \in \mathcal{X}^2 \right\}$, where each $\Delta_{x,x'} = (\epsilon - \overline{\epsilon}_{x,x'}) d_{x,x'} - (\epsilon - \overline{\epsilon}_{x_n,x_m}) d_{x_n,x_m}$ quantifies the additional margin required to align the anchor-based constraint with the original mDP guarantee. 
We sort $\Upsilon_{x_n,x_m}$ in ascending order, and without loss of generality, we represent the sorted elements by $
\Delta^1 \leq ... \leq \Delta^{|\mathcal{X}|^2}$. We then search the smallest $\Delta^{\ell}$ such that the cumulative probability of satisfying the mDP constraint exceeds the desired confidence level: 
$1 - \delta \leq h_{x_n,x_m}(\Delta^{\ell})$.

\DEL{The selected $\xi_{x_n,x_m}$ is then used to enforce the refined mDP constraint:
\begin{equation}
\frac{z_{\hat{x}_n, y}}{z_{\hat{x}_m, y}} \leq \exp\left( (\epsilon - \overline{\epsilon}_{\hat{x}_n,\hat{x}_m}) d_{\hat{x}_n,\hat{x}_m} - \xi_{x_n,x_m} \right).
\end{equation}

For each pair $(x, x') \in \mathcal{X}^2$, we define the candidate safety margin $\Delta_{x,x'}$
\begin{equation}
    \Delta_{x,x'} = (\epsilon - \overline{\epsilon}_{x,x'}) d_{x,x'} - (\epsilon - \overline{\epsilon}_{x_n,x_m}) d_{x_n,x_m}.
\end{equation}
such that if $\xi \geq \Delta_{x,x'}$, the pair $(x, x') \in \mathcal{H}_{x_n,x_m}(\xi)$ (according to the definition of $\mathcal{H}_{x_n,x_m}(\xi)$ in Eq.~\eqref{eq:H}). We then define the candidate safety margin set $\Upsilon_{x_n,x_m} = \left\{ \Delta_{x,x'} \;\middle|\; (x, x') \in \mathcal{X}^2 \right\}$ and sort its elements. Without loss of generality, we represent the sorted elements by $\Delta^1 \leq ... \leq \Delta^{|\mathcal{X}|^2}$. The goal is find $\Delta^{\ell} \in \Upsilon_{x_n,x_m}$ such that
\begin{equation}
    1-\delta \leq h_{x_n,x_m}(\Delta^{\ell})~\mbox{and}~ 1-\delta > h_{x_n,x_m}(\Delta^{\ell-1})
\end{equation}}
Since $h_{x_n,x_m}(\Delta^1), \ldots, h_{x_n,x_m}(\Delta^{|\mathcal{X}|^2})$ are sorted in increasing order, we employ a linear search to identify the first $\Delta^{\ell}$ that satisfies the target condition (the rationale for preferring linear search over binary search is provided in the Appendix). Specifically, we iterate through the sorted set $\Upsilon_{x_n,x_m}$, sequentially evaluating $h_{x_n,x_m}(\Delta^{\ell})$ until a valid $\Delta^{\ell}$ is found. The overall time complexity of the algorithm is $O(K^4 \log K) + O(\Gamma^2 K^4)$, where $K = |\mathcal{X}|$ is the number records in $\mathcal{X}$. A detailed time complexity analysis presented  in \textbf{Section \ref{sec:algorithmdetail} in Appendix}.

% \textbf{Algorithms \ref{al:recompute} and \ref{al:xisearch} in Appendix} gives the pseudo code of $\xi_{\hat{x}_n,\hat{x}_m}$ search.

\subsubsection{LP Formulation of Anchor Perturbation Optimization.} 
After collecting the anchor sets $\mathcal{A}_1, ..., \mathcal{A}_N$, the server calculates the safety margin $\hat{\xi}_{x, x'}$ for each pair of $(x,x')\in \mathcal{A}_n \times \mathcal{A}_m$ ($1\leq n, m \leq N$), and enforces the following mDP constraints: 
\begin{eqnarray}
\label{eq:PAnDA_mDP_LP}
z_{x,y} - e^{(\epsilon-\overline{\epsilon}_{n,m}) d_{x, x'}-\hat{\xi}_{x, x'}} z_{x',y} \leq 0,~\forall (x,x')\in \mathcal{A}_n \times \mathcal{A}_m% ~\forall (x_n, x_m) \in \mathcal{E}_{[1, N]}, 
\end{eqnarray}
We define the utility loss function of $\mathbf{Z}_{\mathcal{A}_{[1, N]}}$ by
\vspace{-0.05in}
\begin{equation}
\textstyle 
\mathcal{L}\left(\mathbf{Z}_{\mathcal{A}_{[1, N]}}\right) = \sum_{(x,y) \in \mathcal{A}_{{[1, N]}}\times \mathcal{Y}}  p_x c_{x,y} z_{x,y}.
\end{equation}

\begin{definition}
(\underline{An}chor \underline{P}erturbation \underline{O}ptimization (\textsc{AnPO}) problem) The \textsc{AnPO} problem can be formulated as the following LP problem: 
\begin{eqnarray}
\label{eq:LPPOobj}
\min && \mathcal{L}\left(\mathbf{Z}_{\mathcal{A}_{[1, N]}}\right) \\ 
\mathrm{s.t.} && \sum_{y_k \in \mathcal{Y}}z_{x,y} = 1,~ \forall x \in \mathcal{A}_{{[1, N]}}, \\ \label{eq:LPPOconstr}
&& \mbox{Eq. (\ref{eq:PAnDA_mDP_LP}) is satisfied $\forall (x, x') \in \mathcal{A}_n \times \mathcal{A}_m$}.
\end{eqnarray}  
\end{definition}
% After the optimal $\mathbf{Z}_{\mathcal{A}_{[1, N]}}$ is derived, each user with anchor record set $\mathcal{A}$ downloads the corresponding perturbation matrix $\mathbf{Z}_i = \left\{z_{l,k}\right\}_{(x_n,y_k) \in \mathcal{A}_{n}\times \mathcal{Y}}$ that includes the perturbation vectors of secret records $\mathcal{A}$. 
Compared to the original secret dataset $\mathcal{X}$, the anchor record set $\mathcal{A}_{{[1, N]}}$ has smaller size, therefore decreasing the complexity of mDP optimization. % {\rev However, given some examples to show that $\mathcal{A}_{{[1, N]}}$ is still large}. 

{\rev \noindent\textbf{Discussion: Missing values in the perturbation matrix}. In the PAnDA framework, only a subset of the full perturbation matrix, corresponding to the selected anchor records, is optimized. Notably, most of the missing values are not needed by users, as users only require perturbation vectors for their local anchor sets. However, when the true record of a user does not belong to their anchor set, the surrogate perturbation vector may cause the server to underestimate indistinguishability across users. This problem has been addressed using the safety margins (Section~\ref{subsubsec:safetymargin}), which ensures privacy constraints are met even under approximation.}

\vspace{0.03in}
{\rev \noindent\textbf{Benders decomposition to improve computation efficiency}. To further improve the computational efficiency of \textsc{AnPO}, we adopt Benders’ decomposition \cite{Qiu-IJCAI2024}. Specifically, we divide the anchor set $\mathcal{A}_{{[1, N]}}$ into several smaller groups based on how closely records are connected under the mDP constraints. With this partitioning, the optimization proceeds in two stages. In Stage 1, a \emph{master program} optimizes the perturbation for records located near the boundaries between groups (referred to as \emph{boundary records}). In Stage 2, a set of \emph{subproblems} handle the optimization for records located well within each group (referred to as \emph{internal records}). This process is iterative: the master program first proposes perturbation strategies for the boundary records. Each subproblem then checks whether the current solution satisfies the constraints within its group. If not, it returns a new constraint (called a \emph{cut}) to refine the master program. This cycle continues until both the master and subproblems reach an optimal solution. Notably, we observe that the anchor sets selected through \textsc{PAnDA} exhibit higher clustering coefficients compared to the original secret data domain $\mathcal{X}$, indicating stronger locality. This structural property facilitates decomposition by reducing coupling between subproblems. Detailed clustering analysis results can be found in \textbf{Section \ref{subsec:clustering_add} in Appendix}.} \looseness = -1

% As a solution, we apply Benders decomposition to further reduce the computation overhead of perturbation optimization. Particularly, Benders' decomposition is an optimization method that breaks a complex problem into two simpler stages: a \emph{master program} and a set of \emph{subproblems}.

\vspace{0.03in}
\noindent\textbf{Lower bound on utility loss via Relaxed \textsc{AnPO}.}  
For theoretical interest, we derive a lower bound on the utility loss $\mathcal{L}\left(\mathbf{Z}^*_{[1, N]}\right)$, where $\mathbf{Z}^*$ denotes the optimal solution to the original perturbation optimization problem defined in Eq.~(\ref{eq:PPO}), and $\mathbf{Z}^*_{[1, N]}$ is the submatrix corresponding to the joint anchor set $\mathcal{A}_{[1, N]}$. By comparing the solution obtained from \textsc{AnPO} to this lower bound, we assess how closely \textsc{AnPO} approximates the optimal utility.

To compute this bound, we consider a relaxed version of the mDP constraints in Eq.~(\ref{eq:PAnDA_mDP_LP}), where the constraint set is modified as:
\vspace{-0.05in}
\begin{equation}
\label{eq:PAnDA_mDP_LP_relax}
z_{x,y} - e^{\epsilon d_{x, x'}} z_{x',y} \leq 0, \quad \forall (x, x') \in \mathcal{A}_n \times \mathcal{A}_m.
\end{equation}
This yields the \emph{Relaxed \textsc{AnPO}} problem, which remains a LP and is solvable efficiently.

\begin{proposition}
\label{prop:ULbound}
Let $\tilde{\mathbf{Z}}_{[1, N]}$ be the optimal solution to the Relaxed \textsc{AnPO} problem. Then, the minimum utility loss achieved by Relaxed \textsc{AnPO} provides a lower bound on the utility loss of the original optimization over the anchor set: $\mathcal{L}\left(\tilde{\mathbf{Z}}_{[1, N]}\right) \leq \mathcal{L}\left(\mathbf{Z}^*_{[1, N]}\right)$. 
\end{proposition}
\vspace{-0.10in}
{\rev 
\begin{proof}[Proof Sketch] 
% The goal is to compare the utility loss of the optimal solution to the Relaxed \textsc{AnPO} with that of the original optimization problem, restricted to the same set of anchors $\mathcal{A}_{[1,N]}$. The key observation is that the relaxed problem admits a larger feasible set. Specifically, every feasible solution to the original problem is also feasible in the relaxed problem, but not vice versa: $\mathcal{F}_{\text{orig}} \subseteq \mathcal{F}_{\text{relaxed}}$. Since the relaxed problem minimizes the same objective function (utility loss $\mathcal{L}$) over a \emph{larger} feasible set, its optimal value cannot be worse (i.e., larger) than that of the original problem restricted to the same submatrix: $\mathcal{L}(\tilde{\mathbf{Z}}_{[1, N]}) \leq \mathcal{L}(\mathbf{Z}^*_{[1, N]})$. This establishes the claim that the relaxed solution is at least as good in terms of utility loss when evaluated over the common anchor submatrix. The detailed proof can be found in \textbf{Section \ref{subsec:proof:prop:ULbound} in Appendix}. 
We compare the utility loss of the Relaxed \textsc{AnPO} solution with that of the original problem, both restricted to the same anchor set $\mathcal{A}{[1,N]}$. Since the relaxed problem has a larger feasible set ($\mathcal{F}_{\text{orig}} \subseteq \mathcal{F}_{\text{relaxed}}$) and minimizes the same objective, its utility loss is no worse:
$\mathcal{L}(\tilde{\mathbf{Z}}_{[1, N]}) \leq \mathcal{L}(\mathbf{Z}^*_{[1, N]})$.
The detailed proof can be found in \textbf{Section \ref{subsec:proof:prop:ULbound} in Appendix}. 
\end{proof}}

\vspace{-0.15in}
\section{Performance Evaluation}
\label{sec:performance}

In this section, we evaluate \textsc{PAnDA}'s performance. Section~\ref{subsec:settings} describes the experimental setup. Section~\ref{subsec:anchorselectpara} analyzes the impact of anchor selection parameters on the required safety margin, guiding parameter choices for subsequent evaluations. Section~\ref{subsec:benchmarks} compares \textsc{PAnDA} with baseline methods, including EM, LP-based, and hybrid approaches. Finally, Section~\ref{subsec:privacycosts} examines the distribution of privacy budgets across \textsc{PAnDA}’s two phases.

\vspace{-0.10in}
\subsection{Settings}
\label{subsec:settings}
\vspace{-0.05in}
\subsubsection*{\textbf{Datasets}} 
We conducted experiments on road map datasets for three cities: \emph{Rome, Italy}, \emph{New York City (NYC), USA}, and \emph{London, UK}. {\rev These datasets define the secret record domains (i.e., users’ possible locations) rather than observed data.} Each dataset comprises nodes representing intersections, junctions, and other critical points in the city's road network, with data sourced from OpenStreetMap~\cite{openstreetmap}. % Edges connect nodes according to actual road segments, and distances are computed using the Haversine formula, which measures great-circle distances on the Earth's surface \cite{Primal-Costomizable2022}. 

The geographic boundaries and dataset statistics for each city are summarized below:
\begin{itemize}[left=0.3em, labelsep=0.5em]
    \item \textbf{Rome, Italy}, bounded by (lat = 41.66, lon = 12.24) in the southwest and (lat = 42.10, lon = 12.81) in the northeast, including 43,160 nodes. % and 89,739 edges.
    \item \textbf{NYC, USA}, bounded by (lat = 73.70, lon = 40.50) in the southwest and (lat = 74.25, lon = 40.91) in the northeast, including 55,353 nodes. % and 139,638 edges.
    \item \textbf{London, UK}, bounded by (lat = 51.29, lon = -0.51) in the southwest and (lat = 51.69, lon = 0.28) in the northeast, including 12,820 nodes. % and 299,524 edges.
\end{itemize}
To construct the secret data domain $\mathcal{X}$ for evaluation, we randomly sample a subset of nodes from each city's road map dataset, varying the size from {\rev 500 to 5,000}. This sampling simulates different domain scales and allows us to assess the scalability and performance of \textsc{PAnDA} under increasing secret data domain size. 

{\rev Since real user data is unavailable for some cities, we simulate user locations uniformly over the domain. To enhance realism, we also include experiments based on a real-world taxicab location dataset~\cite{roma-taxi-20140717} for the Rome road network.}

Our tests were conducted on a machine equipped with an Intel Core i9-13900F processor, featuring 24 cores with a base clock speed ranging from 2.00 GHz to a maximum of 5.60 GHz. The system was configured with 32 GB of DDR5 memory (32 GB, 4800 MHz) and included an NVIDIA GeForce RTX 4090 graphics card with 24 GB of GDDR6X VRAM. We use the MATLAB optimization toolbox "$\mathsf{linearprog}$" to solve the LP problems.
% The grid map for each city is divided into a $100 \times 100$ grid, with the center point of each grid cell serving as its representative coordinate. The distance between the grid cells is calculated using the Haversine distance formula, which measures the angular distance between two points on the Earth's surface.

% How to choose $\lambda$, $\alpha$ (emprical study)

% Fig. xxx: $\delta$, $\alpha$ are contants: Change $\lambda$ -> observe $\xi$ 

% Fig. xxx: $\delta$, $\lambda$ are contants: Change $\alpha$ -> observe $\xi$ 

% Fig. xxx: $\lambda$, $\alpha$ are constants: Change $\delta$ -> observe $\xi$ 

% Fig. xxxx: $\lambda$, $\alpha$ are constants: Change $\delta$ -> observe {\rev actually mDP violation ratio.} 

\vspace{-0.10in}
\subsection{Impact of Anchor Selection Parameters on Safety Margin}
\label{subsec:anchorselectpara}

\begin{figure}[t]
\centering
\hspace{0.00in}
\begin{minipage}{0.50\textwidth}
\centering
  \subfigure[Exponential decay]{
\includegraphics[width=0.31\textwidth, height = 0.10\textheight]{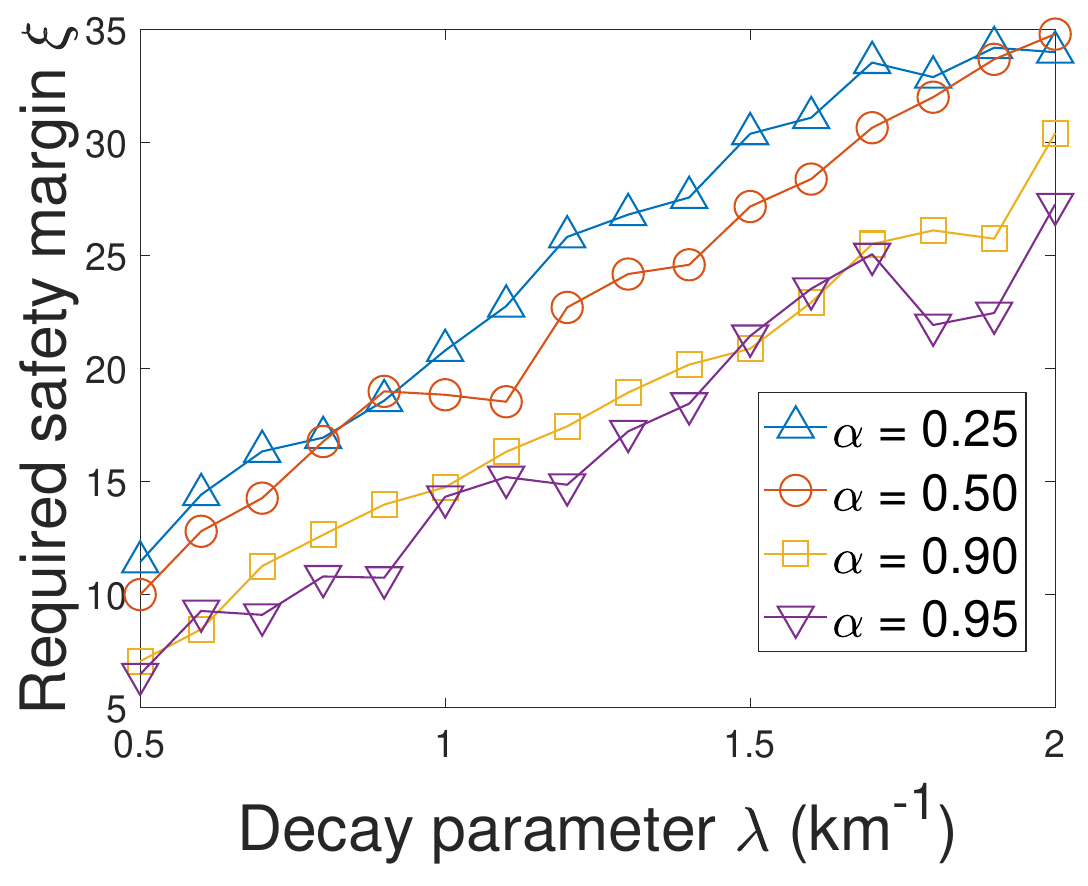}}
  \subfigure[Power law decay]{
\includegraphics[width=0.31\textwidth, height = 0.10\textheight]{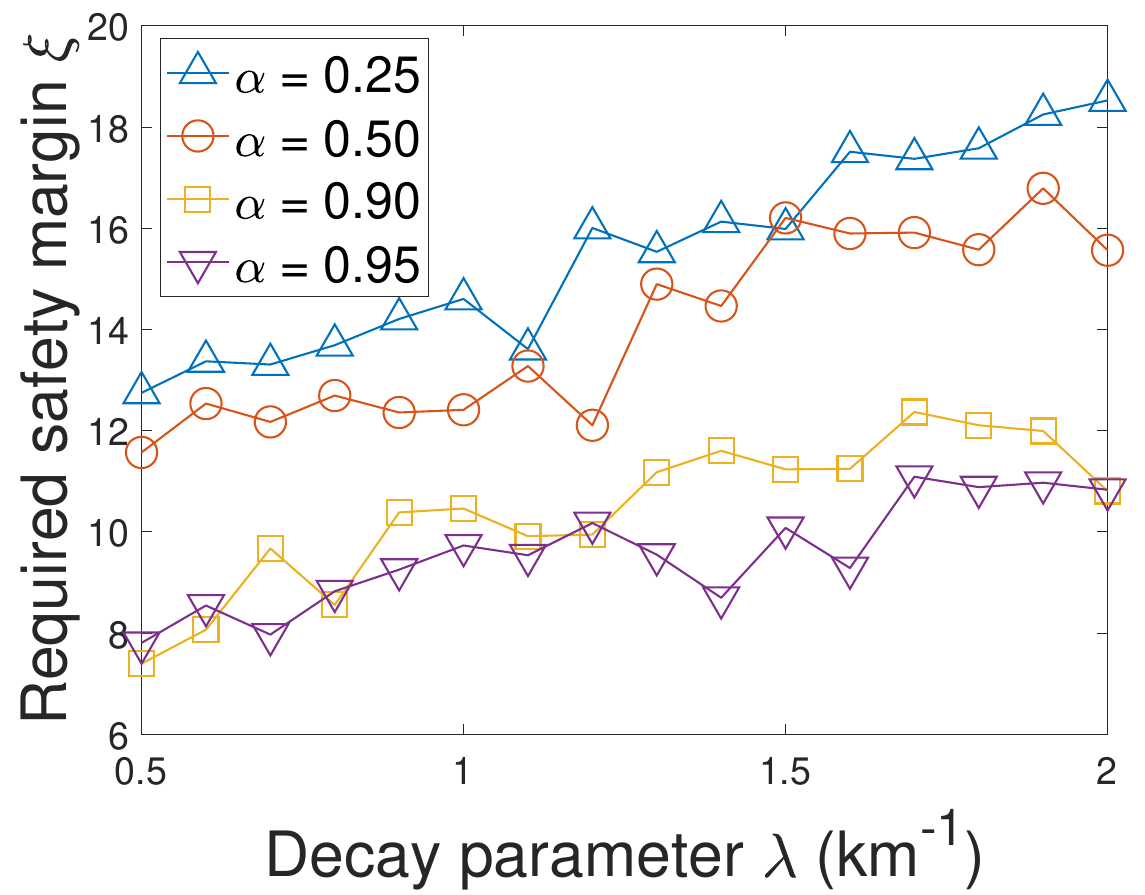}}
  \subfigure[Logistic]{
\includegraphics[width=0.31\textwidth, height = 0.10\textheight]{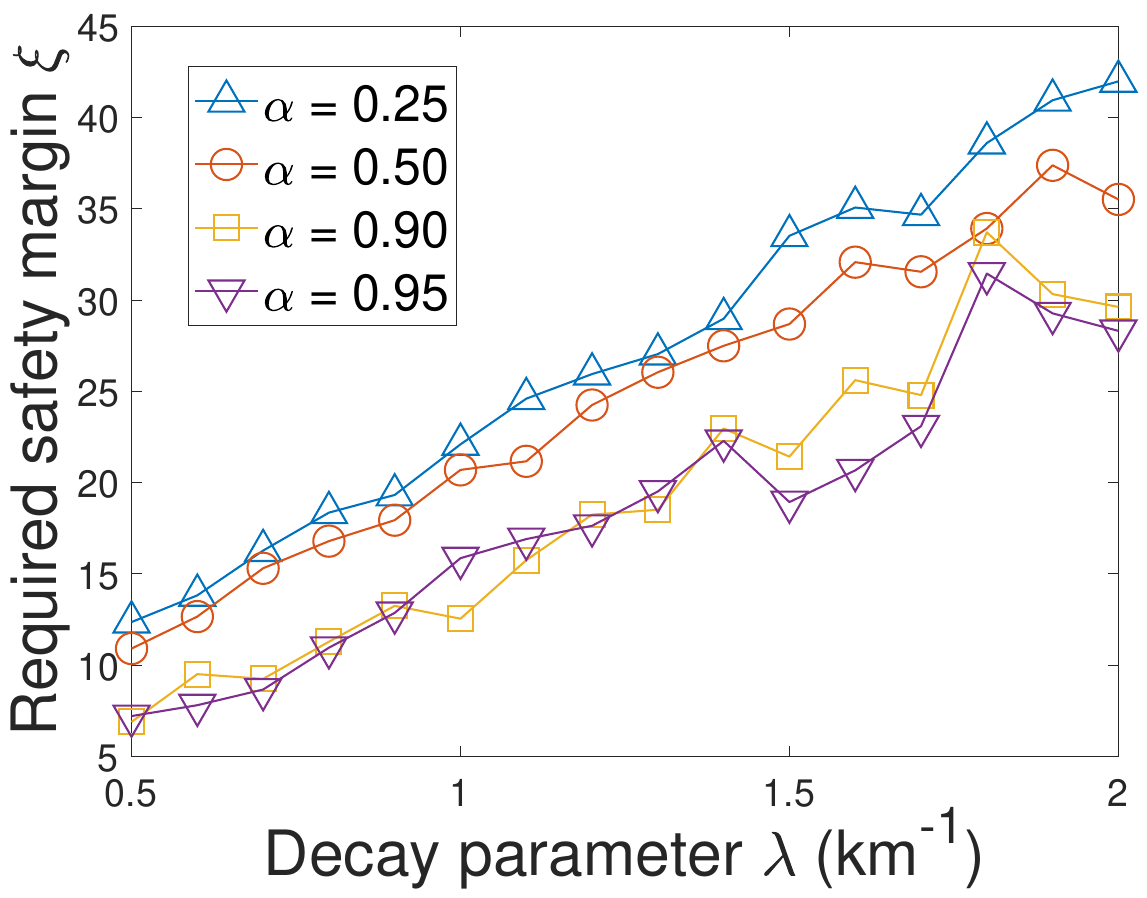}}
\vspace{-0.15in}
\end{minipage}
\caption{Expected safety margin $\xi$ with different decay parameter $\lambda$ and scaling parameter $\alpha$ (Rome).}
\label{fig:xi_lambda_rome}
\vspace{-0.20in}
\end{figure}

We begin by examining how the anchor selection parameters, \textit{the scaling factor} $\alpha$ and \textit{the decay parameter} $\lambda$, impact the required safety margin $\xi$ for achieving $(\epsilon, \delta)$-PmDP, which guides the parameter selection for the subsequent experiment. Recall that in the three anchor selection methods---\emph{exponential decay}, \emph{power-law decay}, and \emph{logistic} (defined in Table~\ref{Tb:anchorrecordselection}), $\alpha$ controls the overall probability of selecting anchors, and $\lambda$ determines how sharply this probability decreases with with increasing distance between a true record and a candidate anchor record. 

Fig.~\ref{fig:xi_lambda_rome}(a)(b)(c) show how the average required safety margin $\xi$ varies with different $\lambda$ and $\alpha$ across the three methods using Rome roadmap dataset. Due to the limited space, the results of NYC and London can be found in Fig. \ref{fig:xi_lambda_nyc}(a)(b)(c) and \ref{fig:xi_lambda_london}(a)(b)(c) in Appendix. In all cases, we observe that increasing $\lambda$ or decreasing $\alpha$ leads to a larger safety margin $\xi$. This is because a steeper decay (larger $\lambda$) or lower scaling factor (smaller $\alpha$) results in fewer selected anchors, which can degrade the approximation of true distances between users' records. Consequently, a larger $\xi$ is needed to compensate for such higher approximation errors and to ensure that the total privacy cost remains within the budget $\epsilon$ with high probability $1-\delta$. \emph{These results suggest that using a higher $\alpha$ and a lower $\lambda$ leads to a smaller required safety margin $\xi$, thereby leaving more privacy budget for perturbation optimization.} 

However, increasing $\alpha$ and decreasing $\lambda$ also result in more anchor records being selected, which increases the computational overhead of anchor perturbation optimization in the second phase. As a trade-off between utility and computational efficiency, in the following experiment, {\rev we set $\alpha$ and $\lambda$, and $\gamma = d_{\max}/50$ by 0.95 and 0.50, respectively, where $d_{\max}$ denotes the maximum pairwise distance among locations in the target region.} In addition, we set {\rev the acceptable mDP violation ratio $\delta = 10^{-7}$ and the privacy budget $\epsilon  = 5km^{-1}$ \cite{Andres-CCS2013} by default}.

%Fig. \ref{fig:budget_vs_delta}: $\lambda$, $\alpha$ are constants: Change $\delta$ -> observe {\rev privacy budget $\overline{\epsilon}_{x_n,x_m}$.} 

% Experiment III: $\lambda$, $\alpha$ are constants: Change $\delta$ -> observe $\xi$ and actually mDP violation ratio. 

% $\xi$ (analytic results) 

% $\lambda$ is constant, change $\xi$

% $\xi$ is constant, change $\lambda$
\vspace{-0.10in}
\subsection{Comparison with the Benchmarks}
\label{subsec:benchmarks}
\vspace{-0.05in}
In this part, we compare \textsc{\textsc{PAnDA}} against several representative baselines that achieve $\epsilon$-mDP:
\vspace{-0.05in}
\begin{itemize}[left=0.1em, labelsep=0.5em]  
    \item \textbf{Exponential Mechanism (EM) \cite{Chatzikokolakis-PoPETs2015}}, which is a \emph{predefined noise distribution mechanism}.  EM perturbs true record $x \in \mathcal{X}$ by following the probability distribution $\Pr[\mathcal{M}(x) = y] \propto \exp\left(-\epsilon  d_{x, y}\right)$,
    where $d_{x,y}$ is the distance between true record $x$ and perturbed record $y$. This mechanism satisfies pure $\epsilon$-mDP and serves as a generic, non-adaptive baseline \cite{dwork2014algorithmic}. 
    \item \textbf{Linear programming (LP) \cite{Bordenabe-CCS2014}.} % This approach formulates the perturbation optimization as a LP that minimizes the expected utility loss subject to mDP constraints. While LP yields an optimal mechanism, it suffers from poor scalability when covering the entire secret data domain $\mathcal{X}$, resulting in $|\mathcal{X}|^2$ decision variables.
    This approach formulates perturbation as a linear program minimizing expected utility loss under mDP constraints, but incurs $|\mathcal{X}|^2$ variables, which limits scalability. \looseness = -1
    \item \textbf{Coarse Approximation of LP (LP+CA).} Considering the high complexity of LP, we also select a coarse-grained approximation strategy of LP as a baseline. In LP+CA, the target region is partitioned into an $8 \times 8$ grid and each location is approximated to its corresponding grid cell. This reduces the domain to 64 representative records, lowering computation but potentially increasing utility loss and mDP violation.  
    \item \textbf{Benders Decomposition (LP+BD) \cite{Qiu-IJCAI2024}.} To improve LP scalability, BD clusters the secret data domain and decomposes the mDP optimization into a master problem and subproblems. The master coordinates cross-cluster perturbation, while each subproblem optimizes perturbation within a cluster.
    \item \textbf{ConstOPTMech (LP+EM) \cite{ImolaUAI2022}.} LP+EM uses LP to optimize perturbation probabilities over a candidate set of outputs with high utility sensitivity, while relying on EM to handle the remaining outputs. This design reduces computational overhead while preserving strong privacy guarantees under mDP.
 
    \item \textbf{Bayesian Remapping (EM+BR) \cite{Chatzikokolakis-PETS2017}.} BR is a post-processing technique that improves utility without violating mDP. Given a perturbed record $y$ (output by EM), the remapped output $y'$ is selected to minimize expected loss: 
    \vspace{-0.05in}
    $$\textstyle y' = \arg\min_{y' \in \mathcal{Y}} \mathbb{E}_{x \sim \Pr[x \mid y]} \left[\mathcal{L}(x, y')\right],$$
    \vspace{-0.10in}
    \newline where the posterior $\Pr[x| y]$ is derived from the prior $p_x$ and the perturbation matrix $\mathbf{Z}_{\mathcal{X}}$. This remapping improves utility while preserving the original privacy guarantees.
    % \begin{equation}
    % f(y_k) = \arg\min_{y' \in \mathcal{Y}} \sum_{x\in \mathcal{X}} \Pr\left[X = x \left|\mathcal{M}(X) = y_k\right.\right] c_{x, y'}
    % \end{equation}
\end{itemize}
We refer to our approach as \textbf{\textsc{PAnDA}}, using the anchor selection methods \emph{exponential decay}, \emph{power-law decay}, and \emph{logistic}, denoted as \textbf{\textsc{PAnDA}-e}, \textbf{\textsc{PAnDA}-p}, and \textbf{\textsc{PAnDA}-l}, respectively.

\subsubsection{Computation efficiency.} We evaluate the computation time of different data perturbation methods across the three roadmap datasets in \emph{Rome}, \emph{NYC}, and \emph{London}, when secret data domain size (denoted by $K = |\mathcal{X}|$) is increased from {\rev 500 to 5,000}. {\rev Unless otherwise specified, we set $K = 2,000$ by default.} We repeat each experiment by 10 times. {\rev Table~\ref{Tb:exp:time_scalability} reports the average computation time (in seconds) for each algorithm under a uniform user location distribution, while Table~\ref{Tb:exp:time_scalability_realdata} shows results on the Rome roadmap using the real-world Rome taxicab location dataset~\cite{roma-taxi-20140717}.} For fairness, we consider an LP-based method as ``unsolved' if its computation time exceeds 1,800 seconds.

{\rev By default, we set the number of users to 50, which aligns with typical applications of mDP (geo-indistinguishability), such as in spatial crowdsourcing (SC)~\cite{Qiu-CIKM2020}. In SC, the perturbation matrix is computed over active users or workers within the target task’s influence region during a given time window, rather than across the entire user population. In practice, a setting with approximately 10--50 users is common \cite{Qiu-CIKM2020}.}

\vspace{-0.00in}
\begin{table}[t]
\centering
\small 
% \footnotesize 
\scriptsize 
\setlength{\tabcolsep}{3pt}  
\begin{tabular}{ c|c|c|c|c|c|c}
%\cline{2-13}
%\hline
\toprule
\multicolumn{7}{ c  }{Rome}\\ 
\cline{1-7}
\multicolumn{1}{ c|  }{Method}  & $K = 500$ & $K = 1,000$  & $K = 2,000$ & $K = 3,000$& $K = 4,000$& $K = 5,000$\\
\hline
\hline
\multicolumn{1}{ c|  }{EM} & $\leq0.005$ & $\leq0.005$ & $\leq0.005$ & $\leq0.005$ & $\leq0.005$ & 
$\leq0.005$ \\ 
% \multicolumn{1}{ c|  }{Laplace} & x.xx$\pm$x.x & x.xx$\pm$x.x & x.xx$\pm$x.x & x.xx$\pm$x.x & x.xx$\pm$x.x & x.xx$\pm$x.x \\ 
\multicolumn{1}{ c|  }{LP} & >1,800 & >1,800 & >1,800 & >1,800 & >1,800 & >1,800 \\ 
\multicolumn{1}{ c|  }{LP+CA} & 6.10±0.60 & 7.86±0.48 & 14.99±0.71 & 19.73±0.44 & 23.17±0.54 & 27.07±0.45 \\ 
\multicolumn{1}{ c|  }{LP+BD} & 5.45±7.10 & >1,800 & >1,800 & >1,800 & >1,800 & >1,800 \\ 
\multicolumn{1}{ c|  }{LP+EM} & 0.67±0.31 & >1,800 & >1,800 & >1,800 & >1,800 & >1,800 \\ 
\multicolumn{1}{ c|  }{EM+BR} & 0.04±0.00 & 0.11±0.00 & 0.17±0.00 & 0.35±0.00 & 0.51±0.01 & 0.65±0.02  \\ 
\hline
\multicolumn{1}{ c|  }{ \textbf{\textsc{PAnDA}-e}} & 0.21±0.06 & 0.36±0.02 & 4.32±0.11 & 8.17±4.52 & 16.57±7.02 & 25.67±8.69 \\
\multicolumn{1}{ c|  }{ \textbf{\textsc{PAnDA}-p}} & 0.33±0.08 & 0.58±0.10 & 5.30±3.07 & 10.43±5.18 & 18.51±9.31 & 29.67±12.10 \\
\multicolumn{1}{ c|  }{ \textbf{\textsc{PAnDA}-l}} & 0.32±0.11 & 0.75±0.05 & 6.69±2.69 & 12.18±5.24 & 25.90±11.87 & 34.77±17.71 \\
\hline
% \multicolumn{1}{ c|  }{ \textbf{LB-e}} & 8.55±13.18 & 0.93±0.59 & 3.56±2.78 & 4.91±4.08 & 13.77±9.48 & 30.08±18.63 \\ 
%\multicolumn{1}{ c|  }{ \textbf{LB-p}} & 3.30±6.01 & 1.12±0.69 & 17.32±17.82 & 8.42±7.15 & 15.71±10.66 &  33.51±20.47 \\
%\multicolumn{1}{ c|  }{ \textbf{LB-l}} & 3.17±5.35 & 1.50±0.888 & 1.51±1.09 & 12.87±12.64 & 14.91±6.56 &  31.80±19.35 \\
% \hline
\toprule
\multicolumn{7}{ c  }{NYC}\\ 
\cline{1-7}
\multicolumn{1}{ c|  }{Method}  & $K = 500$ & $K = 1,000$  & $K = 2,000$ & $K = 3,000$& $K = 4,000$& $K = 5,000$\\
\hline
\hline
\multicolumn{1}{ c|  }{EM} & $\leq0.005$ & $\leq0.005$ & $\leq0.005$ & $\leq0.005$ & $\leq0.005$ & 
$\leq0.005$ \\ 
% \multicolumn{1}{ c|  }{Laplace} & x.xx$\pm$x.x & x.xx$\pm$x.x & x.xx$\pm$x.x & x.xx$\pm$x.x & x.xx$\pm$x.x &  x.xx$\pm$x.x \\ 
\multicolumn{1}{ c|  }{LP} & >1,800 & >1,800 & >1,800 & >1,800 & >1,800 & >1,800 \\ 
\multicolumn{1}{ c|  }{LP+CA} & 5.24±0.17 & 5.55±0.14 & 19.26±0.31 & 25.61±0.43 & 32.15±0.56 & 38.49±0.96 \\ 
\multicolumn{1}{ c|  }{LP+BD} & 6.60±0.39 & >1,800 & >1,800 & >1,800 & >1,800 & >1,800 \\ 
\multicolumn{1}{ c|  }{LP+EM} & 0.95±0.15 & >1,800 & >1,800 & >1,800 & >1,800 & >1,800 \\ 
\multicolumn{1}{ c|  }{EM+BR} & 0.02±0.00 & 0.07±0.00 & 0.12±0.00 & 0.22±0.05 & 0.34±0.07 & 0.46±0.09  \\ 
\hline
\multicolumn{1}{ c|  }{ \textbf{\textsc{PAnDA}-e}} & 0.26±0.03 & 0.23±0.08 & 7.91±0.37 & 12.84±4.18 & 18.43±7.32 & 23.87±6.01 \\
\multicolumn{1}{ c|  }{ \textbf{\textsc{PAnDA}-p}} & 0.33±0.01 & 0.56±0.06 & 7.61±2.08 & 14.69±4.80 & 18.96±13.17 & 27.16±20.25 \\
\multicolumn{1}{ c|  }{ \textbf{\textsc{PAnDA}-l}} & 0.38±0.01 & 0.56±0.06 & 7.24±0.43 & 19.13±18.91 & 29.33±10.38 & 39.07±18.41 \\
\hline
% \multicolumn{1}{ c|  }{ \textbf{LB-e}} & 0.247±0.028 & 0.483±0.127 & 2.35±1.75 & 6.13±2.92 & 9.07±3.83 & 15.32±9.64 \\ 
% \multicolumn{1}{ c|  }{ \textbf{LB-p}} & 0.278±0.021 & 0.563±0.209 & 6.06±8.11 & 38.12±29.12 & 18.45±10.08 & 20.58±08.65 \\
% \multicolumn{1}{ c|  }{ \textbf{LB-l}} & 0.270±0.014 & 0.601±0.185 & 2.47±1.07 & 30.43±19.82 & 37.59±19.34 & 53.91±30.43 \\
% \hline
\toprule
\multicolumn{7}{ c  }{London}\\ 
\cline{1-7}
\multicolumn{1}{ c|  }{Method}  & $K = 500$ & $K = 1,000$  & $K = 2,000$ & $K = 3,000$& $K = 4,000$& $K = 5,000$\\
\hline
\hline
\multicolumn{1}{ c|  }{EM} & $\leq0.005$ & $\leq0.005$ & $\leq0.005$ & $\leq0.005$ & $\leq0.005$ & 
$\leq0.005$ \\ 
% \multicolumn{1}{ c|  }{Laplace} & x.xx$\pm$x.x & x.xx$\pm$x.x & x.xx$\pm$x.x & x.xx$\pm$x.x & x.xx$\pm$x.x &  x.xx$\pm$x.x \\ 
\multicolumn{1}{ c|  }{LP} & >1,800 & >1,800 & >1,800 & >1,800 & >1,800 & >1,800 \\ 
\multicolumn{1}{ c|  }{LP+CA} & 17.72±0.24 & 15.24±0.55 & 28.38±0.93 & 48.92±2.34 & 61.78±1.06 & 87.48±0.56 \\ 
\multicolumn{1}{ c|  }{LP+BD} & 5.77±0.46 & >1,800 & >1,800 & >1,800 & >1,800 & >1,800 \\ 
\multicolumn{1}{ c|  }{LP+EM} & 0.41±0.04 & >1,800 & >1,800 & >1,800 & >1,800 & >1,800 \\ 
\multicolumn{1}{ c|  }{EM+BR} & 0.02±0.00 & 0.05±0.00 & 0.11±0.00 & 0.25±0.00 & 0.44±0.03 & 0.67±0.07  \\ 
\hline
\multicolumn{1}{ c|  }{ \textbf{\textsc{PAnDA}-e}} & 0.14±0.04 & 0.55±0.32 & 7.29±0.94 & 22.42±6.90 & 30.73±9.87 & 42.38±12.72 \\
\multicolumn{1}{ c|  }{ \textbf{\textsc{PAnDA}-p}} & 0.16±0.02 & 0.30±0.02 & 6.90±0.47 & 15.44±12.11 & 64.29±16.49 & 56.08±54.04 \\
\multicolumn{1}{ c|  }{ \textbf{\textsc{PAnDA}-l}} & 0.28±0.02 & 0.30±0.04 & 7.20±1.37 & 20.68±22.35 & 58.24±21.51 & 80.61±34.92 \\
\hline
% \multicolumn{1}{ c|  }{ \textbf{LB-e}} & 0.450±0.0418 & 0.326±0.0359 & 1.28±0.641 & 11.11±9.58 & 20.03±14.77 & 44.95±25.19 \\ 
% \multicolumn{1}{ c|  }{ \textbf{LB-p}} & 0.476±0.0413 & 0.342±0.0422 & 1.28±0.666 & 12.36±9.90 & 42.65±28.02 & 40.33±26.86 \\
% \multicolumn{1}{ c|  }{ \textbf{LB-l}} & 0.513±0.0476 & 0.367±0.0621 & 1.41±0.755 & 27.01±19.18 & 51.98±34.69 & 71.64±43.59 \\ 
% \hline
\end{tabular}
\vspace{0.00in}
\caption{Computation time of different algorithms (uniform user location distribution). Mean$\pm$1.96$\times$std. deviation. }
\label{Tb:exp:time_scalability}
\vspace{-0.25in}
\end{table}

\vspace{-0.00in}
\begin{table}[t]
\centering
\small 
% \footnotesize 
\scriptsize 
\setlength{\tabcolsep}{3pt}  
\begin{tabular}{ c|c|c|c|c|c|c}
%\cline{2-13}
%\hline
\toprule
\multicolumn{7}{ c  }{Rome}\\ 
\cline{1-7}
\multicolumn{1}{ c|  }{Method}  & $K = 500$ & $K = 1,000$  & $K = 2,000$ & $K = 3,000$& $K = 4,000$& $K = 5,000$\\
\hline
\hline
\multicolumn{1}{ c|  }{EM} & $\leq0.005$ & $\leq0.005$ & $\leq0.005$ & $\leq0.005$ & $\leq0.005$ & 
$\leq0.005$ \\ 
\multicolumn{1}{ c|  }{LP} & >1,800 & >1,800 & >1,800 & >1,800 & >1,800 & >1,800 \\ 
\multicolumn{1}{ c|  }{LP+CA} & 6.41±0.54 & 8.10±0.53 & 15.11±0.73 & 21.09±0.42 & 24.61±0.42 & 29.36±0.51 \\ 
\multicolumn{1}{ c|  }{LP+BD} & 6.19±4.05 & >1,800 & >1,800 & >1,800 & >1,800 & >1,800 \\ 
\multicolumn{1}{ c|  }{LP+EM} & 0.45±0.18 & >1,800 & >1,800 & >1,800 & >1,800 & >1,800 \\ 
\multicolumn{1}{ c|  }{EM+BR} & 0.05±0.00 & 0.10±0.00 & 0.17±0.00 & 0.39±0.00 & 0.49±0.01 & 0.73±0.02 \\  
\hline
\multicolumn{1}{ c|  }{ \textbf{\textsc{PAnDA}-e}} & 0.19±0.07 & 0.38±0.02 & 4.10±0.12 & 8.55±5.58 & 17.81±8.31 & 26.31±6.89 \\ 
\multicolumn{1}{ c|  }{ \textbf{\textsc{PAnDA}-p}} & 0.31±0.06 & 0.53±0.11 & 5.35±2.75 & 9.15±5.45 & 17.98±11.47 & 26.15±11.69 \\ 
\multicolumn{1}{ c|  }{ \textbf{\textsc{PAnDA}-l}} & 0.28±0.09 & 0.71±0.04 & 7.24±2.59 & 12.90±5.65 & 22.99±13.19 & 30.91±17.99 \\ 
\hline
\end{tabular}
\vspace{0.00in}
\caption{Computation time of different algorithms (Rome taxicab location dataset). Mean$\pm$1.96$\times$std. deviation. }
\label{Tb:exp:time_scalability_realdata}
\vspace{-0.25in}
\end{table}

From both Table \ref{Tb:exp:time_scalability} and Table \ref{Tb:exp:time_scalability_realdata}, we observe that \textsc{PAnDA} demonstrates a much shorter computation time compared to LP, LP+BD, and LP+EM. In particular, for all the methods, the computation time increases with the increase of number of records. In particular, in Table \ref{Tb:exp:time_scalability}, on average, \textsc{PAnDA}-e, \textsc{PAnDA}-p, and \textsc{PAnDA}-l achieve {\rev 30.64s, 37.63s, and 51.48s} across the three datasets with 5,000 records. In Table \ref{Tb:exp:time_scalability_realdata}, \textsc{PAnDA}-e, \textsc{PAnDA}-p, and \textsc{PAnDA}-l achieve {\rev 26.31, 26.15s, and 30.91s}. While, for the LP-based methods LP, LP+BD, and LP+EM, their computation considerable inefficiency, especially at {\rev larger data domain size}.

\textsc{PAnDA} achieves lower computation time by avoiding the need to solve a full-scale LP over the entire data domain $\mathcal{X}$. Traditional LP-based methods, including the LP+BD variant, require optimizing over all pairs of real and perturbed records, which becomes computationally intractable as $|\mathcal{X}|$ increases. Although LP+EM reduces the LP size by focusing on a utility-sensitive subset of outputs, it still incurs a large number of decision variables because it aims to derive perturbation vectors for all records in $\mathcal{X}$. This limitation remains significant when $|\mathcal{X}|$ reaches 1,000 or more. In contrast, \textsc{PAnDA} allows each user to select a small, personalized set of anchor records, and the server optimizes only over the union of these anchors, drastically reducing the LP size. Figure~\ref{fig:variablesvsnrrecords} illustrates the number of anchor records selected by \textsc{PAnDA}-e, \textsc{PAnDA}-p, and \textsc{PAnDA}-l as the secret data domain size grows from {\rev 500 to 5,000}. On average, each variant selects anchor sets that cover only a small fraction, averaging {\rev 14.81\%}, of the full domain. Even at a domain size of {\rev 5,000}, the number of anchor records remains bounded (up to {\rev 450}), enabling the resulting LPs to be efficiently solved using Benders decomposition~\cite{Qiu-IJCAI2024}. This substantial reduction in optimization variables is key to the improved scalability of \textsc{PAnDA}. {\rev Additionally, the figures show that the number of records covered in the London dataset is lower than in the other two cities. This is due to variations in the density of location points retrieved from OpenStreetMap across the three maps. In particular, London has a significantly lower point density, 62.74\% lower than NYC and 52.50\% lower than Rome, resulting in fewer selected points compared to the other two datasets.}

\begin{figure}[t]
\centering
\hspace{0.00in}
\begin{minipage}{0.50\textwidth}
  \subfigure[Rome]{
\includegraphics[width=0.32\textwidth, height = 0.10\textheight]{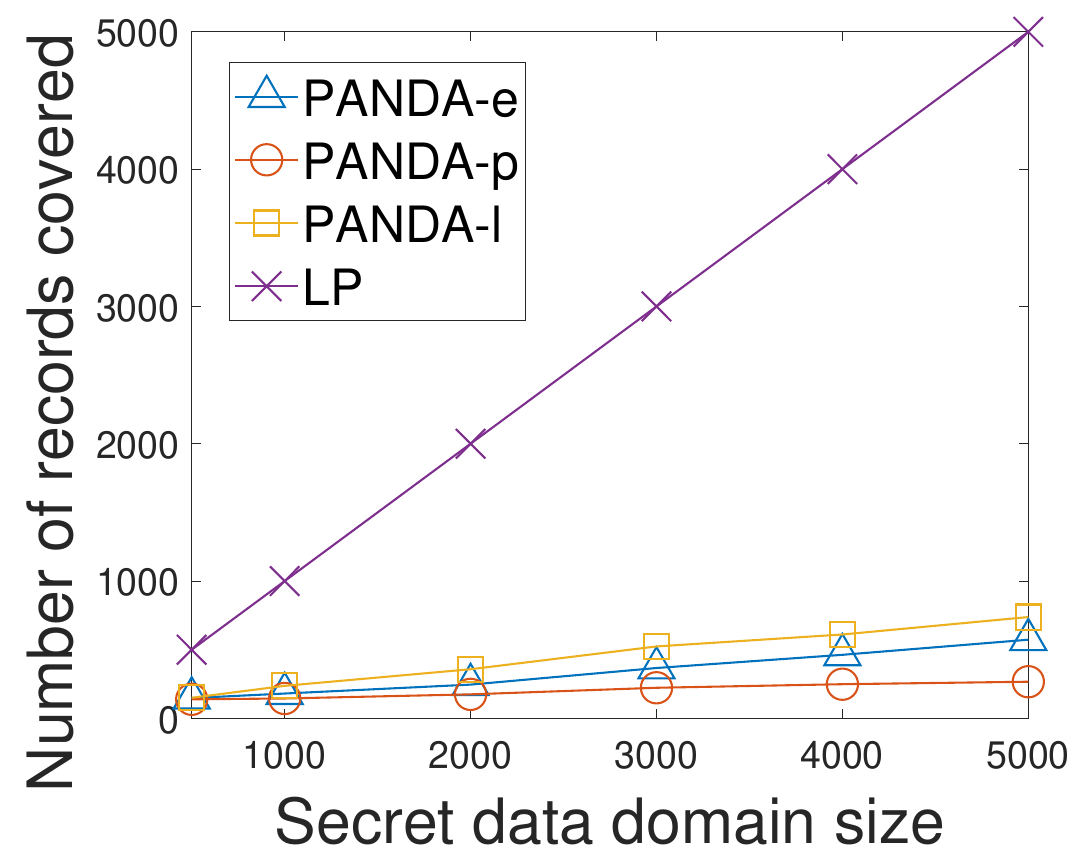}}
  \subfigure[NYC]{
\includegraphics[width=0.32\textwidth, height = 0.10\textheight]{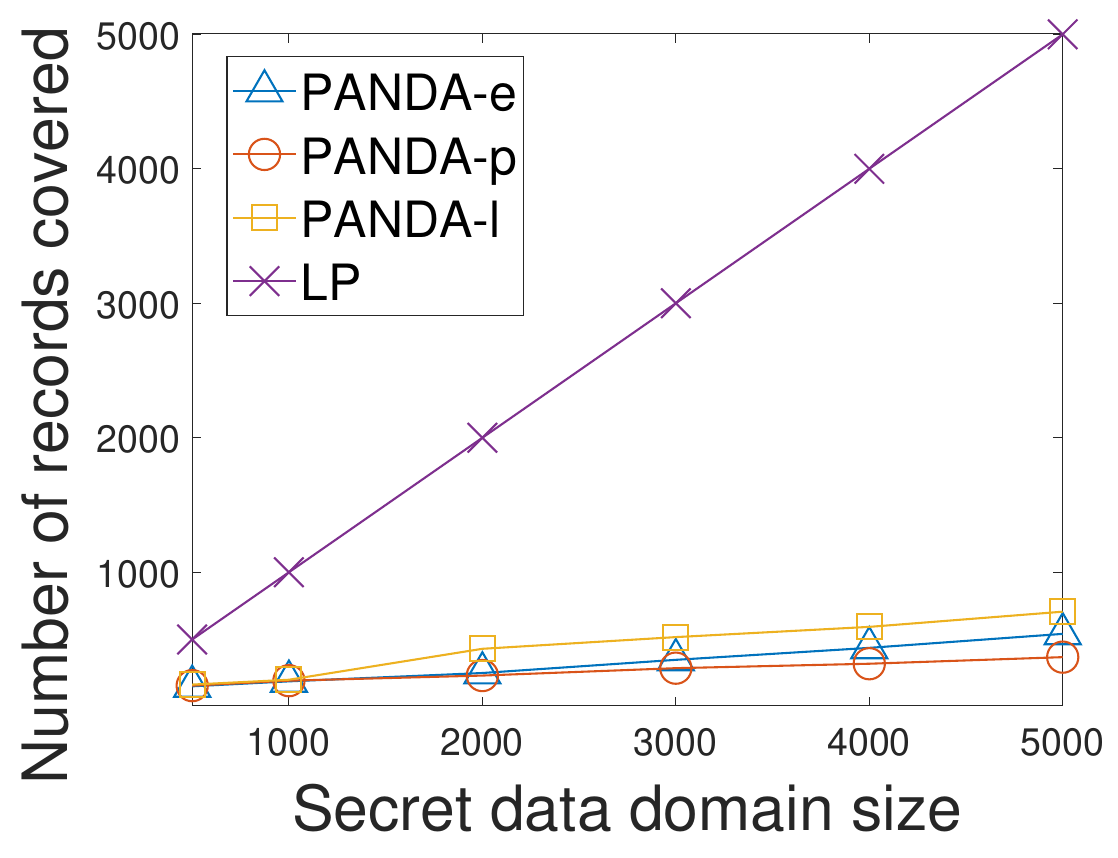}}
  \subfigure[London]{
\includegraphics[width=0.32\textwidth, height = 0.10\textheight]{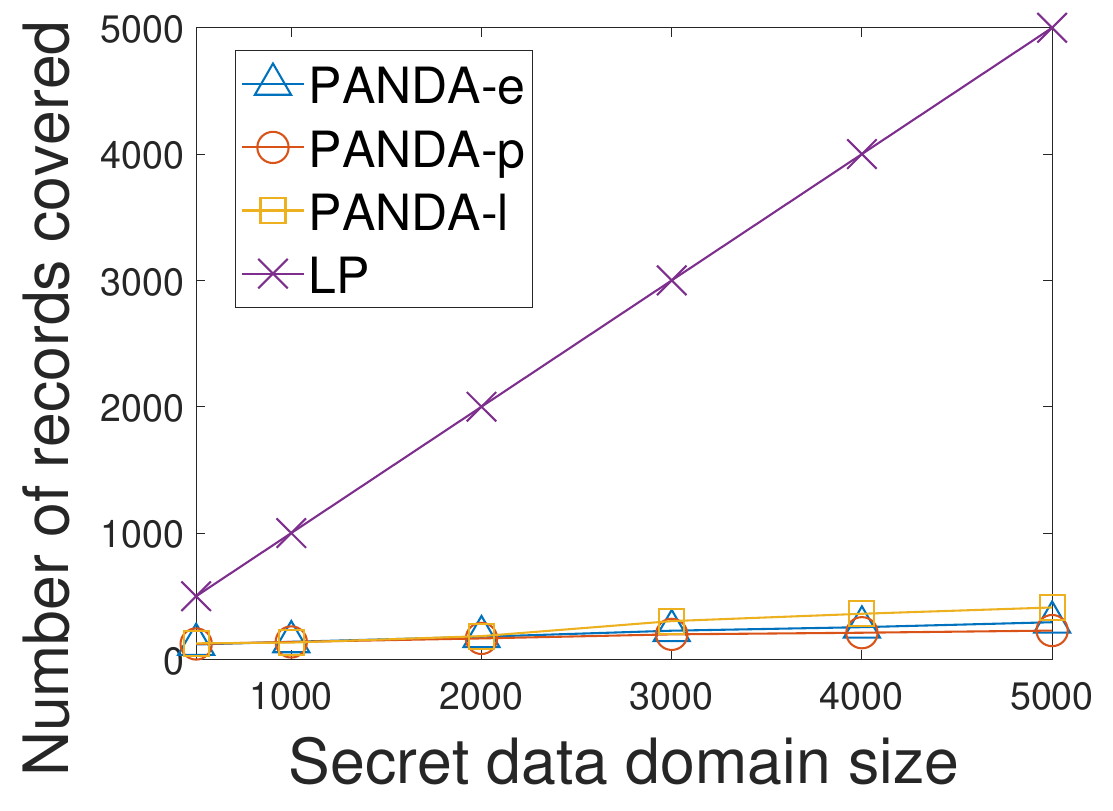}}
\vspace{-0.20in}
\end{minipage}
\caption{Number of records covered by LP vs. secret data domain size $K$.}
\label{fig:variablesvsnrrecords}
\vspace{-0.20in}
\end{figure}

Additionally, among the three methods, \textsc{PAnDA}-e achieves the fastest computation time ({\rev 22.18s}), followed by \textsc{PAnDA}-p ({\rev 34.64s}), and \textsc{PAnDA}-l is the slowest ({\rev 47.24s}) when the domain size is 3,000. This difference arises from the nature of their anchor selection distributions: exponential decay (used in \textsc{PAnDA}-e) concentrates anchor selection on nearby records, resulting in smaller, more localized optimization problems; power-law decay (\textsc{PAnDA}-p) selects a wider range of anchors, increasing problem size; and logistic selection (\textsc{PAnDA}-l) introduces a more uniform spread across distances, leading to denser anchor sets and higher computational overhead.

 % As shown in Table~\ref{tab:computation_time} (Table 2), PANDA maintains low computation time even as the secret data domain size grows to 3,000, in contrast to LP-based methods that become infeasible beyond 1,000 records. Together, Figure 7 and Table 2 highlight how PANDA’s anchor-based approximation substantially reduces optimization complexity without compromising performance.

\begin{figure}[t]
\centering
\hspace{0.00in}
\begin{minipage}{0.50\textwidth}
  \subfigure[Rome]{
\includegraphics[width=0.32\textwidth, height = 0.10\textheight]{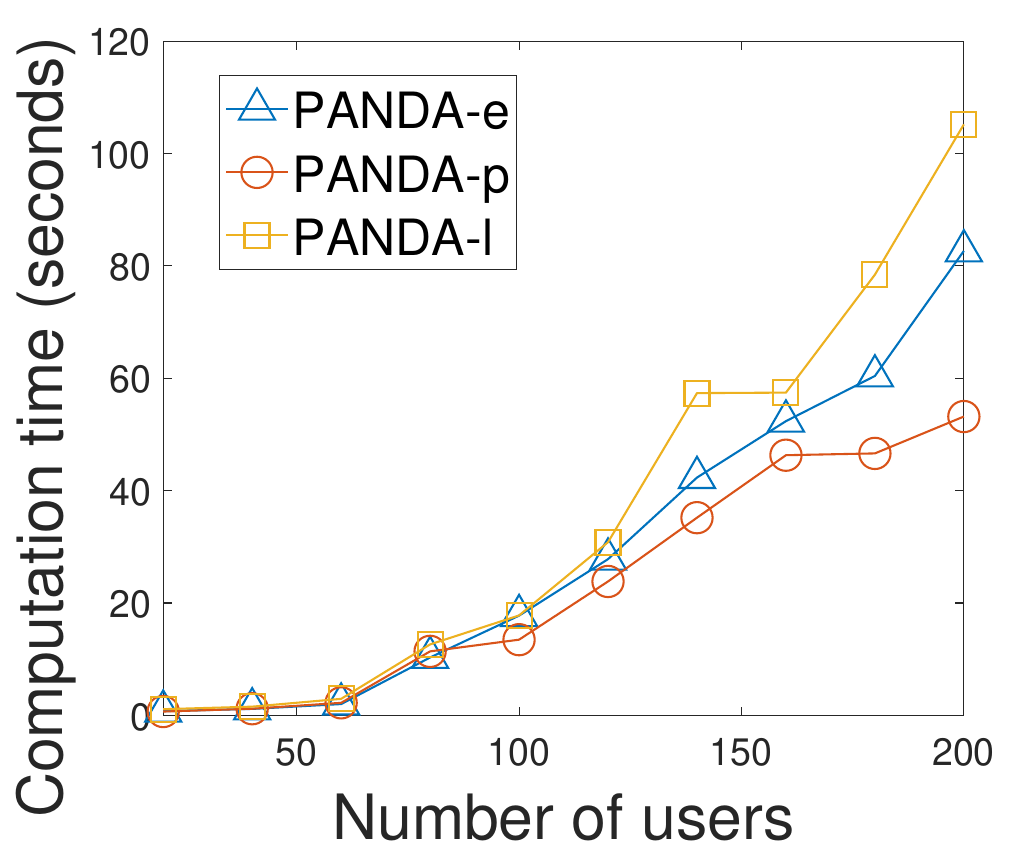}}
  \subfigure[NYC]{
\includegraphics[width=0.32\textwidth, height = 0.10\textheight]{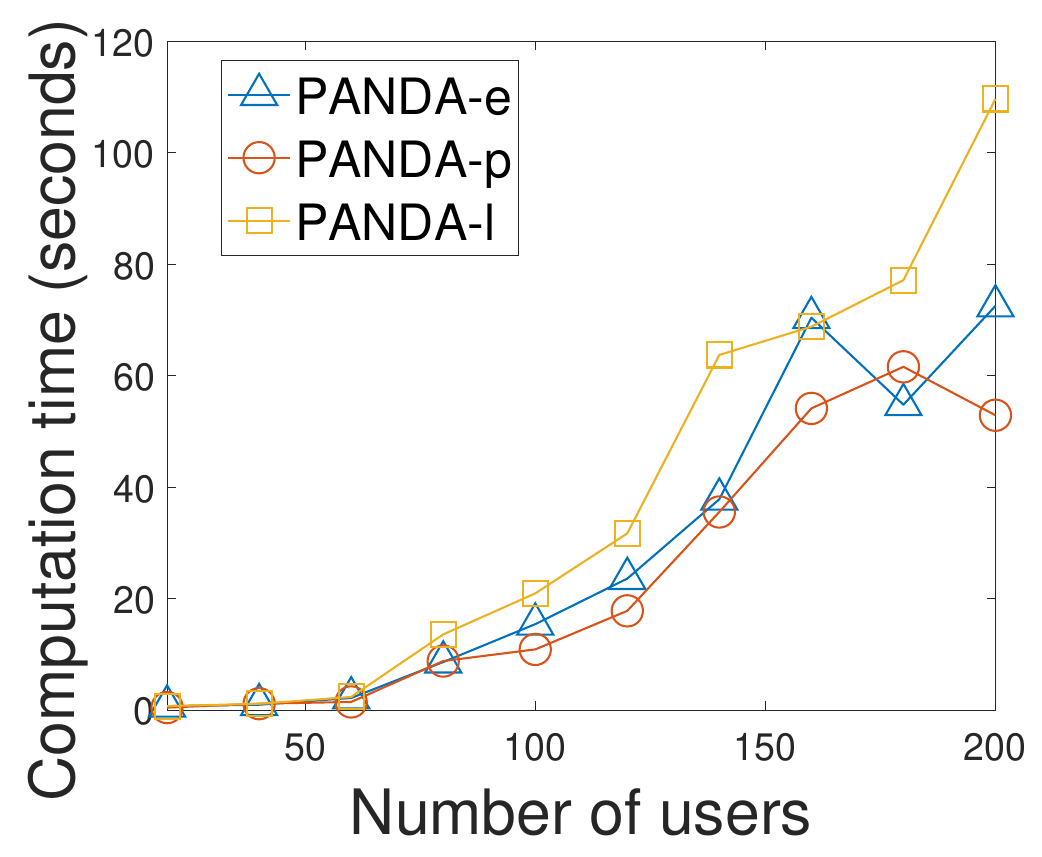}}
  \subfigure[London]{
\includegraphics[width=0.32\textwidth, height = 0.10\textheight]{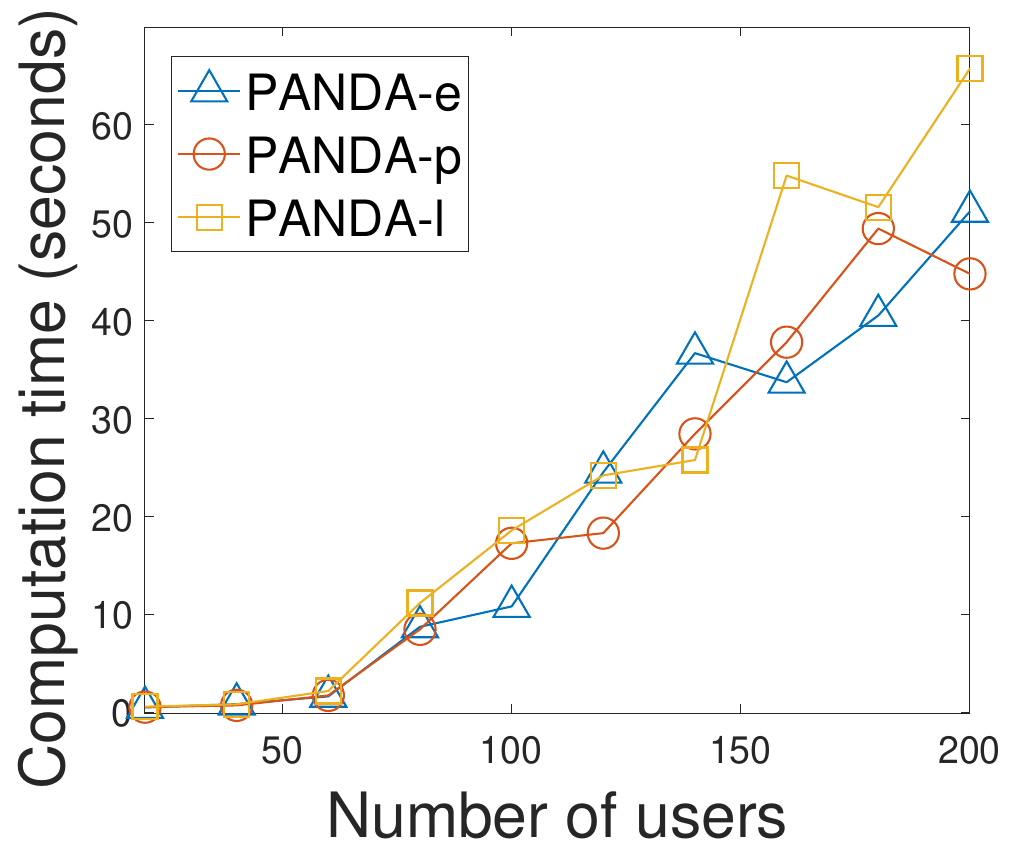}}
\vspace{-0.15in}
\end{minipage}
\caption{Computation time vs. numbers of users.}
\label{fig:timevsnrusers}
\begin{minipage}{0.50\textwidth}
  \subfigure[Rome]{
\includegraphics[width=0.32\textwidth, height = 0.10\textheight]{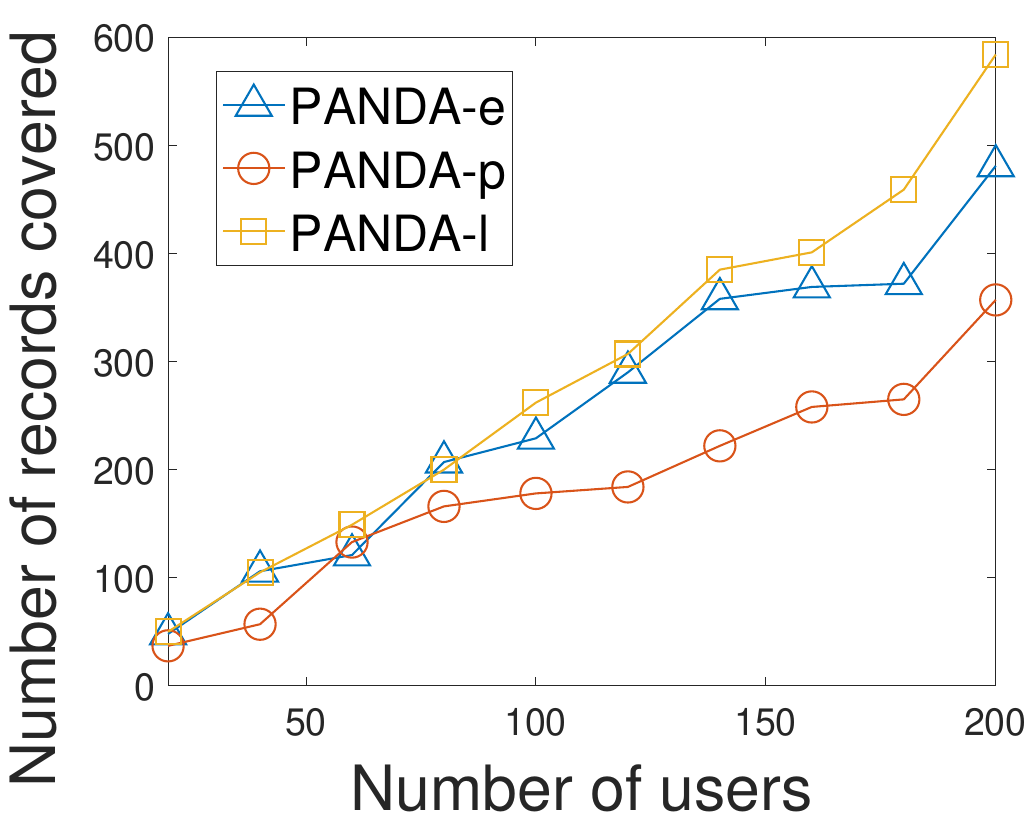}}
  \subfigure[NYC]{
\includegraphics[width=0.32\textwidth, height = 0.10\textheight]{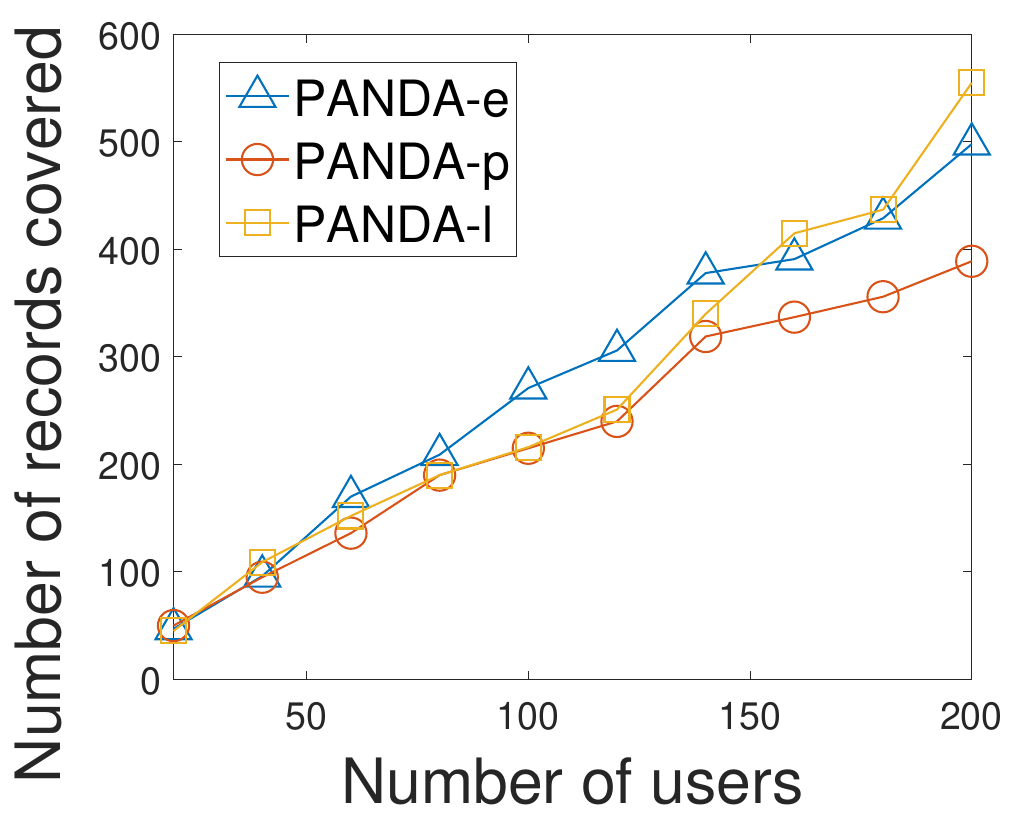}}
  \subfigure[London]{
\includegraphics[width=0.32\textwidth, height = 0.10\textheight]{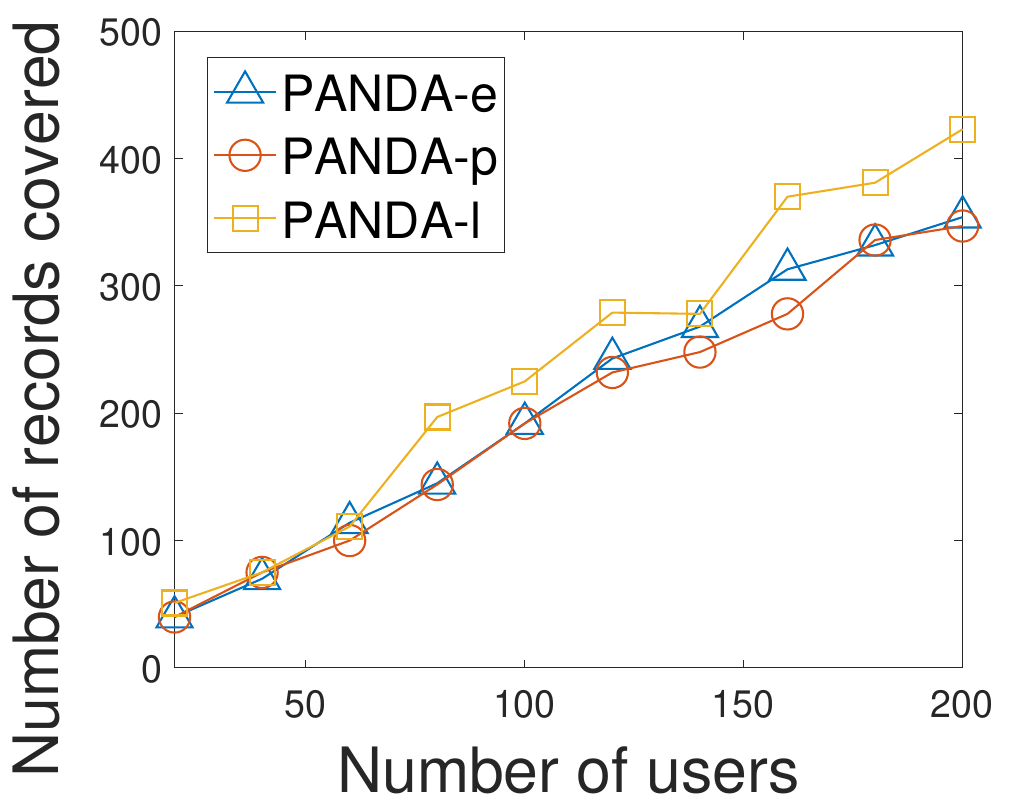}}
\vspace{-0.15in}
\end{minipage}
\caption{Number of decision variables vs. number of users.}
\label{fig:variablesvsnrusers}
\vspace{-0.25in}
\end{figure}

Unlike other methods, \textsc{PAnDA}'s computation time depends on both the domain size and the number of users, as each user contributes a personalized anchor set that expands the LP problem. To evaluate this user-dependent scalability, we measure the computation time of \textsc{PAnDA}-e, \textsc{PAnDA}-p, and \textsc{PAnDA}-l as the number of users increases from {\rev 20 to 200} across three datasets: Rome, NYC, and London (see Fig.\ref{fig:timevsnrusers} and Fig.\ref{fig:variablesvsnrusers}). This experiment is exclusive to \textsc{PAnDA}, as baseline methods do not perform user-level anchor selection and thus have user-independent computation time. 
As expected, \textsc{PAnDA}'s runtime grows with more users due to the expanding anchor set. Among variants, \textsc{PAnDA}-e is the fastest, followed by \textsc{PAnDA}-p and then \textsc{PAnDA}-l, reflecting differences in anchor sparsity: \textsc{PAnDA}-e favors localized anchors, while \textsc{PAnDA}-l selects broader sets, yielding denser LPs. % Fig.~\ref{fig:variablesvsnrusers} confirms that the total number of anchor records grows slowly with users, supporting \textsc{PAnDA}’s scalability via anchor-based approximation.

{\rev Notably, \textsc{PAnDA}'s computation time is largely driven by the size of the global anchor set, which aggregates anchors from all users. When this set grows close to the full domain—e.g., due to a large user population with many anchors per user—\textsc{PAnDA} may lose its computational advantage over LP-based methods. However, this can be mitigated by tuning anchor selection parameters such as the decay rate $\lambda$, scaling factor $\alpha$, and distance threshold $\gamma$ to limit the number of anchors per user.}

While LP+CA, EM, and EM+BR offer lower computational costs, this advantage comes at the expense of higher utility loss or greater mDP violation rates, as evidenced in Table~\ref{Tb:exp:ULscalability} {\rev and Table~\ref{Tb:exp:ULscalability_realdata}}.

% This occurs because these methods either approximate data perturbation to a coarse grained field, like LP+CA, or use a pre-defined noise distrbution without considering the diverse utility of data perturbation in different directiosn. 

% employ simpler perturbation strategies that prioritize computational efficiency over data utility. By applying more straightforward or aggressive perturbations, they reduce the complexity and time required for computations but at the cost of greater distortion to the original data, leading to increased utility loss.​

% This comparison clearly highlights the efficiency gains provided by the our approach in handling large-scale mDP problems across different geographic locations and varying privacy budgets.

\subsubsection{Utility loss.} We compare the utility loss of different data perturbation methods on the three roadmap datasets. {\rev Utility loss is defined as the expected discrepancy in travel costs between the true location and the perturbed location, measured with respect to a set of randomly placed tasks. Specifically, the utility loss incurred by reporting a perturbed record $y$ instead of the true record $x$ is computed as: $\textstyle c_{x,y} = \sum_{x_{\mathrm{task}} \in \mathcal{X}} p_{x_{\mathrm{task}}} \left| \text{travel}(x, x_{\mathrm{task}}) - \text{travel}(y, x_{\mathrm{task}}) \right|$, 
where $p_{x_{\mathrm{task}}}$ denotes the probability of each task location $x_{\mathrm{task}} \in \mathcal{V}$, and $\text{travel}(x, x_{\mathrm{task}})$ (resp. $\text{travel}(y, x_{\mathrm{task}})$) represents the travel cost from the true location $x$ (resp. perturbed location $y$) to the task location $x_{\mathrm{task}}$.
}

Tables~\ref{Tb:exp:ULscalability} and \ref{Tb:exp:ULscalability_realdata} report the utility loss (in meters) for each algorithm as the secret domain size $K = |\mathcal{X}|$ increases from 500 to 5,000, under uniform and real user location distributions, respectively. Our approach, \textsc{PAnDA}, shows competitive utility loss across all datasets compared to all other methods. In Table \ref{Tb:exp:ULscalability}, on average across all domain sizes, \textsc{PAnDA}-e achieves the high utility improvement, reducing utility loss by {\bl82.37\%} over EM, {\bl78.32\%} over EM+BR, and {\bl70.16\%} over LP+EM. \textsc{PAnDA}-p follows closely, with improvements of {\bl82.98\%}, {\bl79.08\%}, and {\bl71.73\%} over the same baselines. \textsc{PAnDA}-l also demonstrates substantial gains, reducing utility loss by {\bl82.44\%} compared to EM, {\bl78.42\%} compared to EM+BR, and {\bl68.73\%} compared to LP+EM. In Table \ref{Tb:exp:ULscalability_realdata}, on average across all domain sizes, \textsc{PAnDA}-e achieves the high utility improvement, reducing utility loss by {\bl 80.75\%} over EM, {\bl 76.39\%} over EM+BR, and {\bl 69.27\%} over LP+EM. \textsc{PAnDA}-p follows closely, with improvements of {\bl 80.77\%}, {\bl 76.47\%}, and {\bl 69.45\%} over the same baselines. \textsc{PAnDA}-l also demonstrates substantial gains, reducing utility loss by {\bl 79.95\%} compared to EM, {\bl 75.46\%} compared to EM+BR, and {\bl 68.46\%} compared to LP+EM. Note that while LP+BD achieves the lowest utility loss at a domain size of 500, it fails to scale to domain sizes of 1,000 or more due to its high computational overhead.

\vspace{-0.00in}
\begin{table}[t]
\centering
\small 
% \footnotesize 
\scriptsize 
\setlength{\tabcolsep}{1pt}  
\begin{tabular}{ c|c|c|c|c|c|c}
%\cline{2-13}
%\hline
\toprule
\multicolumn{7}{ c  }{Rome}\\ 
\cline{1-7}
\multicolumn{1}{ c|  }{Method}  & $K = 500$ & $K = 1,000$  & $K = 2,000$  & $K = 3,000$& $K = 4,000$& $K = 5,000$\\
\hline
\hline
\multicolumn{1}{ c|  }{EM} & 1315.27±98.35 & 1322.43±85.16  & 1402.29±95.17  & 1360.51±74.82 & 1412.43±97.45 & 1387.18±72.94\\
% \multicolumn{1}{ c|  }{Laplace} & x.xx$\pm$x.x & x.xx$\pm$x.x & x.xx$\pm$x.x & x.xx$\pm$x.x & x.xx$\pm$x.x &  x.xx$\pm$x.x \\ 
\multicolumn{1}{ c|  }{LP} & - & -  & -  & 
- & - & -\\ 
\multicolumn{1}{ c|  }{LP+CA} & 2127.95±191.52 & 2169.68±95.67 & 2172.65±68.04 & 2170.7±127.4 & 2157.70±81.41 & 2200.32±99.75\\ 
\multicolumn{1}{ c|  }{LP+BD} & \textbf{167.34±46.52} & -  & - & - & - & -\\ 
\multicolumn{1}{ c|  }{LP+EM} & 792.81±74.78 & -  & - & 
-  & - & -\\ 
\multicolumn{1}{ c|  }{EM+BR} & 964.25±36.04 & 1019.93±55.25 & 1179.39±52.41  & 1142.15±82.24 & 1152.90±51.86 & 1170.21±81.78 \\ 
\hline
\multicolumn{1}{ c|  }{ \textbf{\textsc{PAnDA}-e}} & 238.77±39.76 & 251.64±16.35  & 269.88±55.63  & 277.88±39.74 & 267.04±54.35 & 275.01±40.15\\
\multicolumn{1}{ c|  }{ \textbf{\textsc{PAnDA}-p}} & 243.59±44.33 & 256.87±33.75  & \textbf{269.10±48.16} & \textbf{272.35±23.91} & 273.93±56.38 & 276.11±40.69\\
\multicolumn{1}{ c|  }{ \textbf{\textsc{PAnDA}-l}} & 249.74±55.41 & \textbf{240.81±20.73}  & 271.27±43.45  & 289.37±73.14 & 277.36±43.29 & 296.55±73.12 \\
\toprule
\multicolumn{7}{ c  }{NYC}\\ 
\cline{1-7}
\multicolumn{1}{ c|  }{Method}  & $K = 500$ & $K = 1,000$  & $K = 2,000$  & $K = 3,000$& $K = 4,000$& $K = 5,000$\\
\hline
\hline
\multicolumn{1}{ c|  }{EM} & 1472.57±144.95 & 1651.86±180.61  & 1849.84±157.82  & 1828.18±169.63 & 1894.24±156.53 & 1907.32±153.93\\
% \multicolumn{1}{ c|  }{Laplace} & x.xx$\pm$x.x & x.xx$\pm$x.x & x.xx$\pm$x.x & x.xx$\pm$x.x & x.xx$\pm$x.x &  x.xx$\pm$x.x \\ 
\multicolumn{1}{ c|  }{LP} & -  & - & - & 
- & -  & -\\ 
\multicolumn{1}{ c|  }{LP+CA} & 2121.57±121.71 & 2126.75±89.13 & 2108.98±54.97 & 2134.18±59.51 & 2199.11±90.60 & 2170.25±83.22\\ 
\multicolumn{1}{ c|  }{LP+BD} & \textbf{174.25±58.91} & - & -  & - & -  & -\\ 
\multicolumn{1}{ c|  }{LP+EM} & 817.66±96.23  & - & - & 
-  & -  & -\\ 
\multicolumn{1}{ c|  }{EM+BR} & 1353.97±116.35 & 1411.72±95.63 & 1419.84±84.85  & 1446.86±48.38  & 1467.51±122.15 & 1457.32±47.42\\ 
\hline
\multicolumn{1}{ c|  }{ \textbf{\textsc{PAnDA}-e}}& 246.41±105.18 & 304.14±89.95  & 336.08±101.03 & \textbf{301.25±104.75} & 320.44±63.98 & 344.81±82.48
 \\
\multicolumn{1}{ c|  }{ \textbf{\textsc{PAnDA}-p}} & 212.48±90.52 & \textbf{251.65±59.26}  & \textbf{275.34±86.68}  & 319.86±115.08 & 339.58±105.44 & 346.97±112.33\\
\multicolumn{1}{ c|  }{ \textbf{\textsc{PAnDA}-l}} & 257.34±82.15 & 279.24±72.28  & 294.49±103.75  & 329.14±106.46 & 322.86±107.63 & 324.47±116.21\\
\hline
\toprule
\multicolumn{7}{ c  }{London}\\ 
\cline{1-7}
\multicolumn{1}{ c|  }{Method}  & $K = 500$ & $K = 1,000$  & $K = 2,000$ & $K = 3,000$& $K = 4,000$& $K = 5,000$\\
\hline
\hline
\multicolumn{1}{ c|  }{EM} & 1364.90±96.62 & 1422.15±115.36  & 1551.58±86.94  & 1643.75±80.41 & 1644.96±58.94 & 1689.06±80.06\\
% \multicolumn{1}{ c|  }{Laplace} & x.xx$\pm$x.x & x.xx$\pm$x.x & x.xx$\pm$x.x & x.xx$\pm$x.x & x.xx$\pm$x.x &  x.xx$\pm$x.x \\ 
\multicolumn{1}{ c|  }{LP} & -  & - & - & 
- & -  & -\\ 
\multicolumn{1}{ c|  }{LP+CA} & 1961.07±154.34 & 1928.75±71.72 & 1963.92±71.11 & 2080.37±66.07 & 2063.05±73.56 & 2167.44±74.19 \\ 
\multicolumn{1}{ c|  }{LP+BD} & \textbf{168.75±40.34} & - & -  & -& -  & -\\ 
\multicolumn{1}{ c|  }{LP+EM} & 776.87±59.83  & - & - & 
-  & -  & -\\ 
\multicolumn{1}{ c|  }{EM+BR} & 1148.13±48.75 & 1192.43±40.34  & 1295.75±62.38  & 1306.15±44.17  & 1358.14±68.48 & 1389.20±36.95\\ 
\hline
\multicolumn{1}{ c|  }{ \textbf{\textsc{PAnDA}-e}} & 227.23±60.42 & \textbf{184.22±72.28}  & 262.84±70.14  & 273.74±56.26 & 281.26±98.98 & 296.67±86.03\\
\multicolumn{1}{ c|  }{ \textbf{\textsc{PAnDA}-p}} & 218.90±53.05 & 192.78±74.17  & \textbf{229.77±44.35}  & 258.36±45.72 & 271.58±110.04 & 277.08±47.09\\
\multicolumn{1}{ c|  }{ \textbf{\textsc{PAnDA}-l}} & 239.50±52.27 & 241.65±46.92  & 242.54±111.48  & \textbf{247.78±67.39} & 264.19±101.97 & 270.08±42.22\\
\hline
\end{tabular}
\vspace{0.00in}
\caption{Utility loss (in meters) of different algorithms (uniform user location distribution). Mean$\pm$1.96$\times$std. deviation.}
\label{Tb:exp:ULscalability}
\vspace{-0.25in}
\end{table}

\vspace{-0.00in}
\begin{table}[t]
\centering
\small 
% \footnotesize 
\scriptsize 
\setlength{\tabcolsep}{1pt}  
\begin{tabular}{ c|c|c|c|c|c|c}
%\cline{2-13}
%\hline
\toprule
\multicolumn{7}{ c  }{Rome}\\ 
\cline{1-7}
\multicolumn{1}{ c|  }{Method}  & $K = 500$ & $K = 1,000$  & $K = 2,000$ & $K = 3,000$& $K = 4,000$& $K = 5,000$\\
\hline
\hline
\multicolumn{1}{ c|  }{EM} & 1284.57±88.66 & 1291.24±105.95 & 1401.66±74.91 & 1380.45±85.99 & 1409.63±113.08 & 1375.11±83.01 \\
\multicolumn{1}{ c|  }{LP} & - & - & - & - & - & - \\ 
\multicolumn{1}{ c|  }{LP+CA} & 2114.77±228.56 & 2146.51±94.29 & 2156.15±68.28 & 2211.00±138.64 & 2148.34±86.01 & 2176.42±96.01 \\ 
\multicolumn{1}{ c|  }{LP+BD}  & \textbf{165.35±44.11} & - & - & - & - & - \\ 
\multicolumn{1}{ c|  }{LP+EM} & 787.33±72.01 & - & - & - & - & - \\ 
\multicolumn{1}{ c|  }{EM+BR} & 987.44±27.17 & 1037.52±62.61 & 1168.49±55.43 & 1159.00±62.87 & 1148.25±59.58 & 1153.99±82.64 \\  
\hline
\multicolumn{1}{ c|  }{ \textbf{\textsc{PAnDA}-e}} & 241.97±46.32 & \textbf{244.24±16.59} & 262.12±41.85 & 264.79±41.61 & 285.67±41.70 & \textbf{272.18±34.80} \\ 
\multicolumn{1}{ c|  }{ \textbf{\textsc{PAnDA}-p}} & \textbf{240.53±38.80} & 254.01±41.85 & \textbf{256.69±39.44} & \textbf{264.48±25.87} & \textbf{267.02±54.64} & 282.85±48.60 \\ 
\multicolumn{1}{ c|  }{ \textbf{\textsc{PAnDA}-l}} &248.30±55.28 & 249.98±20.21 & 269.71±33.07 & 297.14±79.77 & 284.43±47.39 & 283.20±55.28 \\ 
\hline
\end{tabular}
\vspace{0.00in}
\caption{Utility loss (in meters) of different algorithms (Rome taxicab location dataset). Mean$\pm$1.96$\times$std. deviation. }
\label{Tb:exp:ULscalability_realdata}
\vspace{-0.25in}
\end{table}

\begin{figure}[t]
\centering
\hspace{0.00in}
\begin{minipage}{0.50\textwidth}
  \subfigure[Rome]{
\includegraphics[width=0.32\textwidth, height = 0.10\textheight]{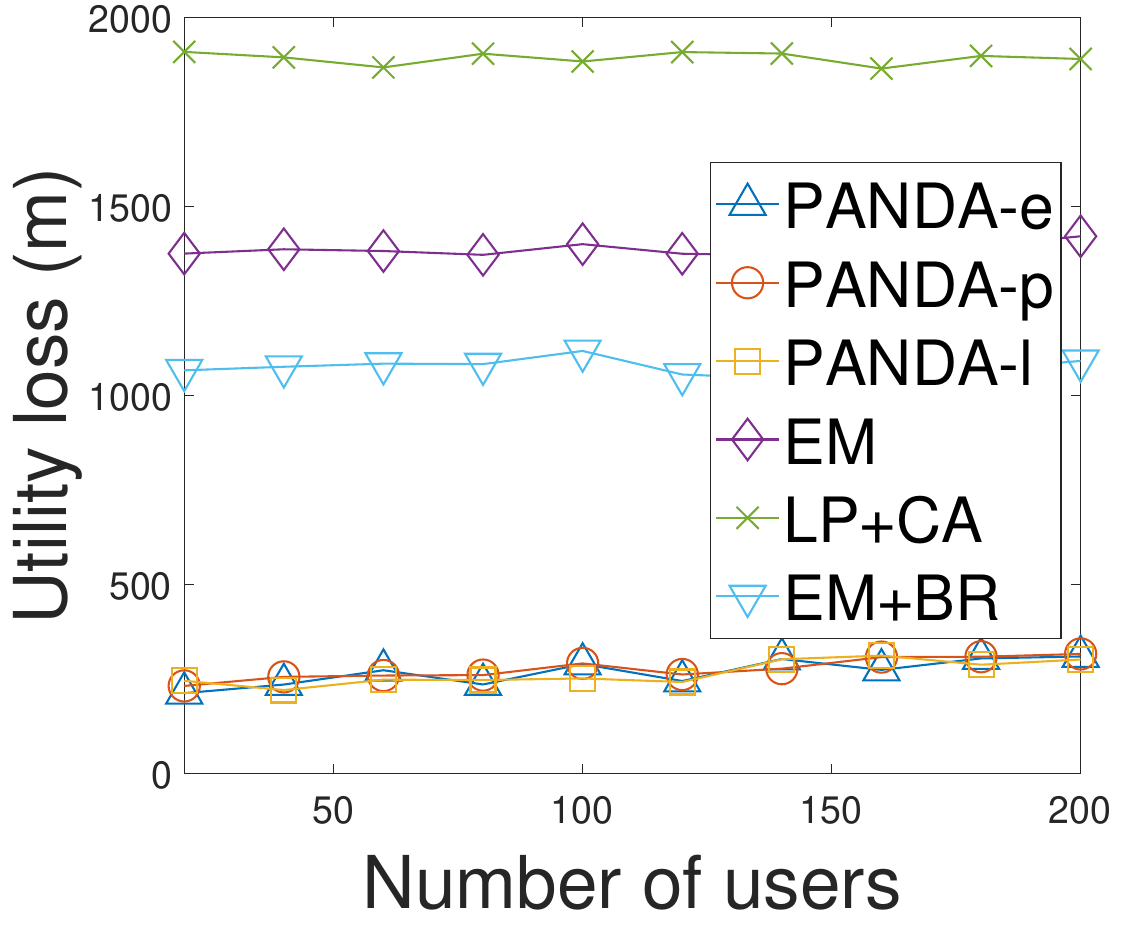}}
  \subfigure[NYC]{
\includegraphics[width=0.32\textwidth, height = 0.10\textheight]{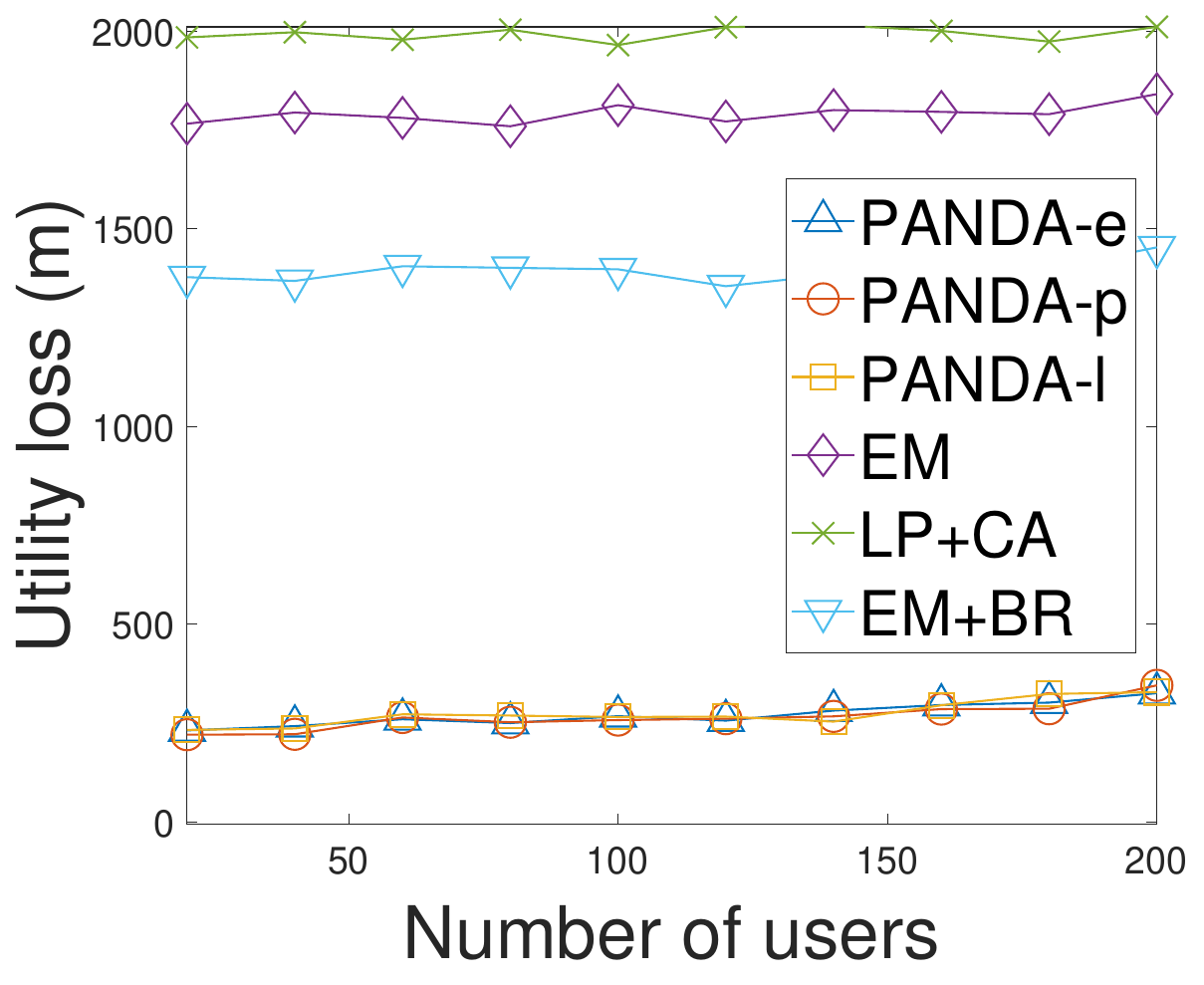}}
  \subfigure[London]{
\includegraphics[width=0.32\textwidth, height = 0.10\textheight]{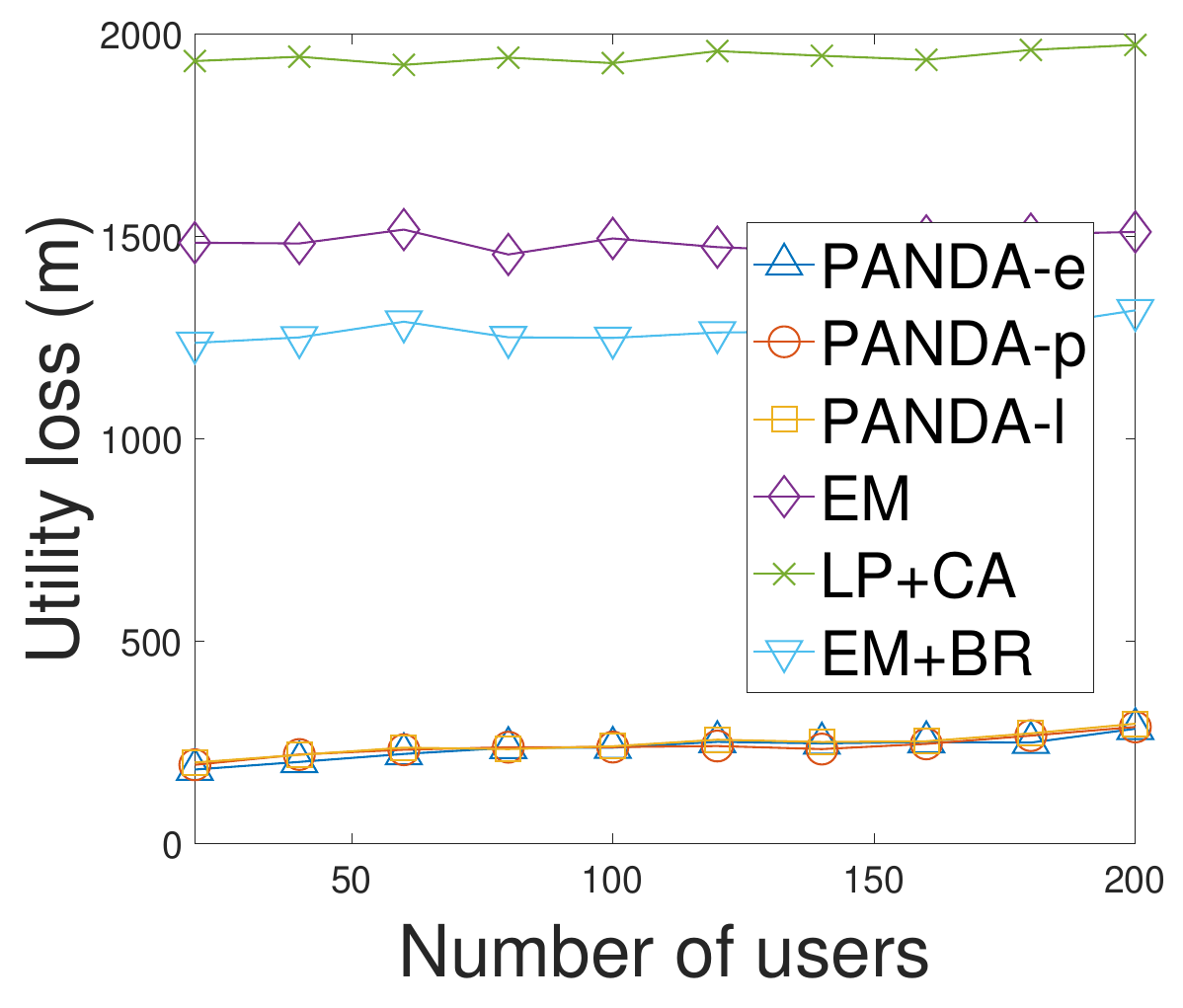}}
\vspace{-0.20in}
\end{minipage}
\caption{Utility loss vs. number of users.}
\label{fig:ULvsnrusers}
\vspace{-0.20in}
\end{figure}

\begin{figure}[t]
\centering
\hspace{0.00in}
\begin{minipage}{0.50\textwidth}
  \subfigure[Rome]{
\includegraphics[width=0.32\textwidth, height = 0.10\textheight]{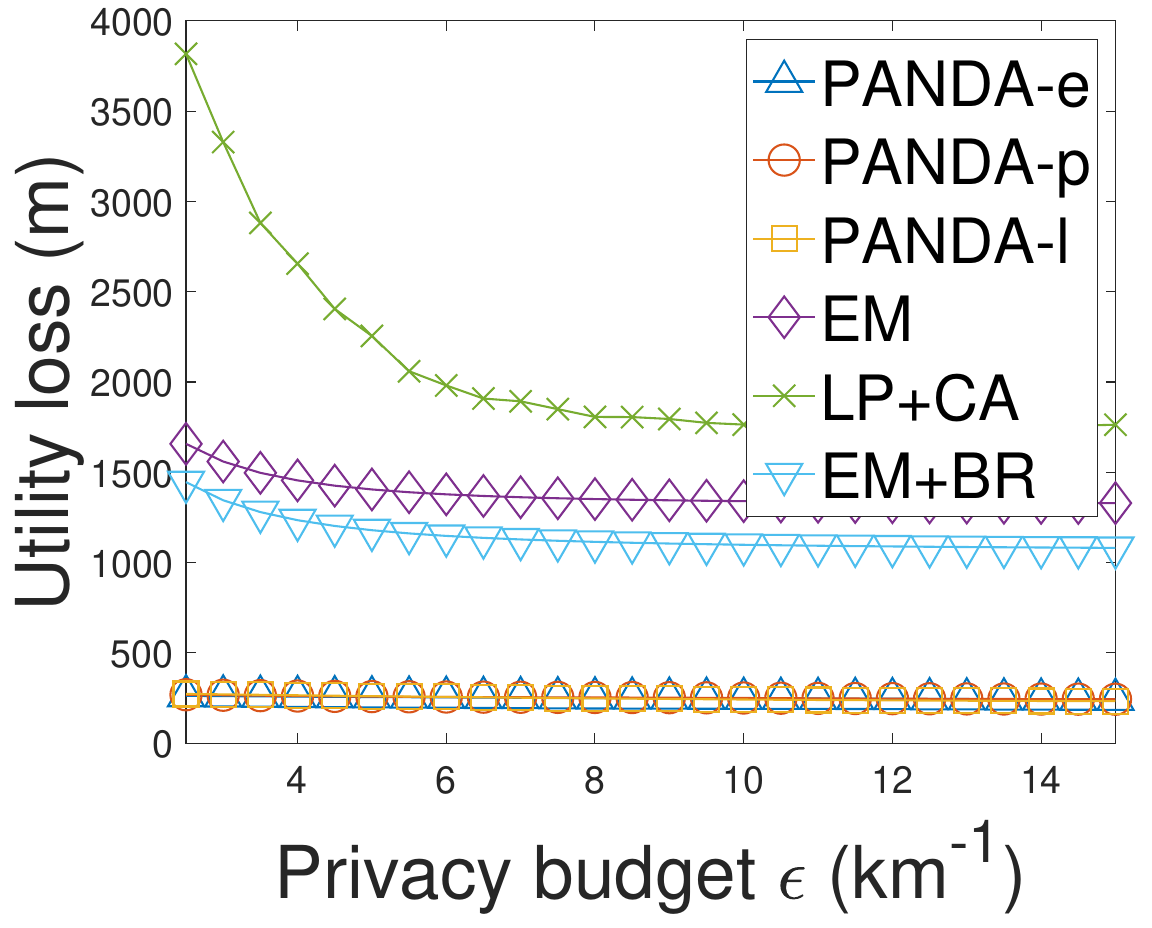}}
  \subfigure[NYC]{
\includegraphics[width=0.32\textwidth, height = 0.10\textheight]{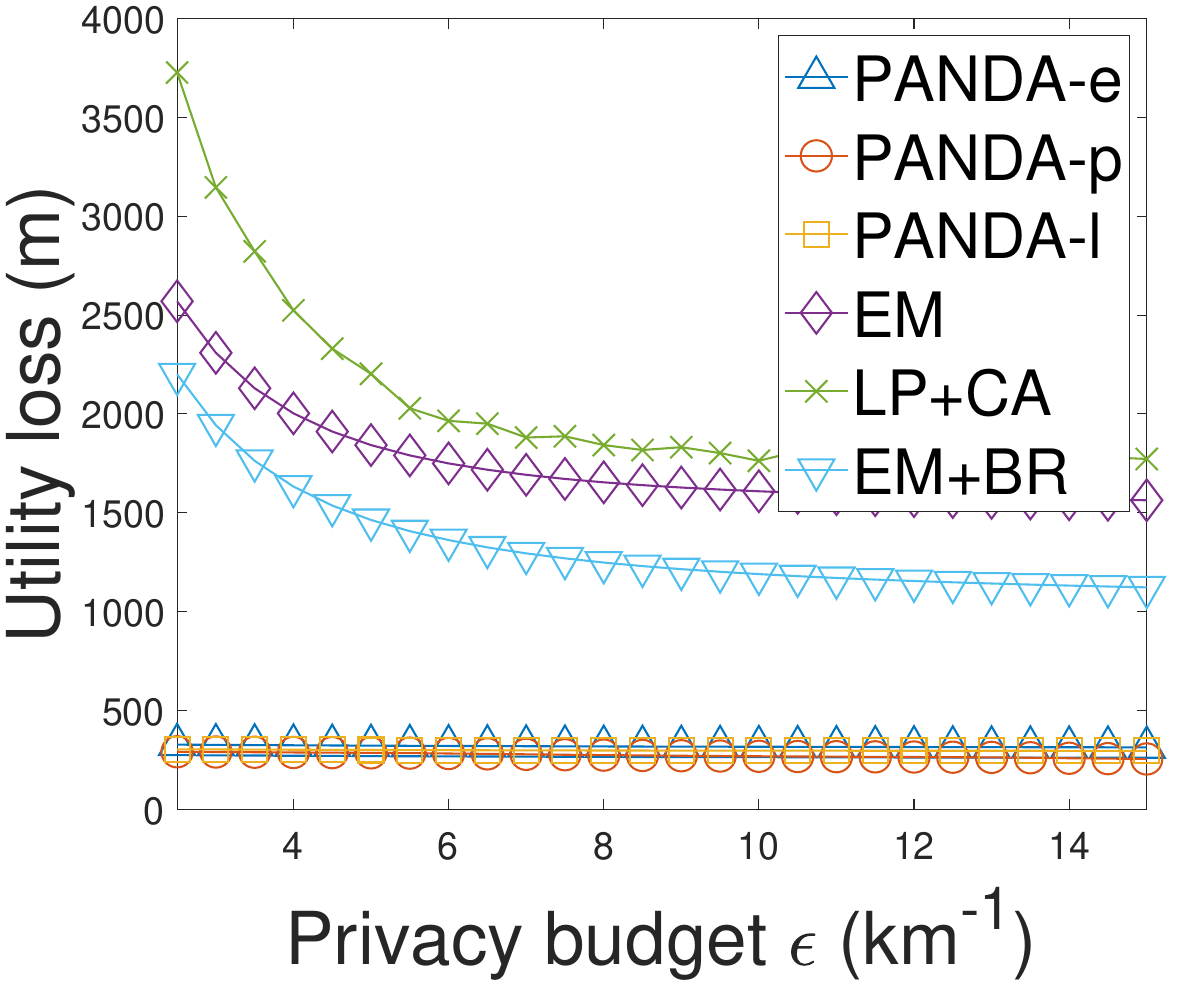}}
  \subfigure[London]{
\includegraphics[width=0.32\textwidth, height = 0.10\textheight]{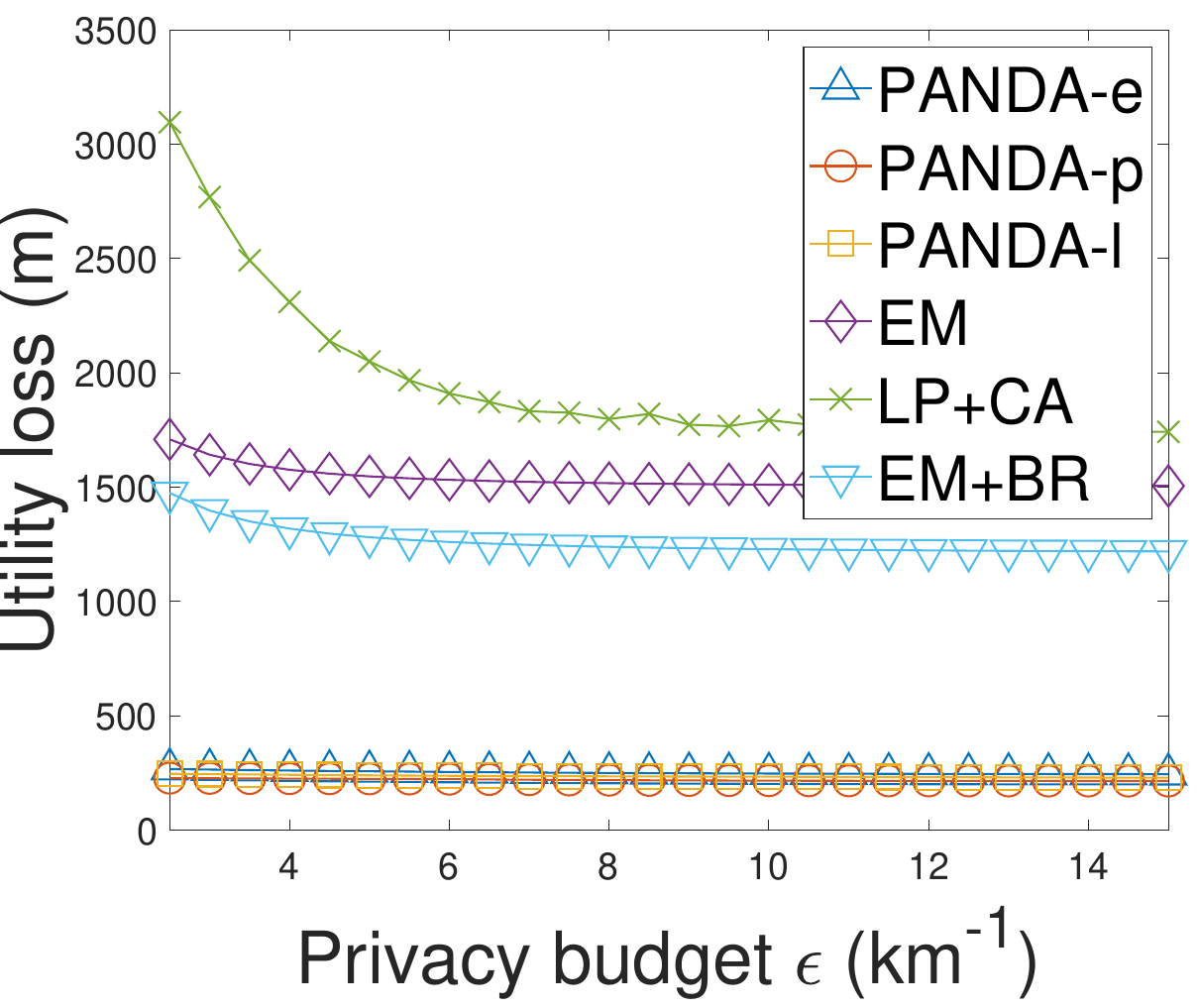}}
\vspace{-0.20in}
\end{minipage}
\caption{Utility loss vs. privacy budget$\epsilon$.}
\label{fig:ULvsbudget}
\vspace{-0.25in}
\end{figure}

% \textsc{PAnDA} outperforms EM, LP+CA, LP+EM, and EM+BR in terms of utility as its personalized anchors-based design. 
Unlike EM and LP+CA mechanisms, which apply fixed or coarse perturbation strategies across the entire data domain, \textsc{PAnDA} selects personalized anchor sets based on each user's real record. These anchors capture the local geometry of the input space, allowing the server to optimize perturbation probabilities that more accurately reflect the utility sensitivity of nearby records. Furthermore, compared to LP+EM, which relies on static anchors and interpolation, \textsc{PAnDA}'s anchor selection is probabilistic and data-aware, reducing approximation errors. Unlike EM+BR, which applies EM noise and then remaps, \textsc{PAnDA} performs end-to-end utility optimization tailored to privacy constraints. As a result, \textsc{PAnDA} consistently achieves lower utility loss while preserving strong privacy guarantees. \looseness = -1

% Additionally, among the \textsc{PAnDA} variants, \textsc{PAnDA}-e achieves the best utility performance across all domain sizes, followed closely by \textsc{PAnDA}-p and then \textsc{PAnDA}-l. This difference arises from their respective anchor selection strategies: \textsc{PAnDA}-e adopts an exponential decay function, strongly favoring nearby anchors, which results in tighter local approximation and lower utility loss. \textsc{PAnDA}-p uses a modified power-law decay, which decays more gradually, allowing slightly more diverse anchor selections while still retaining locality. \textsc{PAnDA}-l, which relies on a logistic selection function, spreads anchor selection more evenly across distances. While this improves coverage, it weakens the locality of perturbation approximation, leading to marginally higher utility loss, especially as domain size increases. 

For theoretical interest, we compare \textsc{PAnDA}-e, \textsc{PAnDA}-p, and \textsc{PAnDA}-l with their respective lower bounds (derived from \textbf{Proposition~\ref{prop:ULbound}}) to evaluate how closely each method approaches the optimal utility. On average, the approximation ratios for \textsc{PAnDA}-e, \textsc{PAnDA}-p, and \textsc{PAnDA}-l are {\bl1.3608, 1.3708, and 1.3630}, respectively, across the three datasets and different domain size $K$. The \emph{approximation ratio} is defined as the ratio between the utility loss of the evaluated mechanism and the lower bound on utility loss. A ratio close to 1 indicates that the mechanism achieves near-optimal performance. Detailed results are presented in \textbf{Figures~\ref{fig:boundPAnDA_rome}, \ref{fig:boundPAnDA_nyc}, and~\ref{fig:boundPAnDA_london} in Section~\ref{sec:addexp} in Appendix}.

% \noindent \textbf{A guideline to help future users choose among the three PAnDA variants:} 
{\rev Based on our experimental results in Tables \ref{Tb:exp:time_scalability}--\ref{Tb:exp:ULscalability_realdata}, and Figures~\ref{fig:boundPAnDA_rome}--\ref{fig:boundPAnDA_london}, we also provide a guideline of choosing the three PAnDA variants: (1) \textbf{\textsc{PAnDA}-l (Logistic decay)} features a smoother decay and performs well when the privacy metric varies gradually with distance. It demonstrates stable performance and is relatively insensitive to hyperparameter choices. (2) \textbf{\textsc{PAnDA}-p (Power-law decay)} strikes a strong balance between locality and diversity in anchor selection. It often delivers high utility, particularly when the utility loss is sensitive to moderately distant records. (3) \textbf{\textsc{PAnDA}-e (Exponential decay)} yields more localized anchor sets, which can be advantageous when the utility loss is sharply concentrated around the true record. However, it may reduce anchor diversity and be less effective in sparse domains. As a general recommendation, we suggest using \textsc{PAnDA}-p as a default, especially in large-scale domains with non-uniform density. Ultimately, the choice of variant should be guided by the application's utility function and its sensitivity to anchor locality and diversity.}

{\rev Figure~\ref{fig:ULvsnrusers}(a)(b)(c) compares the utility loss of \textsc{PAnDA}-e, \textsc{PAnDA}-p, and \textsc{PAnDA}-l as the number of users increases from 20 to 200. On average, \textsc{PAnDA}-e achieves utility loss reductions of {\bl 83.37\%}, {\bl 79.27\%}, and {\bl 86.68\%} over EM, EM+BR, and LP+CA, respectively. \textsc{PAnDA}-p and \textsc{PAnDA}-l follow closely, with respective improvements of {\bl 83.21\%}, {\bl 79.07\%}, {\bl 86.56\%} and {\bl 83.15\%}, {\bl 78.99\%}, {\bl 86.50\%} over the same baselines. \looseness = -1}

{\rev Additionally, Figure \ref{fig:ULvsbudget}(a)(b)(c) compares the utility loss of different methods when $\epsilon$ is incresed from 2.5km$^{-1}$ to 15km$^{-1}$. As shown across all three datasets, our proposed methods (\textsc{PAnDA}-e, \textsc{PAnDA}-p, \textsc{PAnDA}-l) consistently outperform the baselines (EM, LP+CA, and EM-BR) in terms of utility, {\rev achieving at least {\bl 79.12\%} lower utility loss}. Moreover, we observe that the utility loss decreases as $\epsilon$ increases for all perturbation methods. This is expected, as a larger privacy budget $\epsilon$ permits the addition of less noise or distortion, thereby preserving more of the original data utility while relaxing the strictness of the privacy guarantee.}

% Figure \ref{fig:time_scatter} and Figure \ref{fig:loss_scatter} demonstrate the results on two critical metrics: computation time and loss. \emph{Computation Time}: As shown in Figure \ref{fig:time_scatter}, the computation time of our method (red points) remains consistently low, even as secret data domain size exceeds 1000. In contrast, the benchmark method (blue points) and \emph{time\_benders} method (green triangles) show a sharp increase in computation time, with values exceeding the 1,800-second limit for record sizes greater than 600. This illustrates the superior scalability of our method, particularly when handling large datasets.

% \emph{Loss}: Figure \ref{fig:loss_scatter} compares the loss of our method with that of the benchmark and \emph{time\_benders}. Our method (red points) maintains a significantly lower and more stable loss across all record sizes. The benchmark loss (blue points) and \emph{time\_benders} loss (green points for values under 1,800, green triangles for values exceeding 1,800) both exhibit higher and more variable loss values, particularly as secret data domain size increases. This demonstrates the robustness of our approach, achieving lower error rates and providing more reliable results for large datasets.

\begin{figure}[t]
\centering
\hspace{0.00in}
\begin{minipage}{0.50\textwidth}
  \subfigure[Rome]{
\includegraphics[width=0.32\textwidth, height = 0.10\textheight]{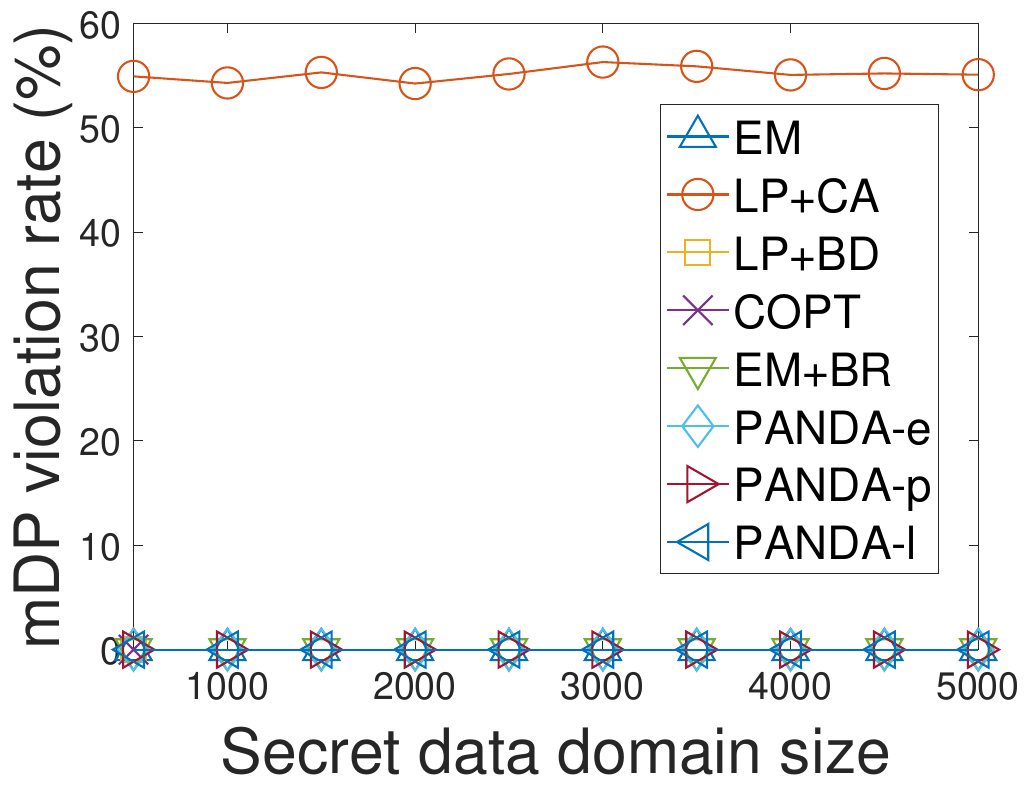}}
  \subfigure[NYC]{
\includegraphics[width=0.32\textwidth, height = 0.10\textheight]{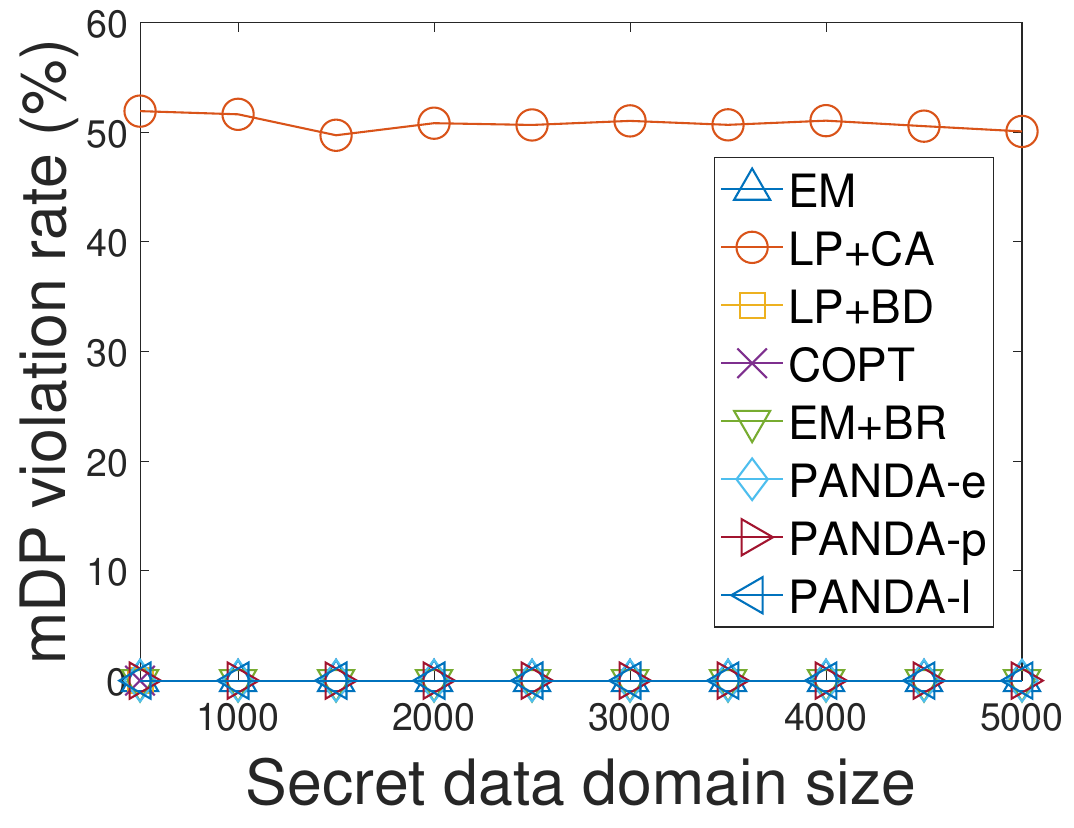}}
  \subfigure[London]{
\includegraphics[width=0.32\textwidth, height = 0.10\textheight]{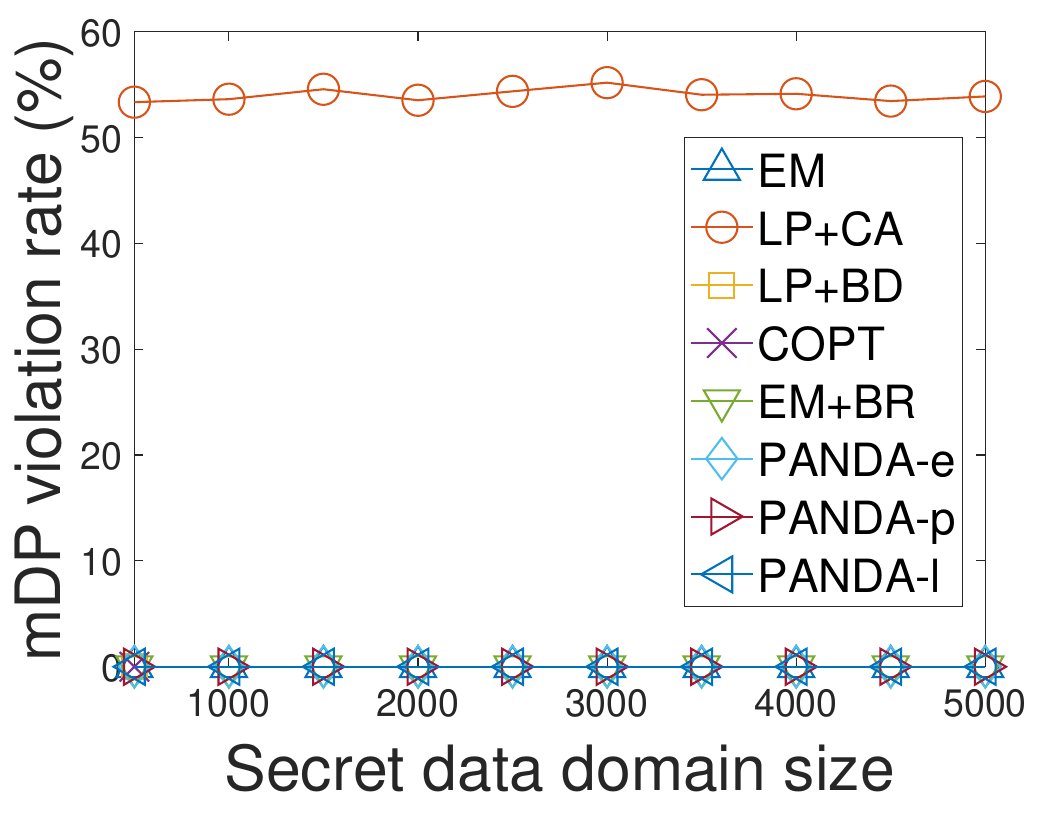}}
\vspace{-0.15in}
\end{minipage}
\caption{mDP violation rate.}
\label{fig:mDPfailture}
\vspace{-0.20in}
\end{figure}

\vspace{-0.05in}
\subsubsection{mDP Violation Rate.}  
Both \textsc{PAnDA} and LP+CA approximate the mDP constraints using surrogate perturbation vectors. Therefore, {\rev it is important to evaluate how often these surrogate-based constraints fail to satisfy the original $\epsilon$-mDP condition between real records, a metric we refer to as the \emph{mDP violation rate}}. 

Figures~\ref{fig:mDPfailture}(a)(b)(c) compare the mDP violation rates of different methods across the {\rev three datasets}. Due to the preserved safety margin, \textsc{PAnDA} consistently achieves a violation rate below $10^{-7}$, which is lower than the predefined threshold $\delta$, thereby demonstrating its strong reliability in meeting the probabilistic privacy guarantee. In contrast, LP+CA exhibits much higher violation rates, averaging {\bl53.33\%}, highlighting its limited robustness in satisfying mDP constraints.
% This result empirically validates the theoretical formulation of \textsc{PAnDA}. 

% {\bl Note: Check how close $\Pr\left[\mbox{Real locations }(x_n, x_m) \in \mathcal{S}_{\hat{x}_n} \times \mathcal{S}_{\hat{x}_m}\right]$ can achieve 1.

% Check how close are $\xi_{\hat{n}, \hat{m}}$ and $h_{x_n,x_m}^{-1}(1-\delta)$.} 

\vspace{-0.08in}
\subsection{Privacy Budget Distribution of \textsc{PAnDA}}
\label{subsec:privacycosts}
\vspace{-0.02in}
\begin{figure}[t]
\centering
\hspace{0.00in}
\begin{minipage}{0.50\textwidth}
  \subfigure[Rome]{
\includegraphics[width=0.32\textwidth, height = 0.10\textheight]{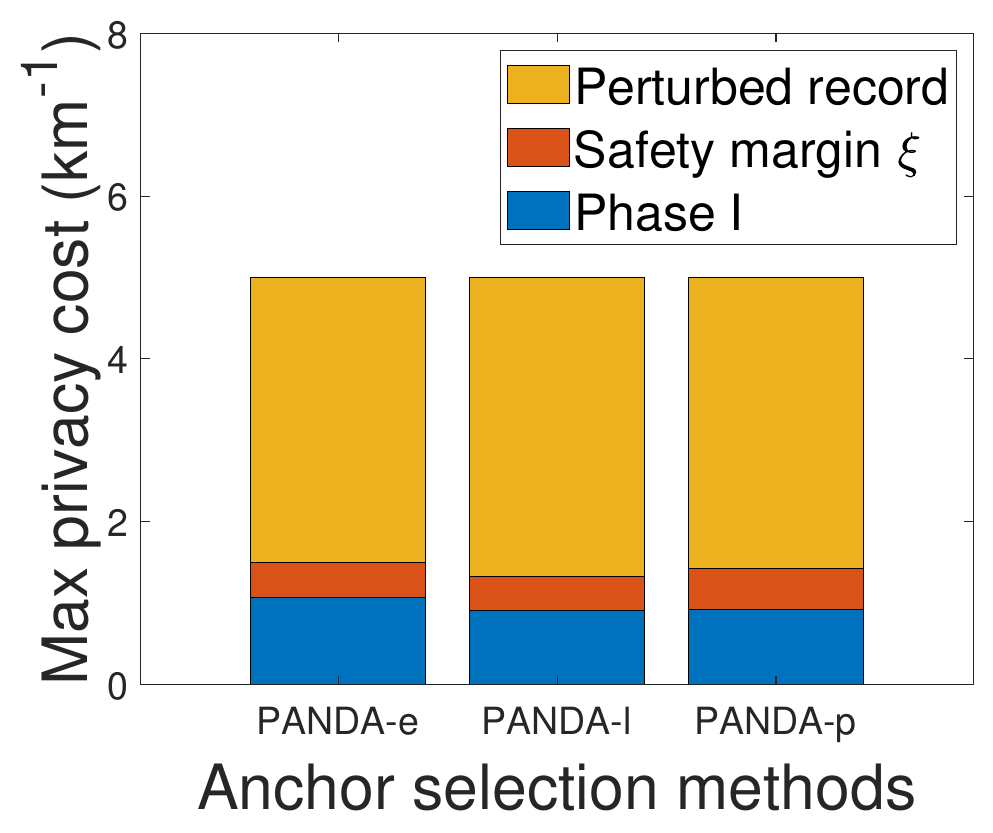}}
  \subfigure[NYC]{
\includegraphics[width=0.32\textwidth, height = 0.10\textheight]{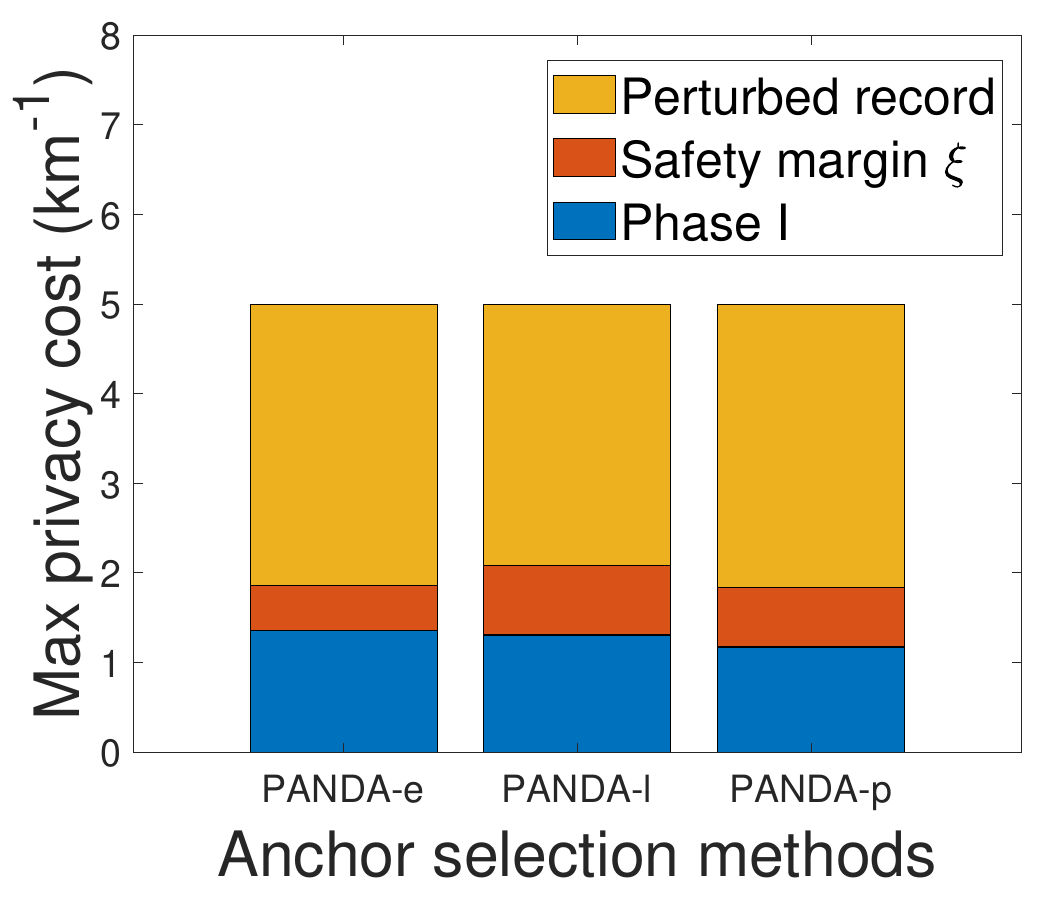}}
  \subfigure[London]{
\includegraphics[width=0.32\textwidth, height = 0.10\textheight]{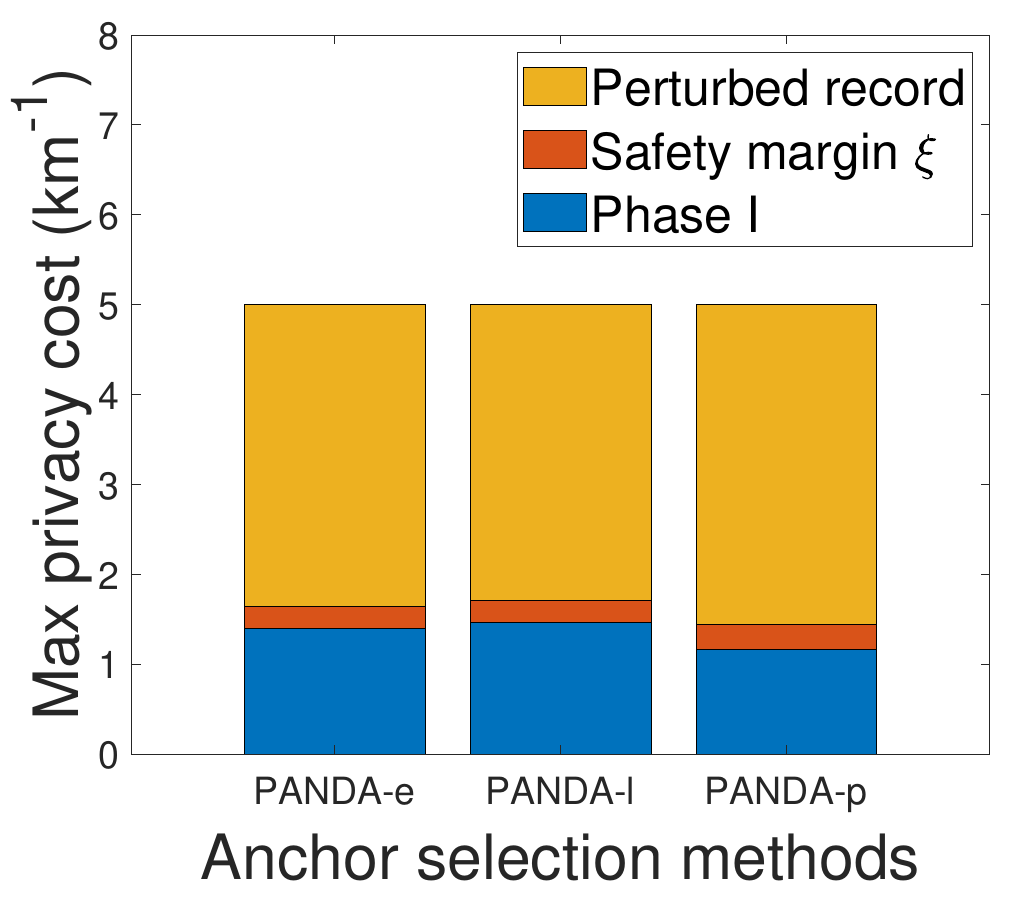}}
\vspace{-0.15in}
\end{minipage}
\caption{Privacy budget allocation across the two phases.}
\label{fig:budgetallocation}
\vspace{-0.25in}
\end{figure}

Finally, Fig.~\ref{fig:budgetallocation}(a)(b)(c) show the distribution of the privacy budget among the three components of \textsc{PAnDA}: anchor selection in Phase I, and the reserved safety margin $\xi$ and perturbation vector optimization in Phase II. The results show that a majority of the privacy budget, {\bl67.03\%} on average, is preserved for Phase II, enabling sufficient flexibility in optimizing utility under the mDP constraints. In contrast, anchor selection consumes {\bl23.91\%} of the budget on average, while the safety margin accounts for {\bl9.06\%}. This allocation reflects a carefully balanced design: while both anchor selection and the safety margin are important for maintaining privacy under probabilistic mDP, over-allocating budget to these components would compromise the effectiveness of the final perturbation mechanism. The results confirm that \textsc{PAnDA} can effectively control privacy leakage in Phase I and judiciously bounds the safety margin, preserving a sufficient portion of $\epsilon$ for utility-aware perturbation optimization. \looseness = -1

\DEL{
\vspace{-0.00in}
\vspace{-0.00in}
\section{Related Works}
\label{sec:related}
\vspace{-0.00in}
The study of record privacy can date back to almost two decades ago, when Gruteser and Grunwald \cite{Gruteser-MobiSys2003} introduced the notion of \emph{record $k$-anonymity}. This notion has been extended to  obfuscate record by means of \emph{$l$-diversity}, i.e., a user's record cannot be distinguished from other $l-1$ records \cite{Yu-NDSS2017}. However, $l$-diversity oversimplifies the threat model by assuming that all the dummy records are equally likely to be the real record from the attacker's view, making it vulnerable to various advanced inference attacks \cite{Andres-CCS2013, Yu-NDSS2017, Qiu-TMC2022}. \looseness = -1

In recent years, Andr{\'e}s \emph{et al.} \cite{Andres-CCS2013} proposed a more practical privacy criterion \emph{Geo-Ind} based on the well-known statistical notion \emph{differential privacy (DP)}. Following this work, a large body of record perturbation strategies have been proposed, e.g., \cite{Mendes-PETS2020, Simon-EuroSPh-II019, Wang-WWW2017, Shokri-TOPS2017, Yu-NDSS2017, Xiao-CCS2015, Qiu-CIKM2020, Qiu-TMC2022, Qiu-SIGSPATIAL2022, Al-Dhubhani-PETS2017, Wang-CIDM2016}. For instance, besides first proposing the notion of Geo-Ind, Andr{\'e}s \emph{et al.}  \cite{Andres-CCS2013} developed a geo-perturbation technique to achieve Geo-Ind by adding noise to the actual record drawn from a polar Laplacian distribution. Moreover, considering that geo-perturbation inevitably introduces errors to users' reported records and causes quality loss in LBS, a key issue that has been discussed by many works is how to trade off quality of service and privacy. For example, given the restriction of Geo-Ind, Bordenabe \emph{et al.}  \cite{Bordenabe-CCS2014} proposed an optimization framework for geo-perturbation to minimize the quality loss for every single user. Chatzikokolakis \emph{et al.} \cite{Chatzikokolakis-PoPETs2015} defines privacy mass over the points of interest % on the plane 
and decides the privacy budget $\epsilon$ of Geo-Ind for a record given the local features of each area. Wang \emph{et al.}  \cite{Wang-WWW2017} considered the quality loss generated by all the users (users) as a whole and proposed a record privacy-preserving target assignment algorithm to minimize the total traveling cost. 

Most of those works follow an LP framework, of which the time complexity is relatively high if no decision variable/constraint reduction is applied. For the sake of time efficiency, the existing works limit the scope of geo-perturbation to either low granularity in a large-scale region (e.g., \cite{Yu-NDSS2017} targets a city-scale region and discretize record field to a grid with each cell size 766m$\times$766m), or high granularity in a small-scale region (e.g., \cite{Qiu-TMC2022} targets a small town with a record point sampled in every 500 square meters). }

\vspace{-0.12in}
\section{Related Work}
\label{sec:related}
\vspace{-0.02in}

Since the introduction of mDP by \cite{Chatzikokolakis-PETS2013}, substantial research efforts have focused on developing data perturbation mechanisms. These mechanisms generally fall into two categories: \textbf{predefined noise distribution mechanisms} and \textbf{optimization-based mechanisms}.

\vspace{-0.15in}
\subsection{Predefined Noise Distribution Mechanisms}

The \textbf{Laplace mechanism}, one of the most widely used data perturbation mechanisms, was originally designed for DP~\cite{Dwork-TC2006} and later adapted to mDP by Chatzikokolakis et al.~\cite{Chatzikokolakis-PETS2013}. Instead of scaling noise to global sensitivity (as in standard DP), \cite{Dwork-TC2006} uses \emph{metric distance} $d_{\mathcal{X}}$ between records to set the Laplace noise scale to $\frac{d_{\mathcal{X}}}{\epsilon}$, ensuring outputs for nearby records are statistically indistinguishable.  
This adaptation underpins \emph{Geo-Indistinguishability}~\cite{Andres-CCS2013}, where mDP protects location data. Early work~\cite{Chatzikokolakis-PETS2013, Andres-CCS2013} secured sporadic locations (e.g., check-ins), while later extensions support continuous tracking~\cite{chatzikokolakis2014predictive}, real-time trajectories~\cite{hua2017geo}, and personalized privacy budgets~\cite{yu2023privacy, min2023personalized}, addressing high-frequency mobility challenges under mDP. Despite its computational efficiency~\cite{Carvalho2021TEMHU}, the Laplace mechanism struggles with structured data (e.g., graphs, word embeddings), where additive noise can corrupt structural integrity or ignore metric space density. For word embeddings, which occupy discrete semantic spaces, perturbation often produces invalid vectors, prompting nearest-neighbor approximations~\cite{fernandes2018author, Feyisetan-ICDM2019}. These methods, however, assume embedding spaces are non-sensitive, potentially exposing privacy risks during neighbor retrieval.  

The \textbf{Exponential Mechanism (EM)}~\cite{Carvalho2021TEMHU, ImolaUAI2022} addresses these limitations by selecting outputs from discrete sets via utility-driven quality scores, bypassing continuous perturbation. This proves particularly effective in location privacy (trajectory protection~\cite{Chatzikokolakis-PoPETs2015, Yu-NDSS2017}; spatial crowdsourcing~\cite{gursoy2019secure, niu2020eclipse, ren2022distpreserv, dong2019preserving}), textual data (word substitution~\cite{wang2017local, Feyisetan-WDSM2020}; embeddings~\cite{Carvalho2021TEMHU, meisenbacher20241}), and graph data (edge preservation~\cite{raskhodnikova2016lipschitz}; query responses~\cite{fioretto2019privacy, kamalaruban2020not}).  

While EM offers greater flexibility and often superior utility compared to the Laplace mechanism, it bases selection on perturbation magnitude $d_{x,y}$, which may not fully capture directional variations in utility loss across the output space. Therefore, many recent efforts have prioritized \emph{optimization-based mechanisms}, which minimize data utility loss by explicitly modeling both the magnitude and directional impact of each candidate perturbation.

% In contrast, optimization-based mechanisms can achieve lower data utility loss by explicitly considering the utility loss of each possible perturbation choice, including both magnitude and direction. This approach enables these mechanisms to minimize the expected utility loss, thereby achieving superior performance.
\vspace{-0.08in}
\subsection{Optimization (or LP)-based Mechanisms} 
% Although EM is more flexible and often provides better utility compared to Laplace, it still base selection on the perturbation magnitude $d_{x, y}$, which may not accurately capture varying utility losses in different directions of the output space. 

Since the seminal work by Bordenabe et al.~\cite{Bordenabe-CCS2014}, which introduced the use of LP to optimize perturbation mechanisms under mDP, significant effort has been devoted to developing LP-based mechanisms. These approaches generally aim to minimize utility loss by explicitly considering both the magnitude and direction of the loss for each possible perturbation. As described in Section~\ref{sec:preliminary}, such mechanisms are typically formalized as stochastic matrices, with the LP objective minimizing the expected utility loss subject to mDP constraints. However, these formulations involve $O(|\mathcal{X}|^2)$ decision variables, making them theoretically optimal but computationally prohibitive for large domains. As shown in Fig.~\ref{fig:relatedwork}, many existing implementations (e.g., \cite{Bordenabe-CCS2014, Wang-ICDM2016, Yu-NDSS2017, Wang-WWW2017}) are limited to relatively small domains with $|\mathcal{X}| \leq 100$. Although Bordenabe et al.~\cite{Bordenabe-CCS2014} proposed using spanners to approximate $d_{\mathcal{X}}$ via shortest path distances in a graph, the scalability remains limited, for example, solving the LP takes over 1,800 seconds when $|\mathcal{X}| = 400$ \cite{ImolaUAI2022}.

Recent research has prioritized improving the scalability of LP-based methods. Ahuja et al.~\cite{AhujaEDBT2019} proposed a hierarchical index structure to protect location data under Geo-Indistinguishability (Geo-Ind), while others leveraged decomposition techniques, such as Dantzig–Wolfe~\cite{Qiu-TMC2022,Qiu-CIKM2020} and Benders decomposition~\cite{Qiu-IJCAI2024}, to partition large-scale LP problems into tractable subproblems. Hybrid frameworks combining LP with probabilistic mechanisms have also emerged, adapting noise using local utility/sensitivity. For example, Imola et al.~\cite{ImolaUAI2022} integrated the weighted EM with selective LP optimization on sensitive inputs. Despite these advances, Fig.~\ref{fig:relatedwork} shows state-of-the-art methods still only handle domains with $|\mathcal{X}| \leq 1{,}000$. 
% A more recent work by Qiu et al.~\cite{Qiu-PETS2025} proposes a locality-aware optimization framework that confines computation to spatial neighborhoods while aiming to preserve global privacy guarantees. This approach significantly improves scalability, demonstrating applicability to domains with $|\mathcal{X}| = 1{,}600$. However, it does so at the cost of strict compliance with metric differential privacy (mDP), as the use of neighborhood-specific parameters introduces potential information leakage, posing a critical limitation in applications requiring strong privacy assurances.
A more recent work by Qiu et al.~\cite{Qiu-PETS2025} introduces a locality-aware optimization approach, restricting computation to spatial neighborhoods. While scaling to $|\mathcal{X}| = 1{,}600$, it does so at the cost of strict compliance with mDP, as the use of neighborhood-specific parameters introduces potential information leakage.

% Similarly, \cite{ImolaUAI2022} proposed a hybrid approach that applies EM to data with minimal utility impact and LP to data with greater utility sensitivity, improving computational efficiency while preserving privacy guarantees in large metric spaces. 

\vspace{-0.15in}
\section{Conclusions and Discussions}
\label{sec:conclude}
\vspace{-0.02in}
In this paper, we presented \textsc{PAnDA}, a two-phase framework that addresses the scalability limitations of the mDP optimization problem. By introducing an anchor-based approximation approach, \textsc{PAnDA} significantly reduces the computational overhead of LP while preserving privacy guarantees under a relaxed probabilistic mDP model. Theoretical analysis demonstrates bounded privacy leakage and near-optimal utility loss, while extensive experiments on real-world geo-location datasets confirm that \textsc{PAnDA} scales to domains with up to {\bl5,000} records, tripling the capacity of prior methods, while achieving superior utility and computational efficiency. These results highlight \textsc{PAnDA} as a practical and theoretically sound solution for scalable mDP optimization.

We envision several promising directions to further improve \textsc{PAnDA}. In particular, the current implementation of \textsc{PAnDA} conservatively assigns a failure probability of $\delta_1 = 0$ to Phase I (anchor selection), which may constrain flexibility in optimizing the overall privacy-utility tradeoff. A promising enhancement is to allocate a small, non-zero portion of the total violation probability $\delta$ to Phase I, i.e., set $\delta_1 > 0$ and $\delta_2 = \delta - \delta_1$, and require anchor selection to satisfy $(\overline{\epsilon}_{x_n,x_m}, \delta_1)$-PmDP for each pair of users $n$ and $m$.

To enable this, as the next step, we will carefully calibrate the anchor selection distributions, e.g., by employing steeper decay functions or reducing overlap between anchor sets, so that the selected anchors more accurately reflect the utility and privacy characteristics of the true records, while keeping the probability of $\overline{\epsilon}_{x_n,x_m}$-mDP violations in Phase I within $\delta_1$. This relaxation permits sharper locality in anchor selection, enhancing utility without breaching the overall $(\epsilon, \delta)$-PmDP guarantee. Moreover, allocating part of the failure budget to Phase I can reduce the required safety margin $\xi$ in Phase II due to the higher accuracy of anchor approximation. This joint optimization of $(\epsilon_1, \delta_1)$ and $(\epsilon_2, \delta_2)$ may yield a more balanced privacy budget allocation and improved utility.

Moreover, our current work assumes honest anchor selection. As a future direction, we plan to address potential manipulation by malicious users. Practical deployments may require additional safeguards, such as lightweight defenses like robust aggregation (e.g., median-based filtering) and distance-based outlier detection. We also aim to explore stronger theoretical guarantees via Byzantine-resilient optimization, drawing on robust distributed learning to preserve privacy even under adversarial behavior.

\vspace{-0.10in}
\section{Acknowledgements}
\vspace{-0.02in}
This research was partially supported by U.S. NSF grants CNS-2136948 and CNS-2313866.


% References

% \clearpage 
%\section*{Appendix}

% \maketitle
%\appendix ||seems to work without this
\clearpage 
\section*{Appendix}
\section{Summary of Notations}

% \subsection{Data perturbation} Highlight the difference between data generalization between data perturbation

% data generalization: Data generalization involves abstracting detailed information to a higher-level, less specific form. 

% data perturbation: involves modifying the original data in a way that the modified data remains useful for analysis but does not compromise the confidentiality of the data. 
\DEL{
\begin{table}[h]
\caption{Main notations and their descriptions}
\vspace{-0.10in}
\label{Tb:Notationmodel}
\centering
\normalsize 
\small 
\begin{tabular}{l l}
\toprule
Symbol                  & Description \\
\hline
\hline
$\mathcal{X}$           & $\mathcal{X} = \left\{1, ..., K\right\}$ denotes the secret record index set. \\ 
$\mathcal{Y}$           & $\mathcal{Y} = \left\{1, ..., L\right\}$ denotes the perturbed record index set. \\ 
$\mathcal{M}$           & Randomized mechanism that maps a secret record index \\
& $x \in \mathcal{X}$ to a perturbed record index $y \in \mathcal{Y}$. \\ 
$\mathcal{G}$           & mDP graph $\mathcal{G} = \left(\mathcal{X}, \mathcal{E}\right)$; where $\mathcal{X}$ and $\mathcal{E}$ are the secret  \\ 
& record index set and the edge set, respectively.\\
% $i$                   & Location $i$ \\
% $x_{\mathrm{t}}$        & Task record $x_{\mathrm{t}} \in \mathcal{X}$ \\
% $e_{n,m}$               & Edge connecting $x_n$ and $x_m$ \\
% $d(x_n,x_m)$               & Weight of $e_{n,m}$, i.e., the traveling cost from $i$ to $j$ \\
$d_{x,x'}$               & Distance between $x$ and $x'$ \\
$\mathbf{Z}_{\mathcal{X}}$            & Perturbation matrix $\mathbf{Z}_{\mathcal{X}}$ \\
$z_{x,y}$               & The probability of selecting $y$ as the perturbed record \\
& given the real record $x$ \\ 
$\mathbf{z}_x$            & The perturbation vector of $x$: $\mathbf{z}_x = [z_{x,y}]_{y \in \mathcal{Y}}$ \\
$\epsilon$           & Privacy budget \\ 
$\eta$           & Neighbor threshold \\ 
$c_{x,y}$   & Cost coefficient of $z_{x,y}$ \\ 
% $X$                     & Random variable to denote real record \\ 
% $Y$                     & Random variable to denote perturbed record \\ 
\hline
\end{tabular}
\normalsize
\vspace{-0.00in}
\end{table}}

\begin{table}[h]
\normalsize 
\small 
\centering
\caption{Main notations and their descriptions}
\begin{tabular}{p{1.2cm} p{6.5cm}}
\toprule
\textbf{Symbol} & \textbf{Description} \\
\midrule
$\mathcal{X}$ & Set of real (secret) records \\
$\mathcal{Y}$ & Set of perturbed (output) records \\
$x_n, x_m$ & Real record held by user $n$ or $m$ \\
$y_n, y_m$ & Perturbed record held by user $n$ or $m$ \\
$z_{x,y}$ & Probability of selecting $y$ as the perturbed record given the real record $x$ \\
$\mathbf{Z}_\mathcal{X}$ & Full perturbation matrix over $\mathcal{X} \times \mathcal{Y}$ \\
$\mathbf{Z}_{\mathcal{A}_n}$ & Perturbation matrix over anchor set $\mathcal{A}_n$ \\
$\mathcal{A}_n$ & Anchor record set selected by user $n$ \\
$d_{x,x'}$ & Distance between records $x$ and $x'$ in the metric space \\
$w_{x_n,x}$ & Probability of selecting anchor $x$ given true record $x_n$ \\
$p_x$ & Prior probability of real record $x$ \\
$c_{x,y}$ & Utility loss when reporting $y$ given true record $x$ \\
$\mathcal{L}(\mathbf{Z})$ & Expected utility loss from perturbation matrix $\mathbf{Z}$ \\
$\tilde{x}_n$ & Surrogate anchor for user $n$'s true record $x_n$ \\
$\mathbf{z}_{\tilde{x}_n}$ & Surrogate perturbation vector for $\tilde{x}_n$ \\
$\epsilon$ & Privacy budget \\
$\delta$ & Violation tolerance parameter of PmDP \\
$\alpha$ & Scaling factor in anchor selection probability function \\
$\lambda$ & Decay parameter in anchor selection probability function \\
\midrule
$\mathcal{M}$ & General randomized mechanism that maps a secret record index $x \in \mathcal{X}$ to a perturbed record index $y \in \mathcal{Y}$.\\
$\mathcal{M}_\mathrm{I}$ & Phase I randomized mechanism (anchor selection) \\
$\mathcal{M}_{\mathrm{II}}$ & Phase II randomized mechanism (data perturbation) \\
$\mathcal{M}_{\text{\textsc{PAnDA}}}$ & Combined mechanism of \textsc{PAnDA} (Phase I + II) \\
$\xi$ & Safety margin to account for approximation error caused by surrogate perturbation vector\\
$\mathcal{H}_{x_n,x_m}(\xi)$ & Set of anchor record pairs $(x,x')$ that satisfy the safety condition for $(x_n,x_m)$ with margin $\xi$ \\
$h_{x_n,x_m}(\xi)$ & Probability of successful mDP constraint satisfaction between $x_n$ and $x_m$ given margin $\xi$ \\
$\mathcal{U}_{x_q,\hat{x}_q}$ & Set of records closer to $x_q$ than the surrogate $\hat{x}_q$, used to compute surrogate probabilities \\
$v_{x_{q}, \hat{x}_{q}}$ & Probability that $\hat{x}_{q}$ is $x_{q}$'s surrogate \\ 
\bottomrule
\end{tabular}
\label{tab:notation}
\end{table}

\section{Omitted Proofs}
\label{sec:proofs}
\DEL{
\begin{lemma}
\label{prop:PL}
% (PL decomposition)
According to Bayes' formula, the posterior leakage $\mathrm{PL}\left((x_n, x_m); \mathcal{M}_{\mathrm{I}}, \mathcal{M}_{\mathrm{II}}\right)$ can be decomposed to  
\begin{eqnarray}
&&\mathrm{PL}\left((x_n, x_m); \mathcal{M}_{\mathrm{I}}, \mathcal{M}_{\mathrm{II}}\right) \\\
&=& 
\mathrm{PL}\left((x_n, x_m); \mathcal{M}_{\mathrm{I}}\right) + \mathrm{PL}\left((x_n, x_m); \mathcal{M}_{\mathrm{II}}\right)
\\ \nonumber 
&=& \underbrace{\sup_{\mathcal{A}\ni x_n, \mathcal{A}' \ni x_m} \overbrace{\frac{u_{i, \mathcal{A}} \sum_{x_\ell \in \mathcal{X}}z_{l,k}u_{x_\ell, \mathcal{A}'}p_{x_\ell}}{u_{x_m, \mathcal{A}'}\sum_{x_\ell \in \mathcal{X}}z_{l,k}u_{x_\ell, \mathcal{A}}p_{x_\ell}}}^{\small \mbox{PL caused by $\mathcal{A}$ and $\mathcal{A}'$}}}_{\small \mbox{Least upper bound of PL caused by anchor records}} \times \frac{z_{n,k}}{z_{m,k}}% \\
% &=& \frac{z_{n,k}\Pr\left(\mathcal{M}_{\mathrm{I}}(X_n) = \mathcal{A}|X_n = x_n\right)}{z_{m,k}\Pr\left(\mathcal{M}_{\mathrm{I}}(X_n) = \mathcal{A}'|X_n = x_m\right)} \\
% &\times & \frac{\sum_{x_m}z_{m,k}\Pr\left(\mathcal{M}_{\mathrm{I}}(X_n) = \mathcal{A}'|X_n = x_m\right)\Pr\left(X_n = x_m\right)}{\sum_{x_n}z_{n,k}\Pr\left(\mathcal{M}_{\mathrm{I}}(X_n) = \mathcal{A}|X_n = x_n\right)\Pr\left(X_n = x_n\right)} 
% &\leq & e^{\epsilon d_{x_n,x_m}}
\end{eqnarray}
\DEL{
where 
\begin{eqnarray}
\nonumber && \mathrm{PL}\left((x_n, x_m); \mathcal{M}_{\mathrm{I}}\right) = sup_{\mathcal{A}\ni x_n, \mathcal{A}' \ni x_m} \frac{u_{i, \mathcal{A}} \sum_{x_\ell \in \mathcal{X}}z_{l,k}u_{x_\ell, \mathcal{A}'}p_{x_\ell}}{u_{x_m, \mathcal{A}'}\sum_{x_\ell \in \mathcal{X}}z_{l,k}u_{x_\ell, \mathcal{A}}p_{x_\ell}} \\ 
\nonumber && \mathrm{PL}\left((x_n, x_m); \mathcal{M}_{\mathrm{II}}\right) = \frac{z_{n,k}}{z_{m,k}}, 
\end{eqnarray}
corresponds to the PL caused by $\mathcal{M}_\mathrm{I}$ and $\mathcal{M}_{\mathrm{II}}$, respectively.}
\end{lemma}
\begin{proof}
The detailed proof can be found in Appendix. 
\end{proof}}

\subsection{Proof of Theorem \ref{thm:composition}}
\label{subsec:proof:thm:composition}
Before proving Theorem \ref{thm:composition}, we first introduce \textbf{Lemma \ref{lem:comp}} as a preparation: 
\begin{lemma}
[Sequential composition of mDP for \textsc{PAnDA}] 
\label{lem:comp}
\begin{eqnarray}
\nonumber && \mathcal{M}_{\mathrm{I}}(x_n) \stackrel{\epsilon_{1}}{\approx} \mathcal{M}_{\mathrm{I}}(x_m) \mbox{ and } \mathcal{M}_{\mathrm{II}}(x_n) \stackrel{\epsilon_{2}}{\approx} \mathcal{M}_{\mathrm{II}}(x_m) \\
&\Rightarrow& \mathcal{M}_{\mathrm{\textsc{PAnDA}}}(x_n) \stackrel{\epsilon_1+\epsilon_2}{\approx} \mathcal{M}_{\mathrm{\textsc{PAnDA}}}(x_m).
\end{eqnarray}
\end{lemma}
\begin{proof}
If $\mathcal{M}_{\mathrm{I}}(x_n) \stackrel{\epsilon_{1}}{\approx} \mathcal{M}_{\mathrm{I}}(x_m)$ and $\mathcal{M}_{\mathrm{II}}(x_n) \stackrel{\epsilon_{2}}{\approx} \mathcal{M}_{\mathrm{II}}(x_m)$, then 
\begin{eqnarray}
&& \Pr\left[\mathcal{M}_{\mathrm{\textsc{PAnDA}}}(x_n) =(\mathcal{A}, y)\right] \\
&=& \Pr\left[\mathcal{M}_{\mathrm{I}}(x_n) = \mathcal{A} \wedge \mathcal{M}_{\mathrm{II}}(x_n) = y)\right] \\
&=& \Pr\left[\mathcal{M}_{\mathrm{I}}(x_n) = \mathcal{A} \right] \Pr\left[\mathcal{M}_{\mathrm{II}}(x_n) = y)|\mathcal{M}_{\mathrm{I}}(x_n) = \mathcal{A}\right] \\
&\leq& e^{\epsilon_1 d_{x_n,x_m}}\Pr\left[\mathcal{M}_{\mathrm{I}}(x_n) = \mathcal{A} \right] \\
&\times& e^{\epsilon_2 d_{x_n,x_m}} \Pr\left[\mathcal{M}_{\mathrm{II}}(x_m) = y)|\mathcal{M}_{\mathrm{I}}(x_m) = \mathcal{A}\right] 
\\
&=&  e^{(\epsilon_1+\epsilon_2) d_{x_n,x_m}}\Pr\left[\mathcal{M}_{\mathrm{I}}(x_m) = \mathcal{A} \wedge \mathcal{M}_{\mathrm{II}}(x_m) = y)\right] \\
&=&  e^{(\epsilon_1+\epsilon_2) d_{x_n,x_m}}\Pr\left[\mathcal{M}_{\mathrm{\textsc{PAnDA}}}(x_m) =(\mathcal{A}, y)\right] 
\end{eqnarray}
The proof is completed. 
\end{proof}
\begin{reptheorem}[Sequential composition of PmDP for \textsc{PAnDA}] 
\begin{eqnarray}
\nonumber && \mathcal{M}_{\mathrm{I}}(x_n) \stackrel{(\epsilon_{1}, \delta_{1})}{\sim} \mathcal{M}_{\mathrm{I}}(x_m) \mbox{ and } \mathcal{M}_{\mathrm{II}}(x_n) \stackrel{(\epsilon_{2}, \delta_{2})}{\sim} \mathcal{M}_{\mathrm{II}}(x_m) \\
&\Rightarrow& \mathcal{M}_{\mathrm{\textsc{PAnDA}}}(x_n) \stackrel{(\epsilon_1+\epsilon_2, \delta_1+\delta_2)}{\sim} \mathcal{M}_{\mathrm{\textsc{PAnDA}}}(x_m).
\end{eqnarray}
\end{reptheorem}
\begin{proof} 
First, by definition, we can obtain that 
\begin{eqnarray} 
&& \mathcal{M}_{\ell}(x_n) \stackrel{(\epsilon_{\ell},\delta_{\ell})}{\sim} \mathcal{M}_{\ell}(x_m) \\ 
&\Leftrightarrow &
\Pr\left[\mathcal{M}_{\ell}(x_n) \stackrel{\epsilon_{\ell}}{\not\approx} \mathcal{M}_{\ell}(x_m)\right] < \delta_{\ell}.
\vspace{-0.00in}
\end{eqnarray}
According to Lemma \ref{lem:comp}, 
\begin{eqnarray}
\nonumber && \mathcal{M}_{\mathrm{I}}(x_n) \stackrel{\epsilon_{1}}{\approx} \mathcal{M}_{\mathrm{I}}(x_m) \mbox{ and } \mathcal{M}_{\mathrm{II}}(x_n) \stackrel{\epsilon_{2}}{\approx} \mathcal{M}_{\mathrm{II}}(x_m) \\
&\Rightarrow& \mathcal{M}_{\mathrm{\textsc{PAnDA}}}(x_n) \stackrel{\epsilon_1+\epsilon_2}{\approx} \mathcal{M}_{\mathrm{\textsc{PAnDA}}}(x_m).
\end{eqnarray}
Therefore, if $\mathcal{M}_{\mathrm{I}}(x_n) \stackrel{(\epsilon_{1}, \delta_{1})}{\sim} \mathcal{M}_{\mathrm{I}}(x_m)$ and $\mathcal{M}_{\mathrm{II}}(x_n) \stackrel{(\epsilon_{2}, \delta_{2})}{\sim} \mathcal{M}_{\mathrm{II}}(x_m)$, then 
\begin{eqnarray}
% \nonumber 
&& \Pr\left[\mathcal{M}_{\mathrm{\textsc{PAnDA}}}(x_n) \stackrel{\epsilon_1+\epsilon_2}{\approx} \mathcal{M}_{\mathrm{\textsc{PAnDA}}}(x_m)\right] \\
&\geq& \Pr\left[\mathcal{M}_{\mathrm{I}}(x_n) \stackrel{\epsilon_{1}}{\approx} \mathcal{M}_{\mathrm{I}}(x_m) \wedge \mathcal{M}_{\mathrm{II}}(x_n) \stackrel{\epsilon_{2}}{\approx} \mathcal{M}_{\mathrm{II}}(x_m)\right] \\
&=&  1 - \Pr\left[\mathcal{M}_{\mathrm{I}}(x_n) \stackrel{\epsilon_{1}}{\not\approx} \mathcal{M}_{\mathrm{I}}(x_m) \vee \mathcal{M}_{\mathrm{II}}(x_n) \stackrel{\epsilon_{2}}{\not\approx} \mathcal{M}_{\mathrm{II}}(x_m)\right] \\
% &=&  1 - \Pr\left[\mathcal{M}_\ell(x_n) \stackrel{\epsilon_\ell}{\not\approx} \mathcal{M}_\ell(x_m), \exists \ell = 1, ..., L\right] \\
\nonumber &\geq&  1 - \Pr\left[\mathcal{M}_{\mathrm{I}}(x_n) \stackrel{\epsilon_{1}}{\not\approx} \mathcal{M}_{\mathrm{I}}(x_m)\right] - \Pr\left[\mathcal{M}_{\mathrm{II}}(x_n) \stackrel{\epsilon_{2}}{\not\approx} \mathcal{M}_{\mathrm{II}}(x_m)\right] \\
% &\geq& 1 - \sum_{\ell = 1}^L \Pr\left[\mathcal{M}_\ell (x_n) \stackrel{\epsilon_\ell}{\not\approx} \mathcal{M}_\ell (x_m)\right]\\
&>& 1 - \sum_{\ell = 1}^L\delta_\ell, 
\end{eqnarray}
indicating that $\mathcal{M}_{\mathrm{\textsc{PAnDA}}}(x_n) \stackrel{(\epsilon_1+\epsilon_2, \delta_1+\delta_2)}{\sim} \mathcal{M}_{\mathrm{\textsc{PAnDA}}}(x_m)$. The proof is completed. 
\end{proof}

\DEL{
\subsection{Proof of Proposition \ref{prop:\textsc{PAnDA}comp}}
\begin{proof}
$\mathcal{M}_{\mathrm{\textsc{PAnDA}}}(x_n) \stackrel{(\epsilon,\delta)}{\sim} \mathcal{M}_{\mathrm{\textsc{PAnDA}}}(x_m),~\forall x_n, x_m \in \mathcal{X}.$
\begin{eqnarray}
&& \Pr\left[\mathcal{M}_{\mathrm{\textsc{PAnDA}}}(x_n) =(\mathcal{A}, y)\right] \\
&=& \Pr\left[\mathcal{M}_{\mathrm{I}}(x_n) = \mathcal{A} \wedge \mathcal{M}_{\mathrm{II}}(x_n) = y)\right] \\
&=& \Pr\left[\mathcal{M}_{\mathrm{I}}(x_n) = \mathcal{A} \right] \Pr\left[\mathcal{M}_{\mathrm{II}}(x_n) = y)|\mathcal{M}_{\mathrm{I}}(x_n) = \mathcal{A}\right] \\
&\leq& e^{\epsilon_1 d_{x_n,x_m}}\Pr\left[\mathcal{M}_{\mathrm{I}}(x_n) = \mathcal{A} \right] \\
&\times& e^{\epsilon_2 d_{x_n,x_m}} \Pr\left[\mathcal{M}_{\mathrm{II}}(x_m) = y)|\mathcal{M}_{\mathrm{I}}(x_m) = \mathcal{A}\right] \\
&=&  e^{(\epsilon_1+\epsilon_2) d_{x_n,x_m}}\Pr\left[\mathcal{M}_{\mathrm{\textsc{PAnDA}}}(x_m) =(\mathcal{A}, y)\right] 
\end{eqnarray}
\end{proof}}

\subsection{Proof of Proposition \ref{prop:posteriorbound}}
\label{subsec:prop:posteriorbound}
\begin{reproposition}
For any two records $x_n, x_m \in \mathcal{X}$, we define $\mathcal{A}_{n,m} = \left\{x \in \mathcal{X} \left|w_{x_n,x} > w_{x_m,x}\right.\right\}$ and $\overline{\mathcal{A}}_{n,m} = \left\{x \in \mathcal{X} \left|w_{x_n,x} < w_{x_m,x}\right.\right\}$, 
then either $\epsilon_{x_n,x_m,\mathcal{A}_{n,m}}$ or $\epsilon_{x_n,x_m,\overline{\mathcal{A}}_{n,m}}$ provides the least upper bound of the privacy cost $\epsilon_{x_n,x_m,\mathcal{A}}$, i.e., 
\begin{eqnarray}
\label{eq:PLupperbound2}
\max\left\{\epsilon_{x_n,x_m,\mathcal{A}_{n,m}}, \epsilon_{x_n,x_m,\overline{\mathcal{A}}_{n,m}}\right\} = \sup_{\mathcal{A}\in 2^{\mathcal{X}}} \epsilon_{x_n,x_m,\mathcal{A}}
\end{eqnarray}
\end{reproposition}
\begin{proof}
We aim to prove that for all $\mathcal{A} \in 2^{\mathcal{X}}$,
\begin{eqnarray}
\frac{\Pr\left[\mathcal{M}_{\mathrm{I}}(x_n) = \overline{\mathcal{A}}_{n,m}\right]}{\Pr\left[\mathcal{M}_{\mathrm{I}}(x_m) = \overline{\mathcal{A}}_{n,m}\right]} 
&\leq& \frac{\Pr\left[\mathcal{M}_{\mathrm{I}}(x_n) = \mathcal{A}\right]}{\Pr\left[\mathcal{M}_{\mathrm{I}}(x_m) = \mathcal{A}\right]} \\
&\leq& \frac{\Pr\left[\mathcal{M}_{\mathrm{I}}(x_n) = \mathcal{A}_{n,m}\right]}{\Pr\left[\mathcal{M}_{\mathrm{I}}(x_m) = \mathcal{A}_{n,m}\right]}
\end{eqnarray}

Due to symmetry, it suffices to prove that
\begin{equation}
\frac{\Pr\left[\mathcal{M}_{\mathrm{I}}(x_n) = \mathcal{A}\right]}{\Pr\left[\mathcal{M}_{\mathrm{I}}(x_m) = \mathcal{A}\right]} \leq \frac{\Pr\left[\mathcal{M}_{\mathrm{I}}(x_n) = \mathcal{A}_{n,m}\right]}{\Pr\left[\mathcal{M}_{\mathrm{I}}(x_m) = \mathcal{A}_{n,m}\right]}.
\end{equation}

For convenience, let us define
\begin{eqnarray}
W(\mathcal{A}, \mathbf{w}_{x_n}, \mathbf{w}_{x_m}) &=& \frac{\Pr\left[\mathcal{M}_{\mathrm{I}}(x_n) = \mathcal{A}\right]}{\Pr\left[\mathcal{M}_{\mathrm{I}}(x_m) = \mathcal{A}\right]} \\
&=& \prod_{x \in \mathcal{A}} \frac{w_{x_n,x}}{w_{x_m,x}}  \prod_{x \notin \mathcal{A}} \frac{1 - w_{x_n,x}}{1 - w_{x_m,x}}.
\end{eqnarray}

To derive a contradiction, assume that there exists a subset $\tilde{\mathcal{A}} \in 2^{\mathcal{X}}$ (with $\tilde{\mathcal{A}} \neq \mathcal{A}_{n,m}$) that provides the maximum likelihood ratio, i.e.,
\begin{eqnarray}
\label{eq:contradiction1}
W(\tilde{\mathcal{A}}, \mathbf{w}_{x_n}, \mathbf{w}_{x_m}) \geq W(\mathcal{A}, \mathbf{w}_{x_n}, \mathbf{w}_{x_m}), \quad \forall \mathcal{A} \in 2^{\mathcal{X}},
\end{eqnarray}
and in particular,
\begin{eqnarray}
\label{eq:contradiction2}
W(\tilde{\mathcal{A}}, \mathbf{w}_{x_n}, \mathbf{w}_{x_m}) > W(\mathcal{A}_{n,m}, \mathbf{w}_{x_n}, \mathbf{w}_{x_m}).
\end{eqnarray}
Since $\tilde{\mathcal{A}} \neq \mathcal{A}_{n,m}$, we consider the following two cases:

\begin{itemize}
    \item \textbf{Case (i):} $\mathcal{A}_{n,m} \setminus \tilde{\mathcal{A}} \neq \emptyset$. Then there exists $x \in \mathcal{A}_{n,m} \setminus \tilde{\mathcal{A}}$, meaning $w_{x_n,x} > w_{x_m,x}$ and hence
\begin{equation}
\frac{w_{x_n,x}}{w_{x_m,x}} > 1 \quad \text{and} \quad \frac{1 - w_{x_m,x}}{1 - w_{x_n,x}} > 1.
\end{equation}
Define a new set $\hat{\mathcal{A}} = \tilde{\mathcal{A}} \cup \{x\}$. Then we have:
\begin{eqnarray}
\nonumber W(\hat{\mathcal{A}}, \mathbf{w}_{x_n}, \mathbf{w}_{x_m}) 
&=& \frac{w_{x_n,x}}{w_{x_m,x}}  \frac{1 - w_{x_m,x}}{1 - w_{x_n,x}}  W(\tilde{\mathcal{A}}, \mathbf{w}_{x_n}, \mathbf{w}_{x_m}) \\
&>& W(\tilde{\mathcal{A}}, \mathbf{w}_{x_n}, \mathbf{w}_{x_m}),
\end{eqnarray}
which contradicts the assumption in Eq.~(\ref{eq:contradiction1}).

\item \textbf{Case (ii):} $\mathcal{A}_{n,m} \setminus \tilde{\mathcal{A}} = \emptyset$ and $\tilde{\mathcal{A}} \setminus \mathcal{A}_{n,m} \neq \emptyset$, i.e., $\mathcal{A}_{n,m} \subsetneq \tilde{\mathcal{A}}$. Then there exists $x \in \tilde{\mathcal{A}} \setminus \mathcal{A}_{n,m}$ such that $w_{x_m,x} > w_{x_n,x}$, implying
\begin{equation}
\frac{w_{x_m,x}}{w_{x_n,x}} > 1 \quad \text{and} \quad \frac{1 - w_{x_n,x}}{1 - w_{x_m,x}} > 1.
\end{equation}
Let $\hat{\mathcal{A}} = \tilde{\mathcal{A}} \setminus \{x\}$. Then:
\begin{eqnarray}
\nonumber W(\hat{\mathcal{A}}, \mathbf{w}_{x_n}, \mathbf{w}_{x_m}) 
&=& \frac{w_{x_m,x}}{w_{x_n,x}}  \frac{1 - w_{x_n,x}}{1 - w_{x_m,x}}  W(\tilde{\mathcal{A}}, \mathbf{w}_{x_n}, \mathbf{w}_{x_m}) \\
&>& W(\tilde{\mathcal{A}}, \mathbf{w}_{x_n}, \mathbf{w}_{x_m}),
\end{eqnarray}
which again contradicts Eq.~(\ref{eq:contradiction1}).

Furthermore, if for all $x \in \tilde{\mathcal{A}} \setminus \mathcal{A}_{n,m}$ we have $w_{x_n,x} = w_{x_m,x}$, then:
\begin{eqnarray}
&& W(\tilde{\mathcal{A}}, \mathbf{w}_{x_n}, \mathbf{w}_{x_m}) \\ 
&=& \prod_{x \in \tilde{\mathcal{A}}} \frac{w_{x_n,x}}{w_{x_m,x}}  \prod_{x \notin \tilde{\mathcal{A}}} \frac{1 - w_{x_n,x}}{1 - w_{x_m,x}} \\
&=& \left( \prod_{x \in \tilde{\mathcal{A}} \setminus \mathcal{A}_{n,m}} 1 \right)
 \prod_{x \in \mathcal{A}_{n,m}} \frac{w_{x_n,x}}{w_{x_m,x}}  \prod_{x \notin \mathcal{A}_{n,m}} \frac{1 - w_{x_n,x}}{1 - w_{x_m,x}} \\
&=& W(\mathcal{A}_{n,m}, \mathbf{w}_{x_n}, \mathbf{w}_{x_m}),
\end{eqnarray}
which contradicts the strict inequality in Eq.~(\ref{eq:contradiction2}).
\end{itemize}

In all cases, we arrive at a contradiction. Hence,
\begin{equation}
\frac{\Pr\left[\mathcal{M}_{\mathrm{I}}(x_n) = \mathcal{A}\right]}{\Pr\left[\mathcal{M}_{\mathrm{I}}(x_m) = \mathcal{A}\right]} \leq \frac{\Pr\left[\mathcal{M}_{\mathrm{I}}(x_n) = \mathcal{A}_{n,m}\right]}{\Pr\left[\mathcal{M}_{\mathrm{I}}(x_m) = \mathcal{A}_{n,m}\right]}.
\end{equation}

By symmetry, one can similarly prove that:
\begin{equation}
\frac{\Pr\left[\mathcal{M}_{\mathrm{I}}(x_n) = \overline{\mathcal{A}}_{n,m}\right]}{\Pr\left[\mathcal{M}_{\mathrm{I}}(x_m) = \overline{\mathcal{A}}_{n,m}\right]} 
\leq \frac{\Pr\left[\mathcal{M}_{\mathrm{I}}(x_n) = \mathcal{A}\right]}{\Pr\left[\mathcal{M}_{\mathrm{I}}(x_m) = \mathcal{A}\right]}.
\end{equation}

This completes the proof.
\end{proof}

\subsection{Proof of Lemma \ref{lem:PhI_privacycost}}
\label{subsec:proof:lem:PhI_privacycost}
\begin{relemma}
For each pair $x_n, x_m \in \mathcal{X}$, $\mathcal{M}_{\mathrm{I}}(x_n) \stackrel{(\overline{\epsilon}_{x_n,x_m},0)}{\sim}  \mathcal{M}_{\mathrm{I}}(x_m)$. 
\end{relemma}
\begin{proof}
First, according to the definition of $\epsilon_{x_n,x_m,\mathcal{A}}$
\begin{eqnarray}
&& \epsilon_{x_n,x_m,\mathcal{A}} = \left|\frac{\ln \left(\frac{\Pr\left[\mathcal{M}_{\mathrm{I}}(x_n) = \mathcal{A}\right]}{\Pr\left[\mathcal{M}_{\mathrm{I}}(x_m) = \mathcal{A}\right]}\right)}{d_{x_n,x_m}}\right| \\
&\Rightarrow& \left|\ln \left(\frac{\Pr\left[\mathcal{M}_{\mathrm{I}}(x_n) = \mathcal{A}\right]}{\Pr\left[\mathcal{M}_{\mathrm{I}}(x_m) = \mathcal{A}\right]}\right)\right| = \epsilon_{x_n,x_m,\mathcal{A}}d_{x_n,x_m} \\
&& \leq \overline{\epsilon}_{x_n,x_m}d_{x_n,x_m} \\
&\Rightarrow& \Pr\left[\mathcal{M}_{\mathrm{I}}(x_n) \stackrel{\overline{\epsilon}_{x_n,x_m}}{\sim} \mathcal{M}_{\mathrm{I}}(x_m)\right] = 1 \\
&\Rightarrow& \mathcal{M}_{\mathrm{I}}(x_n) \stackrel{(\overline{\epsilon}_{x_n,x_m},0)}{\sim} \mathcal{M}_{\mathrm{I}}(x_m).
\end{eqnarray}
The proof is complete. 
\end{proof}

\subsection{Proof of Proposition \ref{prop:wsmall}}
\label{subsec:proof:prop:wsmall}
\begin{reproposition}
To ensure that $\mathcal{M}_{\mathrm{I}}(x_n) \stackrel{\epsilon}{\approx} \mathcal{M}_{\mathrm{I}}(x_m)$ for any finite $\epsilon > 0$, it is necessary that if $w_{x_n,x} = 1$, then $w_{x_m,x} = 1$ must also hold.
\end{reproposition}
\begin{proof}
We prove this by contradiction. Suppose there exists some anchor $x \in \mathcal{X}$ such that $w_{x_n,x} = 1$ but $w_{x_m,x} < 1$. This means that anchor $x$ is always included in the subset selected by $\mathcal{M}_{\mathrm{I}}(x_n)$, but may be excluded from $\mathcal{M}_{\mathrm{I}}(x_m)$ with non-zero probability.

Now consider a subset $\mathcal{A} \subseteq \mathcal{X}$ such that $x \notin \mathcal{A}$ and $\Pr[\mathcal{M}_{\mathrm{I}}(x_m) = \mathcal{A}] > 0$. Since $w_{x_n,x} = 1$, any subset output by $\mathcal{M}_{\mathrm{I}}(x_n)$ must include $x$, and hence:
\begin{equation}
\Pr[\mathcal{M}_{\mathrm{I}}(x_n) = \mathcal{A}] = 0.
\end{equation}
Therefore, we have:
\begin{equation}
\frac{\Pr[\mathcal{M}_{\mathrm{I}}(x_n) = \mathcal{A}]}{\Pr[\mathcal{M}_{\mathrm{I}}(x_m) = \mathcal{A}]} = 0 < e^{-\epsilon d_{x_n,x_m}},
\end{equation}
which violates the $\epsilon$-mDP constraint.

This contradiction implies that our initial assumption is false. Thus, to satisfy the mDP constraints, it is necessary that if $w_{x_n,x} = 1$, then it must also hold that $w_{x_m,x} = 1$.
\end{proof}

\subsection{Proof of Proposition \ref{prop:M2budget}}
\label{subsec:proof:prop:M2budget}
\begin{reproposition}
According to the sequential composition of PmDP (Theorem \ref{thm:composition}), when designing $\mathcal{M}_{\mathrm{II}}$, to ensure 
\begin{equation}
    \mathcal{M}_{\mathrm{\textsc{PAnDA}}}(x_n) \stackrel{(\epsilon,\delta)}{\sim} \mathcal{M}_{\mathrm{\textsc{PAnDA}}}(x_m), 
\end{equation}
it is sufficient to enforce 
\begin{eqnarray}
    \mathcal{M}_{\mathrm{II}}(x_n) \stackrel{(\epsilon - \overline{\epsilon}_{x_n,x_m},\delta)}{\sim}  \mathcal{M}_{\mathrm{II}}(x_m)
\end{eqnarray}
\end{reproposition}

\begin{proof}
For any pair $x_n, x_m \in \mathcal{X}$, it follows from Lemma~\ref{lem:PhI_privacycost} that
\begin{equation}
\mathcal{M}_{\mathrm{I}}(x_n) \stackrel{(\overline{\epsilon}_{x_n,x_m},0)}{\sim} \mathcal{M}_{\mathrm{I}}(x_m).
\end{equation}
Now, assume the following holds:
\begin{equation}
\mathcal{M}_{\mathrm{II}}(x_n) \stackrel{(\epsilon - \overline{\epsilon}_{x_n,x_m},\delta)}{\sim} \mathcal{M}_{\mathrm{II}}(x_m).
\end{equation}
Then, by the sequential composition property of $(\epsilon,\delta)$-PmDP (Theorem~\ref{thm:composition}), the composed mechanism satisfies:
\begin{equation}
\mathcal{M}_{\mathrm{\textsc{PAnDA}}}(x_n) \stackrel{(\epsilon,\delta)}{\sim} \mathcal{M}_{\mathrm{\textsc{PAnDA}}}(x_m).
\end{equation}

This completes the proof.
\end{proof}

\DEL{
\subsection{Proof of Proposition \ref{prop:PL_\textsc{PAnDA}}} 
\begin{proof}
Given the definition of PL defined in Eq. (\ref{eq:PL\textsc{PAnDA}}), we can obtain that 
\begin{eqnarray}
&& e^{-\epsilon d_{n, m}} \leq \mathrm{PL}\left((x_n, x_m); \mathcal{M}_{\mathrm{I}}, \mathcal{M}_{\mathrm{II}}\right) 
\leq e^{\epsilon d_{n, m}} \\
&\Leftrightarrow& e^{-\epsilon d_{n, m}} \leq  \sup_{\mathcal{A}\in 2^{\mathcal{X}}}\frac{\frac{\Pr\left(X_n = x_n |\mathcal{M}_{\mathrm{I}}(X_n) = \mathcal{A}, \mathcal{M}_{\mathrm{II}}(X_n) = y_k\right)}{\Pr\left(X_n = x_m | \mathcal{M}_{\mathrm{I}}(X_n) = \mathcal{A}, \mathcal{M}_{\mathrm{II}}(X_n) = y_k\right)}}{\frac{\Pr\left(X_n = x_n\right)}{\Pr\left(X_n = x_m\right)}} \leq e^{\epsilon d_{n, m}} \\
&\Leftrightarrow& e^{-\epsilon d_{n, m}} \leq  \sup_{\mathcal{A}\in 2^{\mathcal{X}}}\frac{\frac{\Pr\left(\mathcal{M}_{\mathrm{I}}(X_n) = \mathcal{A}, \mathcal{M}_{\mathrm{II}}(X_n) = y_k|X_n = x_n \right)\Pr\left(X_n = x_n \right)}{\Pr\left(\mathcal{M}_{\mathrm{I}}(X_n) = \mathcal{A}, \mathcal{M}_{\mathrm{II}}(X_n) = y_k|X_n = x_m \right)\Pr\left(X_n = x_m \right)}}{\frac{\Pr\left(X_n = x_n\right)}{\Pr\left(X_n = x_m\right)}} \leq e^{\epsilon d_{n, m}}%  ~\mbox{(according to Bayes' theorem)} 
\\
&\Leftrightarrow& e^{-\epsilon d_{n, m}} \leq  \sup_{\mathcal{A}\in 2^{\mathcal{X}}}\frac{\Pr\left(\mathcal{M}_{\mathrm{I}}(X_n) = \mathcal{A}, \mathcal{M}_{\mathrm{II}}(X_n) = y_k|X_n = x_n \right)}{\Pr\left(\mathcal{M}_{\mathrm{I}}(X_n) = \mathcal{A}, \mathcal{M}_{\mathrm{II}}(X_n) = y_k|X_n = x_m \right)} \leq e^{\epsilon d_{n, m}}\\
&\Leftrightarrow& e^{-\epsilon d_{n, m}} \leq  \frac{\Pr\left(\mathcal{M}_{\mathrm{I}}(X_n) = \mathcal{A}, \mathcal{M}_{\mathrm{II}}(X_n) = y_k|X_n = x_n \right)}{\Pr\left(\mathcal{M}_{\mathrm{I}}(X_n) = \mathcal{A}, \mathcal{M}_{\mathrm{II}}(X_n) = y_k|X_n = x_m \right)} \leq e^{\epsilon d_{n, m}},~\forall \mathcal{A} \in 2^{\mathcal{X}}
\end{eqnarray}
which is the mDP of \textsc{PAnDA} defined in Eq. (\ref{eq:mDP\textsc{PAnDA}}). The proof is completed. 
\end{proof}}
\DEL{
\begin{eqnarray}
&\Leftrightarrow&  \frac{\Pr\left(\mathcal{M}_{\mathrm{I}}(X_n) = \mathcal{A}, \mathcal{M}_{\mathrm{II}}(X_n) = y_k|X_n = x_n \right)+\delta}{\Pr\left(\mathcal{M}_{\mathrm{I}}(X_n) = \mathcal{A}, \mathcal{M}_{\mathrm{II}}(X_n) = y_k|X_n = x_m \right)} \leq e^{\epsilon d_{n, m}} \\
&\Leftrightarrow&  \frac{\frac{\left(\Pr\left(\mathcal{M}_{\mathrm{I}}(X_n) = \mathcal{A}, \mathcal{M}_{\mathrm{II}}(X_n) = y_k|X_n = x_n \right)+\delta\right)\Pr\left(X_n = x_n \right)}{\Pr\left(\mathcal{M}_{\mathrm{I}}(X_n) = \mathcal{A}, \mathcal{M}_{\mathrm{II}}(X_n) = y_k|X_n = x_m \right)\Pr\left(X_n = x_m \right)}}{\frac{\Pr\left(X_n = x_n\right)}{\Pr\left(X_n = x_m\right)}} \leq e^{\epsilon d_{n, m}} \\
&\Leftrightarrow&  \frac{\frac{\Pr\left(\mathcal{M}_{\mathrm{I}}(X_n) = \mathcal{A}, \mathcal{M}_{\mathrm{II}}(X_n) = y_k, X_n = x_n \right)+\delta\Pr\left(X_n = x_n \right)}{\Pr\left(\mathcal{M}_{\mathrm{I}}(X_n) = \mathcal{A}, \mathcal{M}_{\mathrm{II}}(X_n) = y_k,X_n = x_m \right)}}{\frac{\Pr\left(X_n = x_n\right)}{\Pr\left(X_n = x_m\right)}} \leq e^{\epsilon d_{n, m}} \\
&\Leftrightarrow&  \frac{\frac{\Pr\left(X_n = x_n|\mathcal{M}_{\mathrm{I}}(X_n) = \mathcal{A}, \mathcal{M}_{\mathrm{II}}(X_n) = y_k \right)+\delta\frac{\Pr\left(X_n = x_n \right)}{\Pr\left(\mathcal{M}_{\mathrm{I}}(X_n) = \mathcal{A}, \mathcal{M}_{\mathrm{II}}(X_n) = y_k \right)}}{\Pr\left(X_n = x_m|\mathcal{M}_{\mathrm{I}}(X_n) = \mathcal{A}, \mathcal{M}_{\mathrm{II}}(X_n) = y_k \right)}}{\frac{\Pr\left(X_n = x_n\right)}{\Pr\left(X_n = x_m\right)}} \leq e^{\epsilon d_{n, m}}
\end{eqnarray}

\subsection{Proof of Proposition \ref{prop:PL}}

\begin{proof}
First, we derive the posterior given the perturbed data $y_k$ and the anchor record set $\mathcal{A}$. Note that given $X_n = x_n$, $\mathcal{M}_{\mathrm{II}}(X_n) = y_k$ and $\mathcal{M}_{\mathrm{I}}(X_n) = \mathcal{A}$ are \emph{conditionally independent}, i.e., 
\begin{eqnarray}
&& \Pr\left(\mathcal{M}_{\mathrm{II}}(X_n) = y_k| X_n = x_n, \mathcal{M}_{\mathrm{I}}(X_n) = \mathcal{A}\right) \\
&=& \Pr\left(\mathcal{M}_{\mathrm{II}}(X_n) = y_k| X_n = x_n\right) \\
&=& z_{n,k}. 
\end{eqnarray}
The posterior of the real record $X_n$ is calculated by 
\begin{eqnarray}
&& \Pr\left(X_n = x_n | \mathcal{M}_{\mathrm{II}}(X_n) = y_k, \mathcal{M}_{\mathrm{I}}(X_n) = \mathcal{A}\right) \\ 
&=& \frac{\Pr\left(X_n = x_n, \mathcal{M}_{\mathrm{II}}(X_n) = y_k, \mathcal{M}_{\mathrm{I}}(X_n) = \mathcal{A}\right)}{\Pr\left(\mathcal{M}_{\mathrm{II}}(X_n) = y_k, \mathcal{M}_{\mathrm{I}}(X_n) = \mathcal{A}\right)} \\ 
&=& \frac{\Pr\left(\mathcal{M}_{\mathrm{II}}(X_n) = y_k| X_n = x_n, \mathcal{M}_{\mathrm{I}}(X_n) = \mathcal{A}\right)\Pr\left(\mathcal{M}_{\mathrm{I}}(X_n) = \mathcal{A}|X_n = x_n\right)\Pr\left(X_n = x_n\right)}{\Pr\left(\mathcal{M}_{\mathrm{II}}(X_n) = y_k, \mathcal{M}_{\mathrm{I}}(X_n) = \mathcal{A}\right)} \\ %{\sum_{x_\ell \in \mathcal{X}}\Pr\left(\mathcal{M}_{\mathrm{II}}(X_n) = y_k| X_n = x_\ell, \mathcal{M}_{\mathrm{I}}(X_n) = \mathcal{A}\right)\Pr\left(X_n = x_\ell, \mathcal{M}_{\mathrm{I}}(X_n) = \mathcal{A}\right)}\\ 
% &=& \frac{\Pr\left(\mathcal{M}_{\mathrm{II}}(X_n) = y_k| X_n = x_n\right)\Pr\left(X_n = x_n, \mathcal{M}_{\mathrm{I}}(X_n) = \mathcal{A}\right)}{\sum_{x_\ell \in \mathcal{X}}\Pr\left(\mathcal{M}_{\mathrm{II}}(X_n) = y_k| X_n = x_\ell\right)\Pr\left(X_n = x_\ell, \mathcal{M}_{\mathrm{I}}(X_n) = \mathcal{A}\right)}\\ 
% &=& \frac{\Pr\left(\mathcal{M}_{\mathrm{II}}(X_n) = y_k| X_n = x_n\right)\Pr\left(\mathcal{M}_{\mathrm{I}}(X_n) = \mathcal{A}|X_n = x_n\right)\Pr\left(X_n = x_n\right)}{\sum_{x_n, \mathcal{A}}\Pr\left(\mathcal{M}_{\mathrm{II}}(X_n) = y_k| X_n = x_n\right)\sum_{\mathcal{A}}\Pr\left(\mathcal{M}_{\mathrm{I}}(X_n) = \mathcal{A}|X_n = x_n\right)\Pr\left(X_n = x_n\right)} \\
&=& \frac{z_{n,k}u_{i, \mathcal{A}}\Pr\left(X_n = x_n\right)}{\Pr\left(\mathcal{M}_{\mathrm{II}}(X_n) = y_k, \mathcal{M}_{\mathrm{I}}(X_n) = \mathcal{A}\right)}%{\sum_{x_\ell \in \mathcal{X}}z_{l,k}\Pr\left(\mathcal{M}_{\mathrm{I}}(X_n) = \mathcal{A}|X_n = x_\ell\right)\Pr\left(X_n = x_\ell\right)}
\end{eqnarray}
\DEL{
Accordingly, the posterior ratio given $\mathcal{A}$ is derived by 
\begin{eqnarray}
\frac{\Pr\left(X_n = x_n | \mathcal{M}_{\mathrm{II}}(X_n) = y_k, \mathcal{M}_{\mathrm{I}}(X_n) = \mathcal{A}\right)}{\Pr\left(X_n = x_m | \mathcal{M}_{\mathrm{II}}(X_n) = y_k, \mathcal{M}_{\mathrm{I}}(X_n) = \mathcal{A}\right)} = \frac{z_{n,k}\Pr\left(\mathcal{M}_{\mathrm{I}}(X_n) = \mathcal{A}|X_n = x_n\right)\Pr\left(X_n = x_n\right)}{z_{m,k}\Pr\left(\mathcal{M}_{\mathrm{I}}(X_n) = \mathcal{A}|X_n = x_m\right)\Pr\left(X_n = x_m\right)} % \times \frac{\sum_{x_\ell \in \mathcal{X}}z_{l,k}\Pr\left(\mathcal{M}_{\mathrm{I}}(X_n) = \mathcal{A}|X_n = x_\ell\right)\Pr\left(X_n = x_\ell\right)}{\sum_{x_\ell \in \mathcal{X}}z_{l,k}\Pr\left(\mathcal{M}_{\mathrm{I}}(X_n) = \mathcal{A}|X_n = x_\ell\right)\Pr\left(X_n = x_\ell\right)}
\end{eqnarray}}
from which we can derive the posterior leakage $\mathrm{PL}\left((x_n, x_m); \mathcal{M}_{\mathrm{I}}, \mathcal{M}_{\mathrm{II}}\right)$ as 
\begin{eqnarray}
&& \mathrm{PL}\left((x_n, x_m); \mathcal{M}_{\mathrm{I}}, \mathcal{M}_{\mathrm{II}}\right) \\
&=& \sup_{\mathcal{A}\in 2^{\mathcal{X}}, y_k \in \mathcal{Y}}
 \underbrace{\frac{\Pr\left(X_n = x_n | \mathcal{M}_{\mathrm{II}}(X_n) = y_k, \mathcal{M}_{\mathrm{I}}(X_n) = \mathcal{A}\right)}{\Pr\left(X_n = x_m | \mathcal{M}_{\mathrm{II}}(X_n) = y_k, \mathcal{M}_{\mathrm{I}}(X_n) = \mathcal{A}\right)}}_{\mbox{posterior ratio}}\left\slash\underbrace{\frac{\Pr\left(X_n = x_n\right)}{\Pr\left(X_n = x_m\right)}}_{\mbox{prior ratio}}\right. \\
% &=& \sup_{\mathcal{A} \in 2^{\mathcal{X}}} \frac{\Pr\left(\mathcal{M}_{\mathrm{I}}(X_n) = \mathcal{A}|X_n = x_n\right)}{\Pr\left(\mathcal{M}_{\mathrm{I}}(X_n) = \mathcal{A}|X_n = x_m\right)} \times  \frac{z_{n,k}}{z_{m,k}} % \\
% &=& \frac{z_{n,k}\Pr\left(\mathcal{M}_{\mathrm{I}}(X_n) = \mathcal{A}|X_n = x_n\right)}{z_{m,k}\Pr\left(\mathcal{M}_{\mathrm{I}}(X_n) = \mathcal{A}'|X_n = x_m\right)} \\
% &\times & \frac{\sum_{x_m}z_{m,k}\Pr\left(\mathcal{M}_{\mathrm{I}}(X_n) = \mathcal{A}'|X_n = x_m\right)\Pr\left(X_n = x_m\right)}{\sum_{x_n}z_{n,k}\Pr\left(\mathcal{M}_{\mathrm{I}}(X_n) = \mathcal{A}|X_n = x_n\right)\Pr\left(X_n = x_n\right)} 
% &\leq & e^{\epsilon d_{x_n,x_m}}
&=& \sup_{\mathcal{A}\in 2^{\mathcal{X}}, y_k \in \mathcal{Y}}
 \underbrace{\frac{\frac{z_{n,k}u_{i, \mathcal{A}}\Pr\left(X_n = x_n\right)}{\Pr\left(\mathcal{M}_{\mathrm{II}}(X_n) = y_k, \mathcal{M}_{\mathrm{I}}(X_n) = \mathcal{A}\right)}}{\frac{z_{m,k}u_{j, \mathcal{A}}\Pr\left(X_n = x_m\right)}{\Pr\left(\mathcal{M}_{\mathrm{II}}(X_n) = y_k, \mathcal{M}_{\mathrm{I}}(X_n) = \mathcal{A}\right)}}}_{\mbox{posterior ratio}}\left\slash\underbrace{\frac{\Pr\left(X_n = x_n\right)}{\Pr\left(X_n = x_m\right)}}_{\mbox{prior ratio}}\right. \\
&=& \sup_{\mathcal{A} \in 2^{\mathcal{X}}} \frac{u_{i, \mathcal{A}}}{u_{j, \mathcal{A}}} \times \sup_{y_k \in \mathcal{Y}}\frac{z_{n,k}}{z_{m,k}}.
\end{eqnarray}
\DEL{
According to Definition \ref{def:PL\textsc{PAnDA}}, 
\begin{eqnarray}
&& \mathrm{PL}\left((x_n, x_m); \mathcal{M}_{\mathrm{I}}, \mathcal{M}_{\mathrm{II}}\right) \\
&=& \sup_{\mathcal{A}\ni x_n, \mathcal{A}' \ni x_m }\frac{\frac{\Pr\left(X_n = x_n |\mathcal{M}_{\mathrm{I}}(X_n) = \mathcal{A}, \mathcal{M}_{\mathrm{II}}(X_n) = y_k\right)}{\Pr\left(X_n = x_m | \mathcal{M}_{\mathrm{II}}(X_n) = \mathcal{A}', \mathcal{M}_{\mathrm{II}}(X_n) = y_k\right)}}{\frac{\Pr\left(X_n = x_n\right)}{\Pr\left(X_n = x_m\right)}}\\
\label{eq:PLupperbound1}
&=& \sup_{\mathcal{A}\ni x_n, \mathcal{A}' \ni x_m} \frac{u_{i, \mathcal{A}} \sum_{x_\ell \in \mathcal{X}}z_{l,k}u_{x_\ell, \mathcal{A}'}p_{x_\ell}}{u_{x_m, \mathcal{A}'}\sum_{x_\ell \in \mathcal{X}}z_{l,k}u_{x_\ell, \mathcal{A}}p_{x_\ell}} \times \frac{z_{n,k}}{z_{m,k}}
\end{eqnarray}}
The proof is completed. 
\end{proof}}

\subsection{Proof of Proposition \ref{prop:failuremDPprob}}
\label{subsec:proof:prop:failuremDPprob}
\begin{reproposition}
For each pair of real records $x_n$ and $x_m$, the probability of success is lower bounded by 
\begin{eqnarray}
h_{x_n,x_m}(\xi) = \sum_{\left(\hat{x}_{n},\hat{x}_{m}\right) \in \mathcal{H}_{x_n,x_m}(\xi)}v_{x_{n}, \hat{x}_{n}} v_{x_{m}, \hat{x}_{m}}
\end{eqnarray}
where each $v_{x_{q}, \hat{x}_{q}} = w_{x_q,\hat{x}_q} \prod_{x \in \mathcal{U}_{x_q,\hat{x}_q}}(1-w_{x_q,x}) (q = n, m)$ represents the probability that $\hat{x}_{q}$ is $x_{q}$'s surrogate and $$\mathcal{U}_{x_{q},\hat{x}_q} = \left\{x\in \mathcal{X}| d_{x_{q},x} < d_{x_{q},\hat{x}_q}\right\}$$ represents the set of records that is closer to $x_{q}$ compared to $\hat{x}_q$.
\end{reproposition}

\begin{proof}
Given a real record $x_n$ and a candidate anchor $\hat{x}_n$, define $\mathcal{U}_{x_n,\hat{x}_n} = \left\{ x \in \mathcal{X} \mid d_{x_n,x} < d_{x_n,\hat{x}_n} \right\}$ as the set of records closer to $x_n$ than $\hat{x}_n$.

A record $\hat{x}_n$ becomes the surrogate of $x_n$ if and only if:
\begin{itemize}
    \item[(1)] $\hat{x}_n \in \mathcal{A}_{x_n}$, i.e., $\hat{x}_n$ is selected as one of $x_n$'s anchors, and
    \item[(2)] $\mathcal{U}_{x_n,\hat{x}_n} \cap \mathcal{A}_{x_n} = \emptyset$, i.e., no closer record is selected as an anchor.
\end{itemize}

The probability of this event is:
\begin{eqnarray}
&& \Pr[\text{$\hat{x}_n$ is the surrogate of $x_n$}] 
\\
&=& \Pr[\hat{x}_n \in \mathcal{A}_{x_n}]  \Pr[\mathcal{U}_{x_n,\hat{x}_n} \cap \mathcal{A}_{x_n} = \emptyset] \\
&=& w_{x_n, \hat{x}_n}  \prod_{x \in \mathcal{U}_{x_n, \hat{x}_n}} (1 - w_{x_n, x}).
\end{eqnarray}

Now consider the joint event that $x$ and $x'$ are the surrogates of $x_n$ and $x_m$, respectively, for some pair $(x, x') \in \mathcal{H}_{x_n,x_m}(\xi)$ that satisfies the mDP constraint. Then, the success probability of satisfying mDP is at least:
\begin{eqnarray}
\nonumber && \Pr\left[\mbox{$\mathcal{M}_{\mathrm{\textsc{PAnDA}}}(x_n) \stackrel{\epsilon}{\approx} \mathcal{M}_{\mathrm{\textsc{PAnDA}}}(x_m)$}\right]\\
\nonumber &\geq& \sum_{\left(x, x'\right)\in \mathcal{H}_{x_n,x_m}(\xi)} \Pr\left[\mbox{$x$ and $x'$ are surrogates of $x_n$ and $x_m$}\right]\\ \nonumber
&=& \underbrace{\sum_{\left(x, x'\right)\in \mathcal{H}_{x_n,x_m}(\xi)}\left(\begin{array}{l}
     w_{x_n,x} \prod_{\tilde{x} \in \mathcal{U}_{x_n,x}}(1-w_{x_n,\tilde{x}})\\
     \times w_{x_m,x'} \prod_{\tilde{x} \in \mathcal{U}_{x_m,x'}}(1-w_{x_m,\tilde{x}}) 
\end{array}\right)}_{h(\xi)}
\end{eqnarray}

This completes the proof.
\end{proof}

\subsection{Proof of Property \ref{prop:xi}}
\label{subsec:proof:prop:xi}
\begin{reproperty}
(1) $h_{x_n,x_m}(\xi)$ is \emph{monotonically non-decreasing} in $\xi$, and (2) by setting  $\xi_{x_n, x_m} = h_{x_n,x_m}^{-1}(1-\delta)$
we can guarantee $(\epsilon, \delta)$-PmDP for $(x_n, x_m)$. 
\end{reproperty}
\begin{proof}
(1) Given any $\xi_1 < \xi_2$, 
\begin{eqnarray}
&& (\epsilon-\overline{\epsilon}_{x,x'}) d_{x,x'}-\xi_1 \leq (\epsilon - \overline{\epsilon}_{x_n,x_m}) d_{x_n,x_m} \\
&\Rightarrow & (\epsilon-\overline{\epsilon}_{x,x'}) d_{x,x'}-\xi_2 \leq (\epsilon - \overline{\epsilon}_{x_n,x_m}) d_{x_n,x_m}
\end{eqnarray}
indicating that $\mathcal{H}_{x_n,x_m}(\xi_1) \subseteq \mathcal{H}_{x_n,x_m}(\xi_2)$. 
\normalsize
% \small
\begin{eqnarray}
\nonumber 
&& h_{n,m}(\xi_2) - h_{n,m}(\xi_1)\\ \nonumber 
&=& \sum_{
\left(\hat{x}_{n}, \hat{x}_{m}\right) \in
\mathcal{H}_{x_n,x_m}(\xi_2)\backslash \mathcal{H}_{x_n,x_m}(\xi_1)}
\left(w_{x_n,\hat{x}_n} \prod_{x_\ell \in \mathcal{U}_{x_n,\hat{x}_n}}(1-w_{x_n,x_\ell})\right. \\ \nonumber 
&&\times \left.  w_{x_m,\hat{x}_m} \prod_{x_\ell \in \mathcal{U}_{x_m,\hat{x}_m}}(1-w_{x_m,x_\ell})\right) \\ \nonumber 
&\geq& 0.
\end{eqnarray}
\normalsize
\end{proof}

\subsection{Proof of Proposition \ref{prop:samplebudget}}
\label{subsec:proof:prop:samplebudget}
\begin{reproposition}
By applying the privacy budget $(\epsilon-\overline{\epsilon}_{x,x'}) d_{x,x'}-\hat{\xi}_{x,x'}$, it can guarantee $(\epsilon, \delta)$-PmDP for $(x_n, x_m)$.
\end{reproposition}
\begin{proof}
Note that locations $(x_n,x_m) \in \mathcal{S}_{x} \times \mathcal{S}_{x'} \Rightarrow \hat{\xi}_{x,x'} \geq \xi_{x_n,x_m}$, hence 
\begin{eqnarray}
\nonumber 
\Pr\left[\hat{\xi}_{x,x'} \geq \xi_{x_n,x_m}\right] \geq \Pr\left[(x_n,x_m) \in \mathcal{S}_{x} \times \mathcal{S}_{x'}|x\in \mathcal{A}_n, x'\in \mathcal{A}_m\right] 
\end{eqnarray}
where 
\normalsize
{\rd \begin{eqnarray}
\label{eq:PPP1}
&& \Pr\left[(x_n,x_m) \in \mathcal{S}_{x} \times \mathcal{S}_{x'}|x\in \mathcal{A}_n, x'\in \mathcal{A}_m\right] \\ \nonumber 
&=& \sum_{(\tilde{x},\tilde{x}') \in \mathcal{S}_{x} \times \mathcal{S}_{x'}} \Pr\left[X_n = \tilde{x}| x \in \mathcal{A}_n \right] \times \Pr\left[X_m = \tilde{x}' | x' \in \mathcal{A}_m\right] \\\label{eq:PPP2}
&=& \sum_{(\tilde{x},\tilde{x}') \in \mathcal{S}_{x} \times \mathcal{S}_{x'}} \frac{ p_{\tilde{x}} w_{\tilde{x}, x}}{\sum_{u \in \mathcal{X}}p_u w_{u, x}} \times \frac{p_{\tilde{x}'} w_{\tilde{x}', x'}}{\sum_{u' \in \mathcal{X}}p_{u'} w_{u', x'}}. 
\end{eqnarray}}
\normalsize
Here, we let  
\begin{eqnarray}
1-\delta_{x, x'} = \frac{1-\delta}{\Pr\left[(x_n,x_m) \in \mathcal{S}_{x} \times \mathcal{S}_{x'}|x\in \mathcal{A}_n, x'\in \mathcal{A}_m\right] }.
\end{eqnarray}
If $\Pr\left[(x_n,x_m) \in \mathcal{S}_{x} \times \mathcal{S}_{x'}|x\in \mathcal{A}_n, x'\in \mathcal{A}_m\right] < 1-\delta$,  
\begin{eqnarray}
&& \Pr\left[\mbox{$\mathcal{M}_{\mathrm{\textsc{PAnDA}}}(x_n) \stackrel{\epsilon}{\approx} \mathcal{M}_{\mathrm{\textsc{PAnDA}}}(x_m)$}\right] \\
&=& \Pr\left[\hat{\xi}_{x,x'} \geq \xi_{x_n,x_m}, \mbox{$\mathcal{M}_{\mathrm{\textsc{PAnDA}}}(x_n) \stackrel{\epsilon}{\approx} \mathcal{M}_{\mathrm{\textsc{PAnDA}}}(x_m)$}\right]  \\ \nonumber
&+& \Pr\left[\hat{\xi}_{x,x'} < \xi_{x_n,x_m}, \mbox{$\mathcal{M}_{\mathrm{\textsc{PAnDA}}}(x_n) \stackrel{\epsilon}{\approx} \mathcal{M}_{\mathrm{\textsc{PAnDA}}}(x_m)$}\right] \\
&\geq & \Pr\left[\mbox{$\mathcal{M}_{\mathrm{\textsc{PAnDA}}}(x_n) \stackrel{\epsilon}{\approx} \mathcal{M}_{\mathrm{\textsc{PAnDA}}}(x_m)$}|\hat{\xi}_{x,x'} \geq \xi_{x_n,x_m}\right]\\ \nonumber
&\times& \Pr\left[\hat{\xi}_{x,x'} \geq \xi_{x_n,x_m}\right] \\
&>& (1-\delta_{x, x'}) \Pr\left[\hat{\xi}_{x,x'} \geq \xi_{x_n,x_m}\right]  \\ \nonumber
&=& \frac{1-\delta}{\Pr\left[(x_n,x_m) \in \mathcal{S}_{x} \times \mathcal{S}_{x'}|x\in \mathcal{A}_n, x'\in \mathcal{A}_m\right] } \\
&\times& \Pr\left[\hat{\xi}_{x,x'} \geq \xi_{x_n,x_m}\right] \\
&\geq & 1 - \delta. 
\end{eqnarray}
The proof is completed. 
\end{proof}

\subsection{Proof of Proposition \ref{prop:ULbound}}
\label{subsec:proof:prop:ULbound}
\begin{reproposition}
Let $\tilde{\mathbf{Z}}_{[1, N]}$ denote the optimal solution of the Relaxed \textsc{AnPO}. Let $\mathbf{Z}^*$ denote the optimal soution of the original perturbation optimization (as defined in Eq. (\ref{eq:PPO})) and let $\mathbf{Z}^*_{[1, N]}$ denote the submatrix of $\mathbf{Z}^*$ that covers $\mathcal{A}_{{[1, N]}}$. Then, $
\mathcal{L}\left(\tilde{\mathbf{Z}}_{[1, N]}\right)\leq \mathcal{L}\left(\mathbf{Z}^*_{[1, N]}\right)$. 
\end{reproposition}
\begin{proof}
Let $\mathcal{F}_{\text{orig}}$ denote the feasible region of $\mathbf{Z}_{[1, N]}$ in the original perturbation optimization problem. Let $\mathcal{F}_{\text{relaxed}}$ denote the feasible region of $\mathbf{Z}_{[1, N]}$ in the Relaxed \textsc{AnPO} problem. Clearly, a feasible solution in $\mathcal{F}_{\text{orig}}$, also satisfies all the constraints of $\mathcal{F}_{\text{relaxed}}$, indicating that 
\begin{equation}
\mathcal{F}_{\text{orig}} \subseteq \mathcal{F}_{\text{relaxed}}.
\end{equation}

By definition, $\tilde{\mathbf{Z}}_{[1, N]}$ is the solution that minimizes the utility loss $\mathcal{L}(\cdot)$ over $\mathcal{F}_{\text{relaxed}}$, while $\mathbf{Z}^*_{[1, N]}$ is feasible in $\mathcal{F}_{\text{orig}} \subseteq \mathcal{F}_{\text{relaxed}}$. Therefore,
\begin{equation}
\mathcal{L}\left(\tilde{\mathbf{Z}}_{[1, N]}\right) \leq \mathcal{L}\left(\mathbf{Z}^*_{[1, N]}\right).
\end{equation}

\noindent This concludes the proof.
\end{proof}

\section{Detailed Description of Algorithms}
\label{sec:algorithmdetail}
\subsection{Pseudo Code}
\vspace{-0.00in}
\begin{algorithm}[h]
\SetKwFunction{DFS}{DFS}
\SetKwInOut{Input}{Input}
\SetKwInOut{Output}{Output}
\SetKwProg{Fn}{Function}{}{}
% \small 
% \scriptsize
\Input{Anchor selection method $\mathcal{M}_{\mathrm{I}}$}
\Output{$v_{x_{q}, \hat{x}_{q}}$ for each pair $(x_q, \hat{x}_q) \in \mathcal{X}^2$ and sorted $\Upsilon_{x_n,x_m}$ for each pair of users $n$ and $m$}
    \For{each $(x_q, \hat{x}_q) \in \mathcal{X}^2$}{
        $v_{x_{q}, \hat{x}_{q}} \leftarrow w_{x_q,\hat{x}_q} \prod_{x \in \mathcal{U}_{x_q,\hat{x}_q}}(1-w_{x_q,x})$\; 
    }
    \For{each pair of users $n$ and $m$}{
        $\Upsilon_{x_n,x_m} \leftarrow \emptyset$\;
        \For{each $(x, x') \in \mathcal{X}^2$}{
            $\Delta_{x,x'} \leftarrow (\epsilon-\overline{\epsilon}_{x, x'})d_{x, x'} - (\epsilon - \overline{\epsilon}_{x_n,x_m}) d_{x_n,x_m}$\;
            $\Upsilon_{x_n,x_m} \leftarrow \Upsilon_{x_n,x_m} \cup \Delta_{x,x'}$\;
        }
        Sort the elements in $\Upsilon_{x_n,x_m}$ increasing order\;
    }
    \Return $\left\{v_{x_{q}, \hat{x}_{q}}\right\}_{(x_q, \hat{x}_q) \in \mathcal{X}^2}$ and $\Upsilon_{x_n,x_m}$\; 
\normalsize
\caption{Pre-computation of candidate safety margin set $\Upsilon_{x_n,x_m}$ and probabilities $v_{x_{q}, \hat{x}_{q}}$. % $\mathsf{Bin\_Search}(\Upsilon_{x_n,x_m}, \delta, \ell_{\mathsf{left}}, \ell_{\mathsf{right}})$. 
}
\label{al:recompute}
\end{algorithm}
\vspace{-0.00in}

Algorithm \ref{al:recompute} gives the pseudo code to 
calculate $v_{x_{q}, \hat{x}_{q}}$ (line 1--2), representing the probability that $\hat{x}_{q}$ is $x_{q}$'s surrogate, and the candidate safety margin set $\Upsilon_{x_n,x_m}$ (line 3--8). 

In particular, for each pair of users $n$ and $m$, we initiate $\Upsilon_{x_n,x_m}$ by empty (line 4), and add each $\Delta_{x,x'}$ one by one. Specifically, we assign 
$\Delta_{x,x'} = (\epsilon-\overline{\epsilon}_{x, x'})d_{x,x'} - (\epsilon - \overline{\epsilon}_{x_n,x_m}) d_{x_n,x_m}$ (line 6). After collecting all $\Delta_{x,x'},~x,x'\in \mathcal{X}^2$. We sort the elements in $\Upsilon_{x_n,x_m}$ increasing order (line 8). Without loss of generality, we represent the sorted elements by 
\begin{equation}
\Delta^1 \leq \Delta^2 \leq ... \leq \Delta^{|\mathcal{X}|^2}.
\end{equation}
We can obtain that $(x, x') \in \mathcal{H}_{x_n,x_m}(\xi)$ if and only if $\xi \geq \xi_{x, x'}$.

\vspace{-0.00in}
\begin{algorithm}[h]
\SetKwFunction{DFS}{DFS}
\SetKwInOut{Input}{Input}
\SetKwInOut{Output}{Output}
\SetKwProg{Fn}{Function}{}{}
% \small 
% \scriptsize
\Input{Sorted $\Upsilon_{x_n,x_m}$ for each pair of users $n$ and $m$ and $\delta_{x, x'}$ for each $(x, x') \in \mathcal{A}_{n}\times \mathcal{A}_{m}$}
\Output{Safety margin $\hat{\xi}_{x, x'}$ for each $(x, x') \in \mathcal{A}_{n}\times \mathcal{A}_{m}$}
    \For{each $(x, x') \in \mathcal{A}_{n}\times \mathcal{A}_{m}$}{
        $\hat{\xi}_{x, x'} \leftarrow 0$\;
        \For{each $(\tilde{x},\tilde{x}') \in \mathcal{S}_{x} \times \mathcal{S}_{x'}$}{
            $\ell \leftarrow 1$\;
            \tcp{Search $\xi_{x, x'}$ within $\Upsilon_{x_n,x_m} = \{\Delta^1, ..., \Delta^{|\mathcal{X}|^2}\}$} 
            \While{$h_{\tilde{x},\tilde{x}'}(\Delta^{\ell}) < 1-\delta_{x, x'}$ and $\ell \leq |\mathcal{X}|^2$}{
                $\ell \leftarrow \ell +1$\;
                \tcp{$\mathcal{H}_{\tilde{x},\tilde{x}'} (\Delta^{\ell}) = \mathcal{H}_{\tilde{x},\tilde{x}'}(\Delta^{\ell-1})\cup (x_i, x_j)$}
                $h_{\tilde{x},\tilde{x}'}(\Delta^{\ell}) \leftarrow h_{\tilde{x},\tilde{x}'}(\Delta^{\ell-1}) + v_{\tilde{x}, x_i} v_{\tilde{x}', x_j}$\;
                
            }
            \If{$\Delta^{\ell} > \hat{\xi}_{x, x'}$}{
                $\hat{\xi}_{x, x'} \leftarrow \Delta^{\ell}$\;
            }
        }
        
    }
\normalsize
\caption{Calculate $\xi_{x, x'}$ by linear search. % $\mathsf{Bin\_Search}(\Upsilon_{x_n,x_m}, \delta, \ell_{\mathsf{left}}, \ell_{\mathsf{right}})$. 
}
\label{al:xisearch}
\end{algorithm}
\vspace{-0.00in}

\DEL{
Therefore, by enforcing $\sum_{\left(\hat{x}_{n}, \hat{x}_{m}\right)\in \mathcal{H}_{x_n,x_m}}w_{x_n,\hat{x}_n} w_{x_m,\hat{x}_m} \geq 1-\delta$, we can guarantee the success ratio is not lower than $1-\delta$, or  
\begin{eqnarray}
%&&\sum_{\left(\hat{x}_{n}, \hat{x}_{m}\right)\notin \mathcal{H}_{x_n,x_m}}w_{x_n,\hat{x}_n} w_{x_m,\hat{x}_m} \leq \delta \\
% &\Rightarrow&  
\sum_{\left(\hat{x}_{n}, \hat{x}_{m}\right)\notin \mathcal{H}_{x_n,x_m}}\alpha h\left(d_{n,\hat{n}}\right) \alpha  h\left(d_{m,\hat{m}}\right) \leq \delta 
% \\ &\Rightarrow&  \alpha \leq \sqrt{\frac{\delta}{\sum_{\left(\hat{x}_{n}, \hat{x}_{m}\right)\in \mathcal{H}_{x_n,x_m}}h(d_{n,\hat{n}})h(d_{m,\hat{m}})}}
\end{eqnarray}
{\bl Remark: we need to pay attention as $\alpha$ might be lower bounded ... }. }

\DEL{
\begin{proposition}
\label{prop:adjmDP}
When determining the perturbation matrix $\mathbf{Z}_{\mathcal{X}}$, to achieve mDP constraints (in Eq. (\ref{eq:mDPdiscrete})), it is sufficient to achieve the following constraints: 
\begin{eqnarray}
\frac{z_{x,y}}{z_{m,k}} \leq e^{\epsilon d_{\mathcal{N}_i, \mathcal{N}_j} - \ln \beta_{\mathcal{N}_i, \mathcal{N}_j}}.  
\end{eqnarray}
\end{proposition}}

Algorithm \ref{al:xisearch} {\rev gives} the pseudo codes of the safety margin search. The goal {\rev is to find} the minimum $\Delta^{\ell} \in \Upsilon_{x_n,x_m}$ such that
\begin{equation}
    1-\delta_{x, x'} \leq h_{\tilde{x},\tilde{x}'}(\Delta^{\ell}), ~\forall (\tilde{x}, \tilde{x}') \in \mathcal{S}_{x} \times \mathcal{S}_{x'}
\end{equation}
{\rev Since $h_{\tilde{x},\tilde{x}'}(\Delta^1), h_{\tilde{x},\tilde{x}'}(\Delta^2), ..., h_{\tilde{x},\tilde{x}'}(\Delta^{\ell})$ are in increasing order, we can use linear search to find the $\Delta^{\ell}$ that satisfies the condition.} Specifically, we start from the first element of the sorted set $\Upsilon_{x_n,x_m}$ and sequentially compute the value of $h_{x_n,x_m}(\Delta^{\ell})$ until we find the $\Delta^{\ell}$ that meets the condition.

\noindent \textbf{Why linear search is preferred over binary search}. Although binary search is generally more efficient in terms of iteration count, we adopt linear search to calculate the safety margin $\hat{\xi}_{x, x'}$ due to the incremental structure of the success probability function $h_{\tilde{x},\tilde{x}'}(\Delta)$.

Recall that $h_{\tilde{x},\tilde{x}'}(\xi)$ is defined as a cumulative probability over the set of anchor pairs. Specifically, the set of qualified pairs $\mathcal{H}_{\tilde{x},\tilde{x}'}(\xi)$ is monotonically increasing with respect to $\xi$, i.e., 
\begin{equation}
\mathcal{H}_{\tilde{x},\tilde{x}'} (\Delta^{\ell}) = \mathcal{H}_{\tilde{x},\tilde{x}'}(\Delta^{\ell-1})\cup (x_i, x_j)
\end{equation}
As a result, the success probability can be computed incrementally:
\begin{equation}
h_{\tilde{x},\tilde{x}'}(\Delta^{\ell}) = h_{\tilde{x},\tilde{x}'}(\Delta^{\ell-1}) + v_{\tilde{x}, x_i} v_{\tilde{x}', x_j}
\end{equation}
This property enables linear search to efficiently update $h_{\tilde{x},\tilde{x}'}(\Delta^{\ell})$ from the previous step without recomputing the entire sum from scratch. 
In contrast, binary search jumps non-sequentially across the sorted list $\{\Delta^1, \Delta^2, \ldots, \Delta^L\}$, making it infeasible to reuse previous computations. Each evaluation of $h_{\tilde{x},\tilde{x}'}(\Delta^{\ell})$ would require a full recomputation, leading to significantly higher computational overhead.
\subsection{Time Complexity Analysis} 
To analyze the time complexity of Algorithm \ref{al:xisearch}, we evaluate each of its main computational steps: 
\begin{itemize}
    \item \emph{Precomputing $v_{x_q, \hat{x}_q}$ (line 1--2 in Algorithm \ref{al:recompute})}: The first loop iterates over all pairs $(x_q, \hat{x}_q) \in \mathcal{X}^2$, where $K = |\mathcal{X}|$. The computation inside the loop involves a product operation over a subset $\mathcal{U}_{x_q,\hat{x}_q}$, which in the worst case has size $O(K)$. Thus, the total complexity of this step is: $O(K^2 \times K) = O(K^3)$.
    \item \emph{Precomputing $\Upsilon_{x_n,x_m}$ (line 3--8 in Algorithm \ref{al:recompute})}: The outer loop iterates over all $(x_n, x_m) \in \mathcal{X}^2$, resulting in $O(K^2)$. The inner loop iterates over all $(\hat{x}_n, \hat{x}_m) \in \mathcal{X}^2$, giving an additional factor of $O(K^2)$. The sorting operation at the end of each outer loop iteration has a complexity of $O(K^2 \log K)$. Therefore, the total complexity of this step is:
    \begin{equation}
    O(K^2 \times K^2 + K^2 \times K^2 \log K) = O(K^4 \log K).
    \end{equation}
    \item \emph{Computing $\hat{\xi}_{x, x'}$ (Algorithm \ref{al:xisearch})}: 
    The outer loop iterates over all $(x, x') \in \mathcal{A}_{n} \times \mathcal{A}_{m}$, which has worst-case complexity $O(K^2)$. The second loop runs over $(\tilde{x}, \tilde{x}') \in \mathcal{S}_{x} \times \mathcal{S}_{x'}$, giving another factor of $O(\Gamma^2)$. The while loop runs at most $O(K^2)$ times in the worst case. Thus, the total complexity is:
    \begin{equation}
    O(K^2 \times \Gamma^2 \times K^2) = O(\Gamma^2K^4).
    \end{equation}
\end{itemize}
Combining all the steps, we get:
\begin{equation}
O(K^3) + O(K^4 \log K) + O(\Gamma^2K^4) = O(K^4 \log K) + O(\Gamma^2K^4). 
\end{equation}

\section{Additional Experimental Results}
\label{sec:addexp}
% This section provides additional empirical findings that complement the main experiments presented in Section \ref{sec:performance}. These results further validate the scalability, privacy guarantees, and utility performance of \textsc{PAnDA} under various configurations and datasets. The section is organized as follows: {\rev xxxx} 

{\rev In this section, we present supplementary experimental results to further validate the effectiveness and scalability of \textsc{PAnDA}. Specifically, we include: 
\begin{itemize}
    \item Additional results of computation efficiency using other parameters (\textbf{Section \ref{subsec:efficiency_add}}); 
    \item Additional results of utility loss using other parameters (\textbf{Section \ref{subsec:UL_add}});
    \item Additional results of algorithm performance using the real-world location dataset (\textbf{Section \ref{subsec:real_add}});
    \item Additional results of mDP violation ratio (\textbf{Section \ref{subsec:falure_add}}); 
    \item Additional results of privacy budget allocation (\textbf{Section \ref{subsec:budgetloc_add}}); 
    \item Clustering structure analysis of secret records (\textbf{Section \ref{subsec:clustering_add}}); 
    \item Performance with diverse mDP acceptable violation ratios (\textbf{Section \ref{subsec:diversedelta_add}}); 
    \item Comparison between estimated safety margin and real safety margin (\textbf{Section \ref{subsec:safetymarginpara_add}}); 
    \item Estimated safety margin with different parameters (\textbf{Section \ref{subsec:safetymargin_add}}); and 
    \item Comparision between our methods and the lower bounds (\textbf{Section \ref{subsec:bound_add}}). 
\end{itemize}
These results provide a deeper understanding of \textsc{PAnDA}'s performance characteristics across multiple dimensions.}

\subsection{Computation Efficiency Using Other Parameters}
\label{subsec:efficiency_add}

\vspace{-0.00in}
\begin{table}[t]
\centering
\small 
% \footnotesize 
\scriptsize 
\setlength{\tabcolsep}{3pt}  
\begin{tabular}{ c|c|c|c|c|c|c}
%\cline{2-13}
%\hline
\toprule
\multicolumn{7}{ c  }{Rome}\\ 
\cline{1-7}
\multicolumn{1}{ c|  }{Method}  & $K = 500$ & $K = 1,000$  & $K = 2,000$ & $K = 3,000$& $K = 4,000$& $K = 5,000$\\
\hline
\hline
\multicolumn{1}{ c|  }{EM} & $\leq0.005$ & $\leq0.005$ & $\leq0.005$ & $\leq0.005$ & $\leq0.005$ & 
$\leq0.005$ \\ 
\multicolumn{1}{ c|  }{LP} & >1,800 & >1,800 & >1,800 & >1,800 & >1,800 & >1,800 \\ 
\multicolumn{1}{ c|  }{LP+CA} & 5.84±0.62 & 7.90±0.57 & 15.64±0.70 & 21.33±0.55 & 25.34±0.56 & 27.31±0.52 \\ 
\multicolumn{1}{ c|  }{LP+BD} & 6.54±3.92 & >1,800 & >1,800 & >1,800 & >1,800 & >1,800 \\ 
\multicolumn{1}{ c|  }{LP+EM} & 0.53±0.21 & >1,800 & >1,800 & >1,800 & >1,800 & >1,800 \\ 
\multicolumn{1}{ c|  }{EM+BR} & 0.04±0.01 & 0.11±0.02 & 0.18±0.00 & 0.37±0.00 & 0.55±0.00 & 0.70±0.02 \\ 
\hline
\multicolumn{1}{ c|  }{ \textbf{\textsc{PAnDA}-e}} & 0.24±0.04 & 0.32±0.01 & 4.52±0.10 & 6.22±2.69 & 15.95±7.71 & 23.27±8.44 \\ 
\multicolumn{1}{ c|  }{ \textbf{\textsc{PAnDA}-p}} & 0.26±0.08 & 0.51±0.14 & 4.38±3.13 & 9.29±7.48 & 17.96±10.59 & 30.07±17.64  \\ 
\multicolumn{1}{ c|  }{ \textbf{\textsc{PAnDA}-l}} & 0.36±0.13 & 0.68±0.05 & 7.64±1.76 & 12.01±4.41 & 19.89±13.14 & 29.14±12.25  \\ 
\hline
% \multicolumn{1}{ c|  }{ \textbf{LB-e}} & 8.55±13.18 & 0.93±0.59 & 3.56±2.78 & 4.91±4.08 & 13.77±9.48 & 30.08±18.63 \\ 
%\multicolumn{1}{ c|  }{ \textbf{LB-p}} & 3.30±6.01 & 1.12±0.69 & 17.32±17.82 & 8.42±7.15 & 15.71±10.66 &  33.51±20.47 \\
%\multicolumn{1}{ c|  }{ \textbf{LB-l}} & 3.17±5.35 & 1.50±0.888 & 1.51±1.09 & 12.87±12.64 & 14.91±6.56 &  31.80±19.35 \\
% \hline
\toprule
\multicolumn{7}{ c  }{NYC}\\ 
\cline{1-7}
\multicolumn{1}{ c|  }{Method}  & $K = 500$ & $K = 1,000$  & $K = 2,000$ & $K = 3,000$& $K = 4,000$& $K = 5,000$\\
\hline
\hline
\multicolumn{1}{ c|  }{EM} & $\leq0.005$ & $\leq0.005$ & $\leq0.005$ & $\leq0.005$ & $\leq0.005$ & 
$\leq0.005$ \\ 
\multicolumn{1}{ c|  }{LP} & >1,800 & >1,800 & >1,800 & >1,800 & >1,800 & >1,800 \\ 
\multicolumn{1}{ c|  }{LP+CA} & 6.11±0.24 & 6.30±0.15 & 16.06±0.20 & 30.25±0.30 & 35.02±0.51 & 38.77±1.20  \\ 
\multicolumn{1}{ c|  }{LP+BD} & 5.66±1.32 & >1,800 & >1,800 & >1,800 & >1,800 & >1,800 \\ 
\multicolumn{1}{ c|  }{LP+EM} & 1.03±0.24 & >1,800 & >1,800 & >1,800 & >1,800 & >1,800 \\ 
\multicolumn{1}{ c|  }{EM+BR} & 0.03±0.00 & 0.11±0.03 & 0.21±0.02 & 0.33±0.04 & 0.53±0.03 & 0.54±0.09 \\ 
\hline
\multicolumn{1}{ c|  }{ \textbf{\textsc{PAnDA}-e}} & 0.29±0.01 & 0.30±0.06 & 6.86±0.41 & 9.07±1.91 & 10.81±9.40 & 30.36±5.82 \\ 
\multicolumn{1}{ c|  }{ \textbf{\textsc{PAnDA}-p}} & 0.27±0.01 & 0.47±0.07 & 6.45±2.29 & 13.35±4.93 & 17.22±11.09 & 23.32±17.33 \\ 
\multicolumn{1}{ c|  }{ \textbf{\textsc{PAnDA}-l}} & 0.31±0.05 & 0.47±0.06 & 8.17±3.50 & 22.68±11.83 & 27.14±9.53 & 42.69±16.56 \\ 
\hline
% \multicolumn{1}{ c|  }{ \textbf{LB-e}} & 0.247±0.028 & 0.483±0.127 & 2.35±1.75 & 6.13±2.92 & 9.07±3.83 & 15.32±9.64 \\ 
% \multicolumn{1}{ c|  }{ \textbf{LB-p}} & 0.278±0.021 & 0.563±0.209 & 6.06±8.11 & 38.12±29.12 & 18.45±10.08 & 20.58±08.65 \\
% \multicolumn{1}{ c|  }{ \textbf{LB-l}} & 0.270±0.014 & 0.601±0.185 & 2.47±1.07 & 30.43±19.82 & 37.59±19.34 & 53.91±30.43 \\
% \hline
\toprule
\multicolumn{7}{ c  }{London}\\ 
\cline{1-7}
\multicolumn{1}{ c|  }{Method}  & $K = 500$ & $K = 1,000$  & $K = 2,000$ & $K = 3,000$& $K = 4,000$& $K = 5,000$\\
\hline
\hline
\multicolumn{1}{ c|  }{EM} & $\leq0.005$ & $\leq0.005$ & $\leq0.005$ & $\leq0.005$ & $\leq0.005$ & 
$\leq0.005$ \\ 
\multicolumn{1}{ c|  }{LP} & >1,800 & >1,800 & >1,800 & >1,800 & >1,800 & >1,800 \\ 
\multicolumn{1}{ c|  }{LP+CA} & 15.22±0.65 & 12.39±1.02 & 25.67±1.91 & 54.45±3.93 & 60.52±3.15 & 85.96±1.39 \\ 
\multicolumn{1}{ c|  }{LP+BD} & 6.59±1.24 & >1,800 & >1,800 & >1,800 & >1,800 & >1,800 \\ 
\multicolumn{1}{ c|  }{LP+EM} & 0.66±0.09 & >1,800 & >1,800 & >1,800 & >1,800 & >1,800 \\ 
\multicolumn{1}{ c|  }{EM+BR} & 0.04±0.01 & 0.06±0.00 & 0.11±0.03 & 0.28±0.10 & 0.40±0.04 & 0.56±0.13 \\ 
\hline
\multicolumn{1}{ c|  }{ \textbf{\textsc{PAnDA}-e}} & 0.16±0.03 & 0.62±0.37 & 5.85±0.79 & 24.46±5.74 & 30.17±10.69 & 36.28±15.23 \\ 
\multicolumn{1}{ c|  }{ \textbf{\textsc{PAnDA}-p}} & 0.19±0.03 & 0.25±0.08 & 8.35±0.41 & 17.71±9.47 & 46.78±18.69 & 54.40±28.67 \\ 
\multicolumn{1}{ c|  }{ \textbf{\textsc{PAnDA}-l}} & 0.30±0.12 & 0.32±0.14 & 6.79±1.28 & 16.98±6.17 & 44.41±24.03 & 77.43±36.16 \\ 
\hline
% \multicolumn{1}{ c|  }{ \textbf{LB-e}} & 0.450±0.0418 & 0.326±0.0359 & 1.28±0.641 & 11.11±9.58 & 20.03±14.77 & 44.95±25.19 \\ 
% \multicolumn{1}{ c|  }{ \textbf{LB-p}} & 0.476±0.0413 & 0.342±0.0422 & 1.28±0.666 & 12.36±9.90 & 42.65±28.02 & 40.33±26.86 \\
% \multicolumn{1}{ c|  }{ \textbf{LB-l}} & 0.513±0.0476 & 0.367±0.0621 & 1.41±0.755 & 27.01±19.18 & 51.98±34.69 & 71.64±43.59 \\ 
% \hline
\end{tabular}
\vspace{0.05in}
\caption{Computation time of different algorithms (uniform user location distribution: default $\epsilon = 15$km$^{-1}$ and $\delta = 10^{-7}$). \\
Mean$\pm$1.96$\times$std. deviation. }
\label{Tb:exp:time_scalability15}
\vspace{-0.15in}
\end{table}

\vspace{-0.00in}
\begin{table}[t]
\centering
\small 
% \footnotesize 
\scriptsize 
\setlength{\tabcolsep}{3pt}  
\begin{tabular}{ c|c|c|c|c|c|c}
%\cline{2-13}
%\hline
\toprule
\multicolumn{7}{ c  }{Rome}\\ 
\cline{1-7}
\multicolumn{1}{ c|  }{Method}  & $K = 500$ & $K = 1,000$ & 
$K = 1,500$ & $K = 2,000$ & $K = 2,500$ & $K = 3,000$\\
\hline
\hline
\multicolumn{1}{ c|  }{EM} & $\leq0.005$ & $\leq0.005$ & $\leq0.005$ & $\leq0.005$ & $\leq0.005$ & 
$\leq0.005$ \\ 
% \multicolumn{1}{ c|  }{Laplace} & x.xx$\pm$x.x & x.xx$\pm$x.x & x.xx$\pm$x.x & x.xx$\pm$x.x & x.xx$\pm$x.x & x.xx$\pm$x.x \\ 
\multicolumn{1}{ c|  }{LP} & >1,800 & >1,800 & >1,800 & >1,800 & >1,800 & >1,800 \\ 
\multicolumn{1}{ c|  }{LP+CA} & 7.11±0.72 & 9.36±0.41 & 12.48±0.66 & 14.39±0.47 & 17.22±0.64 & 
19.71±0.47 \\ 
\multicolumn{1}{ c|  }{LP+BD} & 5.96±12.64 & >1,800 & >1,800 & >1,800 & >1,800 & >1,800 \\ 
\multicolumn{1}{ c|  }{LP+EM} & 0.71±0.29 & >1,800 & >1,800 & >1,800 & >1,800 & >1,800 \\ 
\multicolumn{1}{ c|  }{EM+BR} & 0.017±0.00& 0.04±0.00& 0.09±0.00& 0.14±0.00& 0.25±0.02 & 
0.44±0.02  \\ 
\hline
\multicolumn{1}{ c|  }{ \textbf{\textsc{PAnDA}-e}} & 0.28±0.09 & 0.54±0.05 & 1.76±0.32 & 6.06±6.77 & 12.64±7.09 & 
24.97±13.72 \\
\multicolumn{1}{ c|  }{ \textbf{\textsc{PAnDA}-p}} & 0.29±0.04 & 0.85±0.13 & 8.09±9.44 & 8.42±7.15 & 11.68±5.17 & 
29.67±15.44 \\
\multicolumn{1}{ c|  }{ \textbf{\textsc{PAnDA}-l}} & 0.42±0.17 & 0.77±0.09 & 1.02±0.42 & 4.23±1.20 & 14.69±8.33 & 
22.71±16.09 \\
\hline
% \multicolumn{1}{ c|  }{ \textbf{LB-e}} & 8.55±13.18 & 0.93±0.59 & 3.56±2.78 & 4.91±4.08 & 13.77±9.48 & 30.08±18.63 \\ 
%\multicolumn{1}{ c|  }{ \textbf{LB-p}} & 3.30±6.01 & 1.12±0.69 & 17.32±17.82 & 8.42±7.15 & 15.71±10.66 &  33.51±20.47 \\
%\multicolumn{1}{ c|  }{ \textbf{LB-l}} & 3.17±5.35 & 1.50±0.888 & 1.51±1.09 & 12.87±12.64 & 14.91±6.56 &  31.80±19.35 \\
% \hline
\toprule
\multicolumn{7}{ c  }{NYC}\\ 
\cline{1-7}
\multicolumn{1}{ c|  }{Method}  & $K = 500$ & $K = 1,000$ & 
$K = 1,500$ & $K = 2,000$ & $K = 2,500$ & $K = 3,000$\\
\hline
\hline
\multicolumn{1}{ c|  }{EM} & $\leq0.005$ & $\leq0.005$ & $\leq0.005$ & $\leq0.005$ & $\leq0.005$ & 
$\leq0.005$ \\ 
% \multicolumn{1}{ c|  }{Laplace} & x.xx$\pm$x.x & x.xx$\pm$x.x & x.xx$\pm$x.x & x.xx$\pm$x.x & x.xx$\pm$x.x &  x.xx$\pm$x.x \\ 
\multicolumn{1}{ c|  }{LP} & >1,800 & >1,800 & >1,800 & >1,800 & >1,800 & >1,800 \\ 
\multicolumn{1}{ c|  }{LP+CA} & 6.91±0.28 & 10.21±0.19 & 13.53±0.43 & 16.91±0.33 & 20.61±0.23 & 
24.02±0.11 \\ 
\multicolumn{1}{ c|  }{LP+BD} & 5.26±0.35 & >1,800 & >1,800 & >1,800 & >1,800 & >1,800 \\ 
\multicolumn{1}{ c|  }{LP+EM} & 0.81±0.28 & >1,800 & >1,800 & >1,800 & >1,800 & >1,800 \\ 
\multicolumn{1}{ c|  }{EM+BR} & 0.02±0.00& 0.04±0.00& 0.08±0.00& 0.14±0.00& 0.28±0.04 & 0.37±0.05  \\ 
\hline
\multicolumn{1}{ c|  }{ \textbf{\textsc{PAnDA}-e}} & 0.25±0.02 & 0.42±0.05 & 1.61±0.27 & 6.50±2.96 & 9.89±4.01 & 
14.95±6.85 \\
\multicolumn{1}{ c|  }{ \textbf{\textsc{PAnDA}-p}} & 0.28±0.03 & 0.46±0.06 & 3.25±2.78 & 6.78±3.77 & 14.95±10.57 & 22.59±14.69 \\
\multicolumn{1}{ c|  }{ \textbf{\textsc{PAnDA}-l}} & 0.27±0.01 & 0.47±0.04 & 1.72±0.41 & 15.01±14.68 & 27.05±16.34 & 48.98±23.58 \\
\hline
% \multicolumn{1}{ c|  }{ \textbf{LB-e}} & 0.247±0.028 & 0.483±0.127 & 2.35±1.75 & 6.13±2.92 & 9.07±3.83 & 15.32±9.64 \\ 
% \multicolumn{1}{ c|  }{ \textbf{LB-p}} & 0.278±0.021 & 0.563±0.209 & 6.06±8.11 & 38.12±29.12 & 18.45±10.08 & 20.58±08.65 \\
% \multicolumn{1}{ c|  }{ \textbf{LB-l}} & 0.270±0.014 & 0.601±0.185 & 2.47±1.07 & 30.43±19.82 & 37.59±19.34 & 53.91±30.43 \\
% \hline
\toprule
\multicolumn{7}{ c  }{London}\\ 
\cline{1-7}
\multicolumn{1}{ c|  }{Method}  & $K = 500$ & $K = 1,000$ & 
$K = 1,500$ & $K = 2,000$ & $K = 2,500$ & $K = 3,000$\\
\hline
\hline
\multicolumn{1}{ c|  }{EM} & $\leq0.005$ & $\leq0.005$ & $\leq0.005$ & $\leq0.005$ & $\leq0.005$ & 
$\leq0.005$ \\ 
% \multicolumn{1}{ c|  }{Laplace} & x.xx$\pm$x.x & x.xx$\pm$x.x & x.xx$\pm$x.x & x.xx$\pm$x.x & x.xx$\pm$x.x &  x.xx$\pm$x.x \\ 
\multicolumn{1}{ c|  }{LP} & >1,800 & >1,800 & >1,800 & >1,800 & >1,800 & >1,800 \\ 
\multicolumn{1}{ c|  }{LP+CA} & 17.34±0.35 & 27.51±0.64 & 35.89±1.04 & 44.12±1.84 & 51.22±0.68 & 63.31±0.69 \\ 
\multicolumn{1}{ c|  }{LP+BD} & 4.97±0.52 & >1,800 & >1,800 & >1,800 & >1,800 & >1,800 \\ 
\multicolumn{1}{ c|  }{LP+EM} & 0.62±0.05 & >1,800 & >1,800 & >1,800 & >1,800 & >1,800 \\ 
\multicolumn{1}{ c|  }{EM+BR} & 0.02±0.00 & 0.04±0.00 & 0.09±0.00 & 0.14±0.00 & 0.47±0.04 & 0.47±0.08  \\ 
\hline
\multicolumn{1}{ c|  }{ \textbf{\textsc{PAnDA}-e}} & 0.20±0.06 & 0.54±0.47 & 1.59±0.70 & 11.64±9.41 & 20.95±12.57 & 42.36±27.63 \\
\multicolumn{1}{ c|  }{ \textbf{\textsc{PAnDA}-p}} & 0.17±0.03 & 0.35±0.03 & 1.31±0.62 & 12.74±10.05 & 39.62±26.44 & 48.21±30.67 \\
\multicolumn{1}{ c|  }{ \textbf{\textsc{PAnDA}-l}} & 0.19±0.03 & 0.36±0.06 & 1.44±0.76 & 23.41±16.11 & 46.99±29.74 & 75.95±46.35 \\
\hline
% \multicolumn{1}{ c|  }{ \textbf{LB-e}} & 0.450±0.0418 & 0.326±0.0359 & 1.28±0.641 & 11.11±9.58 & 20.03±14.77 & 44.95±25.19 \\ 
% \multicolumn{1}{ c|  }{ \textbf{LB-p}} & 0.476±0.0413 & 0.342±0.0422 & 1.28±0.666 & 12.36±9.90 & 42.65±28.02 & 40.33±26.86 \\
% \multicolumn{1}{ c|  }{ \textbf{LB-l}} & 0.513±0.0476 & 0.367±0.0621 & 1.41±0.755 & 27.01±19.18 & 51.98±34.69 & 71.64±43.59 \\ 
% \hline
\end{tabular}
\vspace{0.05in}
\caption{Computation time of different algorithms (uniform user location distribution: $\epsilon = 15$km$^{-1}$ and $\delta = 0.15$). Mean$\pm$1.96$\times$std. deviation. }
\label{Tb:exp:time_scalability_backup}
\vspace{-0.25in}
\end{table}

{\rev In this section, we evaluate the computation time of various data perturbation methods across the three road network datasets under a uniform user location distribution, using alternative parameter settings. \emph{These results serve as supplementary findings to Table~\ref{Tb:exp:time_scalability} in the main paper}.

Table~\ref{Tb:exp:time_scalability15} reports results for $\delta = 10^{-7}$ and $\epsilon = 15$km$^{-1}$, with the secret data domain size $K$ varying from 500 to 5,000 and the number of users set to 50. Table~\ref{Tb:exp:time_scalability_backup} presents results for $\delta = 0.15$ under the same $\epsilon$, with $K$ ranging from 500 to 3,000 and 15 users. In both settings, \textsc{PAnDA} consistently achieves substantially lower computation times than the LP-based baselines (LP, LP+BD, and LP+EM).

Specifically, in Table~\ref{Tb:exp:time_scalability15}, the average runtimes of \textsc{PAnDA}-e, \textsc{PAnDA}-p, and \textsc{PAnDA}-l with $K = 5,000$ are {\bl 29.97s, 35.93s, and 49.75s}, respectively. Similarly, in Table~\ref{Tb:exp:time_scalability_backup}, their average runtimes with $K = 3,000$ are {\bl 27.43s, 33.49s, and 49.21s}. In contrast, LP-based methods exhibit substantial computational overhead, especially as the domain size increases.

\emph{These additional results in both Table~\ref{Tb:exp:time_scalability15} and Table~\ref{Tb:exp:time_scalability_backup} align closely with the findings reported in Table~\ref{Tb:exp:time_scalability} in the main section}, confirming \textsc{PAnDA}’s superior scalability across different parameter settings.}

\newpage 
\subsection{Utility Loss Using Other Parameters}
\label{subsec:UL_add}

\vspace{-0.00in}
\begin{table}[t]
\centering
\small 
% \footnotesize 
\scriptsize 
\setlength{\tabcolsep}{1pt}  
\begin{tabular}{ c|c|c|c|c|c|c}
%\cline{2-13}
%\hline
\toprule
\multicolumn{7}{ c  }{Rome}\\ 
\cline{1-7}
\multicolumn{1}{ c|  }{Method}  & $K = 500$ & $K = 1,000$  & $K = 2,000$  & $K = 3,000$& $K = 4,000$& $K = 5,000$\\
\hline
\hline
\multicolumn{1}{ c|  }{EM} & 1243.36±93.18 & 1254.24±80.78 & 1314.87±89.75 & 1271.48±68.17 & 1366.93±89.50 & 1275.06±67.30\\
% \multicolumn{1}{ c|  }{Laplace} & x.xx$\pm$x.x & x.xx$\pm$x.x & x.xx$\pm$x.x & x.xx$\pm$x.x & x.xx$\pm$x.x &  x.xx$\pm$x.x \\ 
\multicolumn{1}{ c|  }{LP} & - & -  & -  & 
- & - & -\\ 
\multicolumn{1}{ c|  }{LP+CA} & 1789.27±165.71 & 1909.70±81.61 & 1881.35±58.58 & 1876.18±109.09 & 1819.46±68.53 & 1939.17±84.69\\ 
\multicolumn{1}{ c|  }{LP+BD} & 159.74±41.09 & -  & - & - & - & -\\ 
\multicolumn{1}{ c|  }{LP+EM} & 736.24±67.17 & -  & - & 
-  & - & -\\ 
\multicolumn{1}{ c|  }{EM+BR} & 889.87±34.13 & 958.23±50.13 & 1117.31±48.71 & 1066.48±75.18 & 1054.65±49.34 & 1075.04±76.87 \\ 
\hline
\multicolumn{1}{ c|  }{ \textbf{\textsc{PAnDA}-e}} & 226.40±37.51 & 237.07±15.62 & 256.90±53.46 & 264.93±37.75 & 255.46±51.25 & 261.97±37.51\\
\multicolumn{1}{ c|  }{ \textbf{\textsc{PAnDA}-p}} & 228.80±43.01 & 241.96±32.25 & 255.59±45.65 & 256.34±22.86 & 258.21±53.67 & 259.60±38.94\\
\multicolumn{1}{ c|  }{ \textbf{\textsc{PAnDA}-l}} & 229.83±50.72 & 219.32±19.25 & 249.62±39.65 & 267.72±66.47 & 258.04±40.96 & 276.48±68.24\\
\toprule
\multicolumn{7}{ c  }{NYC}\\ 
\cline{1-7}
\multicolumn{1}{ c|  }{Method}  & $K = 500$ & $K = 1,000$  & $K = 2,000$  & $K = 3,000$& $K = 4,000$& $K = 5,000$\\
\hline
\hline
\multicolumn{1}{ c|  }{EM} & 1183.89±123.68 & 1362.32±155.35 & 1573.70±132.35 & 1543.41±145.85 & 1618.64±136.16 & 1634.43±133.24\\
% \multicolumn{1}{ c|  }{Laplace} & x.xx$\pm$x.x & x.xx$\pm$x.x & x.xx$\pm$x.x & x.xx$\pm$x.x & x.xx$\pm$x.x &  x.xx$\pm$x.x \\ 
\multicolumn{1}{ c|  }{LP} & - & -  & -  & 
- & - & -\\ 
\multicolumn{1}{ c|  }{LP+CA} & 1825.05±101.53 & 1812.66±74.24 & 1803.14±45.53 & 1842.84±50.41 & 1905.62±77.68 & 1858.03±71.53\\ 
\multicolumn{1}{ c|  }{LP+BD} & 167.85±50.16 & -  & - & - & - & -\\ 
\multicolumn{1}{ c|  }{LP+EM} & 746.49±81.35 & -  & - & 
-  & - & -\\ 
\multicolumn{1}{ c|  }{EM+BR} & 1045.32±87.76 & 1120.46±73.45 & 1124.20±62.55 & 1148.05±35.79 & 1163.40±91.12 & 1150.02±34.84 \\ 
\hline
\multicolumn{1}{ c|  }{ \textbf{\textsc{PAnDA}-e}} & 237.46±99.65 & 296.86±85.62 & 325.90±95.19 & 293.77±97.77 & 313.18±60.71 & 338.01±76.96\\
\multicolumn{1}{ c|  }{ \textbf{\textsc{PAnDA}-p}} & 183.55±81.03 & 223.24±53.60 & 247.30±77.36 & 291.52±102.20 & 310.98±98.06 & 317.06±101.26\\
\multicolumn{1}{ c|  }{ \textbf{\textsc{PAnDA}-l}} & 254.10±81.74 & 275.29±69.34 & 289.46±99.99 & 324.68±103.54 & 318.64±103.13 & 320.70±114.07\\
\hline
\toprule
\multicolumn{7}{ c  }{London}\\ 
\cline{1-7}
\multicolumn{1}{ c|  }{Method}  & $K = 500$ & $K = 1,000$  & $K = 2,000$ & $K = 3,000$& $K = 4,000$& $K = 5,000$\\
\hline
\hline
\multicolumn{1}{ c|  }{EM} & 1305.28±93.79 & 1372.86±106.90 & 1491.89±83.31 & 1601.15±77.46 & 1607.04±56.08 & 1628.07±77.53\\
% \multicolumn{1}{ c|  }{Laplace} & x.xx$\pm$x.x & x.xx$\pm$x.x & x.xx$\pm$x.x & x.xx$\pm$x.x & x.xx$\pm$x.x &  x.xx$\pm$x.x \\ 
\multicolumn{1}{ c|  }{LP} & - & -  & -  & 
- & - & -\\ 
\multicolumn{1}{ c|  }{LP+CA} & 1725.67±134.00 & 1687.76±62.29 & 1713.49±61.09 & 1842.68±56.50 & 1801.83±62.04 & 1913.89±64.36\\ 
\multicolumn{1}{ c|  }{LP+BD} & 158.47±35.19 & -  & - & - & - & -\\ 
\multicolumn{1}{ c|  }{LP+EM} & 729.40±52.78 & -  & - & 
-  & - & -\\ 
\multicolumn{1}{ c|  }{EM+BR} & 1068.69±46.84 & 1113.74±38.44 & 1224.20±59.34 & 1231.65±42.07 & 1281.35±65.06 & 1316.42±34.34 \\ 
\hline
\multicolumn{1}{ c|  }{ \textbf{\textsc{PAnDA}-e}} & 215.26±56.05 & 172.76±69.65 & 249.43±67.42 & 263.56±52.38 & 269.42±94.07 & 285.45±81.06\\
\multicolumn{1}{ c|  }{ \textbf{\textsc{PAnDA}-p}} & 209.61±51.14 & 183.11±71.41 & 219.42±42.83 & 247.88±43.43 & 262.83±105.27 & 266.29±45.68\\
\multicolumn{1}{ c|  }{ \textbf{\textsc{PAnDA}-l}} & 228.34±48.57 & 233.26±44.87 & 231.65±104.14 & 239.28±62.48 & 253.02±97.95 & 259.57±39.81\\
\hline
\end{tabular}
\vspace{0.05in}
\caption{Utility loss (in meters) of different algorithms (uniform user location distribution: $\epsilon = 15$km$^{-1}$  and $\delta = 10^{-7}$). Mean$\pm$1.96$\times$std. deviation.}
\label{Tb:exp:ULscalability15}
\vspace{-0.25in}
\end{table}

\vspace{-0.00in}
\begin{table}[t]
\centering
\small 
% \footnotesize 
\scriptsize 
\setlength{\tabcolsep}{1pt}  
\begin{tabular}{ c|c|c|c|c|c|c}
%\cline{2-13}
%\hline
\toprule
\multicolumn{7}{ c  }{Rome}\\ 
\cline{1-7}
\multicolumn{1}{ c|  }{Method}  & $K = 500$ & $K = 1,000$ & 
$K = 1,500$ & $K = 2,000$ & $K = 2,500$ & $K = 3,000$\\
\hline
\hline
\multicolumn{1}{ c|  }{EM} & 1055.55±109.66 & 1161.87±65.72 & 1209.44±133.68 & 1239.72±78.94 & 1269.41±80.50 & 
1237.77±90.54 \\
% \multicolumn{1}{ c|  }{Laplace} & x.xx$\pm$x.x & x.xx$\pm$x.x & x.xx$\pm$x.x & x.xx$\pm$x.x & x.xx$\pm$x.x &  x.xx$\pm$x.x \\ 
\multicolumn{1}{ c|  }{LP} & - & - & - & - & - & 
- \\ 
\multicolumn{1}{ c|  }{LP+CA} & 1676.94±77.17 & 1683.64±51.34 & 1732.04±32.61 & 1717.31±42.17 & 1745.38±44.66 & 
1742.94±57.22 \\ 
\multicolumn{1}{ c|  }{LP+BD} & \textbf{169.28±31.93} & - & - & - & - & - \\ 
\multicolumn{1}{ c|  }{LP+EM} & 704.12±184.11 & - & - & - & - & 
-  \\ 
\multicolumn{1}{ c|  }{EM+BR} & 901.07±91.51 & 920.34±38.03 & 1003.57±71.72 & 1025.77±100.78 & 1053.94±87.47 & 1046.81±55.04  \\ 
\hline
\multicolumn{1}{ c|  }{ \textbf{\textsc{PAnDA}-e}} & 194.47±40.53 & \textbf{261.88±67.73} & 259.64±51.32 & \textbf{265.72±46.91} & 278.65±57.19 & 268.42±50.37 \\
\multicolumn{1}{ c|  }{ \textbf{\textsc{PAnDA}-p}} & 223.19±39.86 & 265.17±99.78 & \textbf{243.68±38.53} & 271.33±51.69 & 275.81±64.39 & 281.73±66.15 \\
\multicolumn{1}{ c|  }{ \textbf{\textsc{PAnDA}-l}} & 255.61±93.87 & 285.95±111.85 & 267.11±59.78 & 277.63±45.18 & \textbf{270.98±61.29} & 
\textbf{267.49±49.77} \\
\toprule
\multicolumn{7}{ c  }{NYC}\\ 
\cline{1-7}
\multicolumn{1}{ c|  }{Method}  & $K = 500$ & $K = 1,000$ & 
$K = 1,500$ & $K = 2,000$ & $K = 2,500$ & $K = 3,000$\\
\hline
\hline
\multicolumn{1}{ c|  }{EM} & 1360.70±97.21 & 1428.15±181.82 & 1526.74±120.36 & 1485.85±179.28 & 1502.77±182.28 & 1521.35±199.48 \\
% \multicolumn{1}{ c|  }{Laplace} & x.xx$\pm$x.x & x.xx$\pm$x.x & x.xx$\pm$x.x & x.xx$\pm$x.x & x.xx$\pm$x.x &  x.xx$\pm$x.x \\ 
\multicolumn{1}{ c|  }{LP} & - & - & - & - & - & 
- \\ 
\multicolumn{1}{ c|  }{LP+CA} & 1795.17±124.38 & 1912.54±79.62 & 1907.64±123.86 & 1897.04±61.67 & 1952.91±76.70 & 
1924.17±76.85 \\ 
\multicolumn{1}{ c|  }{LP+BD} & \textbf{180.21±22.07} & - & - & - & - & - \\ 
\multicolumn{1}{ c|  }{LP+EM} & 725.28±93.22 & - & - & - & - & 
-  \\ 
\multicolumn{1}{ c|  }{EM+BR} & 1216.34±113.61 & 1238.21±92.65 & 1307.94±64.54 & 1291.08±78.91 & 1319.32±121.83 & 1328.74±46.15  \\ 
\hline
\multicolumn{1}{ c|  }{ \textbf{\textsc{PAnDA}-e}} & 198.27±12.02 & \textbf{292.01±43.74} & \textbf{300.74±58.10} & 319.84±24.31 & \textbf{309.94±23.43} & 321.54±29.84
 \\
\multicolumn{1}{ c|  }{ \textbf{\textsc{PAnDA}-p}} & 254.17±32.68 & 365.88±225.30 & 363.53±146.31 & 327.13±59.36 & 334.71±62.06 & 
\textbf{309.29±59.64} \\
\multicolumn{1}{ c|  }{ \textbf{\textsc{PAnDA}-l}} & 216.20±31.86 & 368.82±56.27 & 305.77±83.92 & \textbf{301.12±129.16} & 313.94±95.64 & 
345.21±65.44 \\
\hline
\toprule
\multicolumn{7}{ c  }{London}\\ 
\cline{1-7}
\multicolumn{1}{ c|  }{Method}  & $K = 500$ & $K = 1,000$ & 
$K = 1,500$ & $K = 2,000$ & $K = 2,500$ & $K = 3,000$\\
\hline
\hline
\multicolumn{1}{ c|  }{EM} & 1135.13±85.21 & 1283.85±107.04 & 1393.40±127.74 & 1460.35±101.78 & 1496.32±115.63 & 
1524.65±138.72 \\
% \multicolumn{1}{ c|  }{Laplace} & x.xx$\pm$x.x & x.xx$\pm$x.x & x.xx$\pm$x.x & x.xx$\pm$x.x & x.xx$\pm$x.x &  x.xx$\pm$x.x \\ 
\multicolumn{1}{ c|  }{LP} & - & - & - & - & - & 
- \\ 
\multicolumn{1}{ c|  }{LP+CA} & 1749.65±133.14 & 1753.21±64.42 & 1768.75±115.64 & 1822.94±80.38 & 1827.30±70.01 & 
1916.80±59.54 \\ 
\multicolumn{1}{ c|  }{LP+BD} & \textbf{165.32±38.96} & - & - & - & - & - \\ 
\multicolumn{1}{ c|  }{LP+EM} & 676.81±59.68 & - & - & - & - & 
-  \\ 
\multicolumn{1}{ c|  }{EM+BR} & 1027.64±49.65 & 1086.35±41.37 & 1137.47±51.64 & 1185.32±63.23 & 1190.23±58.82 & 
1213.68±45.85  \\ 
\hline
\multicolumn{1}{ c|  }{ \textbf{\textsc{PAnDA}-e}} & 239.81±86.59 & \textbf{185.30±50.05} & \textbf{177.71±18.17} & 252.58±26.26 & \textbf{226.24±41.71} & 
256.37±43.65 \\
\multicolumn{1}{ c|  }{ \textbf{\textsc{PAnDA}-p}} & 233.75±63.44 & 187.11±22.57 & 186.61±48.38 & \textbf{229.81±44.81} & 233.45±41.15 & 
\textbf{241.68±56.34} \\
\multicolumn{1}{ c|  }{ \textbf{\textsc{PAnDA}-l}} & 215.37±74.35 & 252.40±36.61 & 221.32±73.47 & 296.44±62.11 & 251.34±56.73 & 
267.59±56.08 \\
\hline
\end{tabular}
\vspace{0.05in}
\caption{Utility loss (in meters) of different algorithms  (uniform user location distribution: $\epsilon = 15$km$^{-1}$ and $\delta = 0.15$). Mean$\pm$1.96$\times$std. deviation.}
\label{Tb:exp:ULscalability_backup}
\vspace{-0.15in}
\end{table}

{\rev In this section, we evaluate the utility loss of various data perturbation methods across the three road network datasets under a uniform user location distribution, using alternative parameter settings. These results complement Table~\ref{Tb:exp:ULscalability} in the main paper.

Table~\ref{Tb:exp:ULscalability15} reports results for $\delta = 10^{-7}$ and $\epsilon = 15$km$^{-1}$, with the secret data domain size $K$ ranging from 500 to 5,000 and 50 users. Table~\ref{Tb:exp:ULscalability_backup} presents results for $\delta = 0.15$ under the same $\epsilon$, with $K$ ranging from 500 to 3,000 and 15 users. In both settings, \textsc{PAnDA} consistently delivers competitive utility performance across all baselines.

In Table~\ref{Tb:exp:ULscalability15}, \textsc{PAnDA}-e achieves the highest utility gains, reducing utility loss by {\bl 81.43\%} over EM, {\bl 76.36\%} over EM+BR, and {\bl 69.30\%} over LP+EM on average. \textsc{PAnDA}-p yields comparable improvements of {\bl 82.60\%}, {\bl 77.85\%}, and {\bl 71.88\%}, while \textsc{PAnDA}-l provides substantial reductions of {\bl 81.56\%}, {\bl 76.53\%}, and {\bl 67.80\%}, respectively. Similarly, in Table~\ref{Tb:exp:ULscalability_backup}, \textsc{PAnDA}-e achieves the largest gains, reducing utility loss by 75.5\% over EM, 73.7\% over EM+BR, and 58.5\% over LP+EM. \textsc{PAnDA}-p follows closely, with reductions of 72.9\%, 71.0\%, and 55.1\%, while \textsc{PAnDA}-l achieves improvements of 70.2\%, 68.1\%, and 52.6\%, respectively.

Note that although LP+BD attains the lowest utility loss when $K=500$, it fails to scale to larger domain sizes due to its high computational overhead. \emph{These additional results are consistent with the main findings in Table~\ref{Tb:exp:ULscalability}}.}

\vspace{-0.00in}
\begin{table}[t]
\centering
\small 
% \footnotesize 
\scriptsize 
\setlength{\tabcolsep}{3pt}  
\begin{tabular}{ c|c|c|c|c|c|c}
%\cline{2-13}
%\hline
\toprule
\multicolumn{7}{ c  }{Rome}\\ 
\cline{1-7}
\multicolumn{1}{ c|  }{Method}  & $K = 500$ & $K = 1,000$  & $K = 2,000$ & $K = 3,000$& $K = 4,000$& $K = 5,000$\\
\hline
\hline
\multicolumn{1}{ c|  }{EM} & $\leq0.005$ & $\leq0.005$ & $\leq0.005$ & $\leq0.005$ & $\leq0.005$ & 
$\leq0.005$ \\ 
\multicolumn{1}{ c|  }{LP} & >1,800 & >1,800 & >1,800 & >1,800 & >1,800 & >1,800 \\ 
\multicolumn{1}{ c|  }{LP+CA} &6.84±0.55 &7.96±0.52 &14.66±0.73 &20.12±0.42 &22.96±0.42 &29.73±0.51 \\ 
\multicolumn{1}{ c|  }{LP+BD} &7.03±3.68 & >1,800 & >1,800 & >1,800 & >1,800 & >1,800 \\ 
\multicolumn{1}{ c|  }{LP+EM} &0.61±0.21 & >1,800 & >1,800 & >1,800 & >1,800 & >1,800 \\ 
\multicolumn{1}{ c|  }{EM+BR} &0.05±0.00 &0.11±0.00 &0.15±0.00 &0.46±0.00 &0.47±0.01 &0.66±0.01 \\ 
\hline
\multicolumn{1}{ c|  }{ \textbf{\textsc{PAnDA}-e}} & 0.21±0.05 & 0.42±0.02 & 4.65±0.13 & 8.86±6.14 & 18.96±9.36 & 26.42±5.32 \\ 
\multicolumn{1}{ c|  }{ \textbf{\textsc{PAnDA}-p}} & 0.26±0.05 & 0.60±0.09 & 6.31±3.40 & 8.48±4.17 & 20.44±13.52 & 25.22±11.89 \\ 
\multicolumn{1}{ c|  }{ \textbf{\textsc{PAnDA}-l}} & 0.25±0.10 & 0.69±0.04 & 8.67±2.84 & 13.57±4.81 & 25.25±15.28 & 29.40±14.43 \\ 
\hline
\end{tabular}
\vspace{0.05in}
\caption{Computation time of different algorithms (Rome taxicab location dataset: $\epsilon = 15$km$^{-1}$ and $\delta = 10^{-7}$). \\ 
Mean$\pm$1.96$\times$std. deviation.}
\label{Tb:exp:time_scalability_realdata15}
\vspace{-0.15in}
\end{table}

\vspace{-0.00in}
\begin{table}[t]
\centering
\small 
% \footnotesize 
\scriptsize 
\setlength{\tabcolsep}{1pt}  
\begin{tabular}{ c|c|c|c|c|c|c}
%\cline{2-13}
%\hline
\toprule
\multicolumn{7}{ c  }{Rome}\\ 
\cline{1-7}
\multicolumn{1}{ c|  }{Method}  & $K = 500$ & $K = 1,000$  & $K = 2,000$ & $K = 3,000$& $K = 4,000$& $K = 5,000$\\
\hline
\hline
\multicolumn{1}{ c|  }{EM} & 1213.38±75.82 & 1267.30±96.79 & 1334.35±76.33 & 1269.70±62.70 & 1385.42±80.03 & 1282.36±71.40\\
% \multicolumn{1}{ c|  }{Laplace} & x.xx$\pm$x.x & x.xx$\pm$x.x & x.xx$\pm$x.x & x.xx$\pm$x.x & x.xx$\pm$x.x &  x.xx$\pm$x.x \\ 
\multicolumn{1}{ c|  }{LP} & - & -  & -  & 
- & - & -\\ 
\multicolumn{1}{ c|  }{LP+CA} & 1780.85±138.36 & 1927.02±75.98 & 1906.99±63.44 & 1858.61±135.95 & 1818.22±68.41 & 1921.62±75.97\\ 
\multicolumn{1}{ c|  }{LP+BD} & 152.73±39.73 & -  & - & - & - & -\\ 
\multicolumn{1}{ c|  }{LP+EM} & 731.94±67.10 & -  & - & 
-  & - & -\\ 
\multicolumn{1}{ c|  }{EM+BR} & 876.52±31.81 & 941.68±44.15 & 1133.21±49.53 & 1061.10±75.07 & 1072.56±48.37 & 1095.34±68.22 \\ 
\hline
\multicolumn{1}{ c|  }{ \textbf{\textsc{PAnDA}-e}} & 229.76±42.22 & 230.39±15.39 & 248.60±43.64 & 250.61±43.52 & 273.12±31.34 & 259.35±35.13\\
\multicolumn{1}{ c|  }{ \textbf{\textsc{PAnDA}-p}} & 225.95±36.87 & 237.85±40.13 & 241.49±37.46 & 248.69±24.08 & 251.68±52.14 & 268.18±46.04\\
\multicolumn{1}{ c|  }{ \textbf{\textsc{PAnDA}-l}} & 222.27±49.21 & 224.63±18.51 & 243.81±29.42 & 272.09±70.55 & 258.14±43.13 & 258.82±48.70\\ 
\hline
\end{tabular}
\vspace{0.05in}
\caption{Utility loss (in meters) of different algorithms (Rome taxicab location dataset: $\epsilon = 15$km$^{-1}$ and $\delta = 10^{-7}$). Mean$\pm$1.96$\times$std. deviation. }
\label{Tb:exp:ULscalability_realdata15}
\vspace{-0.15in}
\end{table}

\newpage
\subsection{Computation Efficiency and Utility Loss Using Real-World Dataset}
\label{subsec:real_add}
{\rev In this section, we evaluate the computation efficiency and utility loss of various data perturbation methods across the three road network datasets in Tables~\ref{Tb:exp:time_scalability_realdata15} and \ref{Tb:exp:ULscalability_realdata15}, using the real-world Rome taxicab dataset, with alternative parameter settings. These results complement Tables~\ref{Tb:exp:time_scalability_realdata} and \ref{Tb:exp:ULscalability_realdata} in the main paper. Particularly, we set $\delta = 10^{-7}$ and $\epsilon = 15$km$^{-1}$, with the secret data domain size $K$ varying from 500 to 5,000 and the number of users set to 50. 

In Table~\ref{Tb:exp:time_scalability_realdata15}, the average runtimes of \textsc{PAnDA}-e, \textsc{PAnDA}-p, and \textsc{PAnDA}-l are {\bl 26.42s, 25.22s, and 29.40s}. In contrast, the LP-based methods—LP, LP+BD, and LP+EM—exhibit substantial computational overhead, particularly when handling larger domains. In Table~\ref{Tb:exp:ULscalability_realdata15}, \textsc{PAnDA}-e, \textsc{PAnDA}-p, and \textsc{PAnDA}-l achieve average reductions in utility loss of {\bl 80.76\%}, {\bl 80.99\%}, and {\bl 80.91\%} over EM, {\bl 75.86\%}, {\bl 76.15\%}, and {\bl 76.06\%} over EM+BR, and {\bl 68.61\%}, {\bl 69.13\%}, and {\bl 69.63\%} over LP+EM, respectively. \emph{These additional results are consistent with the main findings in Tables~\ref{Tb:exp:time_scalability_realdata} and \ref{Tb:exp:ULscalability_realdata}}.}

\subsection{mDP Violation Ratio}
\label{subsec:falure_add}

\begin{figure*}[t]
\centering
\hspace{0.00in}
\begin{minipage}{1.00\textwidth}
\centering
  \subfigure[Rome]{
\includegraphics[width=0.24\textwidth, height = 0.15\textheight]{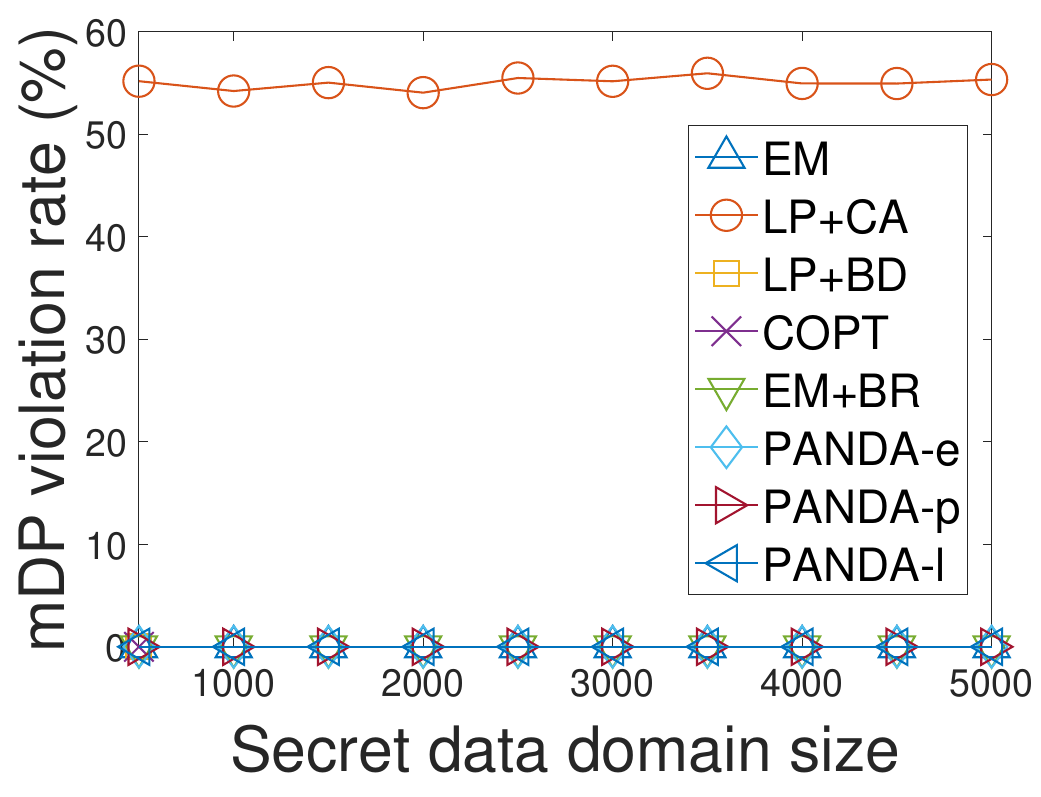}}
  \subfigure[NYC]{
\includegraphics[width=0.24\textwidth, height = 0.15\textheight]{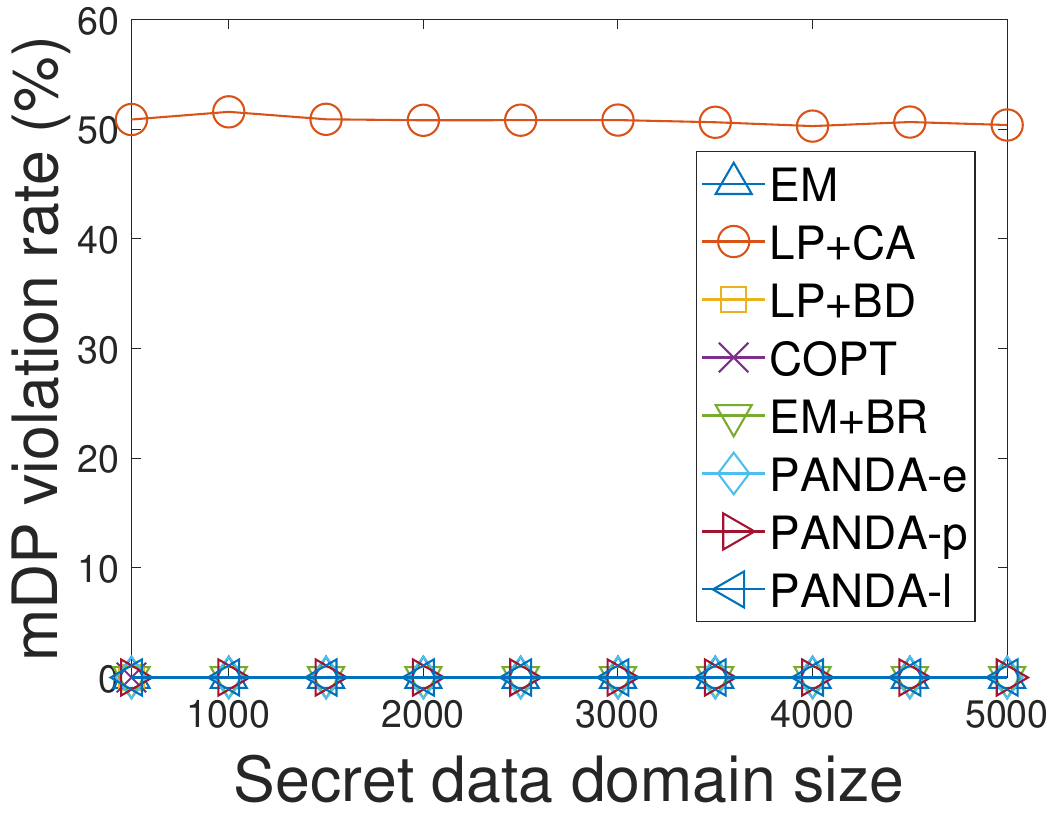}}
  \subfigure[London]{
\includegraphics[width=0.24\textwidth, height = 0.15\textheight]{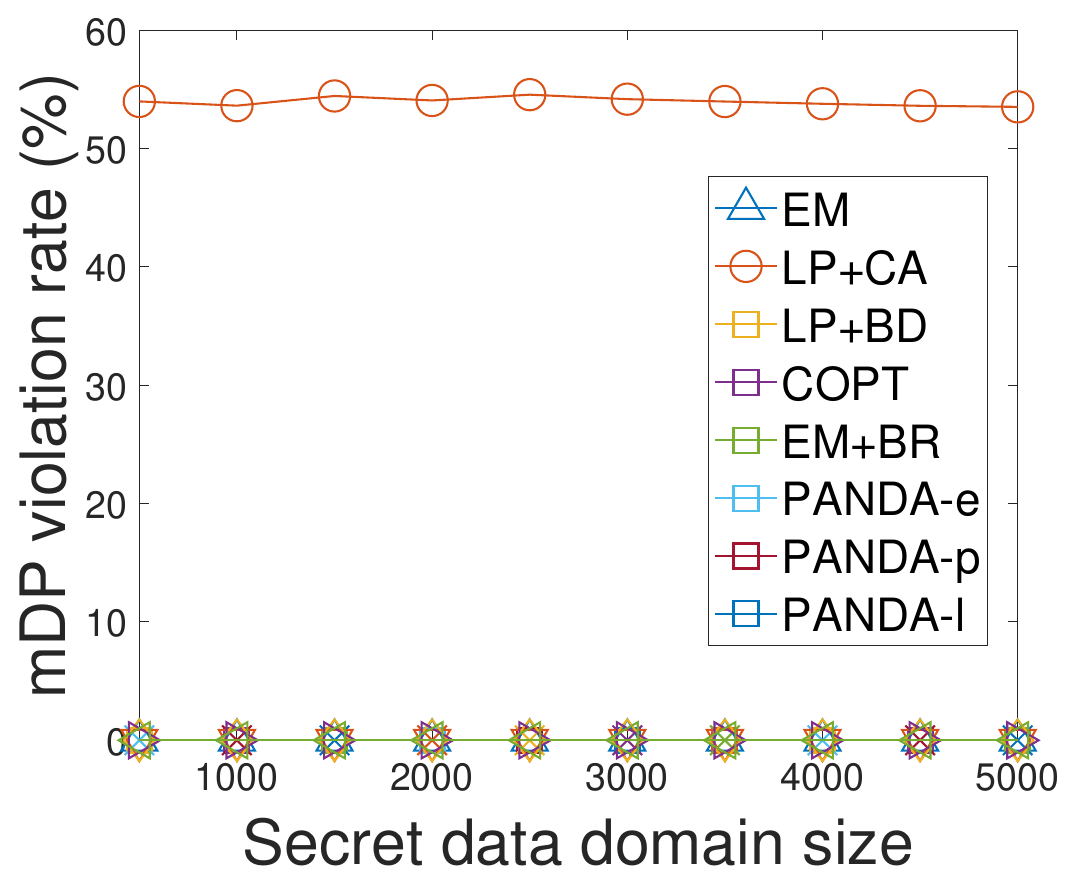}}
\vspace{-0.15in}
\end{minipage}
\caption{mDP failture rate.}
\label{fig:mDPfailture15}
\vspace{-0.00in}
\end{figure*}

{\rev In this section, we evaluate the mDP violation ratio of various data perturbation methods across the three road network datasets, as shown in Figure~\ref{fig:mDPfailture15}(a)(b)(c), under an alternative parameter setting with $\epsilon = 15$km$^{-1}$. These results complement Figure~\ref{fig:mDPfailture}(a)(b)(c) in the main paper, which used $\epsilon = 5$km$^{-1}$. The figure shows that \textsc{PAnDA} consistently maintains a violation rate below $10^{-5}$—well below the predefined threshold of $\delta = 0.15$—demonstrating strong reliability in satisfying probabilistic mDP guarantees. In contrast, LP+CA exhibits significantly higher violation rates, averaging {\bl 53.26\%}, indicating its limited robustness in meeting mDP constraints. Overall, the trends observed in Figure~\ref{fig:mDPfailture15}(a)(b)(c) are consistent with those in Figure~\ref{fig:mDPfailture}(a)(b)(c).} 

\subsection{Privacy Budget Allocation Across the Two Phases}
\label{subsec:budgetloc_add}

\begin{figure*}[t]
\centering
\hspace{0.00in}
\begin{minipage}{1.00\textwidth}
\centering
  \subfigure[Rome]{
\includegraphics[width=0.24\textwidth, height = 0.15\textheight]{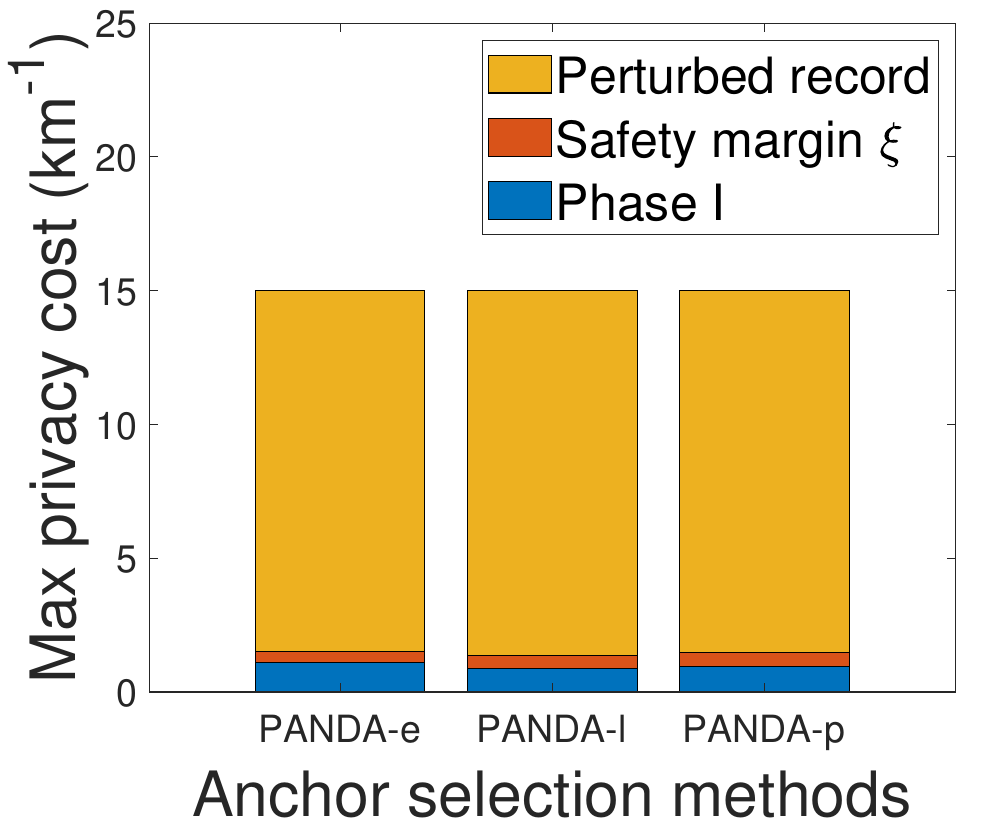}}
  \subfigure[NYC]{
\includegraphics[width=0.24\textwidth, height = 0.15\textheight]{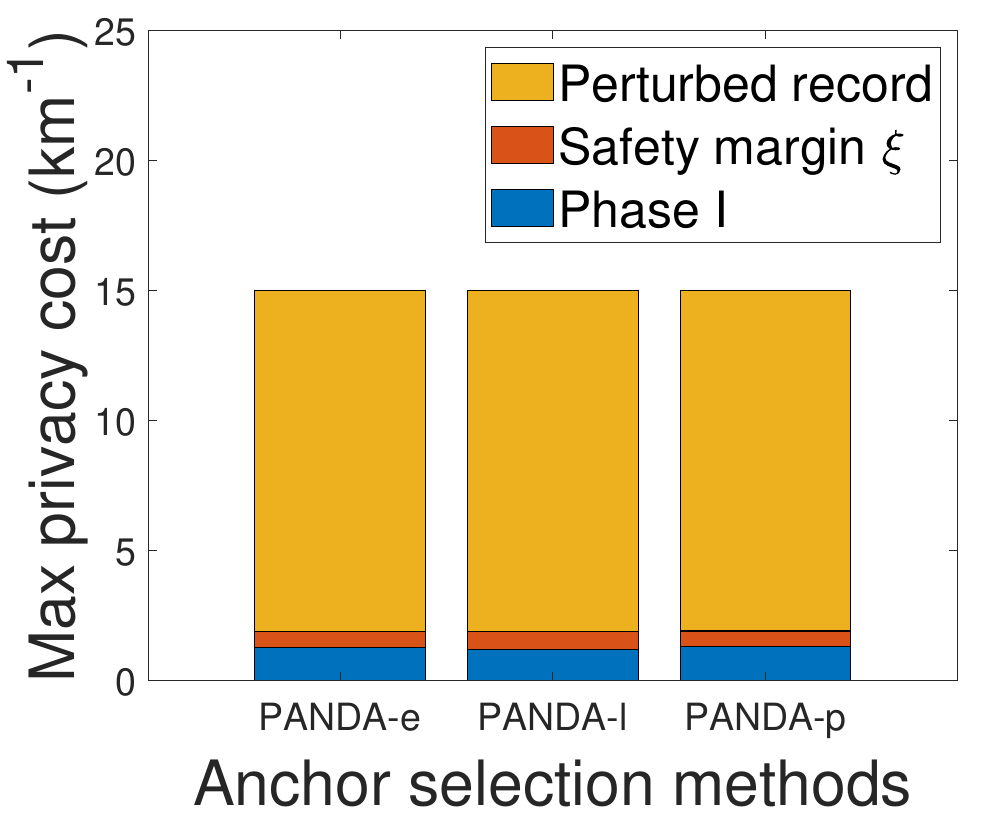}}
  \subfigure[London]{
\includegraphics[width=0.24\textwidth, height = 0.15\textheight]{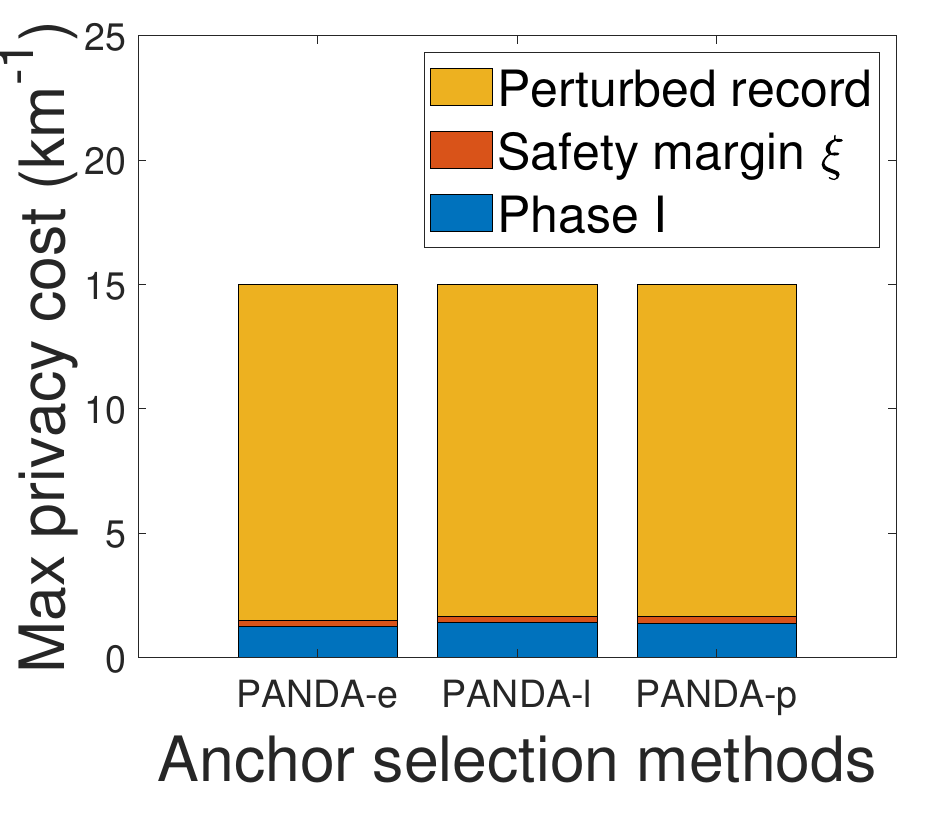}}
\vspace{-0.15in}
\end{minipage}
\caption{Privacy budget allocation across the two phases.}
\label{fig:budgetallocation15}
\vspace{-0.00in}
\end{figure*}

{\rev In this section, we analyze the distribution of the privacy budget across the three components of \textsc{PAnDA} in Figure~\ref{fig:budgetallocation15}(a)(b)(c): anchor selection in Phase I, and the reserved safety margin $\xi$ and perturbation vector optimization in Phase II, under an alternative setting with $\epsilon = 15$km$^{-1}$. These results complement Figure~\ref{fig:budgetallocation}(a)(b)(c) in the main paper, which uses $\epsilon = 5$km$^{-1}$. The results show that a majority of the privacy budget—on average {\bl 89.03\%}—is reserved for Phase II, providing sufficient flexibility to optimize utility under mDP constraints. In comparison, anchor selection consumes {\bl 7.91\%}, while the safety margin accounts for {\bl 3.06\%}. This allocation reflects a balanced design: although anchor selection and the safety margin are essential for ensuring probabilistic mDP guarantees, excessive budget allocation to these components would hinder utility optimization. These findings confirm that \textsc{PAnDA} effectively limits privacy leakage during anchor selection, bounds the safety margin judiciously, and preserves ample budget for utility-aware perturbation, consistent with the observations in Figure~\ref{fig:budgetallocation}(a)(b)(c).}

\subsection{Clustering Coefficients of Secret Records}
\label{subsec:clustering_add}
{\bl 

This experiment evaluates the clustering characteristics of anchor sets selected by \textsc{PAnDA}, compared to those in the original secret dataset. The goal is to assess how well the selected anchors support Benders decomposition, which benefits from more localized (i.e., clustered) data distributions.

The clustering coefficient is computed by converting the distance matrix into a binary adjacency matrix $A$ (setting $A_{ij} = 1$ if the distance between nodes $i$ and $j$ is below a threshold, e.g., the distance threshold $\gamma$, with diagonal entries zeroed). For each node $i$ with $k$ neighbors, we count the number of edges $m$ among those neighbors. The local clustering coefficient is given by $C(i) = \frac{2m}{k(k-1)}$,  which measures how many of the $\frac{k(k-1)}{2}$ possible neighbor-to-neighbor connections actually exist, and finally the average clustering coefficient is simply the arithmetic mean of all $C(i)$ values across the network.
}

{\rev Table~\ref{Tb:exp:cluster} reports the average clustering coefficients of the privacy-preserving graphs induced by each method. Across all datasets, the \textsc{PAnDA} variants consistently achieve higher coefficients than the baseline methods. On average, \textsc{PAnDA}-e, \textsc{PAnDA}-p, and \textsc{PAnDA}-l each attain a clustering coefficient of {\bl 0.6160, 0.5877, and 0.6142}, representing an improvement of {\bl 33.02\%}, {\bl 26.91\%}, and {\bl 32.63\%} over the original (non-selected) domains $\mathcal{X}$, respectively.

These results demonstrate that the localized selection strategies in \textsc{PAnDA} lead to tightly clustered anchor sets, making the resulting optimization problems more amenable to Benders decomposition. Specifically, the sparsity and locality of the anchor sets reduce coupling across subproblems, allowing \textsc{PAnDA} to decompose and solve large-scale LPs more efficiently than traditional full-domain approaches.}

\vspace{-0.00in}
\begin{table}[t]
\centering
\small 
% \footnotesize 
\scriptsize 
\setlength{\tabcolsep}{1pt}  
\begin{tabular}{ c|c|c|c|c|c|c}
%\cline{2-13}
%\hline
\toprule
\multicolumn{7}{ c  }{Rome}\\ 
\cline{1-7}
\multicolumn{1}{ c|  }{Method}  & $K = 500$ & $K = 1,000$  & $K = 2,000$ & $K = 3,000$& $K = 4,000$& $K = 5,000$\\
\hline
\hline
\multicolumn{1}{ c|  }{Origin} &0.1477±0.0029 &0.3328±0.0034 &0.5287±0.0049 &0.5920±0.0045 &0.6206±0.0087 &0.6352±0.0081 \\ 
\hline
\multicolumn{1}{ c|  }{ \textbf{\textsc{PAnDA}-e}} &0.2547±0.0026 &0.5134±0.0027 &0.7144±0.0066 &0.7843±0.0058 &0.8113±0.0089 &0.8269±0.0073 \\
\multicolumn{1}{ c|  }{ \textbf{\textsc{PAnDA}-p}} &0.1833±0.0038 &0.4261±0.0035 &0.6792±0.0071 &0.7528±0.0064 &0.7910±0.0081 &0.8192±0.0063 \\
\multicolumn{1}{ c|  }{ \textbf{\textsc{PAnDA}-l}} &0.2246±0.0042 &0.4818±0.0033 &0.7252±0.0064 &0.7892±0.0081 &0.8151±0.0037 &0.8288±0.0067 \\
\hline
\toprule
\multicolumn{7}{ c  }{NYC}\\ 
\cline{1-7}
\multicolumn{1}{ c|  }{Method}  & $K = 500$ & $K = 1,000$  & $K = 2,000$ & $K = 3,000$& $K = 4,000$& $K = 5,000$\\
\hline
\hline
\multicolumn{1}{ c|  }{Origin} &0.1738±0.0023 &0.3759±0.0031 &0.5617±0.0059 &0.6107±0.0047 &0.6283±0.0061 &0.6353±0.0080 \\ 
\hline
\multicolumn{1}{ c|  }{ \textbf{\textsc{PAnDA}-e}} &0.2493±0.0036 &0.4981±0.0033 &0.7253±0.0041 &0.7529±0.0069 &0.7923±0.0045 &0.7974±0.0088 \\
\multicolumn{1}{ c|  }{ \textbf{\textsc{PAnDA}-p}} &0.2286±0.0046 &0.4673±0.0018 &0.6989±0.0070 &0.7590±0.0063 &0.7900±0.0056 &0.7937±0.0094 \\
\multicolumn{1}{ c|  }{ \textbf{\textsc{PAnDA}-l}} &0.2353±0.0029 &0.4971±0.0034 &0.7381±0.0037 &0.7707±0.0065 &0.7925±0.0084 &0.7956±0.0075 \\
\hline
\toprule
\multicolumn{7}{ c  }{London}\\ 
\cline{1-7}
\multicolumn{1}{ c|  }{Method}  & $K = 500$ & $K = 1,000$  & $K = 2,000$ & $K = 3,000$& $K = 4,000$& $K = 5,000$\\
\hline
\hline
\multicolumn{1}{ c|  }{Origin} &0.0831±0.0019 &0.2193±0.0042 &0.4355±0.0052 &0.5338±0.0034 &0.5927±0.0085 &0.6290±0.0076 \\ 
\hline
\multicolumn{1}{ c|  }{ \textbf{\textsc{PAnDA}-e}} &0.1472±0.0045 &0.3202±0.0038 &0.5940±0.0064 &0.7027±0.0072 &0.7798±0.0039 &0.8231±0.0090 \\
\multicolumn{1}{ c|  }{ \textbf{\textsc{PAnDA}-p}} &0.1155±0.0048 &0.2888±0.0027 &0.5727±0.0037 &0.6806±0.0085 &0.7484±0.0051 &0.7832±0.0058 \\
\multicolumn{1}{ c|  }{ \textbf{\textsc{PAnDA}-l}} &0.1353±0.0054 &0.3430±0.0043 &0.5972±0.0059 &0.7153±0.0077 &0.7510±0.0060 &0.8207±0.0073 \\
\hline
\end{tabular}
\vspace{0.05in}
\caption{Clustering coefficients. Mean$\pm$1.96$\times$std. deviation. }
\label{Tb:exp:cluster}
\vspace{-0.05in}
\end{table}

\subsection{Performance with Diverse Acceptable Violation Ratios $\delta$}
\label{subsec:diversedelta_add}

{\rev In this section, we evaluate the performance of \textsc{PAnDA} under a range of acceptable mDP violation ratios:
\[
\delta \in \{10^{-7}, 10^{-6}, 10^{-5}, 10^{-4}, 10^{-3}, 10^{-2}, 10^{-1}, 0.15\}.
\]
On average, the computation times for \textsc{PAnDA}-e, \textsc{PAnDA}-p, and \textsc{PAnDA}-l are 6.28s, 6.13s, and 6.75s, respectively.

Figure~\ref{fig:UL_delta}(a)(b)(c) shows the utility loss of our methods across different $\delta$ values. \textsc{PAnDA} consistently delivers strong utility performance regardless of the allowed violation threshold. Compared to Table~\ref{Tb:exp:ULscalability}, \textsc{PAnDA}-e, \textsc{PAnDA}-p, and \textsc{PAnDA}-l achieve at least {\bl 82.11\%}, {\bl 77.93\%}, and {\bl 64.00\%} lower utility loss, respectively, relative to EM, EM+BR, and LP+EM.

Figure~\ref{fig:mDPfailure_delta}(a)(b)(c) further confirms the robustness of our methods by reporting the empirical mDP violation rates under varying $\delta$. Across all settings, \textsc{PAnDA}-e, \textsc{PAnDA}-p, and \textsc{PAnDA}-l consistently exhibit zero observed violations, even at the most lenient threshold $\delta = 0.15$. This substantial gap between the theoretical bound and the actual violation rates is primarily due to two factors: (1) the use of the function $h(\xi)$, which enforces more conservative safety margins, and (2) the estimation of $\xi$ in Eq.~(\ref{eq:estxi}), which is derived from worst-case analysis over high-probability location pairs.}

\begin{figure*}[t]
\begin{minipage}{1.00\textwidth}
\centering
  \subfigure[Rome]{
\includegraphics[width=0.24\textwidth, height = 0.15\textheight]{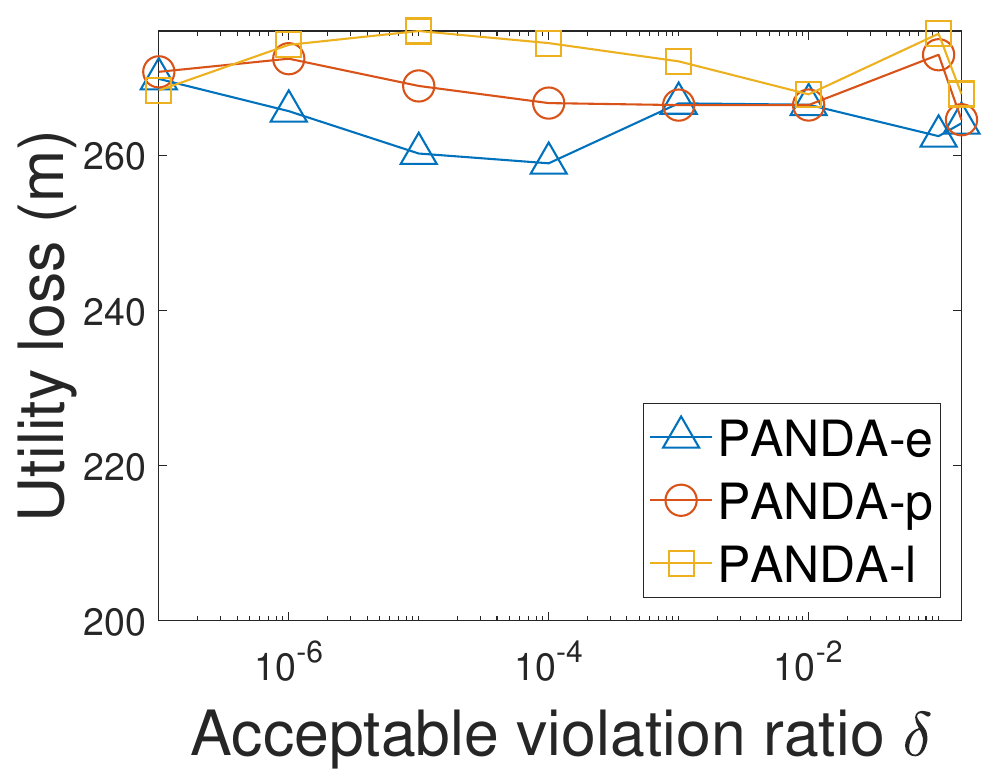}}
  \subfigure[NYC]{
\includegraphics[width=0.24\textwidth, height = 0.15\textheight]{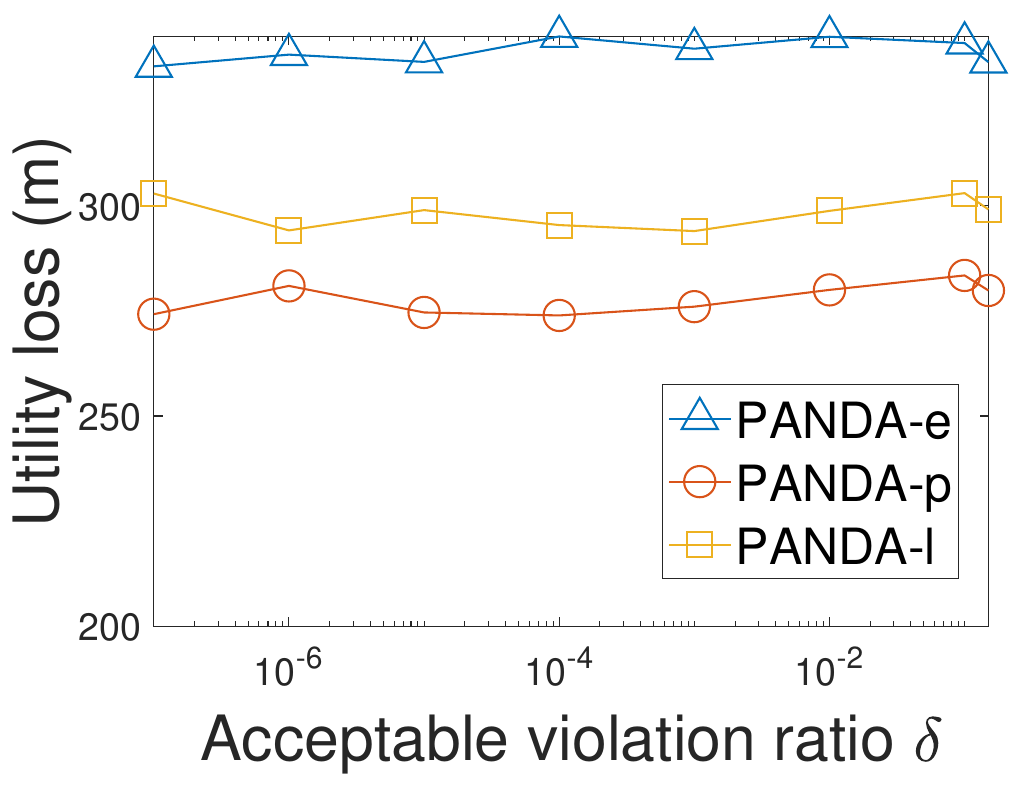}}
  \subfigure[London]{
\includegraphics[width=0.24\textwidth, height = 0.15\textheight]{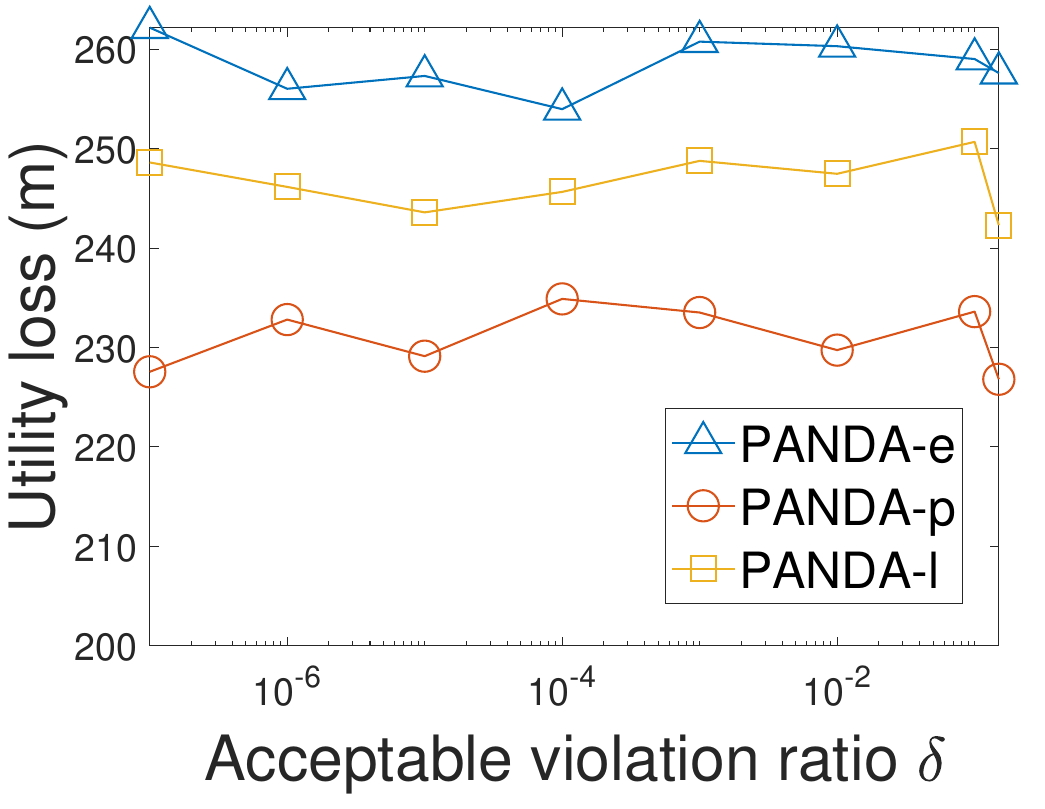}}
\vspace{-0.15in}
\end{minipage}
\caption{Utility loss of \textsc{PAnDA} under varying acceptable mDP violation ratios $\delta$.}
\label{fig:UL_delta}
\vspace{-0.05in}
\end{figure*}

\begin{figure*}[t]
\begin{minipage}{1.00\textwidth}
\centering
  \subfigure[Rome]{
\includegraphics[width=0.24\textwidth, height = 0.15\textheight]{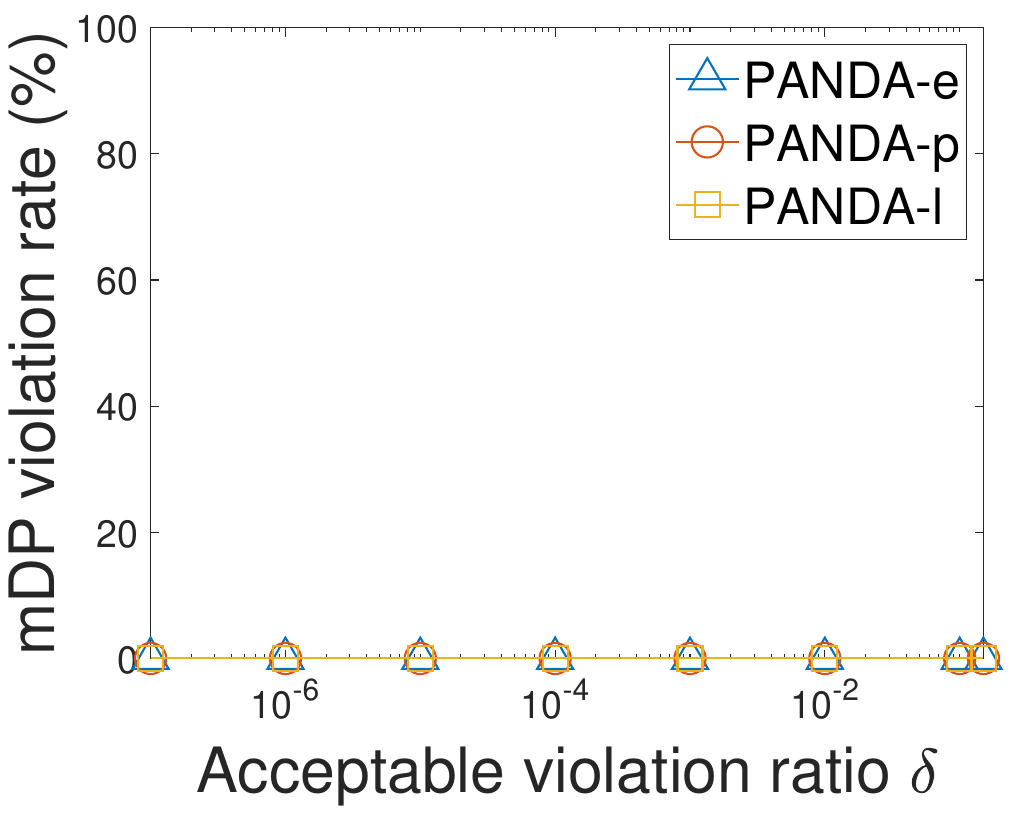}}
  \subfigure[NYC]{
\includegraphics[width=0.24\textwidth, height = 0.15\textheight]{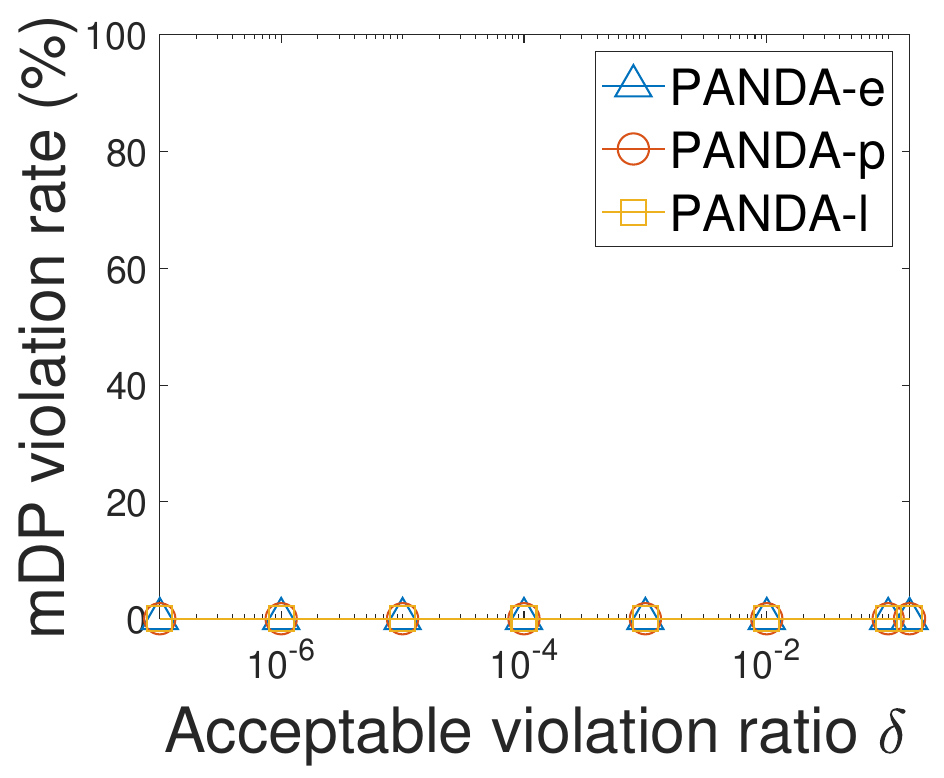}}
  \subfigure[London]{
\includegraphics[width=0.24\textwidth, height = 0.15\textheight]{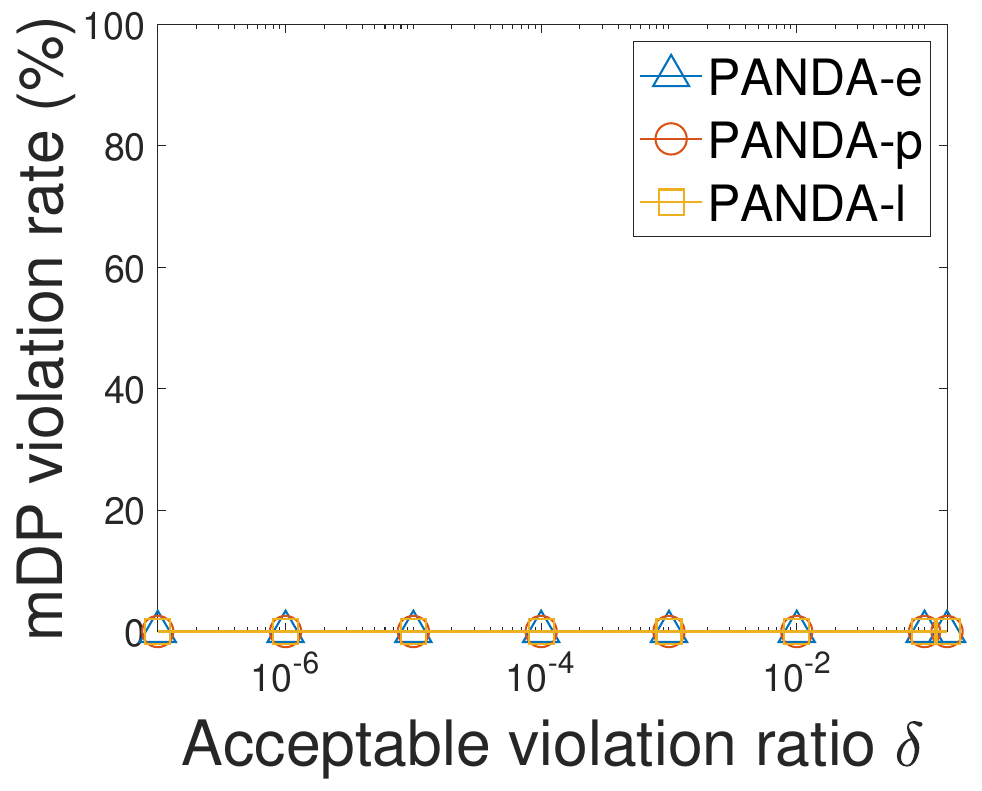}}
\vspace{-0.15in}
\end{minipage}
\caption{mDP violation rate of \textsc{PAnDA} under varying acceptable mDP violation ratios $\delta$.}
\label{fig:mDPfailure_delta}
\vspace{-0.05in}
\end{figure*}

\subsection{Comparison Between Estimated Safety Margin and Real Safety Margin}
\label{subsec:safetymarginpara_add}

\noindent\textbf{Figures~\ref{fig:xi_delta_rome}(a)(b)(c),~\ref{fig:xi_delta_NYC}(a)(b)(c), and~\ref{fig:xi_delta_london}(a)(b)(c)} compare the \emph{actual safety margin}, defined as $\xi_{x_n, x_m} = h_{x_n,x_m}^{-1}(1 - \delta)$, with the \emph{estimated safety margin} $\hat{\xi}_{x,x'}$ computed using Algorithm~2, under varying values of the violation tolerance parameter $\delta$. The comparisons are presented for the Rome, NYC, and London datasets, respectively.

Each figure includes three subplots corresponding to different anchor selection strategies: \textsc{PAnDA}-e, \textsc{PAnDA}-p, and \textsc{PAnDA}-l. Across all datasets and decay models, the estimated safety margins closely follow the actual margins, demonstrating the accuracy of the estimation procedure. On average, the estimated safety margins for \textsc{PAnDA}-e, \textsc{PAnDA}-p, and \textsc{PAnDA}-l exceed the actual safety margins by only {\bl 13.3\%, 13.8\%, and 13.2\%}, respectively.

As expected, the required safety margin increases as $\delta$ decreases, reflecting the need for stricter guarantees under tighter privacy constraints.
\begin{figure*}[t]
\centering
\hspace{0.00in}
\begin{minipage}{1.00\textwidth}
\centering
  \subfigure[\textsc{PAnDA}-e]{
\includegraphics[width=0.24\textwidth, height = 0.15\textheight]{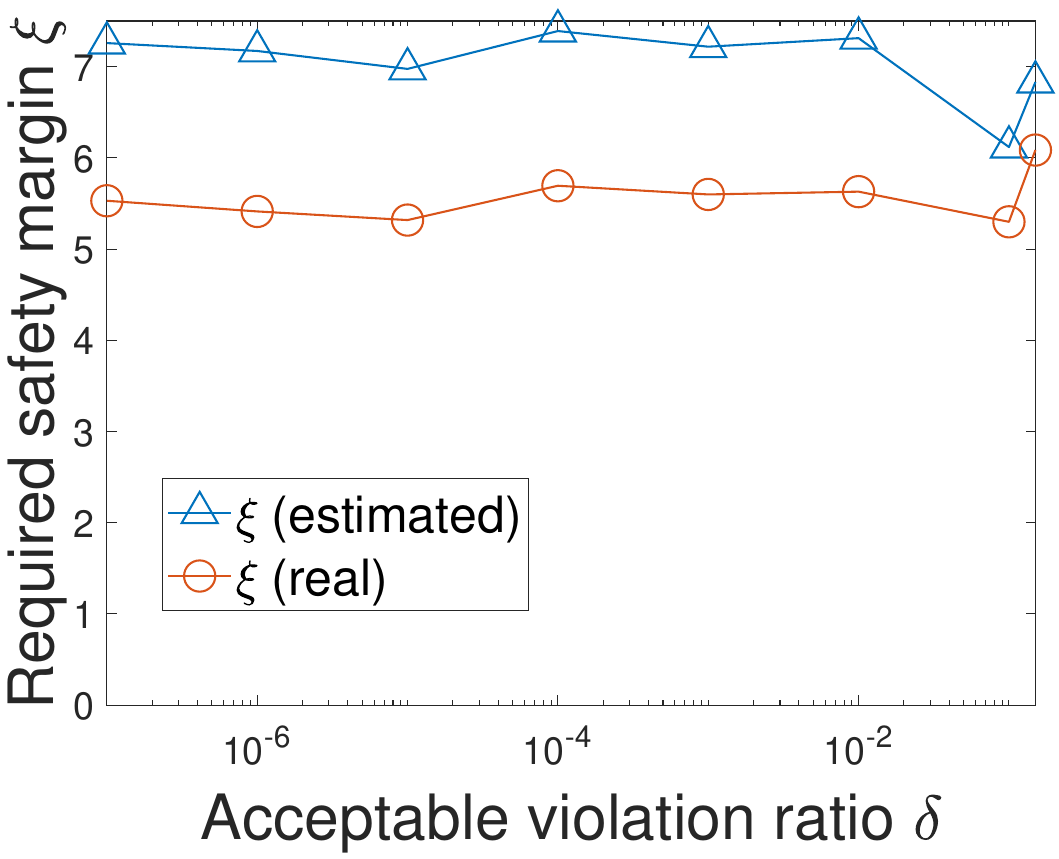}}
  \subfigure[\textsc{PAnDA}-p]{
\includegraphics[width=0.24\textwidth, height = 0.15\textheight]{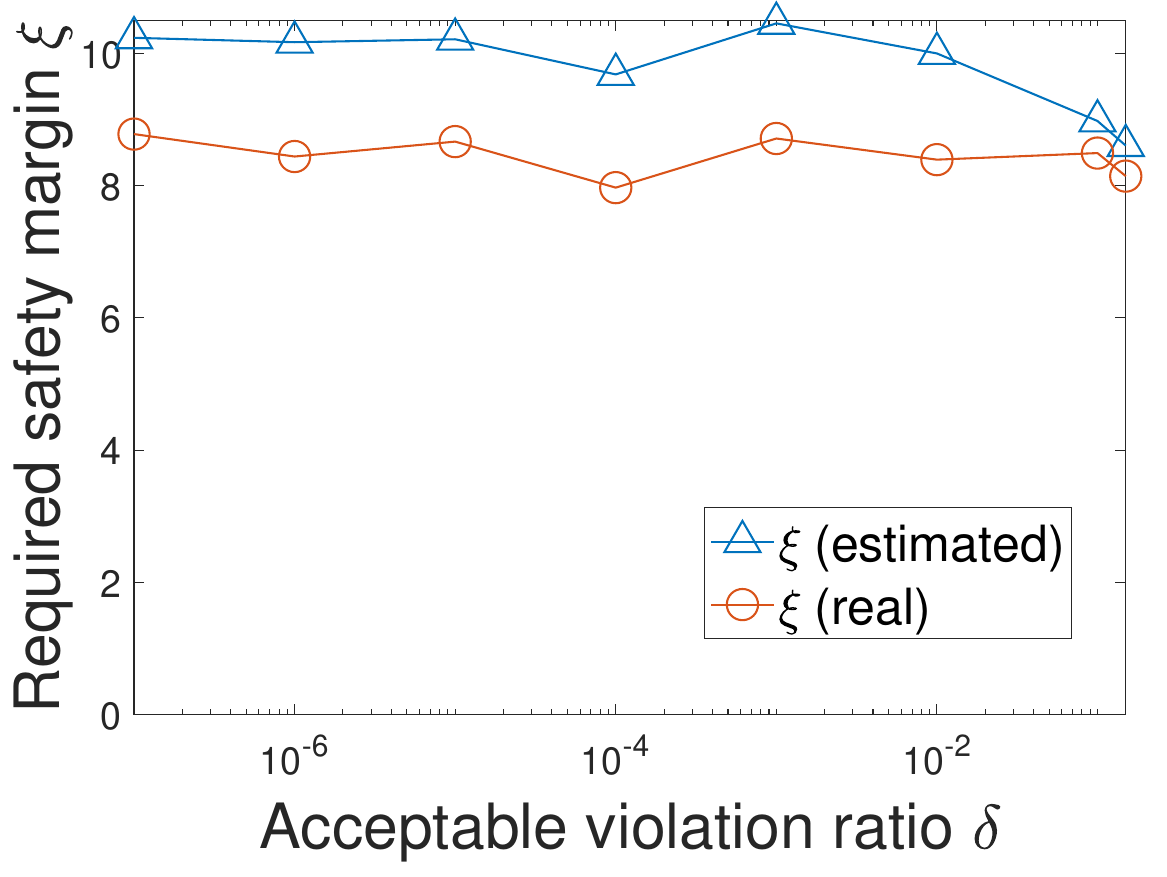}}
  \subfigure[\textsc{PAnDA}-l]{
\includegraphics[width=0.24\textwidth, height = 0.15\textheight]{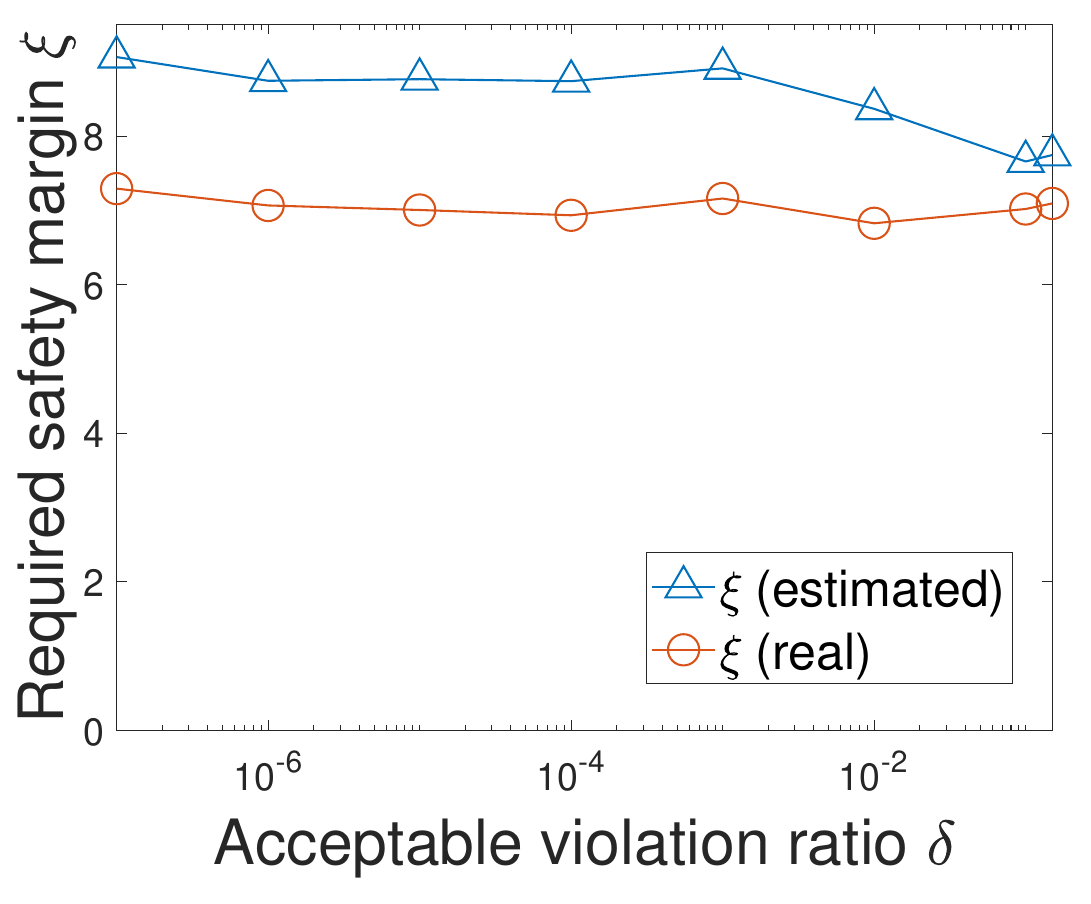}}
\vspace{-0.15in}
\end{minipage}
\caption{Estimated safety margin vs. real safety margin (Rome).}
\label{fig:xi_delta_rome}
\end{figure*}

\begin{figure*}[t]
\begin{minipage}{1.00\textwidth}
\centering
  \subfigure[\textsc{PAnDA}-e]{
\includegraphics[width=0.24\textwidth, height = 0.15\textheight]{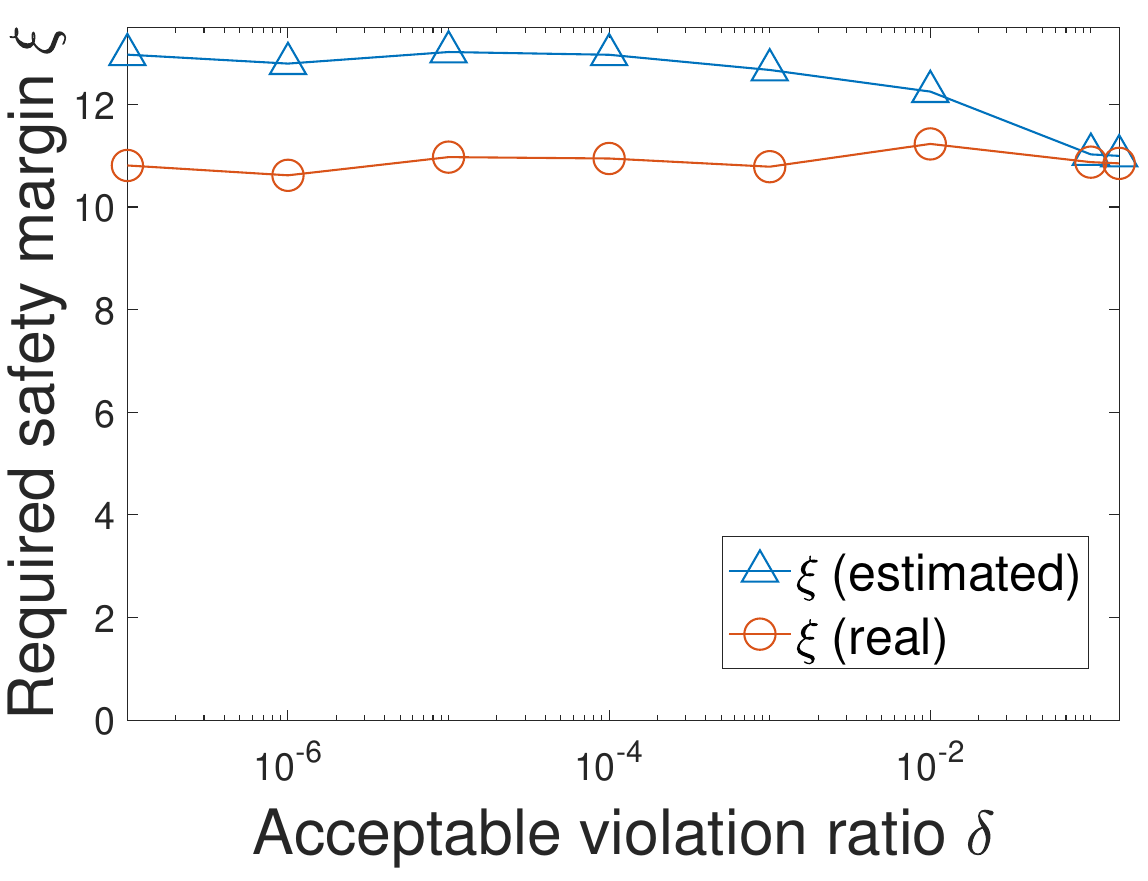}}
  \subfigure[\textsc{PAnDA}-p]{
\includegraphics[width=0.24\textwidth, height = 0.15\textheight]{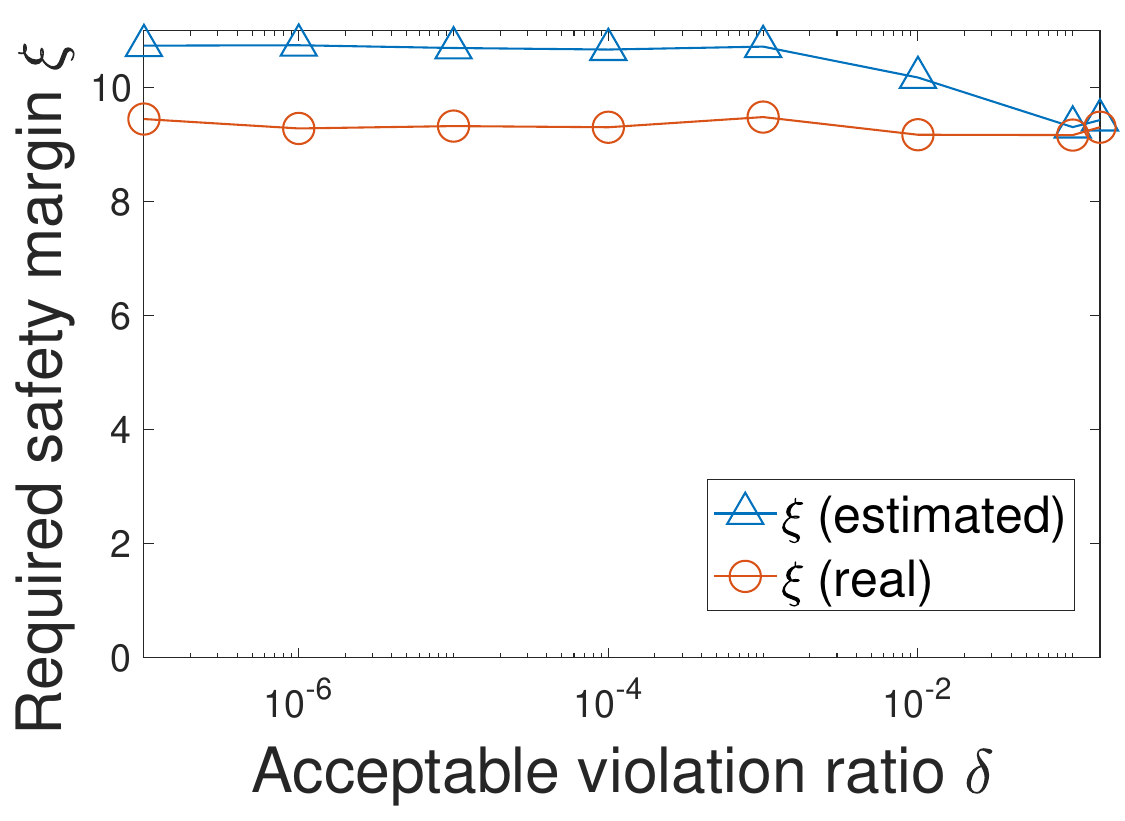}}
  \subfigure[\textsc{PAnDA}-l]{
\includegraphics[width=0.24\textwidth, height = 0.15\textheight]{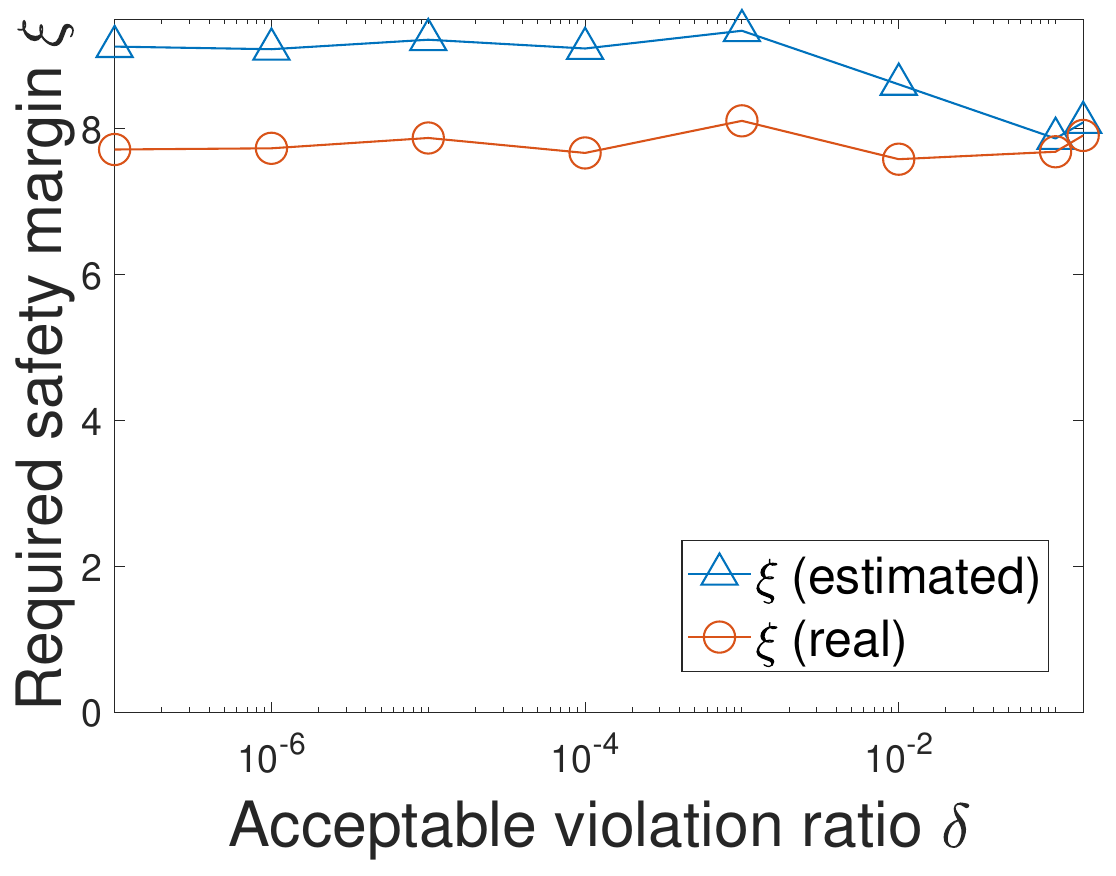}}
\vspace{-0.15in}
\end{minipage}
\caption{Estimated safety margin vs. real safety margin (NYC).}
\label{fig:xi_delta_NYC}
\end{figure*}

\begin{figure*}[t]
\begin{minipage}{1.00\textwidth}
\centering
  \subfigure[\textsc{PAnDA}-e]{
\includegraphics[width=0.24\textwidth, height = 0.15\textheight]{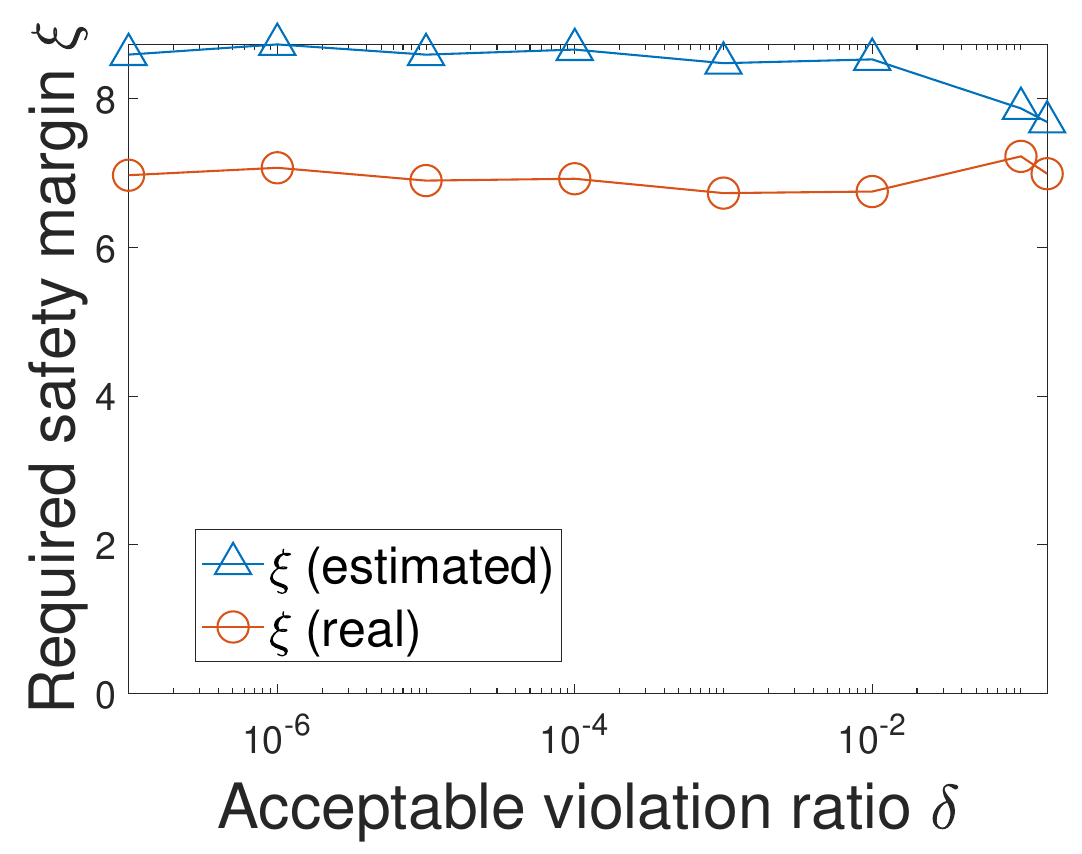}}
  \subfigure[\textsc{PAnDA}-p]{
\includegraphics[width=0.24\textwidth, height = 0.15\textheight]{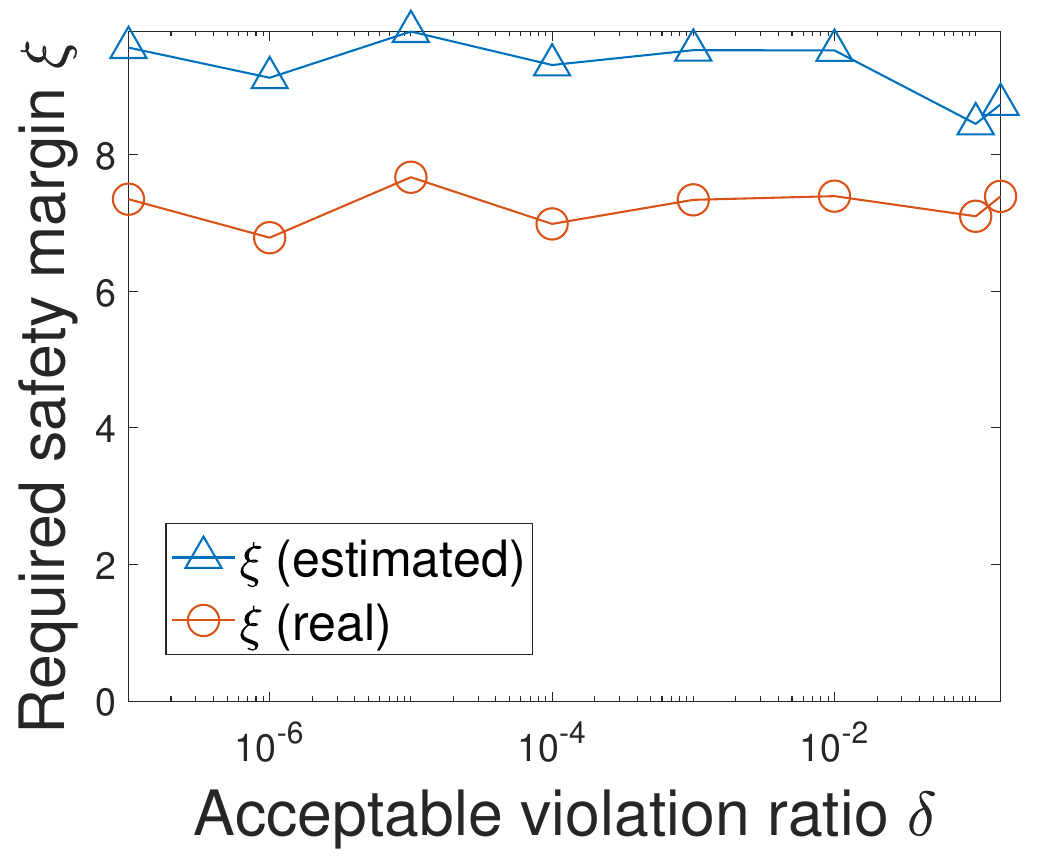}}
  \subfigure[\textsc{PAnDA}-l]{
\includegraphics[width=0.24\textwidth, height = 0.15\textheight]{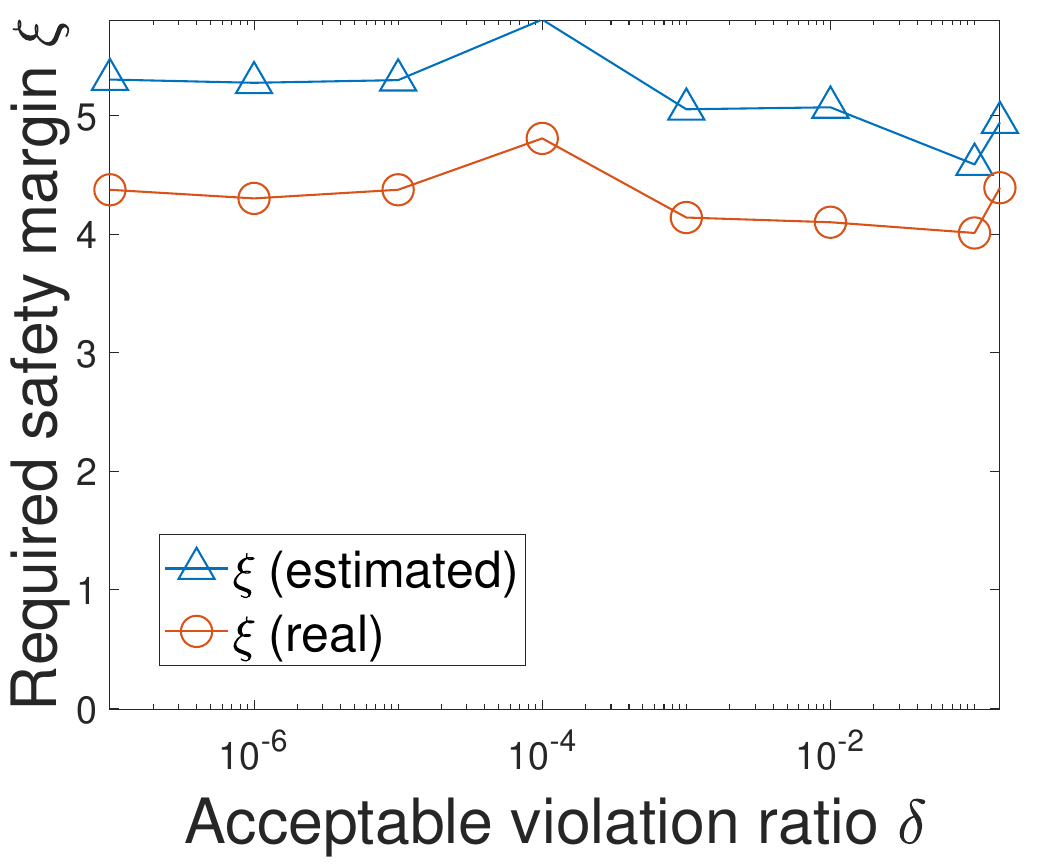}}
\vspace{-0.15in}
\end{minipage}
\caption{Estimated safety margin vs. real safety margin (London).}
\label{fig:xi_delta_london}
\vspace{-0.15in}
\end{figure*}

\subsection{Expected Safety Margin With Different Parameters}
\label{subsec:safetymargin_add}
\begin{figure*}[t]
\centering
\hspace{0.00in}
\begin{minipage}{1.00\textwidth}
\centering
  \subfigure[Exponential decay]{
\includegraphics[width=0.24\textwidth, height = 0.15\textheight]{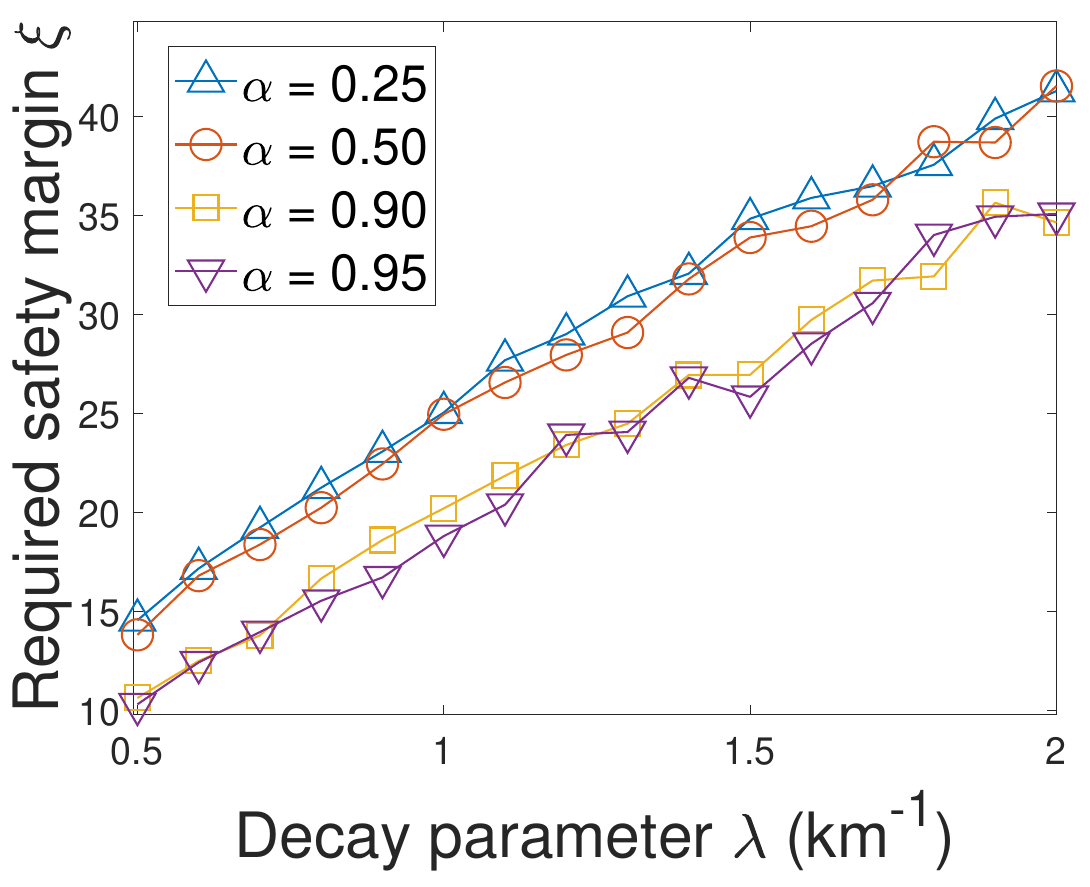}}
  \subfigure[Power law decay]{
\includegraphics[width=0.24\textwidth, height = 0.15\textheight]{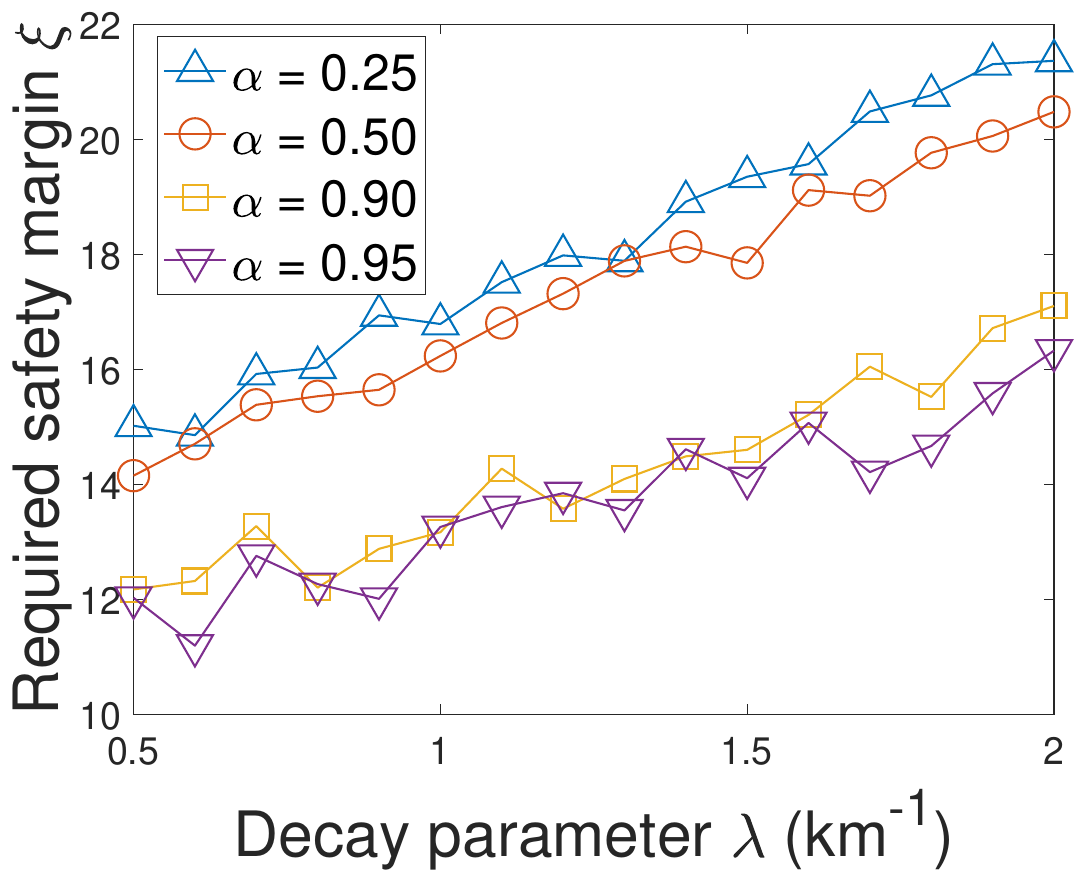}}
  \subfigure[Logistic]{
\includegraphics[width=0.24\textwidth, height = 0.15\textheight]{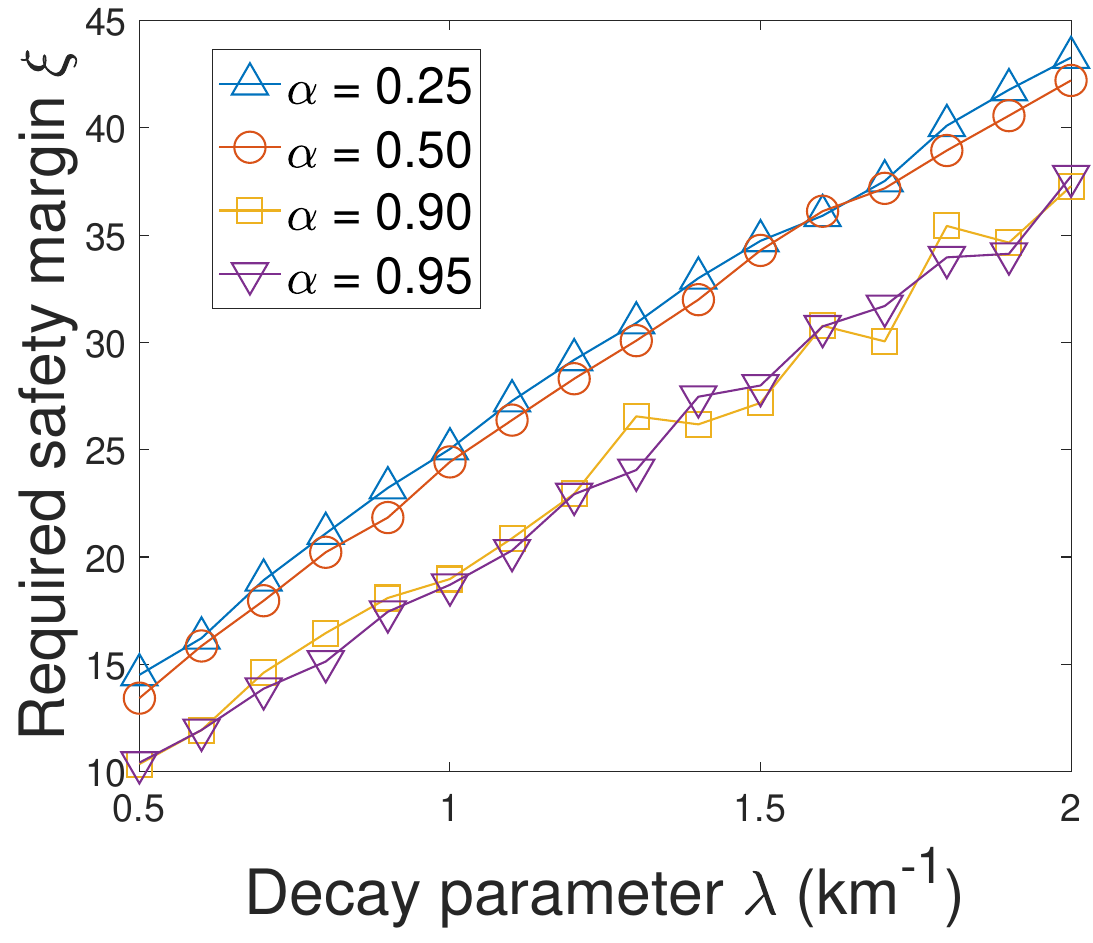}}
\vspace{-0.15in}
\end{minipage}
\caption{Expected safety margin $\xi$ with different decay parameter $\lambda$ and scaling parameter $\alpha$ (NYC).}
\label{fig:xi_lambda_nyc}
\vspace{-0.00in}
\end{figure*}

\begin{figure*}[t]
\centering
\hspace{0.00in}
\begin{minipage}{1.00\textwidth}
\centering
  \subfigure[Exponential decay]{
\includegraphics[width=0.24\textwidth, height = 0.15\textheight]{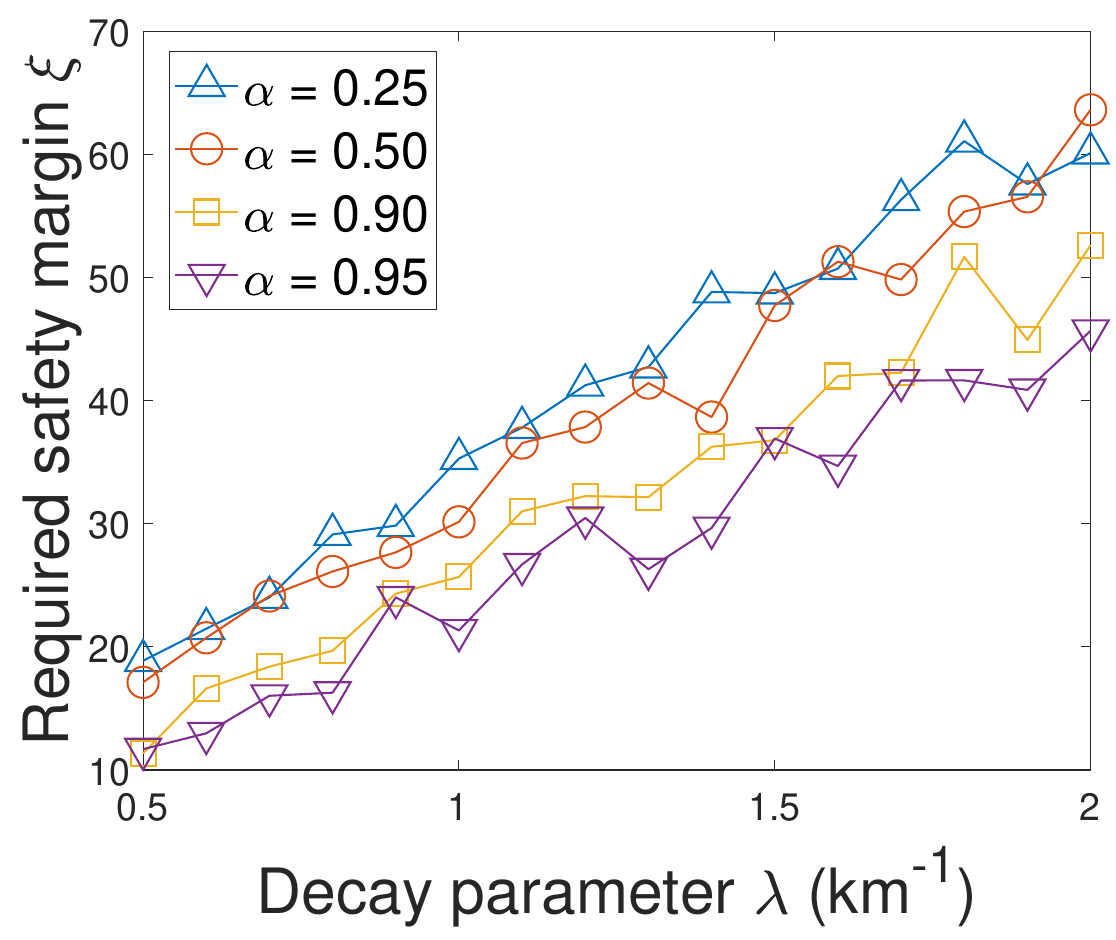}}
  \subfigure[Power law decay]{
\includegraphics[width=0.24\textwidth, height = 0.15\textheight]{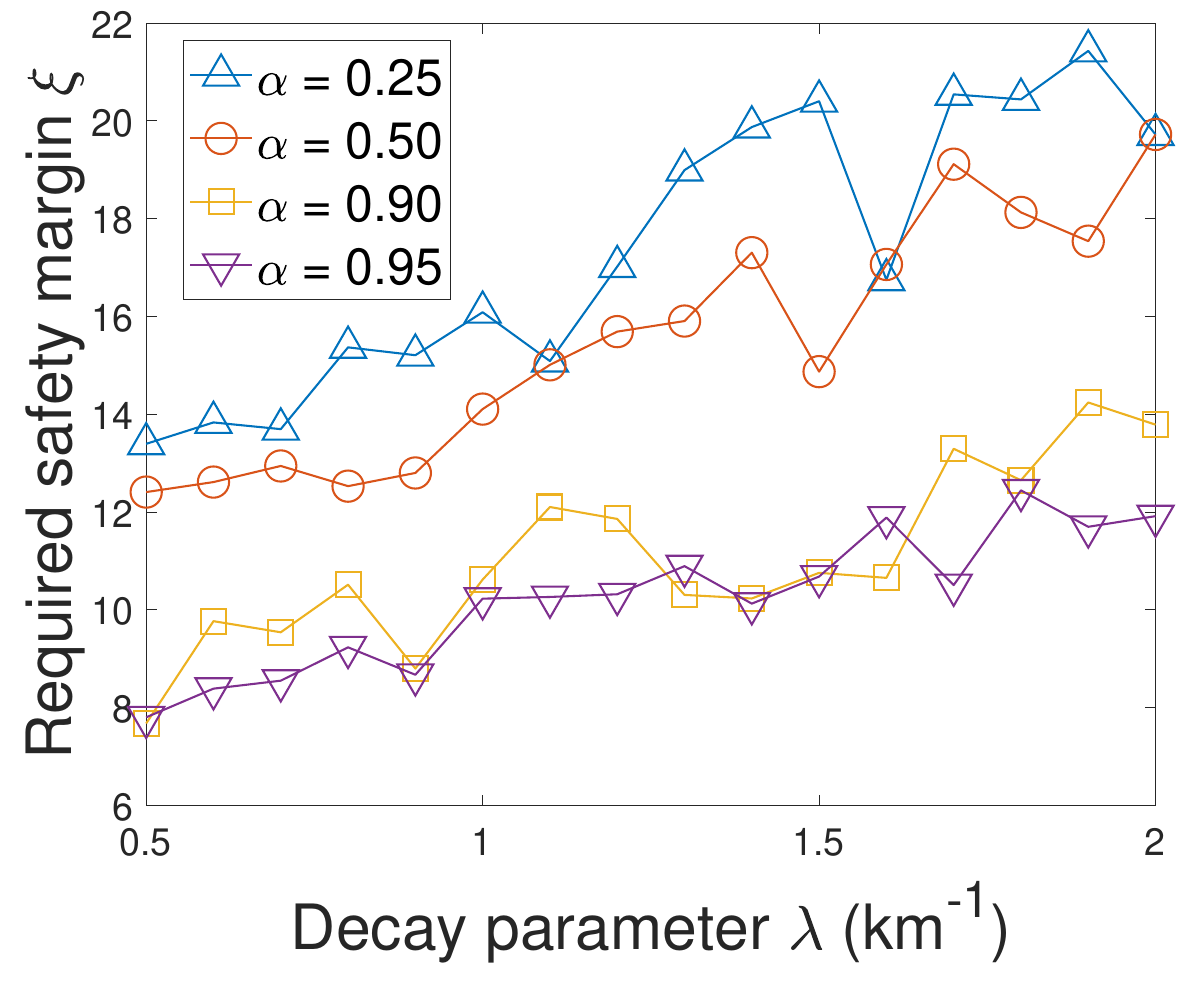}}
  \subfigure[Logistic]{
\includegraphics[width=0.24\textwidth, height = 0.15\textheight]{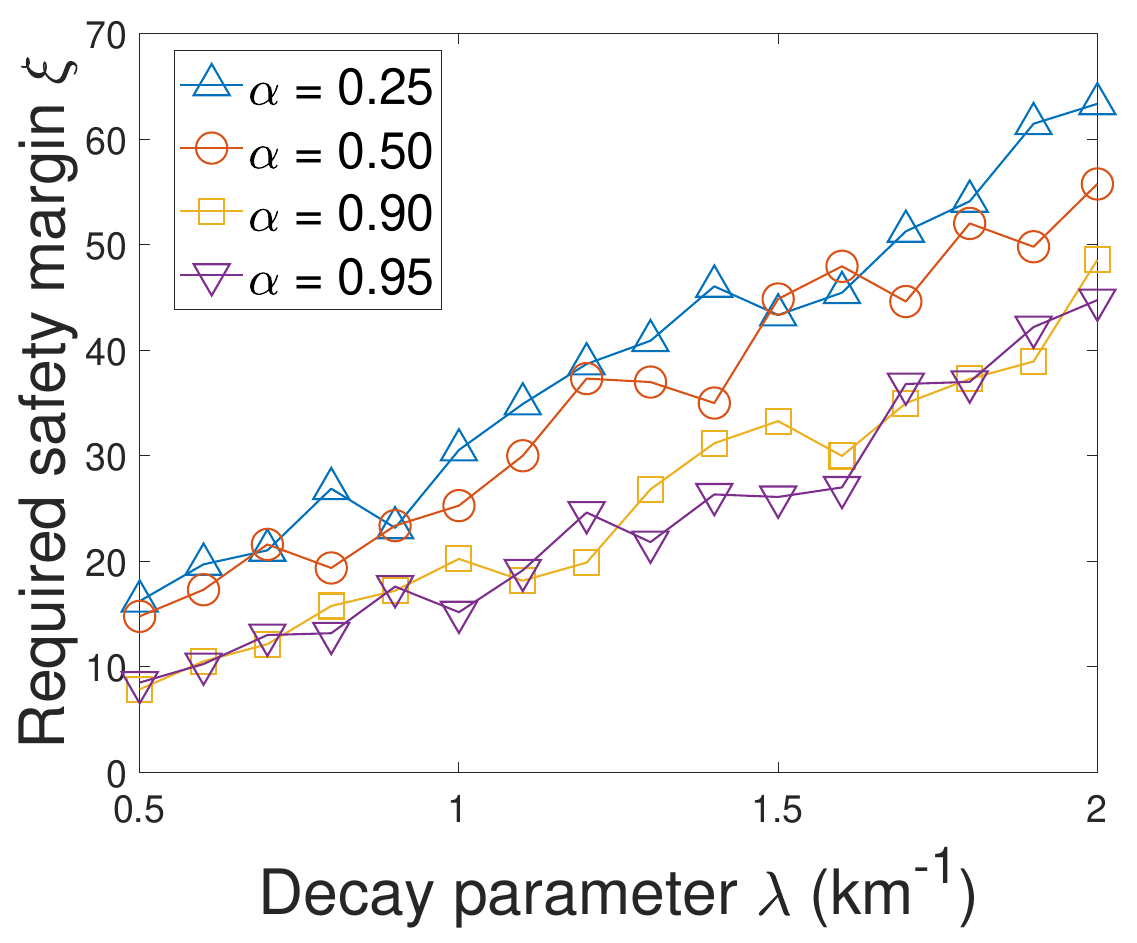}}
\vspace{-0.15in}
\end{minipage}
\caption{Expected safety margin $\xi$ with different decay parameter $\lambda$ and scaling parameter $\alpha$ (London).}
\label{fig:xi_lambda_london}
\vspace{-0.00in}
\end{figure*}

\noindent\textbf{Figures~\ref{fig:xi_lambda_nyc}(a)(b)(c) and~\ref{fig:xi_lambda_london}(a)(b)(c)} illustrate the required safety margin $\xi$ for achieving $(\epsilon, \delta)$-PmDP under different anchor selection strategies across the NYC and London datasets, respectively. Each figure includes three subplots corresponding to exponential decay, power-law decay, and logistic anchor selection methods.

Similar to the results observed on the Rome dataset in Fig.~\ref{fig:xi_lambda_rome}, these figures show how the required safety margin $\xi$ varies as a function of the decay parameter $\lambda$, under different values of the scaling parameter $\alpha \in {0.25, 0.50, 0.90, 0.95}$. The results indicate that as $\lambda$ increases (i.e., the anchor selection probability decays more sharply with distance), the required safety margin $\xi$ consistently decreases. This is because a larger $\lambda$ leads to more localized anchor sets, reducing uncertainty in surrogate selection and thus requiring a smaller $\xi$ to satisfy the PmDP constraint. In contrast, smaller $\lambda$ values produce more dispersed anchor sets, increasing variability in surrogate selection and necessitating a larger safety margin to maintain the privacy guarantee.

Across all selection strategies, the required safety margin $\xi$ also increases as the scaling parameter $\alpha$ decreases. A smaller $\alpha$ leads to fewer anchors being selected, increasing approximation error and requiring a larger $\xi$ to maintain the desired privacy guarantee.

\subsection{Comparison with Lower Bounds}
\label{subsec:bound_add}
\begin{figure*}[t]
\centering
\hspace{0.00in}
\begin{minipage}{1.00\textwidth}
\centering
  \subfigure[\textsc{PAnDA}-e]{
\includegraphics[width=0.24\textwidth, height = 0.15\textheight]{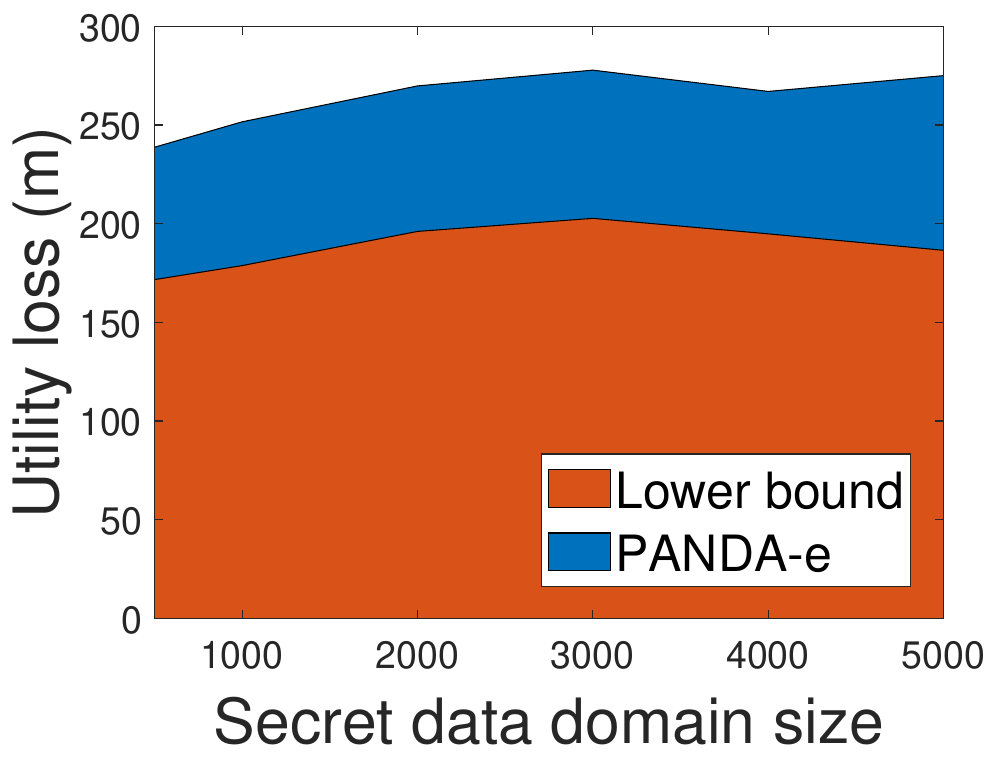}}
  \subfigure[\textsc{PAnDA}-p]{
\includegraphics[width=0.24\textwidth, height = 0.15\textheight]{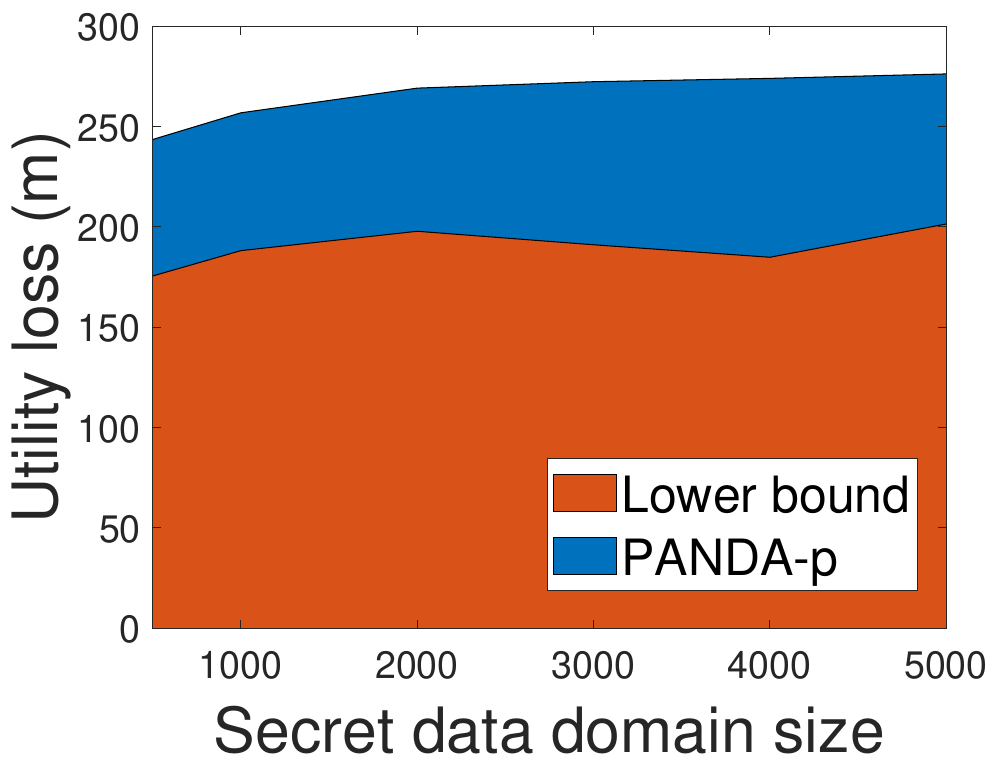}}
  \subfigure[\textsc{PAnDA}-l]{
\includegraphics[width=0.24\textwidth, height = 0.15\textheight]{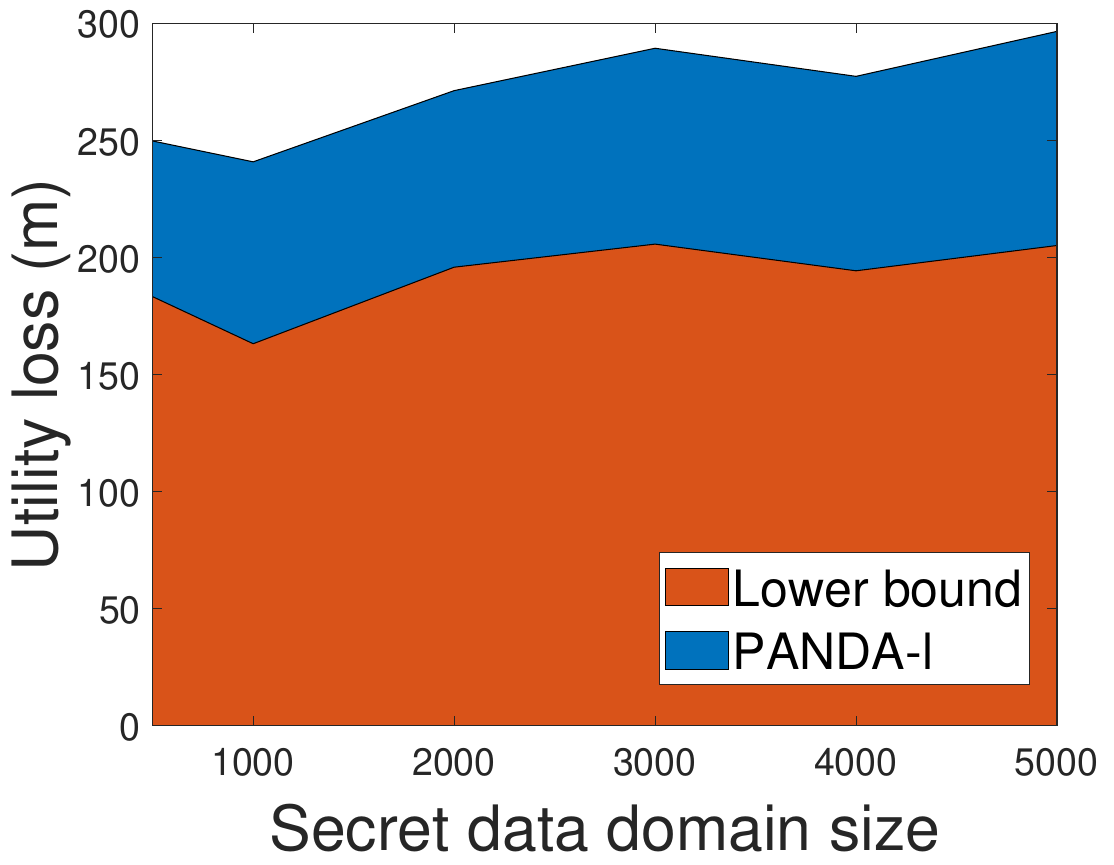}}
\vspace{-0.15in}
\end{minipage}
\caption{\textsc{PAnDA} vs. lower bound (Rome).}
\label{fig:boundPAnDA_rome}
\vspace{-0.00in}
\end{figure*}

\begin{figure*}[t]
\begin{minipage}{1.00\textwidth}
\centering
  \subfigure[\textsc{PAnDA}-e]{
\includegraphics[width=0.24\textwidth, height = 0.15\textheight]{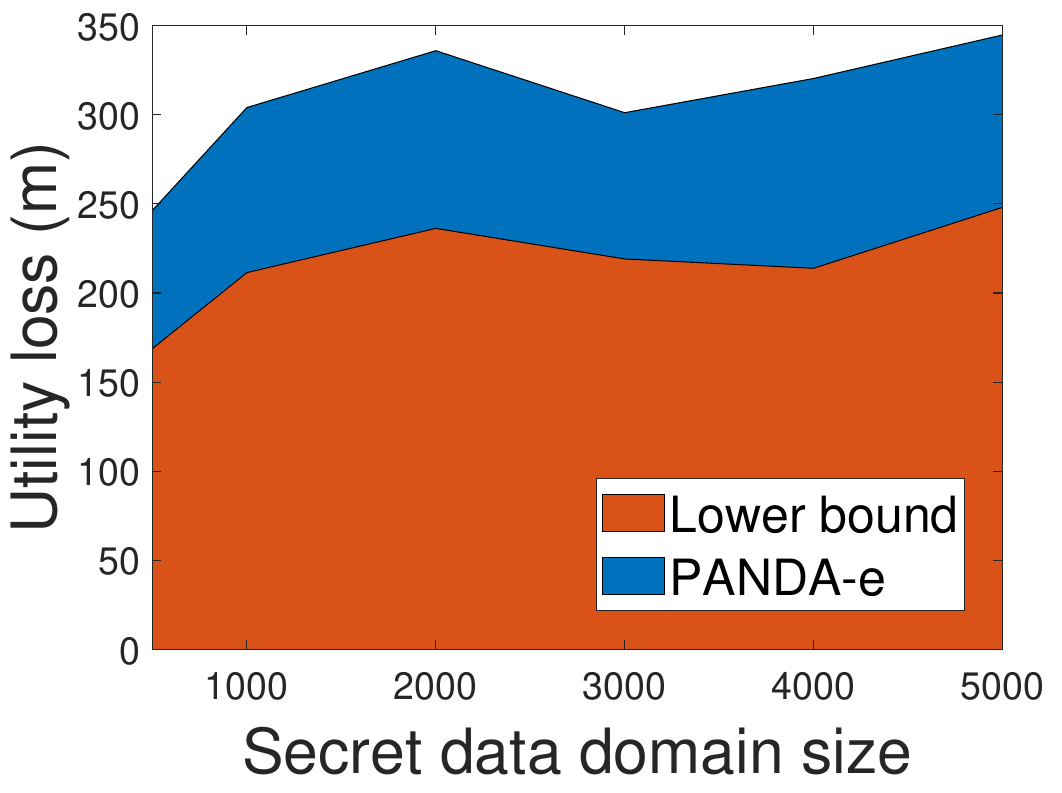}}
  \subfigure[\textsc{PAnDA}-p]{
\includegraphics[width=0.24\textwidth, height = 0.15\textheight]{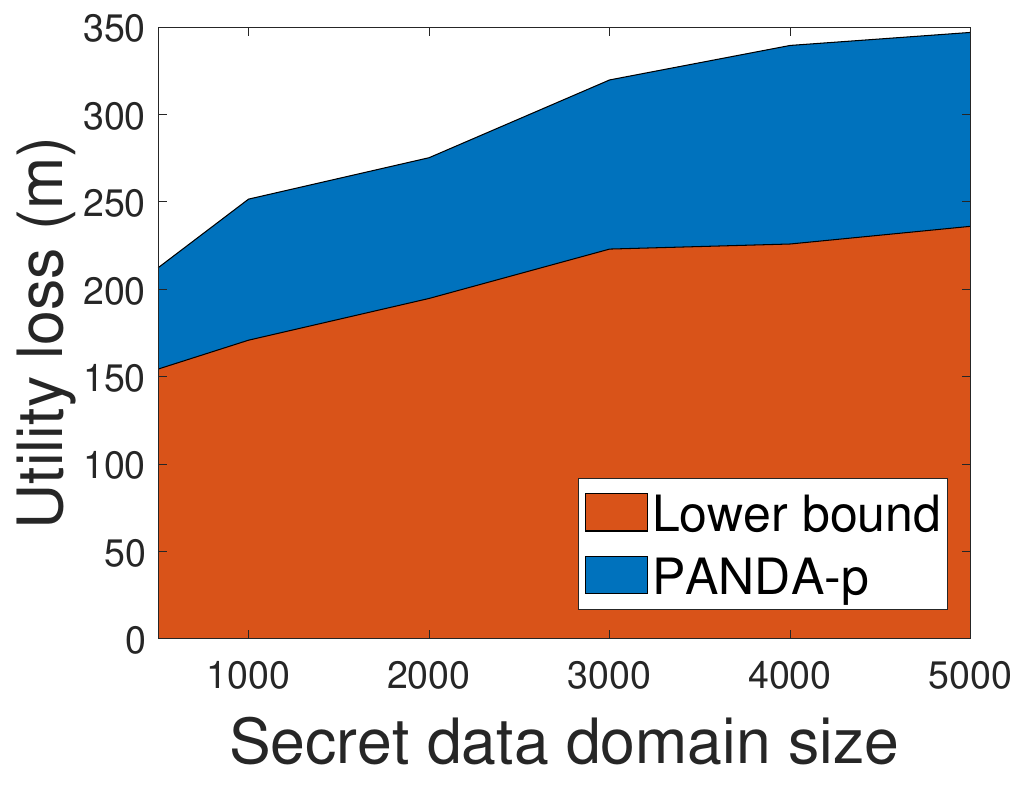}}
  \subfigure[\textsc{PAnDA}-l]{
\includegraphics[width=0.24\textwidth, height = 0.15\textheight]{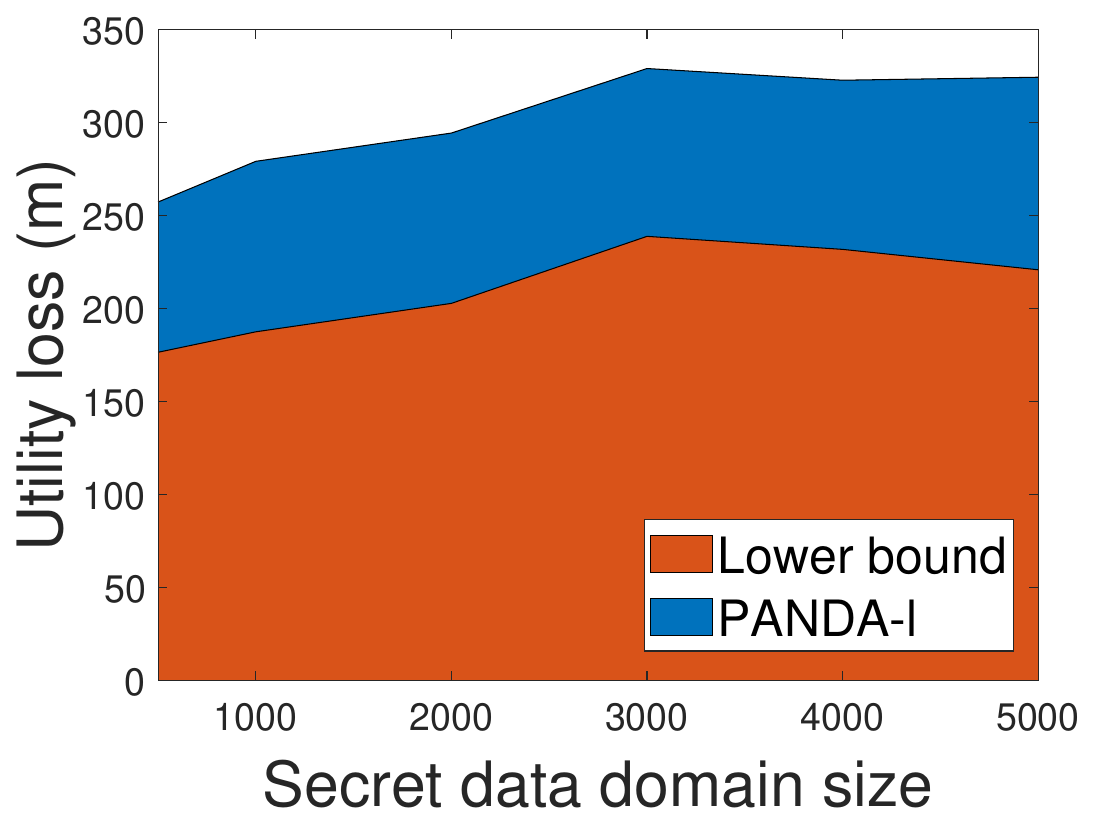}}
\vspace{-0.15in}
\end{minipage}
\caption{\textsc{PAnDA} vs. lower bound (NYC).}
\label{fig:boundPAnDA_nyc}
\vspace{-0.00in}
\end{figure*}

\begin{figure*}[t]
\begin{minipage}{1.00\textwidth}
\centering
  \subfigure[\textsc{PAnDA}-e]{
\includegraphics[width=0.24\textwidth, height = 0.15\textheight]{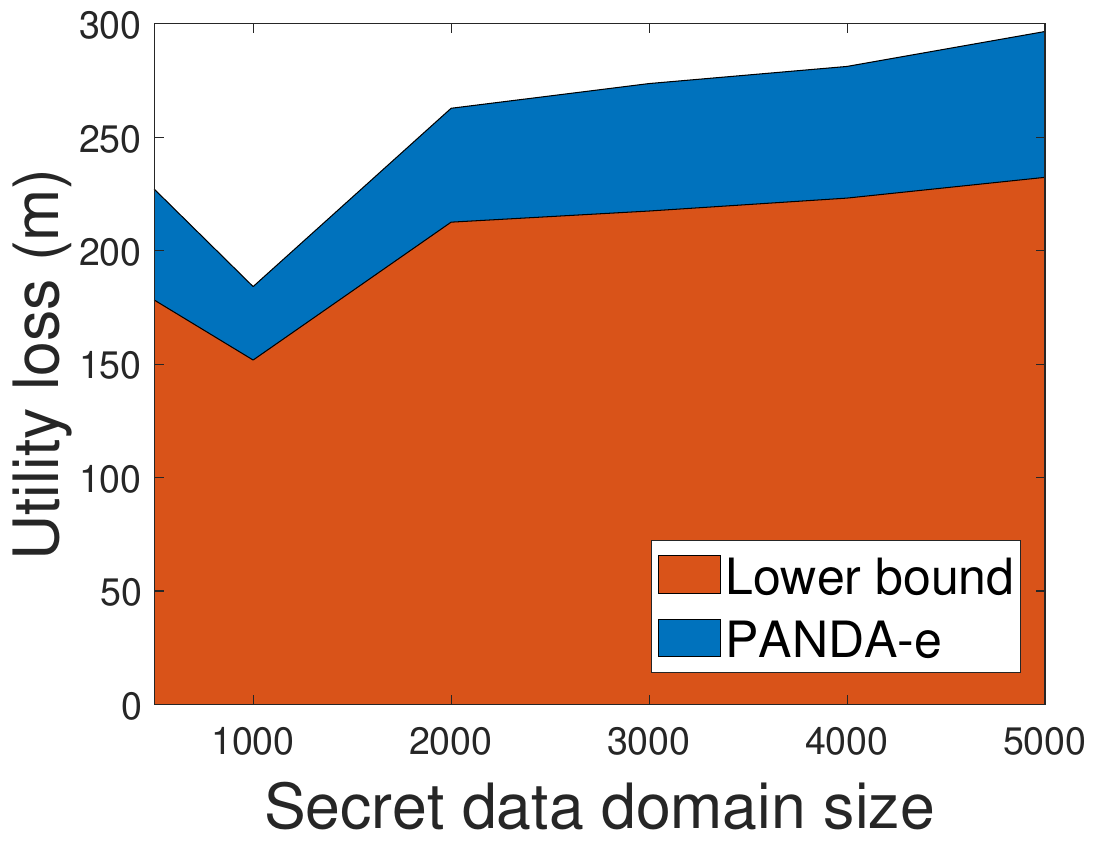}}
  \subfigure[\textsc{PAnDA}-p]{
\includegraphics[width=0.24\textwidth, height = 0.15\textheight]{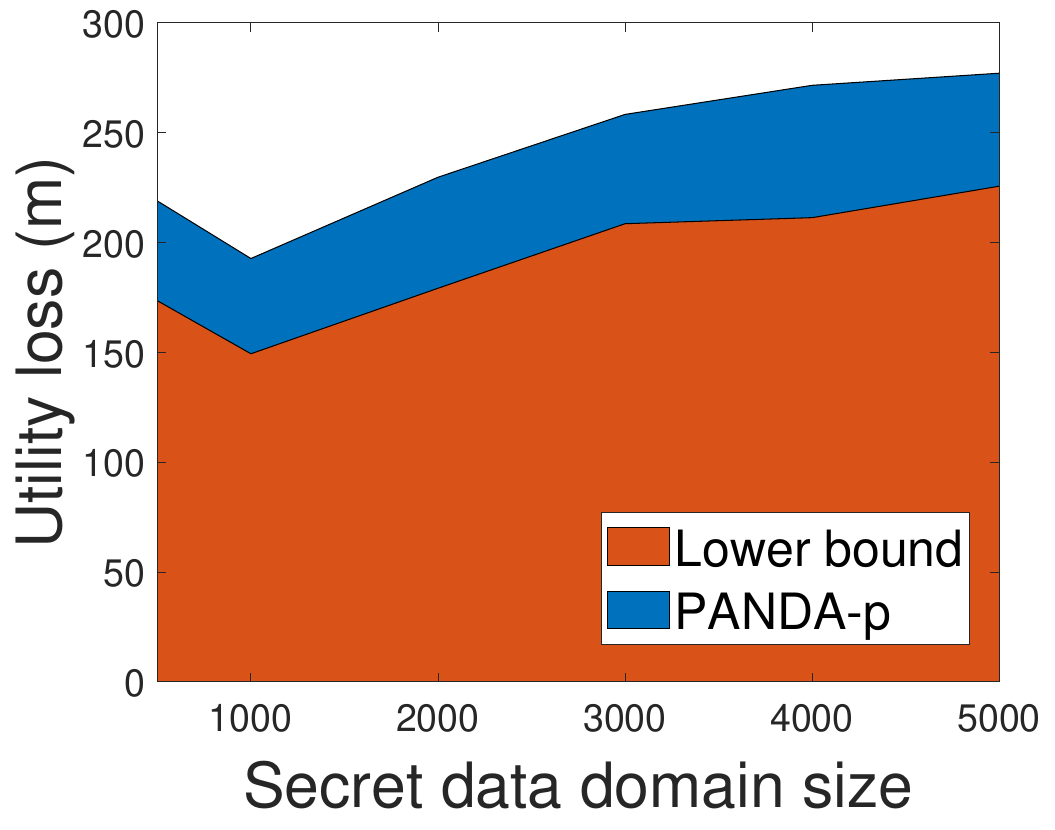}}
  \subfigure[\textsc{PAnDA}-l]{
\includegraphics[width=0.24\textwidth, height = 0.15\textheight]{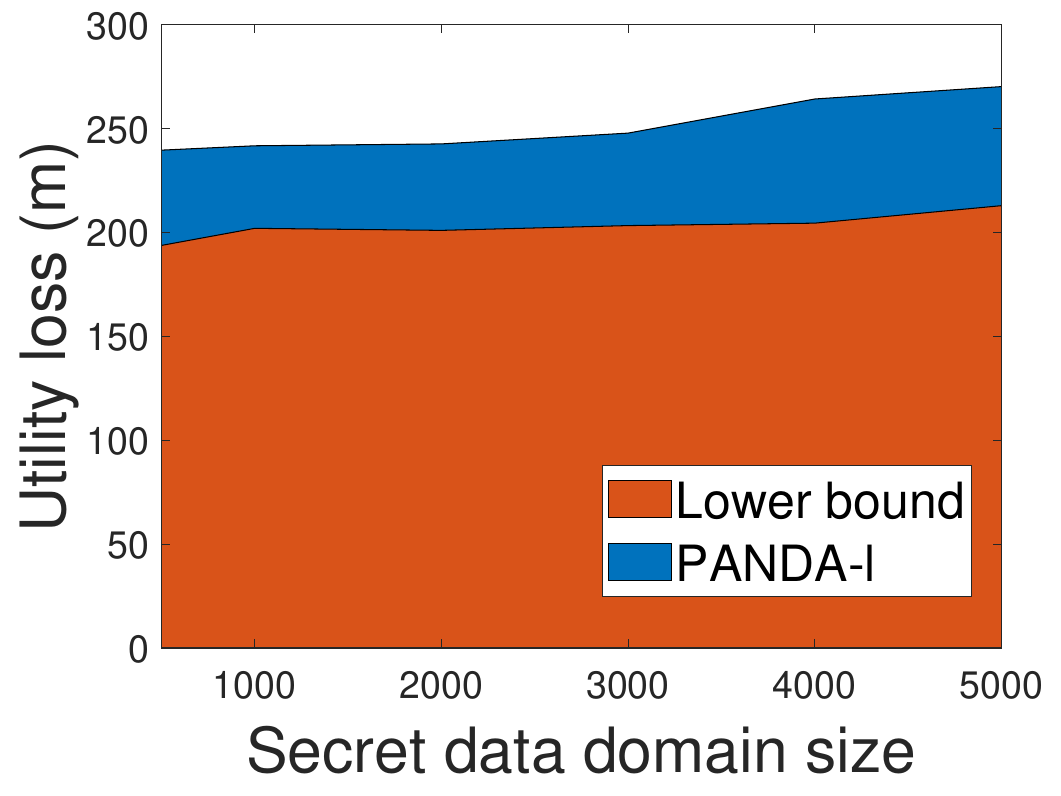}}
\vspace{-0.05in}
\end{minipage}
\caption{\textsc{PAnDA} vs. lower bound (London).}
\label{fig:boundPAnDA_london}
\vspace{-0.00in}
\end{figure*}

\textbf{Figures~\ref{fig:boundPAnDA_rome}(a)(b)(c),~\ref{fig:boundPAnDA_nyc}(a)(b)(c), and~\ref{fig:boundPAnDA_london}(a)(b)(c)} show the utility loss of the three \textsc{PAnDA} variants—\textsc{PAnDA}-e, \textsc{PAnDA}-p, and \textsc{PAnDA}-l—compared to the theoretical lower bound derived from Relaxed \textsc{AnPO} (see \textbf{Proposition~\ref{prop:ULbound}}) on the \textbf{Rome, NYC, and London} datasets, respectively.

Each subplot reports the utility loss as a function of the secret data domain size (from {\bl 500 to 5,000}), to examine how closely each method approaches the optimal solution. Across all datasets, all three \textsc{PAnDA} variants exhibit \emph{bounded approximation ratios}, indicating that the utility loss achieved is close to the theoretical minimum. On average, the approximation ratios for \textsc{PAnDA}-e, \textsc{PAnDA}-p, and \textsc{PAnDA}-l are {\bl 1.3608, 1.3708, and 1.3630}, respectively. These results validate the effectiveness of \textsc{PAnDA} in achieving near-optimal utility.

\section{Artifact Appendix}

%%%%%%%%%%%%%%%%%%%%%%%%%%%%%%%%%%%%%%%%%%%%%%%%%%%%%%%%%%%%%%%%%%%%%

\subsection{Abstract}
This artifact provides the MATLAB (R2024b) implementation of \textsc{PAnDA}, an anchor-based framework for enforcing metric differential privacy, together with baseline scripts, configuration parameters, and small road-network datasets (Rome, NYC, London, plus Rome with real user location distributions). 

The package "\texttt{CCS-2025-\textsc{PAnDA}-main}" majorly contains the following files and folders 
\begin{itemize}
    \item \texttt{main.m}, 
    \item a \texttt{functions/} directory
    \item experiment drivers (\texttt{run\_artifact\_{rome, nyc, london, real\_distribution}.m}, \texttt{run\_LPBD.m}, \texttt{run\_LPEM.m}), and 
    \item a \texttt{dataset/} directory.
\end{itemize} 
When executed, the drivers load the selected dataset, run \textsc{PAnDA} and baselines, and output summary tables/figures and console logs that reproduce the reported trends for utility loss and runtime. The repository’s README documents usage and expected outcomes; as of this release, no explicit open-source license is declared in the repository (authors may add one in the future).
%%%%%%%%%%%%%%%%%%%%%%%%%%%%%%%%%%%%%%%%%%%%%%%%%%%%%%%%%%%%%%%%%%%%%

\subsection{Description \& Requirements}

The artifact is organized as a self-contained MATLAB project. At the repository root, \texttt{README.md} provides top-level instructions and quickstart usage. Core source code lives in \texttt{main.m} (entry point) and the \texttt{functions/} directory. Dataset-specific drivers live alongside the entry point: \texttt{run\_artifact\_rome.m}, \texttt{run\_artifact\_nyc.m}, \\ \texttt{run\_artifact\_london.m}, and  \texttt{run\_artifact\_real\_distribution.m} (for Rome with real distributions).

Configuration is centralized in \texttt{parameters.m} (e.g., number of repeats, seeds, and other run options). Input data and/or download instructions are under \texttt{dataset/}. A paper copy (\texttt{paper.pdf}) is included for reference. To run experiments, open MATLAB in the repo root, adjust \texttt{parameters.m} as needed, then use \texttt{main.m} to select a dataset by uncommenting the relevant line (e.g., \texttt{run\_artifact\_rome.m}; \texttt{run\_artifact\_nyc.m}; \texttt{run\_artifact\_london.m}; or \\ \texttt{run\_artifact\_real\_distribution.m}). This structure lets evaluators quickly locate the README, source code, datasets, and run scripts without navigating extraneous files.

%%%%%

\subsubsection{Security, privacy, and ethical concerns}
There should be no special security, privacy, or ethical concerns. The artifacts are local, compute-only MATLAB scripts that neither require elevated privileges nor access the network, and they perform no destructive operations. They run on public research datasets (Rome Taxi dataset).

%%%%%

\subsubsection{How to access}
% \emph{[Mandatory]}
% \emph{Describe how to access your artifacts.  If your artifacts havereceived the \emph{Artifacts Available} badge, you must provide the DOI(s) for the permanently and publicly available copies of your artifacts.  Most likely, this is the version of your artifacts that you deposited with Zenodo.} \emph{You may describe more than one way to access your artifacts. For example, if the archived versions of your artifacts are available at Zenodo, and you are also making maintained versions of your artifacts available through GitHub, you can describe how to access both versions of your artifacts.}
The artifact is deposited on Zenodo at \\
\texttt{https://zenodo.org/records/17032461}, where the DOI is \\
10.5281/zenodo.17032461.

\emph{Alternative (development repository).} 
A maintained GitHub repository is available at \\ 
\texttt{https://github.com/paopao128/CCS-2025-\textsc{PAnDA}} 
\\ for browsing source code and cloning. 
This artifact corresponds to commit:  \texttt{a3dbf85686153550249e75af486038654c763af3}.

%%%%%

\subsubsection{Hardware dependencies}

% \emph{[Mandatory]}
% \emph{State any specific hardware features that are required to make use of your artifacts: e.g., vendor, CPU/GPU/FPGA, number of processors/cores, microarchitecture, interconnect, memory, hardware counters, etc.  If your artifacts do not have significant hardware dependencies, simply write ``None'' in this section.}
\begin{itemize}
    \item \emph{CPU.} Commodity x86\_64 machine; multi-core processor recommended (4+ cores). 
    \item \emph{Memory.} At least \textbf{64,GB RAM} to run \texttt{run\_artifact\_london.m}; at least \textbf{32,GB RAM} for the other workflows \\ (\texttt{run\_artifact\_rome.m}, \texttt{run\_artifact\_nyc.m}, and \\ \texttt{run\_artifact\_real\_distribution.m}). More memory can reduce runtime by avoiding swapping.
    \item \emph{Accelerators.} Not required (no GPU/FPGA needed).
    \item \emph{Network.} Not required for normal use.
\end{itemize}

%%%%%

\subsubsection{Software dependencies}

% \emph{[Mandatory]}
% \emph{State any specific OS and software packages that are required to make use of your artifacts.  This is particularly important if you share your source code and it must be compiled, or if your artifacts rely on proprietary software that is not included in your artifact packages.  If your artifacts do not have significant software dependencies, simply write ``None'' in this section.}
\begin{itemize}
\item \emph{OS.} Windows 10/11, macOS, or Linux (any recent 64-bit distro).
\item \emph{MATLAB.} MATLAB (R2022b or newer recommended).
\item \emph{Toolboxes.} MATLAB Optimization Toolbox (required for linear programming routines in the baselines) and the Statistics and Machine Learning Toolbox (used for \texttt{randsample} in random sampling). 
\item \emph{Others.} No external solvers, GPUs, or compilers are required for the provided workflows. Network access is not needed for normal use.
\end{itemize}

%%%%%

\subsubsection{Benchmarks}
% \emph{[Mandatory]}
% \emph{Describe any data (e.g., datasets, models, workloads, etc.)\ required by the experiments that are reported in your paper and supported by your artifacts.  If this does not apply to your artifacts, simply write ``None'' in this section.}
\emph{Included datasets.} The artifact bundles all data needed to run the experiments entirely offline. The \texttt{dataset/} directory contains preprocessed road-network graphs and demand distributions for: (i) \textbf{Rome} (uniform vehicle-location distribution), (ii) \textbf{New York City} (uniform), (iii) \textbf{London} (uniform), and (iv) \textbf{Rome–Real} (empirical/real distribution).

\emph{Contents and format.} Each dataset provides a road graph (nodes/edges) within the target region. Files are stored in MATLAB-friendly formats (\texttt{.mat}) and are loaded automatically by the supplied drivers (\texttt{run\_artifact\_{rome, nyc, london, real\_distribution}.m}).

\emph{Compared methods included.} Implementations of the paper’s compared methods are provided alongside \textsc{PAnDA} (e.g., exponential-mechanism and LP-based baselines), exposed through the run scripts and \texttt{functions/}. No external repositories are required.

\emph{No external downloads.} All datasets are included with the package; internet access is not needed to reproduce results.

\emph{Data use and privacy.} The datasets are public/aggregated research data and contain no personally identifiable information.

%%%%%%%%%%%%%%%%%%%%%%%%%%%%%%%%%%%%%%%%%%%%%%%%%%%%%%%%%%%%%%%%%%%%%

\subsection{Set Up}
%%%%%

\subsubsection{Installation}
% \emph{[Mandatory]}
% \emph{Provide instructions for downloading and installing dependencies as well as the main artifacts.  After following these steps, a user of your artifacts should be able to run a simple functionality test.}
\emph{Dependencies.} Install MATLAB (R2022b or newer recommended). Enable the Optimization Toolbox (required for LP-based baselines) and the Statistics and Machine Learning Toolbox (required for random sampling).  

\emph{Download the artifact.}
\begin{compactenum}
\item \textbf{Preferred (archival):} Download the DOI-backed snapshot from Zenodo: \texttt{https://zenodo.org/records/17032461}
. Extract the archive to a local folder (e.g., \texttt{~/\textsc{PAnDA}}).
\item \textbf{Alternative (development mirror):} Clone or download ZIP from GitHub: \\ 
\texttt{https://github.com/paopao128/CCS-2025-\textsc{PAnDA}} \\ 
with the commit: \\
\texttt{a3dbf85686153550249e75af486038654c763af3}
\end{compactenum}

\emph{Set up MATLAB path.}
\begin{compactenum}
\item Start MATLAB and \texttt{cd} to the repository root (where \texttt{main.m} resides).
% \item (Recommended) Add all subfolders to the MATLAB path:
% \begin{verbatim}
% addpath(genpath(pwd)); savepath;
% \end{verbatim}
\item Verify the Optimization Toolbox is available:
\begin{verbatim}
license('test','Optimization_Toolbox')
\end{verbatim}
This should return \texttt{1} (true).
\end{compactenum}

\emph{Configure runs.}
\begin{compactenum}
\item Open \texttt{parameters.m} and (optionally) reduce the number of repeats for a quick smoke test.
\item Open \texttt{main.m} and \emph{uncomment} exactly one dataset selector line, e.g.:
\begin{verbatim}
% run_artifact_rome; 
% run_artifact_nyc; 
% run_artifact_london; 
% run_artifact_real_distribution; 
\end{verbatim}
\item \textbf{Memory note:} Choose \texttt{run\_artifact\_rome.m}, \texttt{nyc}, or \\ \texttt{real\_distribution} if the machine has 
$\geq$32,GB RAM. The London workflow requires 
$\geq$64,GB RAM.
\end{compactenum}

%%%%%

\subsubsection{Basic test}
% \emph{[Mandatory]}
% \emph{Provide instructions for running a simple functionality test. This test does not need to exercise all the features of your artifacts, but ideally, it should allow a user to check that all of the required software components are correctly installed and functioning as intended.  Please include the expected successful output and any required input parameters.}
\begin{compactenum}
\item With \texttt{run\_artifact\_rome.m} selected and a small repeat count in \texttt{parameters.m}, run:
\begin{verbatim}
main
\end{verbatim}
\item Expected outcome: MATLAB prints progress messages and a short summary (utility/runtime) without errors. Output figures/tables and any \texttt{.mat} results (if enabled by the scripts) are written to the working directory or a subfolder specified in \texttt{parameters.m}.
\end{compactenum}

%%%%%%%%%%%%%%%%%%%%%%%%%%%%%%%%%%%%%%%%%%%%%%%%%%%%%%%%%%%%%%%%%%%%%

\subsection{Evaluation Workflow}
%%%%%
\subsubsection{Major claims}
\begin{compactitem}
\item[(C1):] \textsc{PAnDA}-e, \textsc{PAnDA}-p, and \textsc{PAnDA}-l have higher computational time compared to Exponential Mechanism (EM) and Bayesian Remapping (EM+BR), it outperforms optimization-based methods including Linear Programming (LP), Coarse Approximation of LP (LP+CA), Benders Decomposition (LP+BD), and ConstOPTMech (LP+EM) in terms of computation efficiency (described in Section 4.3.1).
\item[(C2):]
\textsc{PAnDA}-e, \textsc{PAnDA}-p, and \textsc{PAnDA}-l achieve lower utility loss compared to EM, LP+CA, and EM+BR (described in the first paragraph of Section 4.3.2).
\end{compactitem}

%%%%%

\subsubsection{Experiments}
\emph{Assumed environment:} MATLAB R2022b+ with Optimization Toolbox; x86\_64 CPU (4–16 cores); 
$\geq$32\,GB RAM for \texttt{rome}, \texttt{nyc}, and \texttt{real\_distribution}; 
$\geq$64\,GB RAM for \texttt{london}. No GPU required.

\begin{compactitem}
\item[(E1):] \textbf{Computation efficiency vs. baselines} (\emph{validates C1}) 
[for each repeat, approximately 1 compute-hour per dataset]
  \begin{asparadesc}
    \item[Preparation:] Set the number of repeats to 1–3 in \texttt{main.m}. 
    When running \texttt{run\_artifact\_rome.m}, \texttt{run\_artifact\_nyc.m}, 
    \\ \texttt{run\_artifact\_london.m}, or \\ \texttt{run\_artifact\_real\_distribution.m}, 
    the methods \textsc{PAnDA}-e, \textsc{PAnDA}-p, \textsc{PAnDA}-l, EM, EM+BR, LP+CA, LP are executed automatically. 
    
    Considering that LP+BD and LP+EM incur much higher computation time and fail to return results when 
    the number of records is $\geq 1{,}000$, they must be run separately using 
    \texttt{run\_LPEM} and \texttt{run\_LPBD}.
    
    \item[Execution:] Uncomment one dataset driver in \texttt{main.m} (e.g., \texttt{run\_artifact\_rome.m}) 
    and run \texttt{main}. Repeat for \texttt{nyc}, \\ \texttt{real\_distribution}, and 
    \texttt{london}.
    
    \item[Results:] Collect wall-clock times per method. 
    Expected outcome: the \textsc{PAnDA} variants are slower than EM/EM+BR, but significantly  faster than optimization-based methods (LP, LP+CA, LP+BD, LP+EM), demonstrating their favorable efficiency position.

    Exact runtimes are difficult to reproduce since they depend on hardware, system load, scheduling, library implementations, and algorithmic randomness. Thus, while relative trends (e.g., scalability across datasets and methods) are reproducible, absolute values may vary across environments.
  \end{asparadesc}

\item[(E2):] \textbf{Utility loss comparison vs. baselines} (\emph{validates C2}) 
[for each repeat, approximately 1 compute-hour per dataset]
  \begin{asparadesc}
  \item[Preparation:] Set the number of repeats to 1–3 in \texttt{main.m}. 
    When running \texttt{run\_artifact\_rome.m}, \texttt{run\_artifact\_nyc.m}, 
    \\ \texttt{run\_artifact\_london.m}, or \\\texttt{run\_artifact\_real\_distribution.m}, 
    the methods \textsc{PAnDA}-e, \textsc{PAnDA}-p, \textsc{PAnDA}-l, EM, EM+BR, LP+CA, LP are still executed automatically. 
    Because LP+BD and LP+EM incur much higher computation time and fail to return results when 
    the number of records is $\geq 1{,}000$, they must be run separately using 
    \texttt{run\_LPEM} and \texttt{run\_LPBD}.
  \item[Execution:] Uncomment one dataset driver in \texttt{main.m} (e.g., \texttt{run\_artifact\_rome.m}) 
    and run \texttt{main}. Repeat for \texttt{nyc}, \\ \texttt{real\_distribution}, and 
    \texttt{london}.
  \item[Results:] Tables and figures report expected utility loss per method. 
  Expected outcome: \textsc{PAnDA}-e/p/l achieve lower utility loss than EM, LP+CA, and EM+BR across datasets.
  
  The exact utility loss values are hard to reproduce, since both location set and users are randomly distributed, leading to variation across runs. However, the overall trend remains consistent: PAnDA-e, PAnDA-p, and PAnDA-l achieve lower utility loss compared to EM, LP+CA, and EM+BR, as reported in the paper.
  \end{asparadesc}
\end{compactitem}

%%%%%%%%%%%%%%%%%%%%%%%%%%%%%%%%%%%%%%%%%%%%%%%%%%%%%%%%%%%%%%%%%%%%%

\subsection{Notes on Reusability}

The artifact can be easily reused and extended. New datasets can be added by formatting 
them like the existing files in \texttt{dataset/} and creating a corresponding run script. 
Key parameters (e.g., $\epsilon$, repeats, grid size, anchor density) are centralized in \texttt{main.m} and 
\texttt{parameters.m}, allowing experiments to be scaled up or down depending on hardware. 
Researchers can also add new perturbation or optimization methods under \texttt{functions/} 
to benchmark against PAnDA and the provided baselines. This modular design makes the 
artifact suitable for reproducing results, adapting to new inputs, and educational use.

%%%%%%%%%%%%%%%%%%%%%%%%%%%%%%%%%%%%%%%%%%%%%%%%%%%%%%%%%%%%%%%%%%%%%

\subsection{Version}
%%%%%%%%%%%%%%%%%%%%
% Obligatory.
% Do not change/remove.
%%%%%%%%%%%%%%%%%%%%
Based on the LaTeX template for Artifact Evaluation V20220926.

\end{document}